{}
{}

\documentclass[11pt,a4paper]{article}

\usepackage{graphicx} 
\usepackage{float}
\usepackage{color}
\usepackage[toc,page]{appendix} 
\usepackage{enumerate}
\usepackage{latexsym,amsthm,amssymb,amsfonts,amsbsy,amsmath}
\usepackage{amsmath}
\usepackage{dsfont}
\usepackage{amsfonts}
\usepackage{amssymb}
\usepackage{mathrsfs}
\usepackage{bbm}
\usepackage{mathtools}

\usepackage{amsthm}
\usepackage{bm} 
\usepackage{bbm} 
\usepackage{slashed} 
\usepackage{cancel}
\usepackage{xcolor}
\usepackage{tikz}
\usetikzlibrary{arrows,matrix,snakes}
\usetikzlibrary{decorations.markings,shapes.geometric,shapes.misc} 
\usepackage{tikz-cd}

\usepackage{pdflscape}
\usepackage{scalefnt}

\definecolor{lightred}{RGB}{255,127,127}
\definecolor{lightgreen}{RGB}{127,255,127}
\definecolor{lightblue}{RGB}{127,127,255}
\definecolor{linkcolor}{rgb}{0,0,0.6}
\usepackage[ pdftex,colorlinks=true,
pdfstartview=FitV,
linkcolor= linkcolor,
citecolor= linkcolor,
urlcolor= linkcolor,
hyperindex=true,
hyperfigures=false]
{hyperref}

\numberwithin{equation}{section}

\usepackage{tabularx} 

\usepackage{fancyhdr} 

\theoremstyle{plain}
\newtheorem{theorem}{Theorem}[section]
\newtheorem{lemma}[theorem]{Lemma}
\newtheorem{proposition}[theorem]{Proposition}
\newtheorem{corollary}[theorem]{Corollary}

\newcommand{\vp}{\varphi}
\newcommand{\noi}{\noindent}
\newcommand{\p}{\partial}
\newcommand{\g}{\mathfrak{g}}
\newcommand{\ti}[1]{_{\bm{\underline{#1}}}}
\newcommand{\fs}[2]{f^{#1}_{{\color{white} #1}\, #2}\;}
\newcommand{\ft}[2]{f^{#1}_{{\color{white} #1}\, #2}}

\newcommand{\dd}{\text{d}}
\newcommand{\lt}{\bm{\ell}}
\newcommand{\Tc}{\mathcal{T}}
\newcommand{\ls}[2]{\ell^{\hspace{1pt}#1}_{#2}}
\newcommand{\lst}[2]{\widetilde{\ell}^{\hspace{2pt}#1}_{#2}}
\newcommand{\ltb}{\bm{\widetilde\ell}}

\newcommand{\kb}[2]{\eta^{\hspace{1pt}#1}_{#2}}

\newcommand{\C}{\mathbb{C}}
\newcommand{\CP}{\mathbb{P}^1}
\newcommand{\R}{\mathbb{R}}
\newcommand{\D}{\mathbb{D}}
\newcommand{\Z}{\mathbb{Z}}
\newcommand{\Ac}{\mathcal{A}}
\newcommand{\rd}{\mathtt{r}}
\newcommand{\cd}{\mathtt{c}}
\newcommand{\Pc}{\mathcal{P}}
\newcommand{\Id}{\text{Id}}

\newcommand{\Sg}{\Gamma}
\newcommand{\Si}{\Sigma}

\newcommand{\Jt}[2]{J^{#1}_{[#2]}}
\newcommand{\J}[2]{\mathcal{J}^{#1}_{[#2]}}
\newcommand{\s}{\sigma}
\newcommand{\po}{z}
\newcommand{\pb}{\bm{z}}
\newcommand{\Rc}{\mathcal{R}}
\newcommand{\Lc}{\mathcal{L}}
\newcommand{\Dc}[2]{\mathcal{D}^{#1}_{[#2]}}
\newcommand{\Hc}{\mathcal{H}}
\newcommand{\Mc}{\mathcal{M}}
\newcommand{\Mcs}[3]{\mathcal{M}(#1\hspace{1.4pt};\hspace{0.4pt}#2,#3)}

\newcommand{\Kc}{\mathcal{K}}
\newcommand{\ze}{\zeta}
\newcommand{\Q}{\mathcal{Q}}
\newcommand{\eb}{\bm{\epsilon}}
\newcommand{\zt}{\widetilde{z}}
\newcommand{\wt}{\widetilde{w}}
\newcommand{\vpt}{\widetilde{\vp}}

\newcommand{\zet}{%
  \mspace{2mu}%
  \widetilde{\mspace{-2mu}\rule{0pt}{1.4ex}\smash[t]{\zeta}}%
}
\newcommand{\pit}{\widetilde{\pi}}
\newcommand{\Jtt}[2]{\widetilde{\mathcal{J}}^{\hspace{1pt}#1}_{[#2]}}

\newcommand{\W}[1]{I_{\text{W}\hspace{-1pt}\text{Z}}\left[#1\right]}
\newcommand{\Ww}[1]{I_{\text{W}\hspace{-1pt}\text{Z}}\bigl[#1\bigr]}
\newcommand{\Te}[2]{T^{#1}_{{\color{white}#1} #2}}

\newcommand{\lc}[1]{\ell_{#1}}
\newcommand{\kc}[1]{\kay_{#1}}
\newcommand{\gb}[1]{g^{(#1)}}
\newcommand{\Xb}[1]{X^{(#1)}}
\newcommand{\jb}[1]{j^{(#1)}}
\newcommand{\Jb}[1]{J^{(#1)}}
\newcommand{\Wb}[1]{W^{(#1)}}
\newcommand{\Rb}[1]{R^{(#1)}}

\newcommand{\Ad}{\text{Ad}}

\newcommand{\hYB}{\text{hYB}}

\newcommand{\vppm}[1]{\varphi_{\pm,#1}}

\newcommand{\vpp}[1]{\varphi_{+,#1}}
\newcommand{\vpm}[1]{\varphi_{-,#1}}
\newcommand{\zi}{\widetilde{z}_{(\infty)}}
\newcommand{\Cc}{\mathcal{C}}
\newcommand{\Sc}[3]{\mathcal{S}^{#1#2}_{{\color{white}#1#2}#3}}
\newcommand{\Tcc}[2]{\mathcal{T}^{#1}_{{\color{white}#1}#2}}
\newcommand{\Uc}[2]{\mathcal{U}^{#1}_{{\color{white}#1}#2}}
\newcommand{\Oo}{\mathcal{O}}
\newcommand{\Oot}{\widetilde{\mathcal{O}}}
\newcommand{\nm}{\cite{Delduc:2019bcl}}
\newcommand{\bg}{\cite{Vicedo:2017cge}}
\newcommand{\bnm}{\cite{Vicedo:2017cge,Delduc:2019bcl}}
\newcommand{\nms}{\cite{Delduc:2019bcl} }
\newcommand{\bgs}{\cite{Vicedo:2017cge} }
\newcommand{\bnms}{\cite{Vicedo:2017cge,Delduc:2019bcl} }
\newcommand{\prl}{\cite{Delduc:2018hty}}
\newcommand{\Sit}{\widetilde{\Sigma}}
\newcommand{\Act}{\widetilde{\mathcal{A}}}
\newcommand{\is}{{(\infty)}}
\newcommand{\Cct}{\widetilde{\mathcal{C}}}
\newcommand{\Gt}{\widetilde{\Gamma}}
\newcommand{\gf}{\equiv}
\newcommand{\Lct}{\widetilde{\mathcal{L}}}
\newcommand{\Rct}{\widetilde{\mathcal{R}}}
\newcommand{\vptpm}[1]{\widetilde{\varphi}_{\pm,#1}}

\newcommand{\vptp}[1]{\widetilde{\varphi}_{+,#1}}
\newcommand{\vptm}[1]{\widetilde{\varphi}_{-,#1}}

\newcommand{\Oc}[2]{\mathcal{O}_{#1 #2}}

\newcommand{\vt}{\vartheta}
\newcommand{\B}{\mathcal{B}}
\newcommand{\Pexp}{\text{P}\overleftarrow{\text{exp}}}
\newcommand{\Kt}{\widetilde{K}}
\newcommand{\cc}{\mathbbm{c}}
\newcommand{\pms}{(\pm)}
\newcommand{\pcm}{{\text{PCM}}}

\newcommand{\elt}[1]{\widetilde{\ell}_{#1}}
\newcommand{\rhot}{\widetilde{\rho}}
\newcommand{\vpi}{\varpi}

\def\beqz{\begin{equation*}}
\def\eeqz{\end{equation*}}

\DeclareFontEncoding{LS1}{}{}
\DeclareFontSubstitution{LS1}{stix}{m}{n}
\DeclareSymbolFont{stixsymbols}{LS1}{stixscr}{m}{n}
\SetSymbolFont{stixsymbols}{bold}{LS1}{stixscr}{b}{n}
\DeclareMathSymbol{\kay}{\mathalpha}{stixsymbols}{"6B}

\def\res{\mathop{\text{res}\,}}

\definecolor{myGreen}{rgb}{0.0,0.4,0.0}

\makeatletter
\let\@keywords\@empty
\let\@subject\@empty
\providecommand{\keywords}[1]{\gdef\@keywords{#1}}
\providecommand{\subject}[1]{\gdef\@subject{#1}}
\def\thetitle{\@title}
\def\theauthor{\@author}
\def\thesubject{\@subject}
\def\thedate{\@date}
\def\thekeywords{\@keywords}
\makeatother
\AtBeginDocument{
\hypersetup{pdftitle={\thetitle}}
\hypersetup{pdfauthor={\theauthor}}
\hypersetup{pdfsubject={\thesubject}}
\hypersetup{pdfkeywords={\thekeywords}}}

\title{Constrained affine Gaudin models and \\ diagonal Yang-Baxter deformations}

\author{S. Lacroix}

\begin{document}

\begin{flushright}
[ZMP-HH/19-12]
\end{flushright}

\begin{center}

\vspace*{2cm}

\begingroup\Large\bfseries\thetitle\par\endgroup

\vspace{1.5cm}

\begingroup
S. Lacroix \footnote{E-mail:~sylvain.lacroix@desy.de}
\endgroup

\vspace{1cm}

\begingroup
\it II. Institut f\"ur Theoretische Physik, Universit\"at Hamburg,
\\
Luruper Chaussee 149, 22761 Hamburg, Germany
\\
Zentrum f\"ur Mathematische Physik, Universit\"at Hamburg, \\
Bundesstrasse 55, 20146 Hamburg, Germany

\endgroup

\end{center}

\vspace{2cm}

\begin{abstract}
We review and pursue further the study of constrained realisations of affine Gaudin models, which form a large class of two-dimensional integrable field theories with gauge symmetries. In particular, we develop a systematic gauging procedure which allows to reformulate the non-constrained realisations of affine Gaudin models considered recently in~\nms as equivalent models with a gauge symmetry. This reformulation is then used to construct integrable deformations of these models breaking their diagonal symmetry. In a second time, we apply these general methods to the integrable coupled $\s$-model introduced recently, whose target space is the $N$-fold Cartesian product $G_0^N$ of a real semi-simple Lie group $G_0$. We present its gauged formulation as a model on $G_0^{N+1}$ with a gauge symmetry acting as the right multiplication by the diagonal subgroup $G_0^{\text{diag}}$ and construct its diagonal homogeneous Yang-Baxter deformation.
\end{abstract}

\newpage

\setcounter{tocdepth}{2}
\tableofcontents

\newpage

\section{Introduction}

\paragraph{Motivations.} Exploring the landscape of integrable two-dimensional field theories is a difficult task. Indeed, the conditions for a field theory to be integrable are quite constraining and there are no systematic procedures to determine whether a given model is integrable or not. This motivates the search for new ways of generating integrable field theories. Recent progresses in this direction have been made using the language of affine Gaudin models. These are two-dimensional field theories associated with affine Kac-Moody algebras, which are constructed in a way which automatically ensures their integrability~\bg. They form a large class of integrable systems, which in particular contains all known integrable $\s$-models, as well as affine Toda field theories~\bg.

Prototypical examples of integrable $\s$-models which can be realised as affine Gaudin models are the Principal Chiral Model (PCM) on a semi-simple real Lie group $G_0$, as well as the model obtained by adding to its action a Wess-Zumino term. Based on this interpretation, the formalism of affine Gaudin models has been used recently in~\cite{Delduc:2018hty,Delduc:2019bcl} to generate an infinite family of new integrable $\s$-models, obtained by coupling together $N$ copies of the PCM with Wess-Zumino term on $G_0$. Let us denote by $\gb r(x,t)$ the $G_0$-valued field attached to the $r^{\rm{th}}$-copy. The action of the model is then expressed in terms of the fields $\gb 1,\cdots,\gb N$. In particular, it is invariant under the global left symmetries $\gb 1 \mapsto h_1^{-1} \gb 1, \, \cdots, \gb N \mapsto h_N^{-1} \gb N$, with constant parameters $h_1,\cdots,h_N$ in $G_0$, as well as under the diagonal right symmetry $\gb 1 \mapsto \gb 1 h, \cdots, \gb N \mapsto \gb N h$, $h\in G_0$, which acts simultaneously on all the copies.

It is well known that the model with only one copy admits a continuous deformation which preserves its integrability: the Yang-Baxter model (see~\cite{Klimcik:2002zj,Klimcik:2008eq,Delduc:2013fga} for the PCM alone,~\cite{Delduc:2014uaa} for the PCM with a Wess-Zumino term and~\cite{Kawaguchi:2014qwa} for homogeneous Yang-Baxter deformations). This deformation breaks the global left symmetry of the model and similarly, there exists an equivalent deformation which breaks its right symmetry. A systematic procedure to construct integrable Yang-Baxter type deformations of the coupled model with $N$ copies has been proposed in~\nm. This procedure can be applied to any of the $N$ copies of the model and breaks the associated left symmetry acting on the field $\gb r$ attached to this copy. A natural question is then whether there also exists an integrable Yang-Baxter type deformation which breaks the diagonal right symmetry of the model. This is one of the main questions addressed in this article: more precisely, we will develop a systematic procedure allowing to deform any realisation of affine Gaudin model considered in~\nm, in a way which breaks the diagonal symmetry naturally associated with this model\footnote{Note that for simplicity, we will restrict in this article to Yang-Baxter deformations without Wess-Zumino term, although such a deformation should also exists in the most general case.}. This general method will in particular be applied to the integrable coupled $\s$-model considered above.

The construction of this diagonal Yang-Baxter deformation relies very strongly on a reformulation of the undeformed model as a theory with a gauge symmetry. This reformulation is the other main result of this article. Indeed, although the models considered in~\nms are not models with a gauge symmetry, we demonstrate in this article that they all possess an alternative formulation as a gauged model, which can also be seen as a realisation of affine Gaudin model. Let us illustrate this general gauging procedure on the example of the integrable coupled $\s$-model on $G_0^N$ considered above. This model possesses an alternative formulation as a model on $G_0^{N+1}$, described by $N+1$ fields $\gb 1(x,t),\cdots,\gb {N+1}(x,t)$, but which now possesses a gauge symmetry acting as the diagonal right translation
\begin{equation*}
\gb 1(x,t) \longmapsto \gb 1(x,t)h(x,t), \;\;\;\;\;\; \cdots, \;\;\;\;\;\; \gb {N+1}(x,t) \longmapsto \gb {N+1}(x,t)h(x,t),
\end{equation*}
with $h(x,t)$ a local $G_0$-valued parameter. This gauge symmetry can be fixed by imposing $\gb{N+1}=\Id$: the resulting model on $\gb 1,\cdots,\gb N$ then coincides with the initial coupled $\s$-model. This gauged reformulation of the model allows a more uniform treatment of its global symmetries. Indeed, the gauged model possesses $N+1$ global symmetries, acting by left translation $\gb 1 \mapsto h_1^{-1} \gb 1, \, \cdots, \gb {N+1} \mapsto h_{N+1}^{-1} \gb {N+1}$ on the fields $\gb r$. Under the gauge-fixing condition $\gb{N+1}=\Id$, the symmetries acting on the fields $\gb 1$ to $\gb N$ are identified with the $N$ independent left symmetries of the non-gauged model, whereas the symmetry acting on the field $\gb{N+1}$ becomes the diagonal right symmetry. As we shall explain, in the gauged formulation, one can apply a Yang-Baxter deformation to any of the $N+1$ copies of the model: after gauge-fixing, the deformation of the $(N+1)^{\rm{th}}$-copy will then be the diagonal deformation announced above.

In the rest of this introduction, we will explain in more details how the gauging and diagonal deformation procedures are constructed in the language of affine Gaudin models. 

\paragraph{The two formulations of the PCM.} A good starting point to understand and motivate the development of the general gauging procedure is the simplest example of the PCM. Indeed, it is already well known that the PCM possesses two different formulations, either as a $\s$-model with target space the manifold $G_0$, or as a model on the Cartesian product $G_0\times G_0$, together with a $G_0$ gauge symmetry. Let us briefly review these non-gauged and gauged formulations, with a particular emphasis on how the integrable structure of the model manifests itself in these two formulations.\\

Let us start with the standard non-gauged formulation of the PCM. It describes the dynamic of a unique $G_0$-valued field $\gb 1(x,t)$. Let us introduce the light-cone coordinates $x^\pm=(t\pm x)/2$, the corresponding derivatives $\p_\pm=\p_t\pm\p_x$ and the light-cone currents $\jb 1_\pm=g^{(1)\,-1}\p_\pm \gb 1$, valued in the Lie algebra $\g_0$ of $G_0$. The equation of motion satisfied by the field $\gb 1(x,t)$ comes from the extremisation of the action
\begin{equation}\label{Eq:IntroActionNonGauged}
S_{\pcm}\bigl[\gb 1] = K \iint \dd x^+ \, \dd x^-\; \kappa\bigl(\jb 1_+,\jb 1_-\bigr),
\end{equation}
where $K$ is a constant and $\kappa(\cdot,\cdot)$ denotes the opposite of the Killing form of $\g_0$. The first step towards establishing the integrability of this model is to recast this equation of motion as the zero curvature equation
\begin{equation}\label{Eq:IntroZCE}
\p_+ \Lc_-(z) - \p_- \Lc_+(z) + \bigl[ \Lc_+(z),\Lc_-(z) \bigr] = 0
\end{equation}
on the Lax pair $\Lc_\pm(z) = \dfrac{\jb 1_\pm}{1 \mp z}$. As for all realisations of affine Gaudin models, this Lax pair is valued in the complexification $\g$ of the Lie algebra $\g_0$ and depends on an auxiliary complex parameter $z$, called the spectral parameter. The zero curvature equation \eqref{Eq:IntroZCE}, true for all $z\in\C$, ensures the existence of an infinite number of conserved charges for the model, extracted from the monodromy matrix of the Lax matrix $\Lc(z)=\frac{1}{2}\bigl( \Lc_+(z)-\Lc_-(z) \bigr)$.

The model is then integrable if these conserved charges are in involution one with another. This is ensured if the Poisson bracket of the Lax matrix $\Lc(z)$ with itself is a non-ultralocal Maillet bracket~\cite{Maillet:1985fn,Maillet:1985ek}. For the PCM, and in fact all realisations of affine Gaudin models, the Lax matrix satisfies a Maillet bracket with a particular form, governed by a rational function of the spectral parameter, called the twist function $\vp(z)$. This function then characterises the integrable structure of the model. For the PCM, it is given by
\begin{equation}\label{Eq:IntroTwistPCM}
\vp_{\pcm}(z) = K\frac{1-z^2}{z^2}.
\end{equation}
To compute this twist function, one first has to perform the canonical analysis of the PCM and obtain its Hamiltonian formulation. Let us say a few words about this formulation. It is expressed in terms of the $G_0$-valued ``coordinate'' field $\gb 1$ and a $\g_0$-valued ``momenta'' field $\Xb 1$, which together describe the canonical fields on the cotangent bundle $T^*G_0$. The algebra of Hamiltonian observables of the PCM is then the Poisson algebra $\Ac_{G_0}$ generated by the fields $\gb 1$ and $\Xb 1$, equipped with the canonical Poisson bracket.\\

Let us now turn to the gauged formulation of the PCM. As mentioned above, it is a field theory on the Cartesian product $G_0\times G_0$ and is thus described by two $G_0$-valued fields $\gb 1(x,t)$ and $\gb 2(x,t)$. Denoting by $\jb 1_\pm$ and $\jb 2_\pm$ the corresponding light-cone currents, the action of the gauged formulation of the PCM is then given by
\begin{equation}\label{Eq:IntroActionGauged}
S^{\text{gauge}}_\pcm \bigl[\gb 1,\gb 2 \bigr] = K \iint \dd x^+\,\dd x^- \; \kappa\bigl( \jb 1_+-\jb 2_+, \jb 1_--\jb 2_- \bigr).
\end{equation}
It can seem surprising at first that the PCM, described above by a single $G_0$-valued field $\gb 1$, also possesses a formulation with two fields $(\gb 1,\gb 2)\in G_0 \times G_0$. This apparent addition of degrees of freedom is in fact counterbalanced by the appearance in the model of a $G_0$ gauge symmetry. Indeed, the action \eqref{Eq:IntroActionGauged} is invariant under the local transformation
\begin{equation}\label{Eq:IntroGaugePCM}
\gb 1(x,t) \longmapsto \gb 1(x,t)h(x,t) \;\;\;\;\;\;\;\;\; \text{ and } \;\;\;\;\;\;\;\;\; \gb 2(x,t) \longmapsto \gb 2(x,t)h(x,t),
\end{equation}
where $h(x,t)$ is an arbitrary function of the space-time coordinates $(x,t)$, valued in the Lie group $G_0$. This gauge symmetry acts on the manifold $G_0\times G_0$, parametrised by the fields $(\gb 1,\gb 2)$, as the right multiplication by the diagonal subgroup $G_0^{\text{diag}}=\bigl\lbrace (h,h), \; h \in G_0 \bigr\rbrace \subset G_0\times G_0$. In particular, this gauge symmetry shows that the fields $\gb 1$ and $\gb 2$ are not independent physical degrees of freedom of the model. To remove this additional non-physical gauge freedom, one can gauge-fix the model, for example by considering  the gauge-fixing condition $\gb 2(x,t)=\Id$ (indeed, it is clear that every orbit of the gauge symmetry \eqref{Eq:IntroGaugePCM} possesses a unique representative satisfying this condition). Under this gauge-fixing, one recovers the original non-gauged action \eqref{Eq:IntroActionNonGauged} of the PCM:
\begin{equation*}
S^{\text{gauge}}_\pcm \bigl[\gb 1,\Id] = S_\pcm \bigl[\gb 1 \bigr].
\end{equation*}

Let us now discuss the integrable structure of the PCM in its gauged formulation. Its equation of motion can also be recast as a zero curvature equation \eqref{Eq:IntroZCE}, on the Lax pair 
\begin{equation*}
\Lct_\pm(\zt)=\frac{\jb 1_\pm+\jb 2_\pm}{2} + \zt^{\pm 1} \frac{ \jb 1_\pm - \jb 2_\pm }{2}.
\end{equation*}
We denoted the spectral parameter of this Lax pair using another notation $\zt$ for reasons to be made clear below. The next step in the study of the integrable structure of the gauged model is to study the Poisson bracket of its Lax matrix $\Lct(\zt)=\frac{1}{2}\bigl( \Lct_+(\zt) - \Lct_-(\zt) \bigr)$. For that, one must first perform the canonical analysis of the model. Because the gauged formulation of the model is on the product $G_0\times G_0$, its algebra of Hamiltonian observables is the tensor product $\Ac_{G_0}^{\otimes 2}$ of two copies of the algebra $\Ac_{G_0}$ of canonical fields on $T^*G_0$, generated by $G_0$-valued fields $\gb 1$ and $\gb 2$ and $\g_0$-valued fields $\Xb 1$ and $\Xb 2$. The invariance of the model under the gauge symmetry \eqref{Eq:IntroGaugePCM} implies the existence of a first-class constraint in the Hamiltonian formulation (see for example~\cite{dirac1964lectures,Henneaux:1992ig} for a review about the Dirac formalism of constrained Hamiltonian systems and its relation with gauge symmetries). In the present case, this constraint simply reads
\begin{equation}\label{Eq:IntroConstraintPCM}
\Xb 1+\Xb 2 \approx 0.
\end{equation}
One can express the Lax matrix $\Lct(\zt)$ of the model in terms of the Hamiltonian fields $\gb 1$, $\gb 2$, $\Xb 1$ and $\Xb 2$, modulo the freedom of adding to it any term proportional to the constraint \eqref{Eq:IntroConstraintPCM}. Choosing this term carefully, one then obtains a Lax matrix $\Lct(\zt)$ satisfying a Maillet bracket\footnote{This bracket is in fact strong, in the sense that it holds even without having to impose the constraint \eqref{Eq:IntroConstraintPCM}.} with twist function $\vpt_\pcm(\zt)$. The latter is given by
\begin{equation}\label{Eq:IntroTwistPCMGauged}
\vpt_\pcm(\zt) = \frac{8K\,\zt}{(\zt^2-1)^2}.
\end{equation}

The integrable structure of the gauged and non-gauged formulation of the PCM are not independent. Under the gauge-fixing condition $\gb 2=\Id$, one verifies easily that the Lax pairs of the two formulations correspond to the same Lax pair $\Lc_\pm(z)=\Lct_\pm(\zt)$, expressed with two different spectral parameters $z$ and $\zt$ related by the transformation
\begin{equation}\label{Eq:IntroChangeSpecPCM}
z \longmapsto \zt=f(z), \;\;\;\;\;\;\; \text{ with } \;\;\;\;\;\;\; f(z) = \frac{1+z}{1-z}.
\end{equation}
Moreover, one checks that the twist functions $\vp(z)$ and $\vpt(\zt)$ of the two formulations in fact correspond to the same 1-form
\begin{equation}\label{Eq:IntroTwist1formPCM}
\vp_\pcm(z)\,\dd z = \vpt_\pcm(\zt) \,\dd \zt  = \vpt_\pcm\bigl(f(z)\bigr) f'(z)\,\dd z,
\end{equation}
expressed in terms of the two different coordinates $z$ and $\zt$.

\paragraph{Gauged and non-gauged realisations of affine Gaudin models.} Let us now come back to realisations of affine Gaudin models and describe how the results presented above for the PCM generalise in this more general framework. We start by briefly describing the structure of these models, based on~\bg\footnote{In the nomenclature of~\bg, we are interested here only in the non-cyclotomic affine Gaudin models.}. One of the fundamental object describing an affine Gaudin model is its twist function $\vp(z)$. In particular, its pole structure describes the so-called sites $\Si$ of the model: each pole $z_\alpha$ of the 1-form $\vp(z)\dd z$ corresponds to a site $\alpha\in\Si$ of the model ($z_\alpha$ is then called the position of the site $\alpha$). For a site at position $z_\alpha\in\C$, let $m_\alpha\in\mathbb{Z}_{\geq 1}$ be the multiplicity of the pole of $\vp(z)$ at $z=z_\alpha$ and $\ls\alpha p$, $p\in\lbrace 0,\cdots,m_\alpha-1 \rbrace$, be the corresponding coefficient of $(z-z_\alpha)^{-p-1}$ in the partial fraction decomposition of $\vp(z)$. The numbers $\ls \alpha p$ are called the levels of the affine Gaudin model. To each site $\alpha$ is attached $m_\alpha$ currents $\J\alpha p(x)$ ($\alpha\in\Si$, $p\in\lbrace 0,\cdots,m_\alpha-1 \rbrace$), called Takiff currents. These are Lie algebra valued fields belonging to the Poisson algebra $\Ac$ describing the Hamiltonian observables of the model and satisfying particular Poisson brackets depending on the levels $\ls \alpha p$.\\

The affine Gaudin models considered in the article~\nms possess a particular structure. Indeed, their twist function $\vp(z)$ is such that the 1-form $\vp(z)\dd z$ possesses a double pole at $z=\infty$, corresponding to a constant term in their twist function. According to the nomenclature introduced above, these models then possess a site of multiplicity two at infinity. As explained in the general construction~\bg, a site at infinity is treated in a slightly different way than the sites at finite positions $z_\alpha\in\C$. In particular, in~\nm, there are no Takiff currents attached to this site and the degrees of freedom of the model then only come from the finite sites.

To illustrate this, let us come back to the example of the PCM, in its non-gauged formulation. Its twist function is given by Equation \eqref{Eq:IntroTwistPCM}: in particular, it possesses two double poles at $z=\infty$ and at $z=0$. The first pole corresponds to a site at infinity, which according to the discussion above, is not associated with any degrees of freedom. The double pole at $z=0$ corresponds to a site $(1)$ with multiplicity $2$, which is then associated with two Takiff currents $\J{(1)}0(x)$ and $\J{(1)}1(x)$, expressed in terms of the canonical fields $\gb 1$ and $\Xb 1$ generating the algebra of observables $\Ac_{G_0}$ of the model. The degrees of freedom of the PCM are then attached to its site $(1)$: following the terminology of~\nm, we say that the site $(1)$ is associated with a PCM realisation in $\Ac_{G_0}$. Similarly, one can describe the PCM with Wess-Zumino term as a realisation of affine Gaudin model with a site at infinity (which does not come with any degrees of freedom attached) and a finite site $(1)$ of multiplicity 2, to which is attached a so-called PCM+WZ realisation, also valued in the algebra $\Ac_{G_0}$ but whose Takiff currents are a deformation of the ones of the PCM realisation. Finally, the integrable coupled $\s$-model introduced in~\cite{Delduc:2018hty,Delduc:2019bcl} is obtained as a realisation of affine Gaudin model with a site at infinity (with no degrees of freedom attached) and $N$ finite sites of multiplicity two, each associated with a PCM+WZ realisation. Thus, to each site of the model is attached a copy of the algebra $\Ac_{G_0}$ and the space of all observables of this model is then the tensor product $\Ac_{G_0}^{\otimes N}$. \\

The models described in~\nms do not possess a gauge symmetry. It was explained in~\bgs how to obtain an affine Gaudin model with gauge symmetry. Let us recall that affine Gaudin models are defined in the Hamiltonian formalism. As mentioned earlier for the PCM, in this formalism, gauge symmetries are associated with the presence of first-class constraints in the model. It was proven in~\bgs that such a constraint can be naturally constructed and consistently imposed in affine Gaudin models whose twist function is such that $\vp(z)\dd z$ is regular at $z=\infty$. In particular, as the models of~\nms possess a site of multiplicity two at infinity, they do not belong to this class. The general construction of these gauged realisations of affine Gaudin models and their properties are discussed in details in the section \ref{Sec:AGM} of this article.

An example of such a gauged realisation of affine Gaudin model is given by the gauged formulation of the PCM described above. Indeed, one checks that its twist function \eqref{Eq:IntroTwistPCMGauged} is such that the 1-form $\vpt_{\text{PCM}}(\zt)\dd \zt$ is regular at $\zt=\infty$ (following the notations used earlier, we denote by $\zt$ the spectral parameter used in the gauged formulation of the PCM). This model does not possess a site at infinity but possesses two finite sites of multiplicity 2 at $\zt=+1$ and $\zt=-1$. Each of these sites is associated with a PCM realisation in $\Ac_{G_0}$, corresponding to the canonical fields $(\gb 1,\Xb 1)$ for the site at $\zt=+1$ and $(\gb 2,\Xb 2)$ for the site at $\zt=-1$. The constraint naturally associated with this affine Gaudin model (following the construction of~\bg) coincides with the one \eqref{Eq:IntroConstraintPCM} described above, which generates the gauge symmetry \eqref{Eq:IntroGaugePCM}.

\paragraph{Gauging procedure.} One of the main result of this article, which is the subject of Section \ref{Sec:Gauging}, states that any non-gauged realisation of affine Gaudin model as considered in~\nms admits an equivalent formulation with a gauge symmetry. Let us sketch how this gauged formulation is constructed. As explained in the previous paragraphs, the model we start with cannot possess a gauge symmetry as it has a site at infinity, corresponding to a constant term in its twist function but which is not associated with any degree of freedom of the model. To obtain a model with gauge symmetry, let us then change what we consider to be the point at infinity by performing a change of spectral parameter
\begin{equation*}
z \longmapsto \zt = f(z), \;\;\;\;\;\;\; \text{ with } \;\;\;\;\;\;\; f(z) = \frac{az+b}{cz+d}.
\end{equation*}
In this equation, $a,b,c$ and $d$ are real numbers such that $ad-bc\neq 0$, making $f$ a M\"obius transformation, \textit{i.e.} a biholomorphism from the Riemann sphere $\CP$ to itself. The gauged model is then constructed as a realisation of affine Gaudin model with twist function $\vpt(\zt)$ defined from the twist function $\vp(z)$ of the non-gauged formulation by imposing
\begin{equation}\label{Eq:IntroTwist1form}
\vp(z)\dd z = \vpt(\zt) \dd \zt.
\end{equation}
The change of spectral parameter considered here is a generalisation of the one \eqref{Eq:IntroChangeSpecPCM} considered earlier in this introduction for the PCM. The definition of the twist function $\vpt(\zt)$ through the relation \eqref{Eq:IntroTwist1form} can be thus motivated from the example of the PCM by the similar relation \eqref{Eq:IntroTwist1formPCM} that we observed in this case.

The sites of the gauged model correspond to the poles of the 1-form $\vpt(\zt)\dd\zt$, which are the images under the M\"obius transformation $f$ of the poles of $\vp(z)\dd z$. An important property of this transformation is that it sends the point $z=\infty$ to a finite point at $\zt=f(\infty)=a/c$ (we suppose here that $c$ is non-zero). Thus, the site at infinity of the model we started with now becomes a finite site of multiplicity two at position $\zi=a/c$, which we shall denote as $\is$ to keep track of its status in the initial model. The finite sites $\alpha\in\Si$ of the initial model, with positions $z_\alpha\in\C$, are sent under the change of spectral parameter to $\zt_\alpha=f(z_\alpha)$. We will suppose that the map $f$ is such that all $\zt_\alpha$ are in the finite complex plane, so that all the sites $\alpha\in\Si$ stay finite after performing the change of spectral parameter $z\mapsto\zt=f(z)$. In particular, this means that all poles of the 1-form $\vp(z)\dd z$ are sent to finite points by $f$: thus, the 1-form $\vpt(\zt)\dd\zt$ is regular at $\zt=\infty$. According to the discussion above, this means that we can impose a first-class constraint to the model with twist function $\vpt(\zt)$ and thus that this model possesses a gauge symmetry.\\

We first have to achieve the construction of this model. In particular, we still have to attach to each of its sites some Takiff currents $\Jtt\alpha p(x)$ (whose Poisson brackets depend on the levels $\lst\alpha p$ extracted from the partial fraction decomposition of $\vpt(\zt)$). Let us first consider a site $\alpha\in\Si$ which had a finite position $z=z_\alpha\in\C$ in the model we started with and is now at position $\zt=\zt_\alpha\in\C$. In the initial non-gauged model, this site is associated with Takiff currents $\J\alpha p(x)$, in the algebra of observables $\Ac$ of this model. As we will show in the main body of this article, one can construct naturally the Takiff currents $\Jtt\alpha p(x)$ attached to the site $\alpha$ in the gauged model as a linear combination of these initial currents $\J\alpha p(x)$. In particular, these sites are still associated with the same degrees of freedom in the algebra $\Ac$.

Let us now turn to the site $\is$. In the initial model, this site was situated at infinity and thus was not associated with any Takiff currents (see discussion above). Thus, to realise the Takiff currents $\Jtt\is 0(x)$ and $\Jtt\is 1(x)$ attached to this site in the new model (recall that this site is of multiplicity two and is thus associated with two Takiff currents), one has to introduce new degrees of freedom in the model. We shall do it by attaching to the site $\is$ a PCM+WZ realisation in the algebra $\Ac_{G_0}$: indeed, as explained above, this is a realisation of multiplicity two. The new degrees of freedom added to the model are thus the canonical fields $g(x)$ and $X(x)$ generating the algebra $\Ac_{G_0}$ and the algebra of observables of the gauged model is then $\Act=\Ac\otimes \Ac_{G_0}$. These additional degrees of freedom are the equivalent of the fields $\gb 2$ and $\Xb 2$ introduced in the gauged formulation of the PCM above.

As in the case of the PCM, this apparent addition of new degrees of freedom is counterbalanced by the appearance of a constraint $\Cct(x)\approx 0$ in the model. Indeed, as we will show in the main text of this article, this constraint can be used to express the additional field $X(x)$ in the component $\Ac_{G_0}$ of $\Act$ in terms of the other degrees of freedom of the model. Moreover, the introduction of the first-class constraint $\Cct(x)\approx 0$ also comes with the appearance of a gauge-symmetry, generated by the constraint. As we shall also show, this gauge symmetry acts on the field $g(x)$ in $\Ac_{G_0}$ by the right-multiplication
\begin{equation*}
g(x) \longmapsto g(x) h(x,t),
\end{equation*}
with $h(x,t)$ a local parameter in $G_0$. Thus, this gauge symmetry can be used to eliminate totally the field $g(x)$, for example by imposing the gauge-fixing condition $g(x)=\Id$. As the constraint and the gauge-fixing allow to eliminate both fields $g(x)$ and $X(x)$, and as these fields generate the component $\Ac_{G_0}$ of the algebra $\Act=\Ac\otimes\Ac_{G_0}$, we see that the physical degrees of freedom of the model (after gauge-fixing) are described by the component $\Ac$ of $\Act$. Let us summarise: the new model that we constructed possesses additional degrees of freedom $g$ and $X$ (attached to the site $\is$) compared to the initial one ; however, these degrees are not physical and can be eliminated by the constraint and the gauge symmetry of the model, hence reducing the space of observables of the model to the initial one $\Ac$.

In the main text of this article, we will finally prove that the gauged model with spectral parameter $\zt$, when considered after the gauge fixing $g=\Id$, is equivalent to the initial model we started with, by proving that these models share the same dynamics. The construction sketched above then provides a systematic gauging procedure of the models considered in~\nm, which generalises the example of the PCM introduced above as a motivation (the fact that the gauging of the PCM described above indeed fits into this more general formalism is proved in Subsection \ref{SubSec:PCM} of this article). In particular, this procedure can now be applied to the integrable coupled $\s$-model introduced in~\prl: as mentioned earlier in this introduction, it then starts from the non-gauged formulation of this model as a $\s$-model on $G_0^N$ (recalled in Subsection \ref{SubSec:NonGaugedSigma}) and realises it as a gauged model on $G_0^{N+1}$, with a gauge symmetry acting as the right multiplication by the subgroup $G_0^{\text{diag}}=\bigl\lbrace (h,\cdots,h),\;h\in G_0 \bigr\rbrace$ of $G_0^{N+1}$, as detailed in Subsection \ref{SubSec:GaugedSigma}.

\paragraph{Diagonal Yang-Baxter deformations.} As already announced at the beginning of this introduction, the gauging procedure described above finds a concrete application in the development of the so-called diagonal Yang-Baxter deformation procedure, described in Subsection \ref{SubSec:DiagYB} of this article. Let us explain briefly the principle behind this construction. As explained above, the gauging procedure goes through the introduction of new degrees of freedom, taking the form of a PCM+WZ realisation attached to the site $\is$ of the gauged model. As explained in~\nm, models which possess such a realisation can be deformed by replacing this PCM+WZ realisation by a (homogeneous or inhomogeneous) Yang-Baxter realisation\footnote{Note that for simplicity, we actually consider in this article only deformations of the PCM realisation without Wess-Zumino term. We shall not enter further into these technical details in this introduction.}.

Applying such a construction to the gauged model obtained after the gauging procedure, one then gets an integrable deformation of this model, which also possesses a gauge symmetry. Gauge-fixing this symmetry then provides us with an integrable deformation of the non-gauged model we started with. We developp this idea in details in Subsection \ref{SubSec:DiagYB}. Note that Yang-Baxter deformations generally come with the breaking of a global symmetry of the model. In the present case, we show that this deformation breaks the global diagonal symmetry of the model, which is a natural symmetry of all models considered in~\nm, acting on all the Takiff currents at each site simultaneously. As announced at the beginning of this introduction, this is one of the main motivation of this article.

As an application of the general deformation procedure described in Subsection \ref{SubSec:DiagYB}, we construct in Subsection \ref{SubSec:YBSigma} the homogeneous diagonal Yang-Baxter deformation of the integrable coupled $\s$-model introduced in~\prl. In particular, this deformation breaks the diagonal symmetry of this model, which acts on the $N$ $G_0$-valued fields $\gb 1,\cdots,\gb N$ by a simultaneous right-multiplication $(\gb 1,\cdots,\gb N) \mapsto (\gb 1 h,\cdots,\gb N h)$, with $h\in G_0$ constant.

\paragraph{Plan of this article.} The plan of this article is the following. The section \ref{Sec:AGM} is devoted to the construction and the study of constrained realisations of affine Gaudin models and their gauge symmetry. Subsection \ref{SubSec:PhaseSpace} introduces the main ingredients of realisations of affine Gaudin models: Takiff currents, twist function, Gaudin Lax matrix and spectral parameter dependent quadratic Hamiltonian. Subsection \ref{SubSec:Gauge} concerns the introduction of a gauge symmetry in these models, using the Dirac formalism of first-class constraints. In Subsection \ref{SubSec:Zeroes}, we introduce quadratic charges associated with the zeroes of the twist function and explain their relation with the Hamiltonian and momentum of these models. In Subsection \ref{SubSec:Integrability}, we prove the integrability of constrained realisations of affine Gaudin models by constructing a Lax representation of their equations of motion and showing that their Lax matrix satisfies a Maillet bracket. In Subsection \ref{SubSec:SpaceTime}, we use the parametrisation of the Hamiltonian and momentum in terms of the quadratic charges associated with the zeroes of the twist function to study the space-time symmetries of the models and in particular give a simple condition for their relativistic invariance. Subsection \ref{SubSec:Reality} is a summary about reality conditions. Finally, in Subsection \ref{SubSec:ChangeSpec}, we show that two constrained realisations of affine Gaudin models related by an appropriate change of the spectral parameter are equivalent. Subsections \ref{SubSec:PhaseSpace}, \ref{SubSec:Gauge}, \ref{SubSec:Integrability} and \ref{SubSec:Reality} are mainly a review of the article~\bg, in a language similar to the one of the reference~\nm. Subsections \ref{SubSec:Zeroes}, \ref{SubSec:SpaceTime} and \ref{SubSec:ChangeSpec} contain new results, which are mostly the generalisations to constrained models of results found in~\nms for non-constrained ones.

Section \ref{Sec:Gauging} contains the most important results of this article. Subsections \ref{SubSec:NonConst} to \ref{SubSec:Gauging} are devoted to the gauging procedure which relates the non-constrained models of~\nms to the constrained ones considered in this article. Subsection \ref{SubSec:NonConst} is a brief review of~\nms about the construction and the properties of the non-constrained model the gauging procedure starts with. Subsection \ref{SubSec:ChangeSpec2} contains the construction of a gauged model, obtained from the non-gauged one by performing an appropriate change of spectral parameter. Subsection \ref{SubSec:Equivalence} is devoted to the proof of the equivalence of these two models through a well-chosen gauge-fixing condition. Finally, Subsection \ref{SubSec:Gauging} summarises the gauging procedure and studies its effects on the symmetries of the model and its integrable structure. The section ends with Subsection \ref{SubSec:DiagYB}, which is an application of the gauging procedure developed before: the diagonal Yang-Baxter deformation procedure.

Section \ref{Sec:SigmaModels} is concerned with the applications of the results of Section \ref{Sec:Gauging} to integrable $\s$-models. In Subsection \ref{SubSec:PCM}, we prove that the previously known gauged formulation of the PCM recalled in this introduction as a motivation for this article indeed fits into the general gauging procedure developed in Section \ref{Sec:Gauging}. In Subsection \ref{SubSec:NonGaugedSigma}, we recall from~\nms the construction of the integrable coupled $\s$-model (first introduced in~\prl) as a non-constrained realisation of affine Gaudin model. We then apply the gauging procedure of Section \ref{Sec:Gauging} to this model and describe its gauged formulation in Subsection \ref{SubSec:GaugedSigma}. Finally, we use this gauged formulation to construct the homogeneous diagonal Yang-Baxter deformation of this model in Subsection \ref{SubSec:YBSigma}, based on the general procedure developed in Subsection \ref{SubSec:DiagYB}.

Some technical results are gathered in Appendices \ref{App:Constraints} to \ref{App:Identities}. Appendix \ref{App:Constraints} is a review about the Dirac treatment of constrained Hamiltonian field theories and its relation with gauge symmetries. Appendix \ref{App:ChangeSpec} contains the proof of a result used in Subsections \ref{SubSec:ChangeSpec} and \ref{SubSec:ChangeSpec2} about the way Takiff currents are modified under a change of spectral parameter. Appendix \ref{App:Realisations} recalls the construction of the PCM+WZ and inhomogeneous realisations. Appendix \ref{App:Dirac} contains the details of the computation of certain Maillet brackets using the Dirac bracket. Finally, Appendix \ref{App:Identities} contains a reinterpretation of the coefficients appearing in both the non-gauged and gauged actions of the coupled integrable $\s$-model of~\prl~as residues of certain quantities and the application of that result to the derivation of some identities involving these coefficients.

\section{Realisations of Affine Gaudin Models with gauge symmetries}
\label{Sec:AGM}

In this section, we explain how to construct local affine Gaudin Models (AGM) with gauge symmetries, as originally introduced in~\cite{Vicedo:2017cge}. The presentation that we shall use here follows closely the one of~\cite{Delduc:2019bcl}, which concerned local AGM without gauge symmetries and their realisations. As a consequence, we will discuss only briefly the aspects which are common to both AGM with and without gauge symmetries and refer to~\cite{Vicedo:2017cge} and~\cite{Delduc:2019bcl} for more details. Thus, we will focus more here on the aspects which are related to the introduction of a gauge symmetry in the formalism of local AGM and their realisations.

\subsection{Takiff currents, Gaudin Lax matrix and quadratic hamiltonians}
\label{SubSec:PhaseSpace}

\subsubsection{Conventions}
\label{SubSubSec:Conventions}

Through this article, we will follow the conventions of~\cite{Delduc:2019bcl}. In particular, we will consider Hamiltonian field theories on a one-dimensional space $\D$, parametrised by a coordinate $x$. This space $\D$ can be chosen to be either the real line $\R$, in which case $x\in ]-\infty, +\infty[$, or the circle $\mathbb{S}^1$, in which case $x\in[0,2\pi[$. In the first case, we suppose that the fields of our theory are periodic and in the second case that they decrease sufficiently fast at infinity.

Moreover, we fix a finite-dimensional semi-simple complex Lie algebra $\g$ and a real form $\g_0$ of $\g$, which can be seen as the subalgebra of fixed-points under an involutive antilinear automorphism $\tau$ of $\g$. We denote by $\kappa$ the opposite of the Killing form of $\g$ and by $C\ti{12}$ the corresponding split-quadratic Casimir in $\g \otimes \g$, which then satisfies the identities
\begin{equation}\label{Eq:ComCasimir}
\bigl[ C\ti{12}, X\ti{1} + X\ti{2} \bigr] = 0, \;\;\;\;\;\; \forall \, X\in\g
\end{equation}
and
\begin{equation*}
\kappa\ti{1}\left( C\ti{12}, X\ti{1} \right) = X, \;\;\;\;\;\; \forall \, X\in\g,
\end{equation*}
where we used the standard tensorial notations $\textbf{\underline{i}}$.

Let us fix a basis $\lbrace I^a \rbrace_{a=,1,\cdots,\dim\g_0}$ of $\g_0$ (which is then also a basis of $\g$ over $\C$). We denote by $\kappa^{ab}=\kappa(I^a,I^b)$ the evaluation of the bilinear form $\kappa$ on this basis and $\kappa_{ab}$ its inverse, which satisfies $\kappa^{ac}\kappa_{cb}=\delta^a_{\,b}$. The split-quadratic Casimir is then given by
\begin{equation*}
C\ti{12} = \kappa_{ab} \, I^a \otimes I^b.
\end{equation*}
We will denote by $\ft {ab}c$ the structure constants of $\g_0$ in the basis $\lbrace I^a \rbrace_{a=,1,\cdots,\dim\g_0}$, defined by
\begin{equation*}
[ I^a, I^b ] = \fs {ab}c I^c.
\end{equation*}

\subsubsection{Takiff currents}
\label{SubSubSec:Takiff}

\paragraph{Takiff algebra.} Following the notations of~\cite{Delduc:2019bcl}, we consider the algebra $\Tc_{\lt}$ generated by a finite set of $\g$-valued \textit{Takiff currents} $\Jt\alpha p(x)$ on $\D$. These currents are formally attached to \textit{sites}, represented by labels $\alpha$ in an abstract finite set $\Si$. For each site $\alpha\in\Si$, these currents are further labelled by their \textit{Takiff mode} $p$, which is an integer ranging from $0$ to $m_\alpha-1$, where $m_\alpha\in\Z_{\geq 1}$ is called the \textit{multiplicity} of the site $\alpha$. The set of sites $\Si=\Si_\rd \sqcup \Si_\cd \sqcup \overline{\Si}_\cd$ is separated into real sites $\alpha\in\Si_\rd$, complex sites $\alpha\in\Si_\cd$ and conjugate sites $\overline{\alpha}\in\overline{\Si}_\cd$. The complex and conjugate sites come in pairs $\alpha\in\Si_\cd$ and $\overline{\alpha}\in\overline{\Si}_\cd$, with same multiplicities $m_\alpha=m_{\overline{\alpha}}$. These denominations are justified by the requirement that the currents $\Jt\alpha p$ are real, \textit{i.e.} $\g_0$-valued, for $\alpha\in\Si_\rd$ and complex, \textit{i.e.} $\g$-valued, and conjugate one-to-another for pairs $\alpha\in\Si_\cd$ and $\overline{\alpha}\in\overline{\Si}_\cd$. In other words, we have
\begin{equation*}
\tau\bigl(\Jt\alpha p \bigr) = \Jt\alpha p \;\;\; \text{ for } \;\;\; \alpha\in\Si_\rd \;\;\;\;\;\; \text{ and } \;\;\;\;\;\; \tau\bigl(\Jt\alpha p \bigr) = \Jt{\overline{\alpha}} p \;\;\; \text{ for } \;\;\; \alpha\in\Si_\cd.
\end{equation*}

The algebra of observables of a local AGM is the algebra $\Tc_{\lt}$ equipped with a Poisson bracket $\lbrace \cdot,\cdot \rbrace_{\Tc_{\lt}}$, making it a Poisson algebra. This Poisson bracket is defined at the level of the currents by
\begin{equation}\label{Eq:Takiff}
\left\lbrace \Jt\alpha p\,\ti{1}(x), \Jt\beta q\,\ti{2}(y) \right\rbrace_{\Tc_{\lt}} = \delta_{\alpha\beta} \left\lbrace \begin{array}{ll}
\left[ C\ti{12}, \Jt\alpha{p+q}\,\ti{1}(x) \right] \delta_{xy} - \ls\alpha{p+q}\, C\ti{12}\, \delta'_{xy} & \text{ if } p+q < m_\alpha \\[5pt]
0 & \text{ if } p+q \geq m_\alpha
\end{array}  \right. ,
\end{equation}
for all $\alpha,\beta\in\Si$, $p\in\lbrace 0,\cdots,m_\alpha-1\rbrace$ and $q\in\lbrace 0,\cdots,m_\beta-1\rbrace$. In this equation, $\delta_{xy}$ denotes the Dirac distribution $\delta(x-y)$ on $\D$ and $\delta'_{xy}$ its derivative $\p_x\delta(x-y)$. Finally, the $\ls\alpha p$'s, for $\alpha\in\Si$ and $p\in\lbrace 0,\cdots,m_\alpha\rbrace$, are numbers that we suppose real for $\alpha\in\Si_\rd$ and complex conjugate one-to-another for pairs $\alpha\in\Si_\cd$ and $\overline{\alpha}\in\overline{\Si}_\cd$. The data defining the Poisson algebra $\Tc_{\lt}$ is encoded in a compact way in the notation $\lt$, called the \textit{Takiff datum}, which is defined as
\begin{equation*}
\lt = \Bigl( \bigl( \ls{\alpha}{p} \bigr)^{\alpha \in \Si_\rd}_{p\in\lbrace 0,\cdots,m_\alpha-1\rbrace}, \bigl( \ls{\alpha}{p} \bigr)^{\alpha \in \Si_\cd}_{p\in\lbrace 0,\cdots,m_\alpha-1\rbrace} \Bigr).
\end{equation*}

\paragraph{Generalised Segal-Sugawara integrals.} We end this paragraph by a brief reminder about \textit{generalised Segal-Sugawara integrals}. These are defined, for $\alpha\in\Si$ and $p\in\lbrace 0,\cdots,m_\alpha-1\rbrace$, as
\begin{equation*}
D^\alpha_{[p]} = \frac{1}{2} \sum_{\substack{q,r=0 \\ q+r \geq p}}^{m_\alpha-1} \kb\alpha{q+r-p}  \int_\D \dd x \; \kappa\! \left( \Jt\alpha q(x), \Jt\alpha r(x) \right),
\end{equation*}
where the numbers $\kb \alpha p$ form the unique solution of the system of linear equations
\begin{equation}\label{Eq:DefKb}
\forall\, q,r \in \lbrace 0,\cdots,m_\alpha-1 \rbrace, \;\;\;\;\;\; \sum_{p=0}^{m_\alpha-1-r} \kb\alpha{p+q}\,\ls\alpha{p+r} = \delta_{q,r}.
\end{equation}
We refer to~\cite{Vicedo:2017cge,Delduc:2019bcl} for more details about the numbers $\kb \alpha p$ and the generalised Segal-Sugawara integrals. The main property of these integrals is the fact that they satisfy the following Poisson brackets:
\begin{equation}\label{Eq:PbSS}
\left\lbrace D^\alpha_{[p]}, D^\beta_{[q]} \right\rbrace = 0 \hspace{20pt} \text{and} \hspace{20pt} \left\lbrace D^\alpha_{[p]}, \Jt\beta q(x) \right\rbrace = \delta_{\alpha\beta}\left\lbrace \begin{array}{ll}
\p_x \Jt\alpha {p+q}(x) & \text{ if } p+q < m_\alpha \\[5pt]
0 & \text{ if } p+q \geq m_\alpha
\end{array}  \right.
\end{equation}
In particular, the momentum $\Pc_{\lt}$ of the Poisson algebra $\Tc_{\lt}$, whose Poisson bracket generates the derivative $\p_x$ with respect to the spatial coordinate $x$, is then given by
\begin{equation}\label{Eq:Momentum}
\Pc_{\lt} = \sum_{\alpha\in\Si} D^\alpha_{[0]}.
\end{equation}

Finally, let us recall that Equation \eqref{Eq:DefKb}, which defines the numbers $\kb\alpha p$, only admits a solution if the highest level $\ls \alpha {m_\alpha-1}$ of each site is non-zero. We shall suppose that this is always the case in the rest of this article.

\subsubsection{Realisations of local AGM}
\label{SubSubSec:Real}

As explained above, the algebra of observables of a local AGM is the Poisson algebra $\Tc_{\lt}$. In this article, as in~\cite{Delduc:2019bcl}, we will be interested in \textit{realisations of local AGM}, whose algebras of observables are \textit{Takiff realisations} of $\Tc_{\lt}$. Such a realisation is a Poisson algebra $(\Ac,\lbrace\cdot,\cdot\rbrace)$, together with a conjugacy-equivariant Poisson map
\begin{equation*}
\pi: \Tc_{\lt} \longrightarrow \Ac.
\end{equation*}
It is uniquely characterised by the currents
\begin{equation}\label{Eq:PiReal}
\J\alpha p(x) = \pi \bigl( \Jt\alpha p(x) \bigr), \;\;\;\;\;\; \text{ for } \; \alpha \in \Si \; \text{ and } \; p\in\lbrace 0,\cdots, m_\alpha-1 \rbrace.
\end{equation}
These are fields in the Poisson algebra $\Ac$, satisfying the same Poisson bracket as the abstract Takiff current $\Jt\alpha p(x)$, \textit{i.e.} such that
\begin{equation}\label{Eq:TakiffReal}
\left\lbrace \J\alpha p\,\ti{1}(x), \J\beta q\,\ti{2}(y) \right\rbrace = \delta_{\alpha\beta} \left\lbrace \begin{array}{ll}
\left[ C\ti{12}, \J\alpha{p+q}\,\ti{1}(x) \right] \delta_{xy} - \ls\alpha{p+q}\, C\ti{12}\, \delta'_{xy} & \text{ if } p+q < m_\alpha \\[5pt]
0 & \text{ if } p+q \geq m_\alpha
\end{array}  \right. .
\end{equation}
Moreover, they also satisfy the same reality conditions:
\begin{equation}\label{Eq:RealityJc}
\tau \bigl( \J\alpha p(x) \bigr) = \J\alpha p(x), \;\;\;\; \forall \, \alpha\in\Si_\rd \;\;\;\;\;\;\;\;\; \text{ and } \;\;\;\;\;\;\;\; \tau \bigl( \J\alpha p(x) \bigr) = \J{\bar\alpha} p(x), \;\;\;\; \forall \, \alpha\in\Si_\cd.
\end{equation}
Conversely, a set of currents $\J\alpha p(x)$ in a Poisson algebra $\Ac$ satisfying Equations \eqref{Eq:TakiffReal} and \eqref{Eq:RealityJc} defines a Takiff realisation of $\Tc_{\lt}$ by \eqref{Eq:PiReal}.

As in~\cite{Delduc:2019bcl}, it is useful to introduce the notion of a \textit{suitable realisation}, defined by the requirement that $\Pc_{\Ac} = \pi\bigl(\Pc_{\lt}\bigr)$, where $\Pc_{\Ac}$ and $\Pc_{\lt}$ are respectively the momentum of the Poisson algebras $\Ac$ and $\Tc_{\lt}$. From the expression \eqref{Eq:Momentum} of $\Pc_{\lt}$, it is clear that the realisation is suitable if and only if
\begin{equation}\label{Eq:MomentumReal}
\Pc_{\Ac} =  \sum_{\alpha\in\Si} \Dc\alpha 0,
\end{equation}
where for $\alpha\in\Si$ and $p\in\lbrace 0,\cdots,m_\alpha-1 \rbrace$, we defined
\begin{equation}\label{Eq:ImageSS}
\Dc\alpha p = \pi \bigl( D^\alpha_{[p]} \bigr) = \frac{1}{2} \sum_{\substack{q,r=0 \\ q+r \geq p}}^{m_\alpha-1} \kb\alpha{q+r-p}  \int_\D \dd x \; \kappa\! \left( \J\alpha q(x), \J\alpha r(x) \right),
\end{equation}
the generalised Segal-Sugawara integral $D^{\alpha}_{[p]}$ seen in the realisation $\Ac$.

\subsubsection{Gaudin Lax matrix and twist function}

As in~\cite{Vicedo:2017cge,Delduc:2019bcl}, we define the \textit{Gaudin Lax matrix} of the model as
\begin{equation}\label{Eq:S}
\Sg(z,x) = \sum_{\alpha\in\Si} \sum_{p=0}^{m_\alpha-1} \frac{\J\alpha p(x)}{(z-\po_\alpha)^{p+1}}.
\end{equation}
In this expression, $z$ is the spectral parameter, \textit{i.e.} an auxiliary variable taking values in the complex plane $\C$, and the numbers $\po_\alpha\in\C$, $\alpha\in\Si$, are called the \textit{positions} of the sites and are part of the defining parameters of the model. We suppose that they obey the reality conditions
\begin{equation}\label{Eq:zReal}
\po_\alpha \in \R \;\; \text{ for } \alpha\in\Si_\rd \;\;\;\;\;\; \text{ and } \;\;\;\;\;\; \overline{\po_\alpha} = \po_{\overline{\alpha}} \;\; \text{ for } \alpha\in\Si_\cd
\end{equation}
and encode them into a unique object
\beqz
\pb = \Bigl( ( \po_\alpha )_{\alpha \in \Si_\rd}, ( \po_\alpha )_{\alpha \in \Si_\cd} \Bigr).
\eeqz

The Poisson bracket \eqref{Eq:TakiffReal} of the Takiff currents $\J\alpha p$ translates into the following bracket for the Gaudin Lax matrix~\cite{Vicedo:2017cge,Delduc:2019bcl}
\begin{equation}\label{Eq:PbGaudin}
\bigl\lbrace \Sg\ti{1}(z,x), \Sg\ti{2}(w,y) \bigr\rbrace = \left[ \frac{C\ti{12}}{w-z}, \Sg\ti{1}(z,x)-\Sg\ti{1}(w,x) \right] \delta_{xy} - \bigl(\vp(z)-\vp(w)\bigr) \frac{C\ti{12}}{w-z} \delta'_{xy},
\end{equation}
where $\vp$ is a rational function of the spectral parameter, called the \textit{twist function}. It is defined in terms of the levels $\ls\alpha p$ as
\begin{equation}\label{Eq:Twist}
\vp(z) = \sum_{\alpha\in\Si} \sum_{p=0}^{m_\alpha-1} \frac{\ls\alpha p}{(z-\po_\alpha)^{p+1}}.
\end{equation}
It is clear that Equation \eqref{Eq:PbGaudin} is also satisfied if one adds to the twist function $\vp(z)$ a constant term. This is the situation considered in~\cite{Delduc:2019bcl}, where this constant term was supposed to be non-zero. In this article, we shall consider in the contrary that there is no such constant term, so that the twist function is simply given by \eqref{Eq:Twist}. As we will see in Subsection \ref{SubSec:Gauge}, this is one of the conditions necessary to obtain a realisation of local AGM with a gauge symmetry. The constant term in $\vp(z)$ considered in~\cite{Delduc:2019bcl} corresponds to a double pole of the 1-form $\vp(z)\dd z$ at $z=\infty$\footnote{As explained in~\cite{Vicedo:2017cge}, in the language of formal AGM, it can be seen as a site of multiplicity 2 at infinity.}. In contrast, the twist function \eqref{Eq:Twist} considered here then yields a form $\vp(z)\dd z$ with at most a simple pole at $z=\infty$. We shall come back on these considerations in Subsection \ref{SubSec:Gauge}.

\subsubsection{Quadratic hamiltonians in involution}

Let us define the following function of the spectral parameter:
\begin{equation}\label{Eq:P}
\Pc(z) = \sum_{\alpha\in\Si} \sum_{p=0}^{m_\alpha-1} \frac{\Dc\alpha p}{(z-\po_\alpha)^{p+1}},
\end{equation}
with $\Dc\alpha p$ as in Equation \eqref{Eq:ImageSS}. Let us note that the momentum \eqref{Eq:MomentumReal} of the theory can then be re-expressed as
\begin{equation*}
\Pc_\Ac = - \res_{z=\infty} \Pc(z) \, \dd z.
\end{equation*}
The quantity $\Pc(z)$ satisfies~\cite{Vicedo:2017cge,Delduc:2019bcl}
\begin{equation}\label{Eq:PbPS}
\lbrace \Pc(z), \Pc(w) \rbrace = 0 \;\;\;\;\; \text{ and  } \;\;\;\;\; \lbrace \Pc(z), \Sg(w,x) \rbrace = - \frac{\p_x \Sg(z,x)-\p_x \Sg(w,x)}{z-w}. 
\end{equation}
We now define the spectral parameter dependent \textit{quadratic hamiltonian} as
\begin{equation}\label{Eq:HamSpec}
\Hc(z) = \frac{1}{2} \int_\D \dd x \; \kappa\bigl( \Sg(z,x), \Sg(z,x) \bigr) - \vp(z) \Pc(z).
\end{equation}
One checks that it satisfies~\cite{Vicedo:2017cge,Delduc:2019bcl}
\begin{equation}\label{Eq:InvolutionHamSpec}
\lbrace \Hc(z), \Hc(w) \rbrace = 0,
\end{equation}
for all $z,w \in \C$. Thus, all quantities linearly extracted from $\Hc(z)$ (evaluation at particular points $z$, residues, ...) are quadratic local charges in involution one with another. Moreover, as $\Hc(z)$ is defined as the integral of a local field over the whole spatial domain $\D$, it is invariant under spatial translation and thus satisfies
\begin{equation*}
\lbrace \Pc_{\Ac}, \Hc(z) \rbrace = 0.
\end{equation*}
Thus, the quantities linearly extracted from $\Hc(z)$ are also in involution with the momentum $\Pc_\Ac$.

\subsection{Constrained realisations of local AGM and diagonal gauge symmetry}
\label{SubSec:Gauge}

\subsubsection{The current at infinity}

The realisations of local AGM considered in~\nms possess a global $G_0$-symmetry called the \textit{diagonal symmetry}. The $\g_0$-valued density of the charge generating this symmetry is extracted from the Gaudin Lax matrix at infinity. Following a similar idea, we define here
\begin{equation}\label{Eq:C}
\Cc(x) = - \res_{z=\infty} \Sg(z,x) \dd z = \sum_{\alpha\in\Si} \J\alpha 0(x).
\end{equation}
Although the field $\Cc(x)$ introduced here is defined in exactly the same way as the one denoted $\Kc^\infty(x)$ in~\cite{Delduc:2019bcl}, we choose here a different notation $\Cc(x)$ (which will be justified in Subsection \ref{SubSubSec:Constraint}) to emphasise the qualitative difference with~\cite{Delduc:2019bcl}. One of these differences is the Poisson bracket of $\Cc(x)$ with the spectral parameter dependent Hamiltonian $\Hc(z)$, that we will compute in the rest of this paragraph.\\

We shall need a few intermediate results. It is for example useful to note that the field $\Cc(x)$ can be rewritten as
\begin{equation*}
\Cc(x) = \lim_{u\to 0} \; \frac{1}{u} \; \Sg\left(\frac{1}{u},x\right).
\end{equation*}
Using this expression and Equation \eqref{Eq:PbGaudin}, one easily finds
\begin{equation}\label{Eq:PBSC}
\bigl\lbrace \Cc\ti{1}(x), \Gamma\ti{2}(z,y) \bigr\rbrace = \bigl[ C\ti{12}, \Gamma\ti{1}(z,x) \bigr] \delta_{xy} - \vp(z) \, C\ti{12} \, \delta'_{xy}.
\end{equation}
Similarly, one gets from Equation \eqref{Eq:PbPS} that
\begin{equation*}
\bigl\lbrace \Cc(x), \Pc(z) \bigr\rbrace = -\p_x \Gamma(z,x).
\end{equation*}
Combining these two Poisson brackets and the definition \eqref{Eq:HamSpec} of the Hamiltonian, we obtain
\begin{equation}\label{Eq:ConservedConstraint}
\lbrace \Hc(z), \Cc(x) \rbrace = 0, \;\;\;\;\; \forall \, z\in\C.
\end{equation}
Thus, the field $\Cc(x)$ Poisson commutes with all quadratic charges extracted from $\Hc(z)$, for all values of the space coordinate $x\in\D$. This is in contrast with the situation in~\nms (see Equation (2.42) therein), where the Poisson bracket of $\Hc(z)$ with the field $\Kc^\infty(x)$ is a spatial derivative, so that only its integral over $\D$ is in involution with $\Hc(z)$. As a consistency check we observe that, in Equation (2.42) of~\nm, this spatial derivative is proportional to the constant term $\ell^\infty$ in the twist function, so that we recover Equation \eqref{Eq:ConservedConstraint} by letting $\ell^\infty$ go to $0$, which is the framework of the present article.

\subsubsection{First-class constraint}
\label{SubSubSec:Constraint}

\paragraph{Constraint.} Let us suppose here that we consider an hamiltonian field theory with observables $\Ac$ and with an Hamiltonian extracted linearly from $\Hc(z)$\footnote{Note that we will consider a slightly more general Hamiltonian later (see Subsection \ref{SubSubSec:Hamiltonian}).}. The current $\Cc(x)$ introduced in the previous paragraph would then be conserved for every $x\in\D$ and would thus generate a local symmetry of the model. Let us be more precise about what we mean by that. For every $\g_0$-valued function $x\mapsto\epsilon(x)$ on $\D$, let us consider the infinitesimal transformation
\begin{equation*}
\delta^\infty_\epsilon \mathcal{O}= \left\lbrace \int_\D \dd x \; \kappa\bigl( \epsilon(x), \Cc(x) \bigr), \mathcal{O} \right\rbrace,
\end{equation*}
acting on observables $\mathcal{O}$ in $\Ac$. This is a local transformation (as opposed to a global one), in the sense that its parameter $\epsilon$ is an arbitrary function of $x$ and not a constant. Moreover, this is a symmetry of the model, as it leaves invariant the Hamiltonian. Indeed, Equation \eqref{Eq:ConservedConstraint} implies
\begin{equation*}
\delta^\infty_\epsilon \Hc(z) = 0,
\end{equation*}
for all functions $x\mapsto\epsilon(x)$.

It would be tempting at this point to interpret the infinitesimal symmetry $\delta^\infty_\epsilon$ as a gauge symmetry. However, it is not the case. Indeed, a gauge symmetry would be a transformation leaving the model invariant and which would depend on an arbitrary function $\epsilon(x,t)$ of both space and time coordinates $x$ and $t$. In the situation described above, the transformation $\delta^\infty_\epsilon$ is a symmetry only for $\epsilon$ depending on the space coordinate $x$, and not the time coordinate $t$. To promote the transformation $\delta^\infty_\epsilon$ into a gauge symmetry, one has to treat its generator $\Cc(x)$ as a \textit{constraint} (hence the notation $\Cc$) and study the model in the Dirac formalism of constrained Hamiltonian systems. We refer to the Appendix \ref{App:Constraints} for a brief reminder about constrained Hamiltonian fields theories and in particular the link between constraints and gauge symmetries. From now on, we shall then restrict the phase space of the model to the constrained surface defined by the constraint
\begin{equation}\label{Eq:Constraint}
\Cc(x) \approx 0.
\end{equation}
Here, we used the standard \textit{weak equality} notation $\approx$ of Dirac to designate an equality which is true only on the constrained surface. As $\Cc$ is a $\g_0$-valued field, the constraint \eqref{Eq:Constraint} should be interpreted as $\dim\g_0$ scalar constraints $\Cc^a(x) \approx 0$, $a=1,\cdots,\dim\g_0$, that we extract from $\Cc(x)$ as (see Paragraph \ref{SubSubSec:Conventions} for the definitions of $I^a$ and $\kappa_{ab}$)
\begin{equation*}
\Cc(x) = \kappa_{ab} \, \Cc^a(x) \, I^b.
\end{equation*}

\paragraph{First-class condition.} We now follow the standard treatment of constrained Hamitlonian systems (see for example~\cite{dirac1964lectures,Henneaux:1992ig} for a detailed exposition and Appendix \ref{App:Constraints} for a summary). An additional condition for the constraint $\Cc$ to generate a gauge symmetry is that it is a \textit{first-class constraint}\footnote{More precisely, we ask that the set of constraints $\lbrace \Cc^a(x) \rbrace^{a=1,\cdots,\dim(\g_0)}_{x\in\D}$, obtained by varying the space variable $x$ and the Lie algebra index $a$, defines a system of first-class constraints.}, \textit{i.e.} that
\begin{equation}\label{Eq:FirstClass}
\bigl\lbrace \Cc\ti{1}(x), \Cc\ti{2}(y) \bigr\rbrace \approx 0.
\end{equation}
This first-class condition ensures that the transformation generated by $\Cc(x)$, which we would like to interpret as a gauge transformation of the model, preserves the constrained surface \eqref{Eq:Constraint}.

It is clear from the expression \eqref{Eq:C} of $\Cc(x)$ and the Poisson bracket \eqref{Eq:TakiffReal} of the currents $\J\alpha 0(x)$ that $\Cc(x)$ is a Kac-Moody current:
\vspace{-8pt}\begin{equation}\label{Eq:PBConstraints}
\bigl\lbrace \Cc\ti{1}(x), \Cc\ti{2}(y) \bigr\rbrace = \bigl[ C\ti{12}, \Cc\ti{1}(x) \bigr] \delta_{xy} - \biggl( \; \sum_{\alpha\in\Si} \, \ls\alpha 0 \, \biggr) C\ti{12}\delta'_{xy}. \vspace{-6pt}
\end{equation}
Thus, the constraint $\Cc(x)$ is first-class, \textit{i.e.} satisfies \eqref{Eq:FirstClass}, if and only if the level of $\Cc(x)$ as a Kac-Moody current vanishes:
\begin{equation}\label{Eq:SumLevels}
\sum_{\alpha\in\Si} \, \ls \alpha 0 = 0.
\end{equation}
In the rest of this article, we will suppose that this condition, that we will call \textit{first-class condition}, is satisfied. Let us note here that it can be rewritten as
\begin{equation}\label{Eq:ResidueInf}
\res_{z=\infty} \vp(z) \dd z = 0.
\end{equation}
Thus, the first-class condition ensuring that $\Cc(x)$ generates a gauge symmetry is equivalent to require that the twist function 1-form $\vp(z)\dd z$ is regular at $z=\infty$.

For completeness, let us rewrite the bracket \eqref{Eq:PBConstraints} in terms of the components $\Cc^a(x)$ of the $\g_0$-valued constraint $\Cc(x)$. One finds
\begin{equation*}
\bigl\lbrace \Cc^a(x), \Cc^b(y) \bigr\rbrace = \fs {ab}c \, \Cc^c(x) \, \delta_{xy},
\end{equation*}
with $\ft{ab}c$ the structure constants defined in Paragraph \ref{SubSubSec:Conventions}. Thus, we are in the framework described in Appendix \ref{App:Constraints}, as this Poisson bracket takes the form \eqref{Eq:AppAlgebraConstraints}, with $\Sc abc(x) = \ft{ab}c$.

\subsubsection{Hamiltonian and dynamics}
\label{SubSubSec:Hamiltonian}

\paragraph{First-class Hamiltonian.} In the previous paragraph, we considered an Hamiltonian extracted linearly from the spectral parameter dependent Hamiltonian $\Hc(z)$, so that it Poisson commutes with the constraint. We shall now consider a slightly more general Hamiltonian $\Hc_0$ constructed as a charge linearly extracted from $\Hc(z)$ to which we add a real constant coefficient $\xi$ times the momentum $\Pc_\Ac$ of the model. As the latter generates the derivative with respect to the space coordinate $x$, the Poisson bracket of $\Hc_0$ with the constraint is simply
\begin{equation*}
\lbrace \Hc_0, \Cc^a(x) \rbrace = \xi \, \p_x \Cc^a(x).
\end{equation*}
In particular, although this Poisson bracket does not vanish for a non-zero coefficient $\xi$, it vanishes weakly:
\begin{equation*}
\lbrace \Hc_0, \Cc^a(x) \rbrace \approx 0.
\end{equation*}
Following Appendix \ref{App:Constraints}, we then say that $\Hc_0$ is a first-class Hamiltonian. In the notations of the appendix, $\Hc_0$ satisfies the hypothesis \eqref{Eq:AppPBHConstraints}, where the fields $\Tcc ab(x)$ are simply given by $0$ and the fields $\Uc ab(x)$ are actually non-dynamical and equal to $\xi \, \delta^a_{\,b}$.

\paragraph{Total Hamiltonian and Lagrange multipliers.} Following Appendix \ref{App:Constraints}, we then define the total Hamiltonian of the model as in Equation \eqref{Eq:AppTotalHam}, introducing auxiliary fields $\mu_a(x)$, $a=1,\cdots,\dim\g_0$, called \textit{Lagrange multipliers}. This total Hamiltonian then reads
\begin{equation}\label{Eq:TotalHam}
\Hc = \Hc_0 + \int_{\D} \dd x \; \kappa\bigl( \mu(x), \Cc(x) \bigr),
\end{equation}
with $\mu(x)=\mu_a(x)\,I^a$ in $\g_0$. The time evolution of any observable $\Oo$ in $\Ac$ is then given by Equation \eqref{Eq:AppTime}:
\begin{equation*}
\p_t \Oo \approx \lbrace \Hc, \Oo \rbrace \approx \lbrace \Hc_0, \Oo \rbrace + \int_\D \dd x \; \kappa\bigl( \mu(x), \lbrace \Cc(x), \Oo \rbrace \bigr).
\end{equation*}

\subsubsection{Action of the gauge symmetry}

Let us end this subsection by discussing how the gauge symmetry associated with the constraint $\Cc(x)$ acts on the dynamical fields of the model. Following Appendix \ref{App:Constraints}, this gauge symmetry is the canonical transformation \eqref{Eq:AppGauge} generated by the constraints $\Cc^a(x)$ and with infinitesimal parameters $\epsilon_a(x,t)$, which are arbitrary functions of space-time coordinates $x$ and $t$. Combining these parameters into a unique $\g_0$-valued parameter $\epsilon(x,t)=\epsilon_a(x,t) I^a$, we write the infinitesimal variation of an observable $\Oo$ in $\Ac$ as
\begin{equation*}
\delta^\infty_\epsilon \Oo \approx \int_\D \dd x \; \kappa\bigl( \epsilon(x,t), \lbrace \Cc(x), \Oo \rbrace \bigr).
\end{equation*}
In particular, one can compute the variation of the Gaudin Lax matrix $\Sg(z,x)$, using Equation \eqref{Eq:PBSC}. One then finds
\begin{equation*}
\delta^\infty_\epsilon \Gamma(z,x) = \bigl[ \Gamma(z,x), \epsilon(x,t) \bigr] + \vp(z)\, \p_x \epsilon(x,t).
\end{equation*}
Using the definition \eqref{Eq:S} of $\Gamma(z,x)$ in terms of the Takiff currents and the one \eqref{Eq:Twist} of $\vp(z)$ in terms of the levels, one gets the action of the gauge symmetry on the Takiff currents:
\begin{equation*}
\delta_\epsilon^\infty \J\alpha p(x) = \bigl[\J\alpha p(x), \epsilon(x,t) \bigr] + \ls\alpha p \, \p_x \epsilon(x,t).
\end{equation*}
This transformation then acts on the Takiff currents associated with different sites in a very similar way. For this reason, we shall call it the \textit{diagonal gauge symmetry}.

As explained in Appendix \ref{App:Constraints}, gauge symmetries also transform the Lagrangian multipliers of the model according to Equation \eqref{Eq:AppTransfMult}. Applying this equation in the present case with $\Sc abc(x)=\ft{ab}c$, $\Tcc ab(x)=0$ and $\Uc ab(x) = \xi \, \delta^a_{\,b}$ yields
\begin{equation*}
\delta_\epsilon^\infty \mu_a(x) = \p_t\epsilon_a(x,t) - \xi\, \p_x\epsilon_a(x,t) -  \fs{bc}a \epsilon_b(x,t) \, \mu_c(x).
\end{equation*}
In terms of the $\g_0$-valued Lagrange multiplier $\mu(x)$, this becomes
\begin{equation*}
\delta_\epsilon^\infty \mu(x) = \bigl[ \mu(x), \epsilon(x,t) \bigr] + \p_t\epsilon(x,t) - \xi\, \p_x\epsilon(x,t).
\end{equation*}

The transformations described above correspond to infinitesimal gauge symmetries. They can be integrated to finite gauge symmetries with parameter a $G_0$-valued field $h(x,t)$ (where $G_0$ is a connected Lie group with Lie algebra $\g_0$). It acts on the Gaudin Lax matrix as
\begin{equation}\label{Eq:GaugeTransfGamma}
\Gamma(z,x) \longmapsto h(x,t)^{-1}\Gamma(z,x)\,h(x,t) + \vp(z) \, h(x,t)^{-1} \p_x h(x,t),
\end{equation}
and on the Lagrange multiplier as
\begin{equation}\label{Eq:GaugeTransfMult}
\mu(x) \longmapsto h(x,t)^{-1}\mu(x)\,h(x,t) +  h(x,t)^{-1} \p_t h(x,t) - \xi \,  h(x,t)^{-1} \p_x h(x,t).
\end{equation}

\subsection{Zeroes of the twist function, Hamiltonian and momentum}
\label{SubSec:Zeroes}

\subsubsection{Zeroes of the twist function}

Let us consider the twist function \eqref{Eq:Twist} put in a common denominator form. As we are considering here the case where it does not possess a constant term, it can be written in the form
\begin{equation*}
\vp(z) = \dfrac{f(z)}{\displaystyle \prod_{\alpha\in\Si} (z-\po_r)^{m_\alpha}},\vspace{-4pt}
\end{equation*}
where $f(z)$ is a polynomial of degree at most $\sum_{\alpha\in\Si} m_\alpha - 1$. Recall also that to obtain a model with gauge symmetry, we supposed that the first-class condition \eqref{Eq:ResidueInf} is satisfied, \textit{i.e.} that the 1-form $\vp(z)\dd z$ is regular at $z=\infty$. This implies that $f(z)$ is in fact of degree at most
\begin{equation}\label{Eq:DefMZeros}
M = \sum_{\alpha\in\Si} m_\alpha - 2.
\end{equation}
One can then distinguish three different cases:
\begin{enumerate}[(i)]
\item $f$ is of degree $M$, hence $z=\infty$ is not a zero of the 1-form $\vp(z)\dd z$ ;
\item $f$ is of degree $M-1$, hence $z=\infty$ is a simple zero of the 1-form $\vp(z)\dd z$ ;
\item $f$ is of degree strictly less than $M-1$, hence $z=\infty$ is a multiple zero of the 1-form $\vp(z)\dd z$.
\end{enumerate}
In this article, we shall restrict ourselves to the first two cases. In case (i), the 1-form $\vp(z)\dd z$ has $M$ finite zeroes $\ze_1,\cdots,\ze_M \in \C$. In case (ii), it has $M-1$ finite zeroes $\ze_1,\cdots,\ze_{M-1}\in\C$ ; to keep the discussion as uniform as possible, we then also let $\ze_M=\infty$ so that in both cases the zeroes of $\vp(z)\dd z$ in the Riemann sphere $\mathbb{P}^1$ are $\ze_1,\cdots,\ze_M$. We define $M_f$ as the number of finite zeroes of $\vp(z)\dd z$: the cases (i) and (ii) then respectively correspond to $M_f=M$ and $M_f=M-1$ and we can write
\begin{equation}\label{Eq:TwistZeros}
\vp(z) = - \ell^\infty \frac{\displaystyle \prod_{i=1}^{M_f} (z-\ze_i)}{\displaystyle \prod_{\alpha\in\Si} (z-\po_r)^{m_\alpha}},
\end{equation}
where $\ell^\infty$ is a real non-zero number. We will suppose in the following that the zeroes of the twist function are simple so that the $\ze_i$'s are pairwise distinct.

We chose the notation for the constant $\ell^\infty$ to be the same as the one for the constant term in the twist function of the reference~\nms due to the similarity of Equation \eqref{Eq:TwistZeros} with the corresponding equation in~\nms (page 15). However, its definition in the present context is different from the one of reference~\nm, as here it is not the constant term of the twist function and hence is not a free parameter. In the case where $M_f=M$, it is related to the levels and the positions of the sites of the model by the equation
\begin{equation*}
\ell^\infty = - \lim_{u\to 0}\, \frac{1}{u^2}\, \vp\left(\frac{1}{u}\right) = - \sum_{\alpha\in\Si} \Bigl( \po_\alpha \, \ls \alpha 0 + \ls \alpha 1 \Bigr),
\end{equation*}
where $\ls \alpha 1$ should be understood as zero if the multiplicity $m_\alpha$ is zero. In the case $M_f=M-1$, the above expression vanishes and we have instead
\begin{equation*}
\ell^\infty = - \lim_{u\to 0}\, \frac{1}{u^3}\, \vp\left(\frac{1}{u}\right) = - \sum_{\alpha\in\Si} \Bigl( \po_\alpha^2 \, \ls \alpha 0 + \po_\alpha \, \ls \alpha 1 + \ls \alpha 2 \Bigr),
\end{equation*}
where we let $\ls \alpha 2=0$ if $m_\alpha=1$ and $\ls \alpha 1 = \ls \alpha 2=0$ if $m_\alpha=0$.\\

Let us end this paragraph by introducing the function
\begin{equation}\label{Eq:DefChi}
\chi(u) = - \frac{1}{u^2} \vp\left(\frac{1}{u}\right).
\end{equation}
It describes the 1-form $\vp(z)\dd z$ in the coordinate patch of $\mathbb{P}^1$ around infinity, in the sense that $\vp(z)\dd z=  \chi(u)\dd u$ for $u=1/z$. The first-class condition \eqref{Eq:ResidueInf} translates to the regularity of $\chi(u)$ at $u=0$. The case (i) described above corresponds to $\chi(0) \neq 0$ and we then have $\ell^\infty = \chi(0)$. Similarly, the case (ii) corresponds to $\chi(0)=0$ and $\chi'(0) \neq 0$ and we then have $\ell^\infty=\chi'(0)$.

\subsubsection{Quadratic charges associated with the zeroes}

As in~\nm, we introduce the spectral parameter dependent quadratic charge
\begin{equation}\label{Eq:QSpec}
\Q(z) = - \frac{1}{2\vp(z)} \int_\D \dd x \; \kappa\bigl( \Gamma(z,x), \Gamma(z,x) \bigr).
\end{equation}
The quantity $\Q(z)$ depends rationally on the spectral parameter $z$. In the complex plane, it has poles at the positions $z_\alpha$, $\alpha\in\Si$, of the sites and at the finite zeroes $\ze_i$, $i\in\lbrace1,\cdots,M_f\rbrace$, of the twist function. The coefficients of the poles of $\Q(z)$ at the $\po_\alpha$'s have been studied in the article~\nm, Proposition A.3. Going through the proof of this statement in~\nm, one sees that it does not depend on the presence or not of a constant term in the twist function. Thus, it also applies in the context of the present article and one then has
\begin{equation*}
\Q(z) = -\sum_{\alpha\in\Si} \sum_{p=0}^{m_\alpha-1} \frac{\Dc\alpha p}{(z-z_\alpha)^{p+1}} + \Q_{\text{reg}}(z),
\end{equation*}
where the $\Dc\alpha p$'s are the generalised Segal-Sugawara integrals \eqref{Eq:ImageSS} and $\Q_{\text{reg}}(z)$ is a rational function of $z$ regular at $z=z_\alpha$ for all $\alpha\in\Si$.

In the complex plane, the function $\Q_{\text{reg}}(z)$ then has poles only at the finite zeroes $\ze_i$, $i\in\lbrace1,\cdots,M_f\rbrace$, of the twist function. Moreover, as we supposed that these zeroes are simple, these poles are simple. Let us then define the corresponding residues
\begin{equation}\label{Eq:QRes}
\Q_i = \res_{z=\ze_i} \, \Q(z)\dd z.
\end{equation}
Recall the spectral parameter dependent Hamiltonian $\Hc(z)$, defined in Equation \eqref{Eq:HamSpec}. Using the fact that $\ze_i$ is a simple zero of $\vp(z)$ and that $\Gamma(z,x)$ is regular at $z=\ze_i$, one can rewrite the residue $\Q_i$ introduced above as
\begin{equation}\label{Eq:QHZeros}
\Q_i = - \frac{\Hc(\ze_i)}{\vp'(\ze_i)} = - \frac{1}{2\vp'(\ze_i)} \int_{\D} \dd x \; \kappa\bigl( \Gamma(\ze_i,x), \Gamma(\ze_i,x) \bigr).
\end{equation}
In particular, we see from Equation \eqref{Eq:InvolutionHamSpec} that the charges $\Q_i$'s are then in involution:
\begin{equation}\label{Eq:QiInvolution}
\lbrace \Q_i, \Q_j \rbrace = 0.
\end{equation}

In addition to simple poles at the $\ze_i$'s ($i=1,\cdots,M_f$) with residues $\Q_i$, the partial fraction decomposition of the function $\Q_{\text{reg}}(z)$ and thus of $\Q(z)$ can also contain a polynomial in $z$. To determine its degree, let us first note that
\begin{equation}\label{Eq:StrongAsymptotics}
\Gamma\left(\frac{1}{u},x\right) = O(u) \;\;\;\;\; \text{ and } \;\;\;\;\; \vp\left(\frac{1}{u}\right) = O(u^{M-M_f+2}),
\end{equation}
as one can show from Equations \eqref{Eq:S}, \eqref{Eq:DefMZeros} and \eqref{Eq:TwistZeros}. Thus, $u^{M-M_f}\Q(1/u)$ is regular at $u=0$. The partial fraction decomposition of $\Q$ then contains a polynomial of degree $M-M_f$. We shall now distinguish the two cases (i) and (ii), following the denominations of the previous paragraph.

\subsubsection[Case (i): $\vp(z)\dd z$ non-zero at infinity]{Case (i): $\bm{\vp(z)\dd z}$ non-zero at infinity}
\label{SubSubSec:CaseI}

The case (i) corresponds to $M_f=M$. The polynomial in the partial fraction decomposition of $\Q(z)$ is then of order $0$, \textit{i.e.} only contains a constant term. We thus have
\begin{equation}\label{Eq:DESQi}
\Q(z) =  \sum_{i=1}^M \frac{\Q_i}{z-\ze_i} - \sum_{\alpha\in\Si} \sum_{p=0}^{m_\alpha-1} \frac{\Dc\alpha p}{(z-z_\alpha)^{p+1}}  + \mathcal{S}_{0},
\end{equation}
for some quadratic charge $\mathcal{S}_0\in\Ac$. Using Equation \eqref{Eq:P}, we then find that the spectral parameter dependent Hamiltonian \eqref{Eq:HamSpec} satisfies
\begin{equation}\label{Eq:DSEHamCaseI}
\frac{\Hc(z)}{\vp(z)} = - \sum_{i=1}^M \frac{\Q_i}{z-\ze_i} - \mathcal{S}_0.
\end{equation}
In particular, we see that any quantity extracted linearly from $\Hc(z)$ is a linear combination of the $\Q_i$'s and $\mathcal{S}_0$.\\

The equations above give the ``strong'' partial fraction decompositions of $\Q(z)$ and $\Hc(z)/\vp(z)$, in the sense that these are true even without imposing the constraint \eqref{Eq:Constraint}. Let us now see what happens when we consider weak equalities, \textit{i.e.} when we impose the constraint. It is clear from Equation \eqref{Eq:C} that the 1-form $\Gamma(z,x)\dd z$ is weakly regular at $z=\infty$, hence
\begin{equation*}
\Gamma\left(\frac{1}{u},x \right) \approx O(u^2) \;\;\;\;\; \text{and} \;\;\;\;\; \vp\left(\frac{1}{u}\right) = O(u^2),
\end{equation*}
where the second equality comes from specifying Equation \eqref{Eq:StrongAsymptotics} to $M_f=M$. We then deduce that $\Q\left(\frac{1}{u}\right) \approx O(u^2)$. Thus, we get
\begin{equation}\label{Eq:WeaklyVanishFirstCase}
\mathcal{S}_0 \approx 0 \;\;\;\;\; \text{ and } \;\;\;\;\; \res_{z=\infty} \Q(z) \dd z \approx 0.
\end{equation}
In particular, any quantity linearly extracted from $\Hc(z)$ is then, at least weakly, a linear combination of the $\Q_i$'s, $i\in\lbrace1,\cdots,M\rbrace$, only. Moreover, it is clear from Equation \eqref{Eq:DESQi} that the momentum \eqref{Eq:MomentumReal} of the theory can be rewritten as
\begin{equation*}
\Pc_\Ac = \sum_{\alpha\in\Si} \Dc\alpha 0 = \sum_{i=1}^M \Q_i + \res_{z=\infty} \Q(z)\dd z.
\end{equation*}
Thus, using the second weak equality in Equation \eqref{Eq:WeaklyVanishFirstCase}, we get the following simple (weak) expression of the momentum:
\begin{equation}\label{Eq:PZerosCaseI}
\Pc_\Ac \approx \sum_{i=1}^M \Q_i.
\end{equation}
This result, similar to Proposition 2.3 of~\nms for non-constrained affine Gaudin models, will play a key role in the analysis of the space-time symmetries of the theory (see Subsection \ref{SubSec:SpaceTime}).\\

Recall that in Subsection \ref{SubSubSec:Hamiltonian} we supposed that the first-class Hamiltonian $\Hc_0$ of the theory is given by a charge extracted linearly from $\Hc(z)$ plus a constant $\xi$ times the momentum of the theory. Thus, it is given by a linear combination of the $\Q_i$'s, at least weakly, and hence is strongly equal to such a linear combination plus terms proportional to the constraints. As explained in Appendix \ref{App:Constraints}, such terms can always be reabsorbed in a redefinition of the Lagrange multipliers $\mu_a(x)$ introduced in Subsection \ref{SubSubSec:Constraint}, without changing the total Hamiltonian and therefore the dynamic of the model. Thus, without loss of generality, we can always suppose that the first-class Hamiltonian $\Hc_0$ is given by
\begin{equation*}
\Hc_0 = \sum_{i=1}^M \epsilon_i \Q_i,
\end{equation*}
for some constant parameters $\bm\epsilon = (\epsilon_1,\cdots,\epsilon_M)$.

As each $\Q_i$, $i\in\lbrace1,\cdots,M\rbrace$, can be extracted linearly from $\Hc(z)$ (see Equation \eqref{Eq:DSEHamCaseI}), this choice of $\Hc_0$ corresponds to taking $\xi=0$, with $\xi$ the coefficient introduced in Subsection \ref{SubSubSec:Hamiltonian}. Let us briefly comment on the interpretation of this result. Equation \eqref{Eq:PZerosCaseI} shows that in the case (i) considered in this paragraph, the momentum of the theory can be extracted from $\Hc(z)$, at least weakly. Thus, the additional freedom of adding to the Hamiltonian a term proportional to the momentum that we introduced in Subsection \ref{SubSubSec:Hamiltonian} is redundant in this case, as it can always be reabsorbed into a redefinition of the Lagrange multipliers. Hence, one can always choose $\xi=0$. In fact, we introduced this additional parameter $\xi$ to be able to treat the case (ii) in full generality, as we shall see in the next paragraph.

\subsubsection[Case (ii): $\vp(z)\dd z$ with a simple zero at infinity]{Case (ii): $\bm{\vp(z)\dd z}$ with a simple zero at infinity}
\label{SubSubSec:CaseII}

Let us now consider the case (ii), corresponding to $M_f=M-1$. The polynomial in the partial fraction decomposition of $\Q(z)$ is then of degree $1$ and there exist quadratic charges $\mathcal{S}_0$ and $\mathcal{S}_1$ such that
\begin{equation}\label{Eq:DESQii}
\Q(z) = \sum_{i=1}^{M-1} \frac{\Q_i}{z-\ze_i} - \sum_{\alpha\in\Si} \sum_{p=0}^{m_\alpha-1} \frac{\Dc\alpha p}{(z-z_\alpha)^{p+1}}  + \mathcal{S}_{0} + \mathcal{S}_{1} z,
\end{equation}
and
\begin{equation*}
\frac{\Hc(z)}{\vp(z)} = - \sum_{i=1}^{M-1} \frac{\Q_i}{z-\ze_i} - \mathcal{S}_0 - \mathcal{S}_1 z.
\end{equation*}
Any quantity linearly extracted from $\Hc(z)$ is thus a linear combination of the $\Q_i$'s ($i\in\lbrace1,\cdots,M-1\rbrace$), $\mathcal{S}_0$ and $\mathcal{S}_1$.

As for the case (i) above, let us now study if some of these charges are weakly vanishing. For $M_f=M-1$, we have
\begin{equation*}
\Gamma\left(\frac{1}{u},x \right) \approx O(u^2) \;\;\;\;\; \text{and} \;\;\;\;\; \vp\left(\frac{1}{u}\right) = O(u^3),
\end{equation*}
so that $\Q\left(\frac{1}{u}\right)\approx O(u)$. Thus, we get
\begin{equation}\label{Eq:WeaklyVanishSecondCase}
\mathcal{S}_0 \approx 0 \;\;\;\;\; \text{ and } \;\;\;\;\; \mathcal{S}_1 \approx 0.
\end{equation}
In particular, any quantity linearly extracted from $\Hc(z)$ is weakly equal to a linear combination of the charges $\Q_i$'s, $i\in\lbrace1,\cdots,M-1\rbrace$.\\

Recall that in the case (ii), the 1-form $\vp(z)\dd z$ possesses, in addition to its finite zeroes $\ze_1,\cdots,\ze_{M-1}$, a simple zero at $z=\ze_M=\infty$. To keep the discussion as uniform as possible, we then introduce the charge
\begin{equation}\label{Eq:DefQM}
\Q_M = \res_{z=\infty} \Q(z) \dd z.
\end{equation}
Recall that the charge $\Q_i$ associated with a finite zero $\ze_i$ can be expressed easily in terms of the current $\Gamma(\ze_i,x)$ and the derivative $\vp'(\ze_i)$, as in Equation \eqref{Eq:QHZeros}. To obtain a similar expression for $\Q_M$, let us introduce the current $\B(x)$ such that
\begin{equation}\label{Eq:DefB}
\Gamma\left(\frac{1}{u},x\right) = u \,\Cc(x) - u^2 \,\B(x) +O(u^3).
\end{equation}
As $\Cc(x) \approx 0$, $\B(x)$ is equal weakly to the evaluation of the 1-form $\Gamma(z,x)\dd z$ at $z=\ze_M=\infty$ and is therefore the equivalent of $\Gamma(\ze_i,x)$ for a finite zero.

Recall the function $\chi(u)$ introduced in Equation \eqref{Eq:DefChi}. By assumption, in the case (ii) considered here, $\chi(u)$ has a simple zero at $u=0$ and we then have
\begin{equation}\label{Eq:AsymptoticTwistCaseII}
\vp\left(\frac{1}{u}\right) = -\chi'(0)\, u^3 + O(u^4),
\end{equation}
with $\chi'(0) \neq 0$. A study of the asymptotic of $\Q(z)$ at infinity from Equations \eqref{Eq:DefB} and \eqref{Eq:AsymptoticTwistCaseII} yields the following weak expression of $\Q_M$:
\begin{equation}\label{Eq:QM}
\Q_M \approx - \frac{1}{2\chi'(0)} \int_\D \dd x \, \kappa\bigl(\B(x),\B(x)\bigr),
\end{equation}
which has to be compared with Equation \eqref{Eq:QHZeros} for a finite zero $\ze_i$.\\

Let us now come back to the momentum $\Pc_\Ac$ of the theory. It is clear from Equation \eqref{Eq:DESQii} and the definition \eqref{Eq:DefQM} of $\Q_M$ that, in the case (ii), it can be expressed as
\begin{equation}\label{Eq:PZerosCaseII}
\Pc_\Ac = \sum_{i=1}^M \Q_i.
\end{equation}
This expression is the equivalent of Equation (2.52) of~\nms for non-constrained models and of Equation \eqref{Eq:PZerosCaseI} for constrained models in the case (i) (although, contrarily to Equation \eqref{Eq:PZerosCaseI}, it holds strongly).

As $\Q_1,\cdots,\Q_{M-1}$ are extracted from $\Hc(z)$, they are in involution between themselves and with the momentum $\Pc_\Ac$. Thus, considering the above expression for the latter, we get that all the charges $\Q_1,\cdots,\Q_M$ are in involution (hence Equation \eqref{Eq:QiInvolution} holds for any $i,j\in\lbrace 1,\cdots,M\rbrace$).\\

To end this subsection, let us now consider the first-class Hamiltonian $\Hc_0$ introduced in Subsection \ref{SubSubSec:Hamiltonian}. It was defined as the sum of a charge extracted linearly from $\Hc(z)$ (which as seen above is then a linear combination of $\Q_1,\cdots,\Q_{M-1}$, at least weakly) and of  a constant coefficient $\xi$ times the momentum $\Pc_\Ac$. Thus $\Hc_0$ is equal weakly to a linear combination of all the $\Q_i$'s, $i\in\lbrace 1,\cdots,M\rbrace$. As in the case (i) treated in the previous subsection, we can reabsorb any term in $\Hc_0$ which is proportional to the constraint in a redefinition of the Lagrange multipliers of the model, so that this equality becomes strong. Thus, we can suppose that
\begin{equation*}
\Hc_0 = \sum_{i=1}^M \epsilon_i\Q_i,
\end{equation*}
for some parameters $\bm\epsilon = (\epsilon_1,\cdots,\epsilon_M)$.

The charges $\Q_1,\cdots,\Q_{M-1}$ associated with finite zeroes can be extracted from $\Hc(z)$. Thus, in the expression above for $\Hc_0$, the only quantity which cannot be extracted from $\Hc(z)$ is $\Q_M$, which then has to come from the momentum \eqref{Eq:PZerosCaseII}. In particular, this means that the coefficient $\xi$ appearing in the definition of $\Hc_0$ in Subsection \ref{SubSubSec:Hamiltonian} is given in the case (ii) by $\xi=\epsilon_M$. As we shall see in Section \ref{Sec:SigmaModels} of this article\footnote{See also the treatment of integrable $\s$-model on $\mathbb{Z}_T$-cosets in~\cite{Vicedo:2017cge}.}, the description of certain integrable $\s$-models as constrained realisations of Affine Gaudin models will require having a non-zero coefficient $\epsilon_M$ (in fact, we shall see in Subsection \ref{SubSubSec:Lorentz} that $\epsilon_M$ cannot be zero for a model with relativistic invariance). This, and the fact that the introduction of the charge $\Q_M$ makes the treatment of both cases (i) and (ii) more uniform, leads us to think that the quantities $\Q_i$, $i\in\lbrace 1,\cdots,M\rbrace$, are in fact more fundamental objects of the theory than $\Hc(z)$.

\subsubsection{Summary}
\label{SubSubSec:Summary}

Let us summarise the results obtained above for the two cases (i) and (ii). In both cases, one expresses the momentum and the first-class Hamiltonian of the theory in terms of the quadratic charges $\Q_i$, $i\in\lbrace 1,\cdots,M \rbrace$:
\begin{equation}\label{Eq:PHZeros}
\Pc_\Ac \approx \sum_{i=1}^M \Q_i \;\;\;\;\; \text{ and } \;\;\;\;\; \Hc_0 = \sum_{i=1}^M \epsilon_i \Q_i,
\end{equation}
where the first equality also holds strongly in case (ii) but only weakly in case (i). The coefficients $\bm\epsilon = (\epsilon_1,\cdots,\epsilon_M)$ are thus defining parameters of the model.

In addition to $\eb$, the model is determined by the choice of the Takiff datum $\lt$ and of the positions of the sites $\pb$, which is equivalent to the choice of the twist function $\vp(z)$, as well as the choice of a realisation $\pi$ from the abstract Takiff algebra $\Tc_{\lt}$ to the physical algebra of observables $\Ac$. Following the notations of~\nm, we shall denote the corresponding model $\mathbb{M}^{\vp,\pi}_{\eb}$ and the corresponding Hamiltonian $\Hc^{\vp,\pi}_{\eb}$, in particular if we have to discuss several realisations of Affine Gaudin models at the same time.

\subsection{Lax pair and integrability}
\label{SubSec:Integrability}

\subsubsection{Zero curvature equation and Lax pair}

In the previous subsections, we defined a constrained model on $\Ac$ by specifying its total Hamiltonian \eqref{Eq:TotalHam}, formed by the first-class Hamiltonian $\Hc_0$ of Equation \eqref{Eq:PHZeros} and Lagrange multipliers associated with the constraints. We shall now prove that the dynamic defined by this Hamiltonian is integrable, by first exhibiting a Lax formulation of its equations of motion. Let us define the Lax matrix of the model as the ratio~\bg:
\begin{equation}\label{Eq:Lax}
\Lc(z,x) = \frac{\Gamma(z,x)}{\vp(z)}.
\end{equation}
The results of this paragraph can then be summarised in Theorem \ref{Thm:Lax} below. To state the theorem, let us first recall the function $\chi(u)$ defined in Equation \eqref{Eq:DefChi} and let us also push the asymptotic expansion \eqref{Eq:DefB} of $\Gamma(z,x)$ around infinity to the next order, introducing a current $\B_1(x)$ in addition to the current $\B(x)$:
\begin{equation}\label{Eq:DefB1}
\Gamma\left(\frac{1}{u},x\right) = u \,\Cc(x) - u^2 \,\B(x) - u^3 \, \B_1(x) + O(u^4).
\end{equation}

\begin{theorem}\label{Thm:Lax}
There exists a $\g$-valued field $\Mc(z,x)$ such that the time evolution of $\Lc(z,x)$ takes the form of a zero curvature equation on the Lax pair $\bigl(\Mc(z,x),\Lc(z,x)\bigr)$:
\begin{equation}\label{Eq:Zce}
\p_t \Lc(z,x) - \p_x \Mc(z,x) + \bigl[ \Mc(z,x), \Lc(z,x) \bigr] = 0,
\end{equation}
for all values of the spectral parameter $z$. Moreover:
\begin{enumerate}[(i)]
\item In the case where $\vp(z)\dd z$ is non-zero at infinity, the Lax pair is given by
\begin{equation*}
\Lc(z,x) \approx \sum_{i=1}^M \frac{1}{\vp'(\ze_i)} \frac{\Gamma(\ze_i,x)}{z-\ze_i} + \frac{\B(x)}{\chi(0)} \;\;\;\;\; \text{ and } \;\;\;\;\; \Mc(z,x) = \sum_{i=1}^M \frac{\epsilon_i}{\vp'(\ze_i)} \frac{\Gamma(\ze_i,x)}{z-\ze_i} + \mu(x).
\end{equation*}
\item In the case where $\vp(z)\dd z$ has a simple zero at infinity, the Lax pair is given by
\begin{eqnarray*}
 & & \displaystyle \Lc(z,x) \approx \sum_{i=1}^{M-1} \frac{1}{\vp'(\ze_i)} \frac{\Gamma(\ze_i,x)}{z-\ze_i} + \frac{\B(x)}{\chi'(0)}z + \frac{2\chi'(0)\B_1(x)-\chi''(0)\B(x)}{2\chi'(0)^2}  \\
 \text{and} & & \displaystyle \Mc(z,x) \approx \sum_{i=1}^{M-1} \frac{\epsilon_i}{\vp'(\ze_i)} \frac{\Gamma(\ze_i,x)}{z-\ze_i} + \frac{\epsilon_M\,\B(x)}{\chi'(0)}z + \epsilon_M \frac{2\chi'(0)\B_1(x)-\chi''(0)\B(x)}{2\chi'(0)^2} + \mu(x).
\end{eqnarray*}
\end{enumerate}
\end{theorem}

\begin{proof} ~\\
\underline{Generalities:} From the Poisson brackets \eqref{Eq:PbGaudin} and \eqref{Eq:PbPS}, one shows that the evolution of $\Lc(z,x)$ under the Hamiltonian flow of the spectral parameter dependent Hamiltonian $\Hc(z)$ takes the form of a zero curvature equation
\begin{equation}\label{Eq:ZceSpec}
\lbrace \Hc(w), \Lc(z,x) \rbrace - \p_x \Mcs wzx + \bigl[ \Mcs wzx, \Lc(z,x) \bigr] = 0,
\end{equation}
with
\begin{equation}\label{Eq:MSpec}
\Mcs wzx = \vp(w)\frac{\Lc(z,x)-\Lc(w,x)}{z-w}.
\end{equation}
As the charges $\Q_i$, $i\in\lbrace 1,\cdots,M_f\rbrace$, associated with finite zeroes are extracted linearly from $\Hc(z)$ in Equation \eqref{Eq:QHZeros}, their Hamiltonian flow on $\Lc(z,x)$ also generates a zero curvature equation. More precisely, noting that $\vp(\ze_i)=0$ and that $\vp(w)\Lc(w,x)=\Gamma(w,x)$ is regular at $w=\ze_i$, we get
\begin{equation}\label{Eq:ZceQi}
\lbrace Q_i, \Lc(z,x) \rbrace - \p_x \left( \frac{1}{\vp'(\ze_i)} \frac{\Gamma(\ze_i,x)}{z-\ze_i} \right) + \left[ \frac{1}{\vp'(\ze_i)} \frac{\Gamma(\ze_i,x)}{z-\ze_i}, \Lc(z,x) \right] = 0.
\end{equation}
The total Hamiltonian \eqref{Eq:TotalHam} of the model also contains a term proportional to the Lagrange multipliers. To compute its action on the Lax matrix $\Lc(z,x)$, we start from the Poisson bracket \eqref{Eq:PBSC} between the constraint and $\Gamma(z,x)$, from which we simply deduce that
\begin{equation}\label{Eq:PbCL}
\lbrace \Cc\ti{1}(y), \Lc\ti{2}(z,x) \rbrace = \bigl[ C\ti{12}, \Lc\ti{1}(z,x) \bigr] \delta_{xy} + C\ti{12} \delta'_{xy}.
\end{equation}
Thus the Lagrange multipliers enter the dynamics of the Lax matrix through:
\begin{equation}\label{Eq:ZceMult}
\left\lbrace \int_\D \dd y \, \kappa\bigl( \mu(y), \Cc(y) \bigr), \Lc(z,x) \right\rbrace - \p_x \mu(x) + \bigl[ \mu(x), \Lc(z,x) \bigr] \approx 0.
\end{equation}

\noi \underline{Case (i):} Let us now consider the case (i), where $M_f=M$ and thus all the charges $\Q_i$'s are associated with finite zeroes. The total Hamiltonian $\Hc$ defining the dynamic $\p_t \approx \lbrace \Hc, \cdot \rbrace$ is then
\begin{equation*}
\Hc = \sum_{i=1}^{M_f} \epsilon_i \Q_i + \int_\D \dd y \, \kappa\bigl( \mu(y), \Cc(y) \bigr),
\end{equation*}
and we obtain the zero curvature equation \eqref{Eq:Zce} by summing Equation \eqref{Eq:ZceQi} from $i=1$ to $i=M_f=M$ (with a coefficient $\epsilon_i$) and Equation \eqref{Eq:ZceMult}, yielding the following expression for the matrix $\Mc(z,x)$:
\begin{equation}\label{Eq:MCaseI}
\Mc(z,x) = \sum_{i=1}^M \frac{\epsilon_i}{\vp'(\ze_i)} \frac{\Gamma(\ze_i,x)}{z-\ze_i} + \mu(x).
\end{equation}
Let us now turn to the partial fraction decomposition of $\Lc(z,x)$. In the complex plane, it has poles only at the zeroes $\ze_1,\cdots,\ze_M$ of the twist function\footnote{Recall that we always suppose in this article that the highest level $\ls \alpha {m_\alpha-1}$ of each site is non-zero. Thus, $\Gamma(z,x)$ and $\vp(z)$ both have a pole of order $m_\alpha$ at $z=z_\alpha$, so that $\Lc(z,x)$ is regular at $z=z_\alpha$ (see Subsection \ref{SubSubSec:Takiff}).}. Moreover, as we supposed the zeroes of the twist function to be simple, these poles are simple. The corresponding residues are given by
\begin{equation}\label{Eq:ResLaxZero}
\res_{z=\ze_i} \Lc(z,x) = \lim_{z \to \ze_i} \frac{z-\ze_i}{\vp(z)} \, \Gamma(z,x) \, = \frac{\Gamma(\ze_i,x)}{\vp'(\ze_i)}.
\end{equation}
In addition to simple poles at the $\ze_i$'s, the partial fraction decomposition of $\Lc(z,x)$ contains a polynomial in $z$. To compute it, recall the asymptotic expansion \eqref{Eq:DefB1} of $\Gamma(z,x)$ around infinity. Similarly, as we are considering the case (i), the asymptotic expansion of the twist function around infinity is given by
\begin{equation*}
\vp\left(\frac{1}{u}\right) = -u^2 \chi(0) - u^3 \chi'(0) + O(u^4),
\end{equation*}
with $\chi(u)$ the function introduced in Equation \eqref{Eq:DefChi}. Combining these two expansions, one gets
\begin{equation*}
\Lc\left(\frac{1}{u},x\right) = \frac{1}{u} \frac{\Cc(x)}{\chi(0)} + \frac{\chi(0)\B(x)-\chi'(0)\Cc(x)}{\chi(0)^2} + O(u). 
\end{equation*}
This determines the polynomial part of the partial fraction decomposition of $\Lc(z,x)$. We then get
\begin{equation*}
\Lc(z,x) = \sum_{i=1}^M \frac{1}{\vp'(\ze_i)} \frac{\Gamma(\ze_i,x)}{z-\ze_i} + \frac{\chi(0)\B(x)-\chi'(0)\Cc(x)}{\chi(0)^2} + \frac{\Cc(x)}{\chi(0)} z,
\end{equation*}
hence, weakly,
\begin{equation*}
\Lc(z,x) \approx \sum_{i=1}^M \frac{1}{\vp'(\ze_i)} \frac{\Gamma(\ze_i,x)}{z-\ze_i} + \frac{\B(x)}{\chi(0)}.
\end{equation*}
This expression for $\Lc(z,x)$ bears a striking resemblance with the above expression for $\Mc(z,x)$, similarly to what was observed in~\nms for non-constrained model.\\

\noi \underline{Case (ii):} Let us now consider the case (ii), where $M_f=M-1$. We shall begin by finding the partial fraction decomposition of $\Lc(z,x)$. As in the case (i) treated above, its poles in the complex plane are only situated at the finite zeroes $\ze_1,\cdots,\ze_{M-1}$ of the twist function. Moreover, these poles are simple and the expression \eqref{Eq:ResLaxZero} of the corresponding residues that we found above for the case (i) is easily seen to apply also in the present case. There is still to determine the polynomial part of the partial fraction decomposition of $\Lc(z,x)$. For simplicity, we will consider here only the weak partial fraction decomposition. The asymptotic expansion \eqref{Eq:DefB1} of $\Gamma(z,x)$ around infinity becomes weakly
\begin{equation*}
\Gamma\left(\frac{1}{u},x\right) \approx -u^2 \left( \B(x) + \B_1(x)\, u + O(u^2) \right).
\end{equation*}
Moreover, one also has, in the case (ii) where $\chi(0)=0$ and $\chi'(0) \neq 0$:
\begin{equation*}
\vp\left(\frac{1}{u}\right) = - u^3 \chi'(0) \left( 1 + \frac{\chi''(0)}{2\chi'(0)}\,u + O(u^2) \right).
\end{equation*}
One then gets
\begin{equation*}
\Lc\left(\frac{1}{u},x\right) \approx \frac{1}{u} \frac{\B(x)}{\chi'(0)}  + \frac{2\chi'(0)\B_1(x)-\chi''(0)\B(x)}{2\chi'(0)^2} + O(u)
\end{equation*}
Thus, the partial fraction decomposition of $\Lc(z,x)$ is given weakly by
\begin{equation*}
\Lc(z,x) \approx \sum_{i=1}^{M-1} \frac{1}{\vp'(\ze_i)} \frac{\Gamma(\ze_i,x)}{z-\ze_i} + \frac{2\chi'(0)\B_1(x)-\chi''(0)\B(x)}{2\chi'(0)^2} + \frac{\B(x)}{\chi'(0)}z.
\end{equation*}

Let us now turn to the computation of $\Mc(z,x)$. In the case (ii) considered here, only the first $M-1$ charges $\Q_i$ obey the zero curvature equation \eqref{Eq:ZceQi} so we still have to study the Hamiltonian flow of $\Q_M$ on $\Lc(z,x)$. For that, we will use Equation \eqref{Eq:PZerosCaseII} rewritten as
\begin{equation*}
\Q_M = \Pc_\Ac - \sum_{i=1}^{M-1} \Q_i.
\end{equation*}
Indeed, as observed above, each $\Q_i$, $i\in\lbrace 1,\cdots,M-1\rbrace$ generates a zero curvature equation on $\Lc(z,x)$ and the momentum $\Pc_\Ac$ also generates a trivial zero curvature equation on it:
\begin{equation*}
\bigl\lbrace \Pc_\Ac, \Lc(z,x) \bigr\rbrace - \p_x \Lc(z,x) + \bigl[ \Lc(z,x), \Lc(z,x) \bigr] = 0.
\end{equation*}
We then get
\begin{equation*}
\bigl\lbrace \Q_M, \Lc(z,x) \bigr\rbrace - \p_x \left( \Lc(z,x) - \sum_{i=1}^{M-1} \frac{1}{\vp'(\ze_i)} \frac{\Gamma(\ze_i,x)}{z-\ze_i} \right) + \left[ \Lc(z,x) - \sum_{i=1}^{M-1} \frac{1}{\vp'(\ze_i)} \frac{\Gamma(\ze_i,x)}{z-\ze_i}, \Lc(z,x) \right] = 0.
\end{equation*}
Combining this equation with the zero curvature equations \eqref{Eq:ZceQi} and \eqref{Eq:ZceMult}, we finally get the zero curvature equation \eqref{Eq:Zce} generated by the total Hamiltonian
\begin{equation*}
\Hc = \sum_{i=1}^{M-1} \epsilon_i\Q_i + \epsilon_M \Q_M + \int_\D \dd y \, \kappa\bigl( \mu(y), \Cc(y) \bigr),
\end{equation*}
with the temporal component of the Lax pair reading
\begin{equation}\label{Eq:MCaseII}
\Mc(z,x) = \sum_{i=1}^{M-1} \frac{\epsilon_i}{\vp'(\ze_i)} \frac{\Gamma(\ze_i,x)}{z-\ze_i} + \epsilon_M \left( \Lc(z,x) - \sum_{i=1}^{M-1} \frac{1}{\vp'(\ze_i)} \frac{\Gamma(\ze_i,x)}{z-\ze_i} \right) + \mu(x).
\end{equation}
Using the weak expression found earlier for $\Lc(z,x)$, we then get
\begin{equation*}
\Mc(z,x) \approx \sum_{i=1}^{M-1} \frac{\epsilon_i}{\vp'(\ze_i)} \frac{\Gamma(\ze_i,x)}{z-\ze_i} + \frac{\epsilon_M\,\B(x)}{\chi'(0)}z + \epsilon_M \frac{2\chi'(0)\B_1(x)-\chi''(0)\B(x)}{2\chi'(0)^2} + \mu(x),
\end{equation*}
which ends the proof of the theorem.
\end{proof}

\subsubsection{Gauge transformation of the Lax pair}

The following proposition gives the transformation of the Lax pair $\bigl(\Lc(z,x),\Mc(z,x)\bigr)$ of the model under a gauge transformation.

\begin{proposition}\label{Prop:GaugeLax}
Under a gauge transformation with local parameter $h(x,t) \in G_0$, the Lax pair transforms as
\begin{eqnarray*}
\Lc(z,x) &\longmapsto&  h^{-1}(x,t) \Lc(z,x) h(x,t) + h^{-1}(x,t)\p_x h(x,t), \\
\Mc(z,x) &\longmapsto& h^{-1}(x,t) \Mc(z,x) h(x,t) +  h^{-1}(x,t)\p_t h(x,t).
\end{eqnarray*}
\end{proposition}

Before proving the proposition, let us first comment on its interpretation. One recognises in the transformation a formal gauge transformation $\bigl(\Lc^h(z,x),\Mc^h(z,x)\bigr)$ of the Lax pair $\bigl(\Lc(z,x),\Mc(z,x)\bigr)$. Although this formal gauge transformation coincides in the present case with the ``physical'' gauge transformation generated by the constraint, their origins are different. The formal gauge transformation is a common property of all integrable two-dimensional field theories, which encodes the non-uniqueness of the choice of a Lax pair: indeed, if the pair $\bigl(\Lc(z,x),\Mc(z,x)\bigr)$ satisfies a zero curvature equation, then so does the transformed pair $\bigl(\Lc^h(z,x),\Mc^h(z,x)\bigr)$, for any $h(x,t)\in G_0$. In the present case, this freedom in the choice of the Lax pair coincides with a change of physical gauge. In particular, this gives an alternative proof that the equations of motion of the model are gauge invariant, as they can be recast as a zero curvature equation, which by construction is invariant under formal gauge transformations. Finally, let us note that formal gauge transformations contain a bigger class of transformations than the physical gauge ones, as one can also consider a formal gauge transformation by a function $h(z,x,t) \in G_0$ which depends also on the spectral parameter $z$. Let us now prove Proposition \ref{Prop:GaugeLax}.

\begin{proof}
The transformation of the spatial component $\Lc(z,x)$ is a direct consequence of its definition \eqref{Eq:Lax} and the transformation \eqref{Eq:GaugeTransfGamma} of the Gaudin Lax matrix $\Gamma(z,x)$. Let us now consider a finite zero $\ze_i$, $i\in\lbrace 1,\cdots,M_f \rbrace$: from Equation \eqref{Eq:GaugeTransfGamma}, we find that the current $\Gamma(\ze_i,x)$ transforms covariantly:
\begin{equation}\label{Eq:GaugeGammaZero}
\Gamma(\ze_i,x) \longmapsto h^{-1}(x,t) \Gamma(\ze_i,x) h(x,t).
\end{equation}
Let us now distinguish the cases (i) and (ii).\\

\noi \underline{Case (i):} In case (i), the temporal component $\Mc(z,x)$ of the Lax pair is given by Equation \eqref{Eq:MCaseI} in terms of the currents $\Gamma(\ze_i,x)$ and the Lagrange multiplier. We determined the gauge transformation of $\Gamma(\ze_i,x)$ just above and of the Lagrange multiplier in Equation \eqref{Eq:GaugeTransfMult}. Recalling from Subsection \ref{SubSubSec:CaseI} that in the case (i), the coefficient $\xi$ appearing in the transformation of the Lagrange multiplier is equal to $0$, we obtain the transformation of $\Mc(z,x)$ announced in the proposition.\\

\noi \underline{Case (ii):} Let us now turn to the case (ii), for which the temporal component $\Mc(z,x)$ of the Lax pair is given by Equation \eqref{Eq:MCaseII}. The first sum in Equation \eqref{Eq:MCaseII} transforms covariantly under a gauge transformation, as can be seen from the above equation for the transformation of $\Gamma(\ze_i,x)$. Moreover, combining this equation and the transformation of $\Lc(z,x)$, we see that the second term in Equation \eqref{Eq:MCaseII} transforms as
\begin{equation*}
\epsilon_M \left( \Lc(z) - \sum_{i=1}^{M-1} \frac{1}{\vp'(\ze_i)} \frac{\Gamma(\ze_i)}{z-\ze_i} \right) \longmapsto \epsilon_M h^{-1}\left( \Lc(z) - \sum_{i=1}^{M-1} \frac{1}{\vp'(\ze_i)} \frac{\Gamma(\ze_i)}{z-\ze_i} \right) h + \epsilon_M\, h^{-1} \p_x h.
\end{equation*}
Finally, let us recall from Subsection \ref{SubSubSec:CaseII} that in the case (ii), the parameter $\xi$ appearing in the transformation \eqref{Eq:GaugeTransfMult} of the Lagrange multiplier is equal to $\epsilon_M$. Thus, in this case, it transforms as
\begin{equation*}
\mu \longmapsto h^{-1}\mu\, h + h^{-1} \p_t h - \epsilon_M \, h^{-1} \p_x h.
\end{equation*}
In particular, one sees that in the total transformation of $\Mc(z,x)$, the terms $h^{-1}\p_x h$ present in the two equations above cancel, yielding the expected result.
\end{proof}

\subsubsection{Maillet bracket and integrability}
\label{SubSubSec:Maillet}

\paragraph{Non-local conserved charges.} The existence of a Lax pair of the model guaranteed by Theorem \ref{Thm:Lax} implies the existence of an infinite number of non-local conserved charges extracted as traces of powers\footnote{We suppose here that we work in a matrix representation of the algebra $\g_0$ to make sense of these charges} of the monodromy matrix of the model
\begin{equation*}
\mathcal{T}(z) = \Pexp \left( - \int_\D \dd x\; \Lc(z,x) \right).
\end{equation*}
A natural question at this point is whether these charges are gauge invariant. Recall from Proposition \ref{Prop:GaugeLax} that gauge transformations act on the Lax matrix by formal gauge transformations. It is a standard result about path-ordered exponentials that the monodromy matrix then transforms as
\begin{equation*}
\mathcal{T}(z) \longmapsto h(x_2,t)^{-1} \,\mathcal{T}(z) \,h(x_1,t),
\end{equation*}
where $x_1$ and $x_2$ are the boundary points of the spatial domain $\D$. If we suppose that $\D$ is the circle $\mathbb{S}^1$, then $x_1=0$ and $x_2=2\pi$ and we restrict to gauge parameters $h(x,t)$ periodic in $x$. Thus, gauge transformations act on the monodromy matrix as conjugacy by $h(0,t)=h(2\pi,t)$, and the conserved charges extracted as traces of powers of $\mathcal{T}(z)$ are then gauge-invariant. If we suppose that $\D$ is the real line $\mathbb{R}$, then $x_1=-\infty$ and $x_2=+\infty$ and we restrict to gauge parameters $h(x,t)$ without monodromy. Thus, the gauge transformation of the monodromy matrix takes the form of a conjugacy by $h(-\infty,t)=h(+\infty,t)$ and the non-local conserved charges are also gauge-invariant in this case.

\paragraph{Maillet bracket.} To prove the integrability of the model, we shall show that the conserved charges extracted from the monodromy matrix in the previous paragraph are in involution. Indeed, it is a standard result~\cite{Maillet:1985fn,Maillet:1985ek} that this is the case if the Poisson bracket of the Lax matrix takes the form of a non-ultralocal Maillet bracket:
\begin{align}\label{Eq:PBR}
\hspace{-60pt}\bigl\lbrace \Lc\ti{1}(z,x), \Lc\ti{2}(w,y) \bigr\rbrace
&= \bigl[ \Rc\ti{12}(z,w), \Lc\ti{1}(z,x) \bigr] \delta_{xy} - \bigl[ \Rc\ti{21}(w,z), \Lc\ti{2}(w,x) \bigr] \delta_{xy}\\
& \hspace{40pt} - \left( \Rc\ti{12}(z,w) + \Rc\ti{21}(w,z) \right) \delta'_{xy}, \notag
\end{align}
for some $\g\otimes\g$-valued matrix $\Rc(z,w)$ called the $\Rc$-matrix. Starting from the definition \eqref{Eq:Lax} of the Lax matrix in terms of $\Gamma(z,x)$ and the Poisson bracket \eqref{Eq:PbGaudin} of the latter, one clearly sees that the Lax matrix indeed satisfies the above Maillet bracket (note that this bracket is strong, in the sense that it holds even without having to impose the constraint $\Cc(x)\approx 0$), with the \textit{$\Rc$-matrix}
\begin{equation}\label{Eq:RMat}
\Rc\ti{12}(z,w) = \frac{C\ti{12}}{w-z} \vp(w)^{-1}.
\end{equation}
This then proves the integrability of the model constructed here as a constrained realisation of an Affine Gaudin model. Moreover, the particular form \eqref{Eq:RMat} of the $\Rc$-matrix makes the model part of a more restricted class of integrable field theories, called models with twist function~\cite{Maillet:1985ec, Reyman:1988sf, Sevostyanov:1995hd, Vicedo:2010qd} (see also~\cite{Lacroix:2018njs}), which possess certain additional interesting properties, among which is the existence of an infinite integrable hierarchy of local charges, as we shall briefly recall in the next paragraph.

\paragraph{Local integrable hierarchy.} It was proven in~\cite{Lacroix:2017isl} that models with twist function possess, in addition to the non-local charges extracted from the monodromy, an infinite number of local charges in involution, forming so-called integrable hierarchies. More precisely, there will exist such an hierarchy for each zero of the twist function.

For a finite zero $\ze_i$, $i\in\lbrace 1,\cdots,M_f \rbrace$, the local charges forming this hierarchy take the form
\begin{equation}\label{Eq:Hierarchy}
\Q^d_i = -\frac{1}{(d+1) \, \vp'(\ze_i)} \int_\D \dd x \; \Phi_d\bigl( \Sg(\ze_i,x) \bigr).
\end{equation}
These charges are labelled by the so-called affine exponents $d\in E$ of the Lie algebra $\g$, which form an infinite subset of $\mathbb{Z}_{\geq 1}$. For each $d\in E$, $\Phi_d$ is then a well-chosen~\cite{Evans:1999mj} invariant polynomial on $\g$ of degree $d+1$. In particular, $1$ is always an exponent in $E$, with corresponding polynomial $\Phi_1(\cdot) = \kappa(\cdot,\cdot)$: the corresponding charge $\Q_i^1$ then coincides with the quadratic charge $\Q_i$, according to Equation \eqref{Eq:QHZeros}. The local charges $\left( \Q_i^d \right)^{d\in E}_{i\in\lbrace 1,\cdots,M_f\rbrace}$ are pairwise in involution, as proven in~\cite{Lacroix:2017isl}.

Recall that in this article, we distinguish two cases (i) and (ii), depending whether the infinity is a zero of the twist function or not. Let us consider the case (ii), where $\ze_M=\infty$ is a zero of the twist function. As shown in~\cite{Lacroix:2017isl}, one can also associate an infinite integrable hierarchy $\Q_M^d$, $d\in E$, to this zero. The local charges in this hierarchy are then given, at least weakly, by
\begin{equation}\label{Eq:HierarchyInfinity}
\Q^d_M \approx -\frac{1}{(d+1) \, \chi'(0)} \int_\D \dd x \; \Phi_d\bigl( \B(x) \bigr),
\end{equation}
where $\chi$ was defined in Equation \eqref{Eq:DefChi} and $\B(x)$, introduced in Equation \eqref{Eq:DefB}, can be seen as the evaluation (weakly) of the 1-form $\Gamma(z,x)\dd z$ at $z=\ze_M=\infty$. In particular, the quadratic charge $\Q^1_M$ coincide weakly with the charge $\Q_M$ introduced in Subsection \ref{SubSubSec:CaseII} (see Equation \eqref{Eq:QM}). The charges $\Q_M^d$, $d\in E$, pairwise Poisson commute. Moreover, they are also weakly in involution with the charges $\Q_i^d$, $i\in\lbrace 1,\cdots,M_f\rbrace$, associated with the finite zeroes\footnote{Let us note that the quadratic charge $\Q_M$, which is weakly equal to $\Q^1_M$, actually Poisson commutes strongly with the other $\Q_i$'s. It would be interesting to see if for higher exponents $d>1$, one can find a correction of the charge $\Q^d_M$ by terms proportional to the constraint which would make the charge strongly in involution with all other $\Q_i^d$'s.}.

Let us end this paragraph by mentioning that all the local charges $\Q_i^d$ ($i\in\lbrace 1,\cdots,M\rbrace$, $d\in E$) are gauge-invariant (see Theorem 5.3 of~\cite{Lacroix:2017isl}). For charges \eqref{Eq:Hierarchy} associated with finite zeroes, this can be easily understood from the facts that the polynomials $\Phi_d$ are invariant under conjugacy and that the currents $\Gamma(\ze_i,x)$ transform covariantly under a gauge transformation, as seen in Equation \eqref{Eq:GaugeGammaZero}. Similarly, one can show that in the case (ii), the current $\B(x)$ also transforms covariantly, proving the gauge-invariance of the corresponding charges \eqref{Eq:HierarchyInfinity}.

\subsection{Space-time symmetries}
\label{SubSec:SpaceTime}

\subsubsection{Energy-momentum tensor}

The model defined in the previous subsections is invariant under space-time translations. The conserved quantities associated with these symmetries are the momentum $\Pc_\Ac$ and the Hamiltonian $\Hc_0$ of the model\footnote{When discussing conserved quantities, one can consider the first-class Hamiltonian $\Hc_0$ instead of the total Hamiltonian $\Hc$. Indeed, the two are weakly equal and thus are equal on-shell, as the constraint vanishes on physical solutions of the equations of motion.}. Let us recall from Equation \eqref{Eq:PHZeros} that $\Pc_\Ac$ and $\Hc_0$ have a simple expression in terms of the charges $\Q_i$, $i\in\lbrace 1,\cdots,M\rbrace$. Let us define the densities corresponding to these charges:
\begin{equation*}
q_i(x) = -\res_{z=\ze_i}  \frac{\kappa\bigl( \Gamma(z,x), \Gamma(z,x) \bigr)}{2\vp(z)} \dd z,
\end{equation*}
such that
\begin{equation*}
\Q_i = \int_\D \dd x \; q_i(x).
\end{equation*}
The main result of this subsection is the following Proposition (which is the equivalent of Proposition 2.4 in~\nms for non-constrained model).

\begin{proposition}\label{Prop:EnergyMomentum}
The components of the energy-momentum tensor of the model are
\begin{equation}
\Te01(x) \approx \sum_{i=1}^M q_i(x), \;\;\;\;\; \Te00(x) \approx -\Te11(x) \approx \sum_{i=1}^M \epsilon_i\, q_i(x) \;\;\;\;\; \text{and} \;\;\;\;\; \Te10(x) \approx -\sum_{i=1}^M \epsilon_i^2\, q_i(x),
\end{equation}
where the indices 0 and 1 correspond respectively to the time direction and the space direction.
\end{proposition}

\begin{proof}
The energy-momentum tensor is the conserved current associated with space-time translations, which thus satisfies a local conservation equation
\begin{equation}\label{Eq:LocalEMTensor}
\p_\mu \Te\mu\nu = \p_t \Te0\nu + \p_x \Te1\nu = 0.
\end{equation}
In this equation, the index $\nu=0,1$ is free. The choice $\nu=0$ gives the local conservation equation associated with the time translation symmetry. The corresponding conserved charge is the Hamiltonian:
\begin{equation*}
\Hc_0 \approx \int_\D \dd x \, \Te 00(x).
\end{equation*}
Similarly, the choice $\nu=1$ corresponds to the space translation symmetry, with conserved charge the momentum:
\begin{equation*}
\Pc_\Ac \approx \int_\D \dd x \, \Te 01(x).
\end{equation*}
The expected expressions of the components $\Te 0\nu(x)$ then directly follow from the expression \eqref{Eq:PHZeros} of the Hamiltonian and the momentum in terms of the charges $\Q_i$, whose densities are the $q_i(x)$'s.\\

The local conservation equation \eqref{Eq:LocalEMTensor} of the energy-momentum tensor serves as a definition of the components $\Te1\nu$. To compute these components, one has to find an explicit expression for $\p_t \Te0\nu$ and write is a space derivative $-\p_x \Te1\nu$. To do that, we will need the time evolution of the densities $q_i(x)$. As these densities are first-class (\textit{i.e.} gauge-invariant), their time evolution is governed by the Hamiltonian flow of $\Hc_0$ only:
\begin{equation*}
\p_t q_i(x) \approx \lbrace \Hc_0, q_i(x) \rbrace + \int_\D \dd y \; \kappa\bigl( \mu(y), \lbrace \Cc(y), q_i(x) \rbrace \bigr) \approx \lbrace \Hc_0, q_i(x) \rbrace.
\end{equation*}
Thus, we find
\begin{equation*}
\p_t q_i(x) \approx - \sum_{j=1}^M \epsilon_j \int_\D \dd y \; \bigl\lbrace q_i(x), q_j(y) \rbrace.
\end{equation*}
A direct but slightly lengthy computation from the Poisson bracket \eqref{Eq:PbGaudin} gives\footnote{This computation is very similar to the one yielding Equation (2.56) of~\nms for unconstrained model, the only difference being that in the case (ii), one has to take into account that the zero $\ze_M$ is at infinity.}
\begin{equation}\label{Eq:PbDensities}
\big\lbrace q_i(x), q_j(y) \big\rbrace = -\delta_{ij} \bigl( \p_x q_i(x) \delta_{xy} + 2 q_i(x) \delta'_{xy} \bigr).
\end{equation}
From this equation, we get
\begin{equation*}
\p_t q_i(x) \approx \epsilon_i \p_x q_i(x),
\end{equation*}
from which we deduce the announced expressions of $\Te10(x)$ and $\Te11(x)$.
\end{proof}

\subsubsection{Relativistic invariance}
\label{SubSubSec:Lorentz}

\paragraph{Lorentz symmetry.} Let us consider the Minkowski metric $\eta_{\mu\nu}$ with signature $(+1,-1)$ and the energy-momentum tensor $T_{\mu\nu} =\eta_{\nu\rho}\Te\mu\rho$ with lowered indices. From Proposition \ref{Prop:EnergyMomentum}, one gets
\beqz
T_{01}(x) \approx \sum_{i=1}^M q_i(x), \;\;\;\;\; T_{00}(x) \approx T_{11}(x) \approx \sum_{i=1}^M \epsilon_i\, q_i(x) \;\;\;\;\; \text{and} \;\;\;\;\; T_{10}(x) \approx \sum_{i=1}^M \epsilon_i^2\, q_i(x).
\eeqz
It is a classical result from field theory that the model is relativistic invariant, \textit{i.e.} invariant under the Lorentz group of isometries of the Minkoswki metric $\eta_{\mu\nu}$, if the tensor $T_{\mu\nu}$ is symmetric, hence if $T_{01}=T_{10}$. It is clear from the expressions of $T_{01}$ and $T_{10}$ above that this is the case if
\begin{equation}\label{Eq:Relat}
\epsilon_i^2=1, \;\;\;\;\; \forall\, i\in\lbrace 1,\cdots,M \rbrace,
\end{equation}
\textit{i.e.} if each coefficient $\epsilon_i$ is either $+1$ or $-1$. This quite simple condition for relativistic invariance will be a key element in Section \ref{Sec:SigmaModels}, where we will construct relativistic integrable $\s$-models as constrained Affine Gaudin models.

For the rest of this subsection, we will suppose that the condition \eqref{Eq:Relat} ensuring relativistic invariance of the model is satisfied. Following the notations of~\nm, we then introduce the following subsets of $\lbrace 1,\cdots,M\rbrace$:
\begin{equation*}
I_\pm = \bigl\lbrace i\in\lbrace 1,\cdots,M\rbrace \, \bigl| \, \epsilon_i = \pm 1 \bigr\rbrace,
\end{equation*}
which form a partition of $\lbrace 1,\cdots,M \rbrace = I_+ \sqcup I_-$. Recall from Subsection \ref{SubSec:Zeroes} the number $M_f$ of finite zeroes, which is equal to $M$ in the case (i) and equal to $M-1$ in the case (ii). We will also need the subsets
\begin{equation*}
I_{f,\pm} = \bigl\lbrace i\in\lbrace 1,\cdots,M_f\rbrace \, \bigl| \, \epsilon_i = \pm 1 \bigr\rbrace
\end{equation*}
of $\lbrace 1,\cdots, M_f \rbrace$.

\paragraph{Light-cone Lax pair.} As we are considering a model with Lorentz symmetry, it is convenient to introduce the light-cone coordinates
\begin{equation*}
x_\pm = \frac{t\pm x}{2}
\end{equation*}
and the corresponding derivatives $\p_\pm = \p_t \pm \p_x$. Similarly, one can re-express the Lax pair $\bigl(\Lc(z),\Mc(z)\bigr)$ in light-cone components:
\begin{equation*}
\Lc_\pm(z) = \Mc(z) \pm \Lc(z).
\end{equation*}
The zero curvature equation \eqref{Eq:Zce} then takes the form
\begin{equation*}
\p_+ \Lc_-(z) - \p_- \Lc_+(z) + \bigl[ \Lc_+(z), \Lc_-(z) \bigr] = 0.
\end{equation*}
Theorem \ref{Thm:Lax}, which gave the explicit expression of the Lax pair $\bigl( \Lc(z), \Mc(z) \bigr)$, then gives the following corollary (similar to Equation (2.60) of~\nms for non-constrained model).

\begin{corollary}\label{Cor:LightConeLax}~
\begin{enumerate}[(i)]
\item In the case where $\vp(z)\dd z$ is non-zero at infinity, the light-cone Lax pair is given by
\begin{equation*}
\Lc_\pm (z) \approx \pm 2 \sum_{i \in I_\pm} \frac{1}{\vp'(\ze_i)} \frac{\Gamma(\ze_i)}{z-\ze_i} + \mu \pm \frac{\B}{\chi(0)}.
\end{equation*}
\item In the case where $\vp(z)\dd z$ has a simple zero at infinity, the Lax pair is given by
\begin{equation*}
\Lc_\pm(z) \approx \pm 2 \left( \sum_{i \in I_{f,\pm}} \frac{1}{\vp'(\ze_i)} \frac{\Gamma(\ze_i)}{z-\ze_i} + \delta^\pm_M \frac{\B}{\chi'(0)}z \right) + \mu \pm 2 \delta^\pm_M \frac{2\chi'(0)\B_1-\chi''(0)\B}{2\chi'(0)^2},
\end{equation*}
where
\begin{equation*}
\delta^\pm_M = \left\lbrace \begin{array}{ll}
1 & \text{ if } M \in I_\pm, \\
0 & \text{ if } M \notin I_\pm.
\end{array} \right.
\end{equation*}
\end{enumerate}
\end{corollary}

Corollary \ref{Cor:LightConeLax} exhibits a very interesting property of the light-cone Lax pair $\Lc_\pm(z)$ concerning its dependence on the spectral parameter $z$, which is true only for relativistic models. Although both the temporal and spatial components $\Mc(z)$ and $\Lc(z)$ of the Lax pair possess poles at all the zeroes $\ze_i$, $i\in\lbrace 1,\cdots,M \rbrace$, of the model, the pole structure of the light-cone components $\Lc_\pm(z)$ is much simpler. Indeed, $\Lc_\pm(z)$ only has poles at the zeroes $\ze_i$ for $i$ in $I_\pm$ (where in the case (ii), we understand a pole at $\ze_M=\infty$ as a linear term in $z$). Let us also note here that both components $\Lc_+(z)$ and $\Lc_-(z)$ have a constant term independent of $z$.

\paragraph{Spin of the currents.} Let us end this subsection by briefly discussing the spin of the various currents appearing in the Lax pair $\Lc_\pm(z)$. For the case of non-constrained model treated in~\nm, it was shown in Proposition 2.5 that the current $\Gamma(\ze_i)$ associated with a zero $\ze_i$ has a definite spin equal to $\epsilon_i$, \textit{i.e.} has spin $+1$ or $-1$, depending on whether $i$ belongs to $I_+$ or $I_-$. The proof of this result in~\nms can be adapted to the present case for the currents $\Gamma(\ze_i)$ associated with finite zeroes $\ze_i$, $i\in\lbrace 1,\cdots,M_f \rbrace$, and, in the case (ii) where $\ze_M=\infty$ is a zero of the twist function, for the current $\B$, which plays the role of $\Gamma(\ze_M)$ (see previous subsections). For the sake of brevity, we will not enter into more details about this computation here. As for the case of non-constrained models, an interesting consequence of this result is the fact that the hierarchy \eqref{Eq:Hierarchy} of local charges associated with the zero $\ze_i$ (or the hierarchy \eqref{Eq:HierarchyInfinity} for the zero $\ze_M$ in the case (ii)) is composed of charges $\Q_i^d$ of either increasing or decreasing spin $\pm(d+1)$, depending whether $i$ belongs to $I_+$ or $I_-$.

\subsection{Reality conditions}
\label{SubSec:Reality}

Recall that the elementary bricks of the model defined above, the Takiff currents $\J\alpha p(x)$, satisfy the reality conditions \eqref{Eq:RealityJc}, keeping track of the choice of a real form $\g_0 = \lbrace X\in\g \, | \, \tau(X)=X \rbrace$ of the complex Lie algebra $\g$. Let us briefly discuss the consequences of these reality conditions.

\paragraph{Gaudin Lax matrix and twist function.} The reality conditions \eqref{Eq:RealityJc} of the Takiff currents $\J\alpha p(x)$ and \eqref{Eq:zReal} of the positions $z_\alpha$, translate into the following equivariance property of the Gaudin Lax matrix \eqref{Eq:S}:
\begin{equation}\label{Eq:RealGamma}
\tau\bigl( \Gamma(z,x) \bigr) = \Gamma( \overline{z},x ).
\end{equation}
Recall from the paragraph following Equation \eqref{Eq:Takiff} that the levels $\ls\alpha p$ satisfy a reality condition similar to the one \eqref{Eq:RealityJc} for the Takiff currents:
\begin{equation*}
\overline{ \ls\alpha p } = \ls\alpha p, \;\;\;\; \forall \, \alpha\in\Si_\rd \;\;\;\;\;\;\;\;\; \text{ and } \;\;\;\;\;\;\;\; \overline{ \ls\alpha p } = \ls{\bar\alpha} p, \;\;\;\; \forall \, \alpha\in\Si_\cd.
\end{equation*}
Thus, the twist function \eqref{Eq:Twist} also satisfies a simple equivariance condition under complex conjugation:
\begin{equation}\label{Eq:RealTwist}
\overline{\vp(z)} = \vp(\overline{z}).
\end{equation}
In particular, this implies two possibilities for the zeroes $\ze_i$ of the twist function:
\begin{itemize}
\item either $\ze_i$ is real, in which case the current $\Gamma(\ze_i,x)$ is real (\textit{i.e.} $\tau$-invariant) by Equation \eqref{Eq:RealGamma} ;
\item either $\ze_i$ is not real, in which case $\overline{\ze_i}$ is also a zero of the twist function and thus is equal to $\ze_{\bar i}$ for some $\bar{i}\in\lbrace 1,\cdots,M\rbrace$ and we then have $\tau\bigl(\Gamma(\ze_i,x)\bigr) = \Gamma(\ze_{\bar{i}},x)$.
\end{itemize}
Let us also note that the constraint $\Cc(x)$ and the currents $\B(x)$ and $\B_1(x)$, extracted from the asymptotic expansion \eqref{Eq:DefB1} of $\Gamma(z,x)$ around infinity, are all real, as can be seen by developing the reality condition \eqref{Eq:RealGamma} around infinity.

\paragraph{Hamiltonian and Lax pair.} Let us consider a real finite zero $\ze_i$ of the twist function. As observed above, the current $\Gamma(\ze_i,x)$ is real, and from the reality condition \eqref{Eq:RealTwist} of the twist function, one also gets that $\vp'(\ze_i)$ is real. Thus the quadratic charge $\Q_i$, expressed as in Equation \eqref{Eq:QHZeros}, is real. Let us now consider a pair $\ze_i$ and $\ze_{\bar i}$ of complex conjugate zeroes. The currents $\Gamma(\ze_i,x)$ and $\Gamma(\ze_{\bar i},x)$ are then conjugate under $\tau$ and one has $\overline{\vp'(\ze_i)}=\vp'(\ze_{\bar i})$. Thus, the quadratic charges $\Q_i$ and $\Q_{\bar i}$ are complex conjugate one of another. Finally, let us consider the case (ii), where $\ze_M=\infty$ is a (real) zero of the twist function. In this case, one checks that $\chi'(0)$ is real from the definition \eqref{Eq:DefChi} of $\chi$ and the reality condition \eqref{Eq:RealTwist}. As the current $\B(x)$ is also real, one gets that the charge $\Q_M$ associated with the zero $\ze_M=\infty$ and given by Equation \eqref{Eq:QM} is real.

Recall that the first-class Hamiltonian $\Hc_0$ of the model can be expressed in terms of the charges $\Q_i$ as in Equation \eqref{Eq:PHZeros}. Thus, it is real if and only if the coefficients $\epsilon_i$ satisfy the following reality condition
\begin{equation*}
\epsilon_i \in \R \;\; \text{ if } \;\; \ze_i \in \R \;\;\;\;\;\; \text{ and } \;\;\;\;\;\; \overline{\epsilon_i} = \epsilon_{\bar{i}} \;\; \text{ if } \;\; \ze_i \in \C \setminus \R.
\end{equation*}
If we consider a relativistic model, whose coefficient $\epsilon_i$ then satisfy condition \eqref{Eq:Relat}, we see that the above reality condition translates into the fact that if $\ze_i$ and $\ze_{\bar i}$ are pair of complex conjugate zeroes, then $\epsilon_i$ and $\epsilon_{\bar i}$ should be equal (and be both either $+1$ or $-1$). In other words, the labels $i$ and $\bar i$ of pairs of complex conjugate zeroes must belong to the same subset $I_+$ or $I_-$ of $\lbrace 1,\cdots,M\rbrace$.

From the reality conditions of the currents $\Gamma(\ze_i,x)$, $\B(x)$ and $\B_1(x)$ and of the coefficients $\epsilon_i$ above, one checks that the Lax pair of the model, given by Theorem \ref{Thm:Lax}, satisfy the following equivariance condition:
\begin{equation*}
\tau\bigl( \Lc(z,x) \bigr) = \Lc( \overline{z},x ) \;\;\;\;\; \text{ and } \;\;\;\;\; \tau\bigl( \Mc(z,x) \bigr) = \Mc( \overline{z},x ).
\end{equation*}
As a consistency check, one can verify the condition on the Lax matrix $\Lc(z,x)$ from its definition \eqref{Eq:Lax} in terms of $\Gamma(z,x)$ and $\vp(z)$ and their reality conditions \eqref{Eq:RealGamma} and \eqref{Eq:RealTwist}.

\paragraph{Positivity of the Hamiltonian.} For this paragraph only, we suppose that $\g_0$ is the compact real form of $\g$ and that all finite zeroes $\ze_i$ of the twist function are real. In this case, the corresponding currents $\Gamma(\ze_i,x)$ are valued in $\g_0$. Moreover, the restriction of the bilinear form $\kappa$ on the compact form is positive. Thus, the quantity $\kappa\bigl( \Gamma(\ze_i,x), \Gamma(\ze_i,x)\bigr)$ is real and positive and the charge $\Q_i$, given by Equation \eqref{Eq:QHZeros}, has the same sign as $-\vp'(\ze_i)$. Similarly, if we are in the case (ii) and $\ze_M=\infty$ is a zero of the twist function, the corresponding charge $\Q_M$, given by \eqref{Eq:QM}, has the same sign as $-\chi'(0)$.

If one choose the coefficient $\epsilon_i$ to be of the sign of $-\vp'(\ze_i)$ (and the coefficient $\epsilon_M$ to be of the sign of $-\chi'(0)$ if we are in the case (ii)), the Hamiltonian of the theory is then positive, and in particular bounded below. If one considers a relativistic model, where all coefficients $\epsilon_i$ are either $+1$ or $-1$, this choice for a positive Hamiltonian then reduces to
\begin{equation*}
\epsilon_i = - \text{sign} \bigl( \vp'(\ze_i) \bigr) \;\;\;\;\; \text{and} \;\;\;\;\; \epsilon_M = - \text{sign} \bigl( \chi'(0) \bigr) \;\; \text{ if } \;\; \ze_M = \infty.
\end{equation*}

\subsection{Change of spectral parameter}
\label{SubSec:ChangeSpec}

A key role in the construction of constrained realisations of local affine Gaudin models in the previous subsections was played by the spectral parameter $z$, which is an auxiliary parameter valued in the Riemann sphere $\mathbb{P}^1$. In this subsection, we discuss how appropriate changes of the spectral parameter lead to equivalent models, generalising the results of Subsection 2.3.2 of~\nms for non-constrained models.

\subsubsection{Transformation of the twist function and the Gaudin Lax matrix}

\paragraph{Mobius transformation.} A change of spectral parameter is encoded in a bijective transformation $f:\CP \rightarrow \CP$, mapping the ``old'' spectral parameter $z$ to a ``new'' spectral parameter $\zt=f(z)$. In the following paragraphs, we shall explain what condition the map $f$ should satisfy so that $\zt$ can be interpreted as the spectral parameter of a new constrained realisation of local affine Gaudin model. We will then show afterwards that this new model is equivalent to the one we started with.

The key ingredients of the local affine Gaudin models, such as the Gaudin Lax matrix $\Gamma(z,x)$ and the twist function $\vp(z)$, are defined as rational functions of $z$. Thus, if we want the map $f$ to preserve the structure of local affine Gaudin models, it should send rational functions to rational functions. It is a classical result that such (bijective) maps are given by Mobius transformations:
\begin{equation}\label{Eq:Mobius}
f : z \longmapsto \frac{az+b}{cz+d}, \;\;\;\; \text{ with } \;\;\;\; ad-bc \neq 0.
\end{equation}

Recall also that we are considering models with reality conditions (see Subsection \ref{SubSec:Reality}), which mostly take the form of equivariance properties with respect to the conjugacy of the spectral parameter. In order for the new model to also satisfy such reality conditions, we will suppose that $f$ is conjugacy equivariant:
\begin{equation}\label{Eq:fReal}
\overline{f(z)} = f\bigl( \overline{z} \bigr),
\end{equation} 
which is equivalent to requiring that the numbers $a$, $b$, $c$ and $d$ are reals.\\

Let us give two useful properties of the Mobius transformation \eqref{Eq:Mobius} that we shall need in the rest of this subsection. It is a bijection from $\CP$ to $\CP$ and its inverse is again a Mobius transformation:
\begin{equation}\label{Eq:fInv}
f^{-1}(\zt) = -\frac{d\,\zt-b}{c\,\zt-a}.
\end{equation}
Moreover, its derivative $f'$ satisfies:
\begin{equation}\label{Eq:fDer}
f'\bigl( f^{-1}(\zt) \bigr) = \frac{(c\,\zt-a)^2}{ad-bc}.
\end{equation}

\paragraph{Twist function and positions of the sites.} In the formal construction of local affine Gaudin models in~\bg, the twist function appears naturally as a 1-form $\vp(z)\dd z$. Thus, under a change of spectral parameter $z \mapsto \zt$, it should transform as
\begin{equation}\label{Eq:Twist1form}
\vp(z) \dd z = \vpt(\zt) \dd \zt.
\end{equation}
This equation serves as a definition of the new twist function $\vpt(\zt)$. More explicitly, one gets
\begin{equation*}
\vpt(\zt) = \frac{1}{f'\bigl( f^{-1}(\zt) \bigr)} \vp\bigl( f^{-1}(\zt) \bigr).
\end{equation*}
Using Equations \eqref{Eq:fInv} and \eqref{Eq:fDer}, we then have
\begin{equation}\label{Eq:NewTwist}
\vpt(\zt) = \frac{ad-bc}{(c\,\zt-a)^2} \; \vp\left( -\frac{d\,\zt-b}{c\,\zt-a} \right).
\end{equation}

Recall that the twist function of a local affine Gaudin model encodes the information about its sites, their levels and their positions. In particular the positions of the sites correspond to the poles of $\vpt(\zt)$. It is clear from Equation \eqref{Eq:Twist1form} that the poles of $\vpt(\zt)$ are the image under $f$ of the poles of $\vp(z)$. Thus, one can describe the sites of the new model with the same labels $\alpha\in\Si=\Si_{\rd} \sqcup \Si_{\cd} \sqcup \overline{\Si}_{\cd}$ and the positions of the sites in the new model are the image under $f$ of the positions in the initial model:
\begin{equation}\label{Eq:NewPositions}
\zt_\alpha = f(z_\alpha), \;\;\;\; \forall \alpha\in\Si.
\end{equation}
Note that thanks to the reality condition \eqref{Eq:fReal} on $f$, real sites $\alpha\in\Si_{\rd}$ have real positions $\zt_\alpha\in\R$ in the new model and pairs of conjugate complex sites $\alpha\in\Si_{\cd}$ and $\bar\alpha\in\overline{\Si}_{\cd}$ have complex conjugate positions $\zt_\alpha$ and $\zt_{\overline\alpha}$.

The local affine Gaudin models considered in the previous subsections only have sites at finite positions. Thus, we should restrict here to the transformations $f$ such that the positions \eqref{Eq:NewPositions} of the sites in the new model are finite. Concretely, this means that the coefficients $c$ and $d$ appearing in the definition \eqref{Eq:Mobius} of $f$ are such that
\begin{equation}\label{Eq:NewSitesFinite}
c z_\alpha + d \neq 0, \;\;\;\;\; \forall \, \alpha\in\Si.
\end{equation}
As none of the poles of $\vp(z)\dd z$ are sent to infinity under the Mobius transformation $f$, the new twist function  1-form $\vpt(\zt)\dd\zt$ is regular at $\zt=\infty$. In particular, this means that it satisfies the first-class condition
\begin{equation}\label{Eq:NewResidueInf}
\res_{\zt=\infty} \vpt(\zt) \dd \zt = 0,
\end{equation}
similar to the first-class condition \eqref{Eq:ResidueInf} on $\vp(z)\dd z$ that was used in Subsection \ref{SubSec:Gauge} to construct the gauge symmetry of the initial model. This will be a key feature in the rest of this subsection, as it will allow us to construct a gauge symmetry for the new model with twist function $\vpt(\zt)$.

To end the discussion about the pole structure of $\vpt(\zt)$, let us end by the following remark. As explained above, the new twist function $\vpt(\zt)$ should have poles only at the $\zt_\alpha = f(z_\alpha)$. If $c\neq 0$ (\textit{i.e.} if the Mobius transformation is not just the combination of a dilation and a translation), one could expect naively from Equation \eqref{Eq:NewTwist} that $\vpt(\zt)$ has a pole at $\zt=a/c$. However, one can show that $\vpt(\zt)$ is in fact regular at this point. Indeed, the point $\zt=a/c$ is the image under $f$ of the point $z=\infty$. Yet, let us recall that $\vp(z)$ satisfies the first-class condition \eqref{Eq:ResidueInf}, which implies the regularity of $\vp(z)\dd z$ at $z=\infty$. Thus, the new twist function $\vpt(\zt)$ is also regular at $\zt=f(\infty)=a/c$. This can be verified explicitly using the asymptotic expansion of $\vp(z)$ around infinity to prove that the term $\vp\bigl( -(d\,\zt-b)/(c\,\zt-a) \bigr)$ in Equation \eqref{Eq:NewTwist} evolves as $O\bigl((c\,\zt-a)^2\bigr)$ around $\zt=c/a$.\\

It is interesting to compare the situation treated here with the one of reference~\nm. In particular, recall that the local affine Gaudin models considered in~\nm, contrarily to the ones considered here, possess a constant term in the twist function, which corresponds to a double pole at infinity. To preserve this particular structure, the changes of spectral parameter studied in~\nms are restricted to Mobius transformations which fix the point at infinity, \textit{i.e.} to translations and dilations $f:z\mapsto az+b$. As the models considered here do not possess a site at infinity, one can consider the larger class of transformations described above: the Mobius transformations \eqref{Eq:Mobius}, which satisfy the additional condition \eqref{Eq:NewSitesFinite} ensuring that that the new model does not possess a site at infinity. It is worth mentioning here that in Section \ref{Sec:Gauging}, we will consider certain transformations outside of the ones of~\nms and of this subsection, which will allow us to relate the models considered in~\nms with the ones considered here.

\paragraph{Levels.} We have seen in the previous paragraph that the new twist function $\vpt(\zt)$ has poles at the $\zt_\alpha=f(z_\alpha)$, which thus correspond to the positions of the sites in the new model. Recall that in a local affine Gaudin model, the levels of the sites are the coefficients of the partial fraction decomposition of the twist function. The next step into the construction of a model with twist function $\vpt(\zt)$ is thus to determine the corresponding levels, which is done by the following lemma:

\begin{lemma}\label{Lem:NewLevels}
The partial fraction decomposition of $\vpt(\zt)$ is
\begin{equation*}
\vpt(\zt) = \sum_{\alpha\in\Si} \sum_{p=0}^{m_\alpha-1} \frac{\lst\alpha p}{(\zt-\zt_\alpha)^{p+1}},
\end{equation*}
with the levels
\begin{equation*}
\lst \alpha 0 = \ls \alpha 0 \;\;\;\;\;\; \text{ and } \;\;\;\;\;\; \lst \alpha p = \sum_{k=p}^{m_\alpha-1} {k-1 \choose p-1 } \frac{(-c)^{k-p}(ad-bc)^p}{(cz_\alpha+d)^{p+k}} \ls\alpha k, \;\; \text{ for } \; p>0.
\end{equation*}
\end{lemma}

Before proving the lemma, let us note that the levels $\ls \alpha 0$ of Takiff mode 0 are invariant under the change of spectral parameter. This is simply a consequence of the fact that the residues of a 1-form do not depend on the choice of coordinate:
\begin{equation*}
\ls \alpha 0 = \res_{z=z_\alpha} \vp(z) \dd z = \res_{\zt = \zt_\alpha} \vpt(\zt) \dd\zt = \lst \alpha 0 .
\end{equation*}
This observation allows one to check explicitly that the new twist function $\vpt(\zt)$ satisfies the first-class condition \eqref{Eq:NewResidueInf}, as:
\begin{equation*}
\res_{\zt=\infty} \vpt(\zt) \dd \zt = \sum_{\alpha\in\Si} \lst \alpha 0 = \sum_{\alpha\in\Si} \ls \alpha 0 = 0.
\end{equation*}

\begin{proof}
To prove Lemma \ref{Lem:NewLevels}, let us start from the partial fraction decomposition \eqref{Eq:Twist} of $\vp(z)$ and the explicit expression \eqref{Eq:NewTwist} of $\vpt(\zt)$, which together yield
\begin{equation*}
\vpt(\zt) = \frac{ad-bc}{(c\zt-a)^2} \sum_{\alpha\in\Si} \sum_{p=0}^{m_\alpha-1} \frac{(-1)^{p+1}\ls \alpha p}{\left( \displaystyle  \frac{d\zt-c}{c\zt-a} + z_\alpha \right)^{p+1}}.
\end{equation*}
Elementaty manipulations give
\begin{equation*}
\frac{d\zt-c}{c\zt-a} + z_\alpha = \frac{d+cz_\alpha}{c\zt-a} \bigl( \zt-\zt_\alpha \bigr),
\end{equation*}
and after reinsertion in the equation above (we suppose here that $c\neq 0$ and will comment on the case $c=0$ at the end of the proof):
\begin{equation}\label{Eq:NewTwistAux}
\vpt(\zt) = (ad-bc) \sum_{\alpha\in\Si} \sum_{p=0}^{m_\alpha-1} \frac{(-c)^{p-1}\ls\alpha p}{(cz_\alpha+d)^{p+1}} \frac{(\zt-a/c)^{p-1}}{(\zt-\zt_\alpha)^{p+1}}.
\end{equation}
We shall need the following direct identity:
\begin{equation*}
\zt_\alpha - a/c = - \frac{ad-bc}{c(cz_\alpha+d)}.
\end{equation*}
For $p=0$, we have
\begin{equation*}
\left. \frac{(\zt-a/c)^{p-1}}{(\zt-\zt_\alpha)^{p+1}} \right|_{p=0} = \frac{1}{\zt_\alpha-a/c} \left( \frac{1}{\zt-\zt_\alpha} - \frac{1}{\zt-a/c} \right) = -\frac{c(cz_\alpha+d)}{ad-bc} \left( \frac{1}{\zt-\zt_\alpha} - \frac{1}{\zt-a/c} \right).
\end{equation*}
For $p>0$, we develop $(\zt-a/c)^{p-1} = (\zt-\zt_\alpha+\zt_\alpha-a/c)^{p-1}$ with the Newton binomial identity, yielding after a few manipulations
\begin{equation*}
\left. \frac{(\zt-a/c)^{p-1}}{(\zt-\zt_\alpha)^{p+1}} \right|_{p>0} = \sum_{k=1}^p {{p-1}\choose{k-1}} \left(- \frac{ad-bc}{c(cz_\alpha+d)} \right)^{k-1} \frac{1}{(\zt-\zt_\alpha)^{k+1}}.
\end{equation*}
Separating the sum in Equation \eqref{Eq:NewTwistAux} into the cases $p=0$ and $p>0$ and using the expressions above, we get
\begin{equation*}
\vpt(\zt) = \sum_{\alpha\in\Si} \left( \ls \alpha 0 \left( \frac{1}{\zt-\zt_\alpha}-\frac{1}{\zt-a/c} \right) + \sum_{p=1}^{m_\alpha-1} \sum_{k=1}^p {p-1 \choose k-1 } \frac{(-c)^{p-k}(ad-bc)^k}{(cz_\alpha+d)^{k+p}} \frac{\ls\alpha p}{(\zt-\zt_\alpha)^{k+1}} \right) 
\end{equation*}
Exchanging the sums over $p$ and $k$, we get
\begin{equation}\label{Eq:DESNewTwist}
\vpt(\zt) = \sum_{\alpha\in\Si} \left( \frac{\ls \alpha 0}{\zt-\zt_\alpha} + \sum_{k=1}^{m_\alpha-1} \left( \sum_{p=k}^{m_\alpha-1} {p-1 \choose k-1 } \frac{(-c)^{p-k}(ad-bc)^k}{(cz_\alpha+d)^{k+p}} \ls\alpha p \right) \frac{1}{(\zt-\zt_\alpha)^{k+1}} \right) - \left( \sum_{\alpha\in\Si} \ls\alpha0 \right) \frac{1}{\zt-a/c}.
\end{equation}
The last term vanishes using the first-class condition \eqref{Eq:SumLevels}, hence the result announced in the lemma (modulo an exchange of the labels $k$ and $p$).

Although we supposed $c\neq 0$ to prove this statement, the expression above makes sense when taking the limit $c\to0$, where it simply becomes (as one keeps only $p=k$ in the sum over $p$)
\begin{equation*}
 \vpt(\zt) \Bigr|_{c=0} = \sum_{\alpha\in\Si} \sum_{k=0}^{m_\alpha-1} \left(\frac{a}{d}\right)^k \frac{\ls\alpha k}{(\zt-\zt_\alpha)^{k+1}}.
\end{equation*}
One verifies easily that this is indeed the transformation law of the twist function when $c=0$, \textit{i.e.} when the Mobius transformation $f$ is simply a translation followed by a dilation of factor $a/d$ (see also~\nm, Equation (2.44)).
\end{proof}

\paragraph{Realisation and Gaudin Lax matrix.} The next step in the construction of a model with twist function $\vpt(\zt)$ is to pick a realisation of the Takiff algebra $\Tc_{\ltb}$ corresponding to the levels $\lst \alpha p$ computed in Lemma \ref{Lem:NewLevels}. Recall that for the initial model with levels $\ls \alpha p$, we picked a realisation $\pi: \Tc_{\lt} \rightarrow \Ac$ into the algebra of observables $\Ac$, which sends the abstract Takiff currents $\Jt\alpha p(x)$ of $\Tc_{\lt}$ to Takiff currents $\J\alpha p(x)$ in $\Ac$ satisfying the Poisson bracket \eqref{Eq:TakiffReal} (see Subsection \ref{SubSubSec:Real}). The following proposition shows that this realisation can be used to construct a realisation of the abstract Takiff currents $\widetilde{J}^{\,\alpha}_{[p]}(x)$ of $\Tc_{\ltb}$ in the same algebra of observables $\Ac$.

\begin{proposition}\label{Prop:NewTakiff}
The map
\begin{equation*}
\begin{array}{rccl}
\pit : & \Tc_{\ltb} & \longrightarrow & \Ac \\[3pt]
       & \widetilde{J}^{\,\alpha}_{[p]} & \longmapsto & \Jtt\alpha p
\end{array}
\end{equation*}
defined by
\begin{equation*}
\Jtt\alpha 0(x) =\J\alpha 0(x) \;\;\;\;\;\; \text{ and } \;\;\;\;\;\; \Jtt \alpha p (x) = \sum_{k=p}^{m_\alpha-1} {k-1 \choose p-1 } \frac{(-c)^{k-p}(ad-bc)^p}{(cz_\alpha+d)^{p+k}} \J\alpha k(x), \;\; \text{ for } \; p>0,
\end{equation*}
is a realisation of the Takiff algebra $\Tc_{\ltb}$.
\end{proposition}

\begin{proof}
The proof of this proposition being rather technical, we give it in Appendix \ref{App:ChangeSpec}.
\end{proof}

The expression of the new Takiff currents $\Jtt \alpha p(x)$ in Proposition \ref{Prop:NewTakiff} is in striking resemblance with the expression of the new levels $\lst \alpha p$ in Lemma \ref{Lem:NewLevels}. This is in fact natural, as we shall see below.

Consider the Gaudin Lax matrix $\widetilde{\Gamma}\bigl(\zt,x)$ of the new model with spectral parameter $\zt$ and realisation $\pit$ as in Proposition \ref{Prop:NewTakiff}. It is given from the Takiff currents $\Jtt\alpha p(x)$ by
\begin{equation*}
\widetilde{\Gamma}\bigl(\zt,x)  = \sum_{\alpha\in\Si} \sum_{p=0}^{m_\alpha-1} \frac{\Jtt\alpha p(x)}{(\zt-\zt_\alpha)^{p+1}}.
\end{equation*}
Recall that the expression of the new levels $\lst\alpha p$ found in Lemma \ref{Lem:NewLevels} comes from performing the partial fraction decomposition of the twist function $\vpt(\zt)$, which was defined through Equation \eqref{Eq:Twist1form} by requiring that the twist function transforms as a 1-form. Following the steps of the proof of Lemma \ref{Lem:NewLevels}, but for the Gaudin Lax matrix $\widetilde{\Gamma}\bigl(\zt,x)$ instead of the twist function, one sees that the definition of the new Takiff currents $\Jtt\alpha p(x)$ in Proposition \ref{Prop:NewTakiff}, which mimics the expression of the new levels $\ls\alpha p$, implies that the Gaudin Lax matrix almost transforms as a 1-form, as it satisfies
\begin{equation}\label{Eq:TransfoGammaStrong}
\Gamma(z,x) \dd z = \frac{ad-bc}{(c\,\zt-a)^2} \; \Gamma\left( -\frac{d\,\zt-b}{c\,\zt-a}, x \right) \dd \zt = \widetilde{\Gamma}\bigl(\zt,x) \dd\zt - \Cc(x) \frac{c\,\dd \zt}{c\zt-a}.
\end{equation}
This equation is the equivalent for the Gaudin Lax matrix of the Equation \eqref{Eq:DESNewTwist} for the twist function, where the constraint $\Cc(x) = \sum_{\alpha\in\Si} \J\alpha 0(x)$ now replaces  the term $\sum_{\alpha\in\Si} \ls \alpha 0$ of Equation \eqref{Eq:DESNewTwist}. We thus see that the Gaudin Lax matrix also transforms as a 1-form, up to the constraint $\Cc(x)$. This naturally leads us to discuss the constraint in the new model.

\paragraph{Constraint.} The initial model considered in the previous subsections possesses a constraint $\Cc(x) \approx 0$, defined by Equation \eqref{Eq:C}. Similarly, we define the constraint for the new model (with spectral parameter $\zt$) as
\begin{equation*}
\widetilde{\Cc}(x) = - \res_{\zt = \infty} \widetilde{\Gamma}\bigl(\zt,x) = \sum_{\alpha\in\Si} \Jtt\alpha p(x).
\end{equation*}
However, according to Proposition \ref{Prop:NewTakiff}, the currents of Takiff mode $p=0$ are not modified by the change of spectral parameter. Thus, the constraints of the inital model and the new model coincide:
\begin{equation*}
\widetilde{\Cc}(x) = \Cc(x).
\end{equation*}
This is a particularly important result, as it shows that the constraint and hence the gauge transformation are not modified by the change of spectral parameter and that one can use the same notion of weak equality $\approx$ in both models. In particular, it shows, starting from Equation \eqref{Eq:TransfoGammaStrong}, that the Gaudin Lax matrix indeed transforms as a 1-form in the algebra of physical observables, \textit{i.e.} when considered weakly:
\begin{equation}\label{Eq:Gamma1form}
\Gamma(z,x)\dd z \approx \widetilde{\Gamma}\bigl(\zt,x) \dd \zt.
\end{equation}

\subsubsection{Equivalence of the two models}
\label{SubSubSec:EquivChangeSpec}

\paragraph{Summary.} Before going further, let us summarise what we did so far. We started with the model $\mathbb{M}^{\vp,\pi}_{\eb}$ (see Paragraph \ref{SubSubSec:Summary}), as constructed in the previous subsections, \textit{i.e.} a constrained realisation of a local affine Gaudin model with spectral parameter $z$. In particular, the observables $\Ac$ and the positions of the sites of this model are specified by its twist function $\vp(z)$ and the realisation $\pi:\Tc_{\lt} \rightarrow \Ac$, whereas its Hamiltonian is specified by the parameters $\eb$.

We then considered a change of spectral parameter $z \mapsto \zt=f(z)$ and started the construction of a new model, with spectral parameter $\zt$. For that, we defined the twist function $\vpt(\zt)$ of this new model in such a way that $\vpt(\zt)\dd \zt$ and $\vp(z)\dd z$ are the same 1-form expressed in the two different coordinates $\zt$ and $z$. This defines the Takiff datum $\ltb$ and the positions of the sites $\zt_\alpha$ of the new model (see Lemma \ref{Lem:NewLevels}). We then showed that the Takiff algebra $\Tc_{\ltb}$ also possesses a natural realisation $\pit$ in the same algebra of observables $\Ac$ as the initial model, constructed from the realisation $\pi$ (Proposition \ref{Prop:NewTakiff}). Moreover, we observed that the constraint naturally associated with the new model coincides with the one of the initial model, thus showing that the two models share the same physical phase space.

The new model can then be identified as a model $\mathbb{M}^{\vpt,\pit}_{\bm{\widetilde\epsilon}}$, with spectral parameter $\zt$, twist function $\vpt(\zt)$ and realisation $\pit$. The last element that one needs to define the new model is then its Hamiltonian, or equivalently the parameters $\bm{\widetilde\epsilon}$. We will show in this subsection that there is a canonical choice for $\bm{\widetilde\epsilon}$ and that this choice implies the equivalence of the two models $\mathbb{M}^{\vp,\pi}_{\eb}$ and $\mathbb{M}^{\vpt,\pit}_{\bm{\widetilde\epsilon}}$.

\paragraph{Zeroes of the new twist function and quadratic charges.} As explained in Subsection \ref{SubSec:Zeroes}, the construction of the Hamiltonian of the new model is deeply related to the zeroes of its twist function 1-form $\vpt(\zt)\dd\zt$. It is clear from Equation \eqref{Eq:Twist1form} that the zeroes $\zet_i$ of $\vpt(\zt)\dd\zt$ are the images under the map $f$ of the zeroes $\ze_i$ of $\vp(z)\dd z$:
\begin{equation}\label{Eq:NewZeros}
\zet_i = f(\ze_i), \;\;\;\;\;\;\;\; \forall i\in \lbrace 1,\cdots,M \rbrace.
\end{equation}
Let us recall that in Subsection \ref{SubSec:Zeroes}, we considered two cases (i) and (ii), whether $z=\infty$ is a zero of $\vp(z)\dd z$ or not. In both cases, all the zeroes $\ze_i$, including potentially the infinity, possess an image $f(\ze_i)$ in $\CP$ under $f$, which can be itself finite or infinite, depending on the choice of $f$. Thus, the new model itself can fit either in the cases (i) or (ii), independently of what one had in the initial model. However, as we shall see, we will be able to treat in an uniform way all possible cases that one can consider (the initial model being in case (i) or (ii) and the the new model being in case (i) or (ii)).\\

Following Subsection \ref{SubSec:Zeroes}, the next step in the construction of the Hamiltonian is to associate to each zero $\zet_i$ a quadratic charge $\widetilde{\Q}_i$, defined as
\begin{equation*}
\widetilde{\Q}_i = \res_{\zt = \zet_i} \widetilde{\Q}(\zt)\dd \zt.\vspace{-5pt}
\end{equation*}
In this expression, the spectral parameter dependent quantity $\widetilde{\Q}(\zt)$ is defined as in Equation \eqref{Eq:QSpec} for the inital model, \textit{i.e.} as
\begin{equation*}
\widetilde{\Q}(\zt) = - \frac{1}{2\vpt(\zt)} \int_\D \dd x \; \kappa\bigl( \widetilde{\Gamma}(\zt,x), \widetilde{\Gamma}(\zt,x) \bigr).
\end{equation*}
It is clear from the transformation laws \eqref{Eq:Twist1form} and \eqref{Eq:Gamma1form} that the quantity $\Q(z)$ also transforms as a 1-form (at least weakly):
\begin{equation*}
\Q(z)\dd z \approx \widetilde{\Q}(\zt) \dd \zt.
\end{equation*}
It is a standard result of complex analysis that the residues of a 1-form are independent of the choice of coordinate. Thus, we obtain that the quadratic charges $\Q_i$ and $\widetilde{\Q}_i$ are weakly equal
\begin{equation}\label{Eq:NewQi}
\Q_i \approx \widetilde{\Q}_i.
\end{equation}

\paragraph{Hamiltonian.} Let us recall that the first-class Hamiltonian $\Hc_0$ of the initial model is defined by a choice of parameters $\epsilon_i$ ($i\in\lbrace 1,\cdots,M\rbrace$) associated with the zeroes $\ze_i$, as in Equation \eqref{Eq:PHZeros}. Similarly, one defines the Hamiltonian $\widetilde{\Hc}_0$ of the new model by choosing a parameter $\widetilde\epsilon_i$ for each zero $\zet_i$. There is a one-to-one correspondence \eqref{Eq:NewZeros} between the zeroes $\ze_i$ of the initial model and the zeroes $\zet_i$ of the new one. Thus, a canonical choice for defining the parameters $\widetilde\epsilon_i$ is to take them equal to the parameters $\epsilon_i$ of the initial model, so that the Hamiltonian of the new model reads
\begin{equation*}
\widetilde{\Hc}_0 = \sum_{i=1}^M \epsilon_i \widetilde{\Q}_i.
\end{equation*}
Using the weak invariance \eqref{Eq:NewQi} of the quadratic charges $\Q_i$'s under the change of spectral parameter, we then get that the Hamiltonians of the initial and new models are weakly equal:
\begin{equation*}
\widetilde{\Hc}_0 \approx \Hc_0.
\end{equation*}
In particular, they define the same dynamics on the physical observables. Thus, the models $\mathbb{M}^{\vp,\pi}_{\eb}$ and $\mathbb{M}^{\vpt,\pit}_{\eb}$ are equivalent, as announced.

\paragraph{Transformation of the Lax pair.}\label{Par:TransfoLax} We have just shown that the dynamics of the new and the initial models are equivalent. As this dynamic is encoded in the zero curvature equation \eqref{Eq:Zce} on the Lax pair $\bigl(\Lc(z),\Mc(z)\bigr)$, it is natural to wonder how this Lax pair transforms under the change of spectral parameter. Similarly to Equation \eqref{Eq:Lax} for the initial model, the Lax matrix of the new model is defined by
\begin{equation*}
\widetilde{\Lc}(\zt,x) = \frac{\widetilde{\Gamma}(\zt,x)}{\vpt(\zt)}.
\end{equation*}
Considering the transformation laws \eqref{Eq:Twist1form} and \eqref{Eq:Gamma1form} of the twist function and the Gaudin Lax matrix, one sees that the Lax matrix transforms as a function:
\begin{equation*}
\widetilde{\Lc}(\zt,x) \approx \Lc(z,x).
\end{equation*}
Starting from the expression of the temporal component of the Lax pair $\Mc(z,x)$ in Theorem \ref{Thm:Lax}, one can also compute its behaviour under the change of spectral parameter\footnote{This computation is however quite subtle, as it requires taking into account the transformation of the Lagrange multiplier $\mu(x)$. Indeed, as the Hamiltonian $\Hc_0$ and $\widetilde{\Hc}_0$ are only equal weakly, they define the same dynamics only through a change of Lagrange multipliers.}. One then finds that it also transforms as a function:
\begin{equation*}
\widetilde{\Mc}(\zt,x) \approx \Mc(z,x).
\end{equation*}
It is then clear that the zero curvature equation \eqref{Eq:Zce} for the initial Lax pair $\bigl(\Lc(z),\Mc(z)\bigr)$ is equivalent to the one for the transformed Lax pair $\left(\widetilde{\Lc}(\zt),\widetilde{\Mc}(\zt)\right)$, which provides another check of the equivalence of the dynamics of the two models.

\section{Gauging non-constrained affine Gaudin models and diagonal Yang-Baxter deformations}
\label{Sec:Gauging}

The realisations of affine Gaudin models considered in~\nms are described by sites $\alpha\in\Si$ with positions $z_\alpha$ in the complex plane, as the ones studied in Section \ref{Sec:AGM} of the present article, but also possess a site of multiplicity 2 at infinity, taking the form of a constant term $\ell^\infty$ in the twist function. This site at infinity is the main difference between the models of~\nms and the ones considered in this article: in particular, the presence of this site prevents the introduction of a constraint in the models of~\nm\footnote{Recall that in the present article the constraint $\Cc(x)$ was introduced in Subsection \ref{SubSubSec:Constraint} as (minus) the residue at infinity of the Gaudin Lax matrix $\Gamma(z,x)$. The fact that one can consistently restrict the dynamics of the model to the constrained surface $\Cc(x)\approx 0$ comes from the fact that this constraint is conserved under the time evolution of the model (and more generally under the Hamiltonian flow of the spectral parameter dependent Hamiltonian $\Hc(z)$). In the models considered in~\nm, only the integral of the residue at infinity of $\Gamma(z,x)$ is conserved by the Hamiltonian flow of $\Hc(z)$. Thus, this residue cannot be chosen as a constraint of the model (and instead just generates a global symmetry of the model).}.

In this section, we will show that all non-constrained models considered in~\nms are equivalent to a constrained model of the form considered in Section \ref{Sec:AGM}. This will be done using a change of the spectral parameter of the model, which will send all sites of the non-constrained model, including the site with multiplicity two at infinity, to sites in the complex plane. Before developing this idea further, let us recall briefly the construction of the non-constrained model we start with, based on~\bnm.

\subsection{Non-constrained affine Gaudin model}
\label{SubSec:NonConst}

\paragraph{Twist function and Gaudin Lax matrix.} Let us consider a non-constrained model $\mathbb{M}^{\vp,\pi}_{\eb}$, following the notations of~\nm. Its twist function takes the form
\begin{equation}\label{Eq:TwistNonConst}
\vp(z) = \sum_{\alpha\in\Si} \sum_{p=0}^{m_\alpha-1} \frac{\ls\alpha p}{(z-z_\alpha)^{p+1}} - \ell^\infty.
\end{equation}
The first term in this expression corresponds to sites $\alpha\in\Si$ with finite position $z_\alpha\in\C$, multiplicity $m_\alpha\in\Z_{\geq 1}$ and levels $\ls \alpha p$, $p\in\lbrace 0,\cdots,m_\alpha-1\rbrace$. The constant term $\ell^\infty$ creates a double pole in the 1-form $\vp(z)\dd z$ and thus corresponds to a site of multiplicity two at infinity (see~\bgs for details about the treatment of a site at infinity).

This twist function specify the Takiff datum $\lt$ of the model. Its algebra of observables $\Ac$ is then defined by taking a realisation $\pi : \Tc_{\lt} \rightarrow \Ac$ of the Takiff algebra $\Tc_{\lt}$. This realisation is given concretely by a choice of Takiff currents $\J\alpha p(x)$ in $\Ac$ for $\alpha\in\Si$ and $p\in\lbrace 0,\cdots,m_\alpha-1\rbrace$. These Takiff currents are then gathered into the Gaudin Lax matrix
\begin{equation}\label{Eq:GammaNonConst}
\Gamma(z,x) = \sum_{\alpha\in\Si} \sum_{p=0}^{m_\alpha-1} \frac{\J\alpha p(x)}{(z-z_\alpha)^{p+1}}.
\end{equation}

\paragraph{Hamiltonian.} The twist function \eqref{Eq:TwistNonConst} can be re-expressed as
\begin{equation*}
\vp(z) = - \ell^\infty \frac{\prod_{i=1}^M(z-\ze_i)}{\prod_{\alpha\in\Si} (z-z_\alpha)^{m_\alpha}},
\end{equation*}
upon introducing its zeroes $\ze_1,\cdots,\ze_M$, with $M=\sum_{\alpha\in\Si} m_\alpha$. One can then consider the spectral parameter dependent quantity
\begin{equation*}
\Q(z) = - \frac{1}{2\vp(z)} \int_{\D} \dd x \, \kappa\bigl( \Gamma(z,x), \Gamma(z,x) \bigr)
\end{equation*}
and the quadratic charges $\Q_i = \res_{z=\ze_i} \Q(z)\dd z$ associated with the zeroes $\ze_i$, $i\in\lbrace1,\cdots,M\rbrace$. The Hamiltonian of the non-constrained model is then defined by a choice of parameters $\eb=(\epsilon_1,\cdots,\epsilon_M)$:
\begin{equation*}
\Hc = \sum_{i=1}^M \epsilon_i \Q_i.
\end{equation*}

\subsection{Change of spectral parameter and constrained model}
\label{SubSec:ChangeSpec2}

\subsubsection{Mobius transformation}
\label{SubSubSec:Mobius}

As explained at the beginning of this section, our goal here is to relate the non-constrained model defined above to a constrained model, using a change of the spectral parameter $z$. This idea then follows closely the ones developed in Subsection 2.3.2 of~\nms and Subsection \ref{SubSec:ChangeSpec} of the present article. In particular, following Subsection \ref{SubSec:ChangeSpec}, we consider a change of spectral parameter $z\mapsto \zt = f(z)$, where the bijective map $f:\CP\rightarrow\CP$ sends rational functions of $z$ to rational functions of $\zt$ and hence is a Mobius transformation \eqref{Eq:Mobius}. The goal of this subsection is then to construct a model with spectral parameter $\zt$, as we did in Subsection \ref{SubSec:ChangeSpec} (the difference being that here we start with a non-constrained model).

We would like the model with spectral parameter $\zt$ to be a constrained model, as the ones considered in Section \ref{Sec:AGM}. In particular, this model then should not have a site at infinity. Thus, the positions of the sites of the non-constrained model we start with should not be sent to infinity by the Mobius transformation $f$. A site $\alpha\in\Si$ with finite position $z_\alpha\in\C$ is sent to the position
\begin{equation}\label{Eq:NewPositions2}
\zt_\alpha = f(z_\alpha) = \frac{a z_\alpha+b}{c z_\alpha+d}.
\end{equation}
We then require that the coefficients $c$ and $d$ satisfy
\begin{equation*}
c z_\alpha + d \neq 0, \;\;\;\;\; \forall \, \alpha\in\Si,
\end{equation*}
similarly to Equation \eqref{Eq:NewSitesFinite} in Subsection \ref{SubSec:ChangeSpec}. However, in the case considered in this subsection, one has an additional condition on the Mobius transformation $f$. Indeed, it should also send the site at $z=\infty$ to a finite position $\zi\in\C$. This is simply equivalent to requiring that the coefficient $c$ in \eqref{Eq:Mobius} is non-zero, which we shall suppose through all this subsection. The point at $z=\infty$ is then mapped under $f$ to the finite point
\begin{equation}\label{Eq:PosInf}
\zi = f(\infty) = \frac{a}{c}.
\end{equation}

\subsubsection{Twist function, sites and levels}

Similarly to Subsections 2.3.2 of~\nms and \ref{SubSec:ChangeSpec} of the present article, we define the twist function of the model with spectral parameter $\zt$ by supposing that it transforms as a 1-form, \textit{i.e.} that $\vp(z)\dd z = \vpt(\zt)\dd\zt$. Concretely, this means that the new twist function is given by (see Equation \eqref{Eq:NewTwist})
\begin{equation}\label{Eq:NewTwistConst}
\vpt(\zt) = \frac{ad-bc}{(c\,\zt-a)^2} \; \vp\left( -\frac{d\,\zt-b}{c\,\zt-a} \right),
\end{equation}
with $\vp$ the twist function \eqref{Eq:TwistNonConst} of the non-constrained model.

We now want to construct a realisation of affine Gaudin model with twist function $\vpt(\zt)$. The defining data of such a model (sites, positions, levels, ...) is encoded in the partial fraction decomposition of the twist function \eqref{Eq:NewTwistConst}. This is given by the following lemma.

\begin{lemma}\label{Lem:NewLevels2}
The partial fraction decomposition of $\vpt(\zt)$ is
\begin{equation*}
\vpt(\zt) = \sum_{\alpha\in\Si} \sum_{p=0}^{m_\alpha-1} \frac{\lst\alpha p}{(\zt-\zt_\alpha)^{p+1}} + \frac{\lst{(\infty)}0}{\zt-\zi} + \frac{\lst{(\infty)}1}{(\zt-\zi)^2},
\end{equation*}
with the levels
\begin{equation*}
\lst \alpha 0 = \ls \alpha 0 \;\;\;\;\;\; \text{ and } \;\;\;\;\;\; \lst \alpha p = \sum_{k=p}^{m_\alpha-1} {k-1 \choose p-1 } \frac{(-c)^{k-p}(ad-bc)^p}{(cz_\alpha+d)^{p+k}} \ls\alpha k, \;\; \text{ for } \; p>0,
\end{equation*}
together with
\begin{equation}\label{Eq:LevelInf}
\lst{(\infty)}0 = -\sum_{\alpha\in\Si} \ls\alpha 0 \;\;\;\;\;\; \text{ and } \;\;\;\;\;\; \lst{(\infty)}1 = - \frac{\ell^\infty(ad-bc)}{c^2}.
\end{equation}
\end{lemma}

\begin{proof}
The proof of this lemma is very similar to the one of Lemma \ref{Lem:NewLevels}. Indeed, in the present case one has to apply the transformation \eqref{Eq:NewTwistConst} starting with the twist function \eqref{Eq:TwistNonConst}, whereas in Lemma \ref{Lem:NewLevels}, we applied the same transformation but starting with the twist function \eqref{Eq:Twist}. The two twist functions \eqref{Eq:TwistNonConst} and \eqref{Eq:Twist} differ from the presence of the constant term $\ell^\infty$ in \eqref{Eq:TwistNonConst}. Thus, in the present case, one can perform the computation of the new twist function \eqref{Eq:NewTwistConst} in the same way as in Lemma \ref{Lem:NewLevels} just by adding the terms corresponding to the constant term $\ell^\infty$. More precisely, it is straightforward to adapt Equation \eqref{Eq:DESNewTwist} in the proof of Lemma \ref{Lem:NewLevels} to the present case, yielding
\begin{align*}
\vpt(\zt) =& \sum_{\alpha\in\Si} \left( \frac{\ls \alpha 0}{\zt-\zt_\alpha} + \sum_{k=1}^{m_\alpha-1} \left( \sum_{p=k}^{m_\alpha-1} {p-1 \choose k-1 } \frac{(-c)^{p-k}(ad-bc)^k}{(cz_\alpha+d)^{k+p}} \ls\alpha p \right) \frac{1}{(\zt-\zt_\alpha)^{k+1}} \right) \\
 & \hspace{150pt} - \left( \sum_{\alpha\in\Si} \ls\alpha0 \right) \frac{1}{\zt-a/c} - \frac{\ell^\infty(ad-bc)}{c^2} \frac{1}{(\zt-a/c)^2}.
\end{align*}
Note that in the present case, contrarily to the one of Lemma \ref{Lem:NewLevels}, the twist function of the unconstrained model does not satisfy the first-class condition \eqref{Eq:ResidueInf} so that $\sum_{\alpha\in\Si} \ls \alpha 0$ can be non-zero. The Lemma then directly follows.
\end{proof}

One can read from Lemma \ref{Lem:NewLevels2} the defining data of the model with twist function $\vpt(\zt)$. As expected, in addition to the sites $\alpha\in\Si$ with positions $\zt_\alpha = f(z_\alpha)$, it possesses a site of multiplicity 2 that we shall denote by $\is$, at the position $\zi=f(\infty)$ (we denote this additional site as $\is$ as it is the image under $f$ of the site at infinity of the initial model ; however note that its position $\zi$ in the new model with spectral parameter $\zt$ is finite). Let us introduce the set
\begin{equation*}
\Sit = \Si \sqcup \lbrace (\infty) \rbrace,
\end{equation*}
which labels the sites of the model with twist function $\vpt(\zt)$ (for completeness, let us note here that the site $\is$ is real, hence is in $\Sit_{\rd}$, and has real position $\zi$). The levels of the sites in $\Sit$ are given by Lemma \ref{Lem:NewLevels2}. We will denote by $\ltb$ the corresponding Takiff datum.

\subsubsection{Takiff currents and realisation}
\label{SubSubSec:NewCurrents}

\paragraph{Takiff currents.} Now that we have determined the sites $\alpha\in\Sit$ and the Takiff datum $\ltb$ of the model with twist function $\vpt(\zt)$, the next step into its construction is to define its algebra of observables $\Act$, by choosing a realisation $\pit : \Tc_{\ltb} \rightarrow \Act$ of the Takiff algebra with Takiff datum $\ltb$. Concretely, this requires to find an algebra $\Act$ containing Takiff currents $\Jtt\alpha p(x)$ with levels $\lst\alpha p$, for all $\alpha \in \Sit$ and $p\in\lbrace 0,\cdots,m_\alpha-1\rbrace$.

Recall from the previous paragraph that the sites $\Sit$ of the new model are divided into two parts: the sites $\alpha\in\Si$, which are the image under $f$ of the sites of the initial model with finite positions, and the site $\is$, which is the image of the site at infinity of the initial model. Let us focus first on the sites $\alpha \in \Si$. In the initial model, these are realised in the algebra of observables $\Ac$, by the Takiff realisation $\pi : \Tc_{\lt} \rightarrow \Ac$, with Takiff currents $\J\alpha p(x)$ (see Subsection \ref{SubSec:NonConst}). For $\alpha\in\Si$, Proposition \ref{Prop:NewTakiff} (see also Appendix \ref{App:ChangeSpec}) provides us with a realisation of the Takiff currents with levels $\lst\alpha p$ in $\Ac$, by defining
\begin{equation}\label{Eq:NewCurrents2}
\Jtt\alpha 0(x) =\J\alpha 0(x) \;\;\;\;\;\; \text{ and } \;\;\;\;\;\; \Jtt \alpha p (x) = \sum_{k=p}^{m_\alpha-1} {k-1 \choose p-1 } \frac{(-c)^{k-p}(ad-bc)^p}{(cz_\alpha+d)^{p+k}} \J\alpha k(x), \;\; \text{ for } \; p>0.
\end{equation}

There is left to find a realisation of the site $\is$, which has multiplicity two and levels \eqref{Eq:LevelInf}. Such a realisation has been constructed (and extensively used) in~\nm: the PCM+WZ realisation. It takes value in the Poisson algebra $\Ac_{G_0}$ generated by canonical fields on the cotangent bundle $T^*G_0$ of a real Lie group $G_0$. We recall briefly the definition of $\Ac_{G_0}$ and of the PCM+WZ realisation in Appendix \ref{App:PCM+WZ}. In particular, $\Ac_{G_0}$ is generated by a group-valued field $g(x)\in G_0$ and a Lie algebra-valued field $X(x)\in\g_0$. The currents of the realisation involve the $\g_0$-valued spatial derivative $j(x)=g(x)^{-1}\p_x g(x)$ as well as another $\g_0$-valued field $W(x)$, also constructed from $g(x)$ and its spatial derivative (both $j(x)$ and $W(x)$ then Poisson commute with $g(x)$ and between themselves).

Here we will suppose that the site $\is \in \Sit$ is realised in $\Ac_{G_0}$ by the PCM+WZ realisation and thus that the corresponding Takiff currents are
\begin{equation*}\label{Eq:NewCurrentInf}
\Jtt\is 0(x) = X(x) + \frac{1}{2} \lst\is 0\, W(x) + \frac{1}{2} \lst\is 0 \, j(x) \;\;\;\;\;\; \text{ and } \;\;\;\;\;\; \Jtt\is 1(x) = \lst\is 1 \, j(x),
\end{equation*}
according to Appendix \ref{App:PCM+WZ}.

\paragraph{Realisation and algebra of observables.} Equation \eqref{Eq:NewCurrents2} gives a realisation of the Takiff currents $\Jtt\alpha p(x)$ for $\alpha\in\Si \subset \Sit$, in the algebra of observables $\Ac$ of the initial non-constrained model. Equation \eqref{Eq:NewCurrentInf} gives a realisation of the Takiff currents $\Jtt\is p(x)$ associated with the site $\is \in \Sit$, in the algebra of observables $\Ac_{G_0}$. Together, they define a realisation $\pit:\Tc_{\ltb} \rightarrow \Act$
of the whole Takiff algebra $\Tc_{\ltb}$ in the algebra
\begin{equation*}
\Act = \Ac \otimes \Ac_{G_0}.
\end{equation*}
In the rest of this section, we will denote by $\lbrace\cdot,\cdot\rbrace_{\Act}$ the Poisson bracket on $\Act$.

The realisation $\pit$ sends the abstract Takiff currents $\widetilde{J}^{\hspace{1pt}\alpha}_{[p]}$ ($\alpha\in\Sit$) in $\Tc_{\ltb}$ to
\begin{equation*}
\pit\Bigl( \widetilde{J}^{\hspace{1pt}\alpha}_{[p]} \Bigr) = \Jtt\alpha p \otimes 1 \, \in \Ac \otimes \Ac_{G_0} \;\;\;\; \text{ for } \;\; \alpha \in \Si \subset \Sit
\end{equation*}
and
\begin{equation*}
\pit\Bigl( \widetilde{J}^{\hspace{2pt}\is}_{[p]} \Bigr) = 1 \otimes \Jtt \is p \, \in \Ac \otimes \Ac_{G_0},
\end{equation*}
where the notations $1$ in the above equations represent the unit elements of the algebras $\Ac$ and $\Ac_{G_0}$. In the rest of this section, when there are no needs for the distinction, we will often write elements $\mathcal{F} \otimes 1$ and $1 \otimes \mathcal{G}$ of $\Act$ simply as $\mathcal{F}$ and $\mathcal{G}$, by a slight abuse of notations.

The choice of this representation $\pit$ defines the algebra of observables of the new model with spectral parameter $\zt$ to be the algebra $\Act=\Ac \otimes \Ac_{G_0}$. This can seem problematic at first, as we would like to construct a model which is equivalent to the initial model $\mathbb{M}^{\vp,\pi}_{\eb}$, whose algebra of observables is $\Ac$, which is then smaller that $\Act$. As we shall see in the Subsection \ref{SubSubSec:Gauging}, this apparent issue will be resolved by introducing a constraint in the new model, hence reducing its physical observables $\Act$ to the observables $\Ac$ of the initial model (which is non-constrained).

\paragraph{Gaudin Lax matrix.} Let us end this subsection by describing the Gaudin Lax matrix of the model with spectral parameter $\zt$. It is defined from the choice of Takiff currents $\Jtt\alpha p(x)$ by
\begin{equation}\label{Eq:GammaTilde}
\Gt(\zt,x) = \sum_{\alpha\in\Sit} \frac{\Jtt\alpha p(x)}{(\zt-\zt_\alpha)^{p+1}}.
\end{equation}
We shall need later the following expression for $\Gt(\zt,x)$, which is obtained by separating the set of sites $\Sit$ into $\Si$ and $\lbrace \is \rbrace$ and rearranging the terms:
\begin{equation}\label{Eq:DecompoGt}
\Gt(\zt,x) = \Gt_0(\zt,x) + \Gt_\infty(\zt,x),
\end{equation}
where
\begin{equation*}
\Gt_0(\zt,x) = \sum_{\alpha\in\Si} \sum_{p=0}^{m_\alpha-1} \frac{\Jtt \alpha p(x)}{(\zt-\zt_\alpha)^{p+1}} - \left( \sum_{\alpha\in\Si} \Jtt \alpha 0(x) \right) \frac{1}{\zt-\zi}
\end{equation*}
and
\begin{equation}\label{Eq:GtInf}
\Gt_\infty(\zt,x) = \left( \sum_{\alpha\in\Sit} \Jtt\alpha 0(x) \right) \frac{1}{\zt-\zi} + \frac{\Jtt\is 1}{(\zt-\zi)^2}.
\end{equation}
Note in particular that $\Gt_0(\zt,x)$ depends only on the Takiff currents $\Jtt\alpha p(x)$ associated with sites $\alpha$ in $\Si$. As these are related by Equation \eqref{Eq:NewCurrents2} to the Takiff currents $\J\alpha p(x)$ of the initial non-constrained model, we get that $\Gt_0(\zt,x)$ can be expressed in terms of the same objects than the initial Gaudin Lax matrix $\Gamma(z,x)$. In fact, by applying to the case at hand the same reasoning that led us to Lemmas \ref{Lem:NewLevels} and \ref{Lem:NewLevels2} and to Equation \eqref{Eq:TransfoGammaStrong}, one gets that
\begin{equation*}
\Gt_0(\zt,x) = \frac{ad-bc}{(c\,\zt-a)^2} \; \Gamma\left( -\frac{d\,\zt-b}{c\,\zt-a}, x \right),
\end{equation*}
\textit{i.e.} that
\begin{equation}\label{Eq:Gt1form}
\Gamma(z,x) \dd z = \Gt_0(\zt,x) \dd\zt.
\end{equation}

\subsubsection{Constraint and gauge symmetry}
\label{SubSubSec:Gauging}

\paragraph{Constraint.} The partial fraction decomposition of the twist function $\vpt(\zt)$ is given by Lemma \ref{Lem:NewLevels2}. In particular, let us observe that
\begin{equation}\label{Eq:NewResInf}
\res_{\zt=\infty} \vpt(\zt) \dd \zt = -\sum_{\alpha\in\Sit} \lst\alpha 0 = -\sum_{\alpha\in\Si} \lst \alpha 0 - \lst {(\infty)} 0 = 0.
\end{equation}
One can expect this result: indeed, by the invariance of the residues of a 1-form under a change of coordinate, the residue of $\vpt(\zt) \dd \zt$ at $\zt=\infty$ is equal to the residue of $\vp(z)\dd z$ at $z=f^{-1}(\infty)$. As we supposed that none of the poles $z_\alpha$ and $\infty$ of $\vp(z)\dd z$ are sent to infinity under the Mobius transformation $f$ (see Subsection \ref{SubSubSec:Mobius}), the 1-form $\vp(z)\dd z$ is regular at $f^{-1}(\infty)$, hence the vanishing of the residue.

Equation \eqref{Eq:NewResInf} shows that the new twist function $\vpt(\zt)$ satisfies the first-class condition \eqref{Eq:ResidueInf} considered in Section \ref{Sec:AGM}. According to Subsection \ref{SubSec:Gauge}, this means that one can impose a first-class constraint
\begin{equation}\label{Eq:NewConstraint}
\Cct(x) \approx 0
\end{equation}
to the model with twist function $\vpt(\zt)$, where $\Cct(x)$ is defined as
\begin{equation}\label{Eq:Ctilde}
\Cct(x) = - \res_{\zt=\infty} \Gt(\zt,x)\dd \zt = \sum_{\alpha\in\Sit} \Jtt\alpha 0(x) = \sum_{\alpha\in\Si} \J\alpha 0(x) + \Jtt\is 0(x).
\end{equation}
To obtain the last equality in the above equation, we used the decomposition $\Sit=\Si\sqcup\lbrace\is\rbrace$ of the sites of the new model and Equation \eqref{Eq:NewCurrents2} to express the currents $\Jtt\alpha p(x)$ for $\alpha\in\Si$. From the expression \eqref{Eq:NewCurrentInf} of $\Jtt\is 0(x)$, we then see that, on the constrained surface \eqref{Eq:NewConstraint},
\begin{equation}\label{Eq:XUnderConst}
X(x) \approx -\sum_{\alpha\in\Si} \J\alpha 0(x) - \frac{1}{2} \lst\is 0\, W(x) - \frac{1}{2} \lst\is 0 \, j(x).
\end{equation}
Recall from the previous subsection that the algebra of observables of the new model is $\Act=\Ac\otimes\Ac_{G_0}$ and that the component $\Ac_{G_0}$ of this algebra is generated by the fields $g(x)$ and $X(x)$. According to the above equation, on the constrained surface \eqref{Eq:NewConstraint}, the field $X(x)$ is expressed in terms of the Takiff currents $\J\alpha 0(x)$ of the initial model (which belong to the component $\Ac$ of the algebra $\Act$) and the currents $W(x)$ and $j(x)$, which depend only on the coordinate field $g(x)$. Thus, the field $X(x)$ is not an actual physical degree of freedom of the model.

\paragraph{Gauge symmetry and physical observables.} As explained in Subsection \ref{SubSec:Gauge} (see also Appendix \ref{App:Constraints}), the introduction of the first-class constraint \eqref{Eq:NewConstraint} implies the existence of a gauge symmetry of the model. This gauge symmetry acts on an observable $\Oo$ in the algebra $\Act$ by the following transformation
\begin{equation}\label{Eq:NewGauge}
\delta^\infty_\epsilon \Oo \approx \int_\D \dd x \; \kappa\Bigl( \epsilon(x,t), \bigl\lbrace \Cct(x), \Oo \bigr\rbrace_{\Act} \Bigr),
\end{equation}
where $\epsilon(x,t)$ is a $\g_0$-valued infinitesimal local parameter. The physical observables of the system are then the gauge-invariant quantities in the algebra of observables $\Act$, subject to the constraint \eqref{Eq:NewConstraint}.

Let us study the action of this gauge transformation on the field $g(x)$, which belongs to the component $\Ac_{G_0}$ of the Poisson algebra $\Act = \Ac \otimes \Ac_{G_0}$. This field Poisson commute with the Takiff currents $\J\alpha 0$, $\alpha\in\Si$, as those belong to the component $\Ac$ of the algebra $\Act$. Thus, using the expression \eqref{Eq:Ctilde} of the constraint $\Cct$, one gets
\begin{equation*}
\bigl\lbrace \Cct(y)\ti{1}, g(x)\ti{2} \bigr\rbrace_{\Act} = \bigl\lbrace \Jtt{\is}0(y)\ti{1}, g(x)\ti{2} \bigr\rbrace_{\Act}.
\end{equation*}
From the Poisson brackets \eqref{Eq:PBXg}, \eqref{Eq:PBjg} and \eqref{Eq:PbW1} of the field $g$ with $X$, $j$ and $W$, we then get
\begin{equation}\label{Eq:PBCg}
\bigl\lbrace \Cct(y)\ti{1}, g(x)\ti{2} \bigr\rbrace_{\Act} = g(x)\ti{2}\,C\ti{12}\, \delta_{xy},
\end{equation}
hence
\begin{equation*}
\delta^\infty_\epsilon g(x) \approx g(x)\epsilon(x,t).
\end{equation*}
This infinitesimal transformation lift to a gauge symmetry with a local parameter $h(x,t)$ in the group $G_0$ acting on the field $g(x)$ by a right multiplication:
\begin{equation}\label{Eq:GaugeOnG}
g(x) \longmapsto g(x) h(x,t).
\end{equation}
This gauge symmetry can be used to eliminate entirely the field $g(x)$, which is thus not a physical degree of freedom of the model.  As explained in Equation \eqref{Eq:XUnderConst} and below, the field $X(x)$ is also not a degree of freedom of the model because of the constraint \eqref{Eq:NewConstraint}. As these two fields generate the component $\Ac_{G_0}$ of the algebra $\Act=\Ac\otimes\Ac_{G_0}$, we see that the space of physical observables of the model can be identified with the component $\Ac$ in $\Act$ and thus with the space of observables of the initial non-constrained model we started with. This idea will be made more precise and rigorous in the next subsection by considering the elimination of the field $g(x)$ sketched above as a gauge fixing condition.

\subsection{Equivalence of the models}
\label{SubSec:Equivalence}

\subsubsection{Gauge fixing and identification of the algebras of observables}
\label{SubSubSec:GaugeFixing}

\paragraph{Gauge fixing.} The model with spectral parameter $\zt$ is defined on the Poisson algebra $\Act$. Moreover, it is subject to a constraint \eqref{Eq:NewConstraint} and a gauge symmetry \eqref{Eq:NewGauge}, which implies that all elements of $\Act$ are not physical observables of the theory. One way to get rid of the redundancy introduced by the gauge symmetry is to introduce a gauge fixing condition. Such a condition should be transverse to the gauge symmetry action, in the sense that in every orbit of this action, there is one and only one representative which satisfies this condition. We saw in Equation \eqref{Eq:GaugeOnG} that the gauge symmetry acts on the $G_0$-valued field $g(x)$ by right multiplication. As the right multiplication acts freely and transitively on $G_0$, it is clear that in every orbit of the gauge symmetry, there is a unique representative with the field $g(x)$ equal to the identity. Thus, this defines a good gauge fixing condition.

Mathematically, this gauge fixing can be seen as imposing another constraint on the model, in addition to the initial first-class constraint \eqref{Eq:NewConstraint} (see for example~\cite{dirac1964lectures,Henneaux:1992ig}). More precisely, we shall consider here the algebra of observables $\Act$ subject to the constraints
\begin{equation}\label{Eq:GaugeFix}
\Cct(x) \gf 0 \;\;\;\;\;\;\;\; \text{ and } \;\;\;\;\;\;\;\; g(x) \gf \Id.
\end{equation}
We introduced here a new symbol $\gf$, which denotes the weak equality under the set of constraints \eqref{Eq:GaugeFix}. As these constraints are stronger than the constraint $\Cct(x) \approx 0$ alone, we use another symbol. When the distinction is needed, we will then talk about quantities being $\approx$-weakly equal or $\gf$-weakly equal. Note that two quantities $\approx$-weakly equal are $\gf$-weakly equal but that the converse is not true.

\paragraph{Physical observables.} The algebra of physical observables can be seen as the algebra $\Act$ with $\approx$-weak equality where we identify two fields configurations related by a gauge transformation or as the algebra $\Act$ with $\gf$-weak equality (in which case there is no gauge symmetry, as it as been fixed by the second constraint in \eqref{Eq:GaugeFix}). In this subsection, we will consider the second formulation and will then work with $\gf$-weak equality. To understand what are the physical observables of the model in this formulation, lets us recall that the algebra $\Act$ is given by the tensor product $\Ac\otimes\Ac_{G_0}$, where the second component $\Ac_{G_0}$ is generated by the fields $g(x)$ and $X(x)$.

The gauge-fixing condition $g(x) \gf \Id$ implies that the spatial derivative $\p_x g(x)$ $\gf$-vanishes. As the currents $j(x)$ and $W(x)$ are constructed from these derivatives (see Appendix \ref{App:PCM+WZ}), we have
\begin{equation}\label{Eq:jWWeak}
j(x) \gf 0 \;\;\;\;\;\;\;\; \text{ and } \;\;\;\;\;\;\;\; W(x) \gf 0.
\end{equation}
Let us now consider the $\approx$-weak equality \eqref{Eq:XUnderConst} which expressed the field $X(x)$ using the initial constraint \eqref{Eq:NewConstraint}. Adding the gauge-fixing condition, it now becomes
\begin{equation}\label{Eq:XGaugeFix}
X(x) \gf - \sum_{\alpha\in\Si} \J\alpha 0(x).
\end{equation}
Thus, under $\gf$-weak equalities, the field $g(x)$ is constant and the field $X(x)$ is expressed in terms of the currents $\J\alpha 0(x)$, which belong to the component $\Ac$ of $\Act=\Ac\otimes\Ac_{G_0}$. As these fields generate the component $\Ac_{G_0}$ of $\Act$, one sees that this component does not contain any physical degree of freedom independent from the ones in the component $\Ac$.

Moreover, once we used the constraints \eqref{Eq:GaugeFix} to remove the fields $g(x)$ and $X(x)$, there are no relations left that would constrain further the observables in the component $\Ac$ of $\Act$. Thus, the space of physical observables of the new model can be identified with the subspace $\Ac\otimes 1$ of $\Act$, where $1$ denotes the unit element of $\Ac_{G_0}$. This subspace is naturally isomorphic to the space of observables of the inital model $\Ac$, at least as a vector space.

\paragraph{Dirac bracket.} In order to really identify the physical observables of the new model with the ones of the initial model, one should also prove that their Poisson structures coincide. The Poisson structure on the space of gauge-fixed physical observables is described by the Dirac bracket on $\Act$, as explained for example in~\cite{dirac1964lectures,Henneaux:1992ig}. Indeed, the fact that the gauge-fixing condition $g(x)\gf\Id$ is transverse to the gauge symmetry generated by the initial first-class constraint $\Cct(x)$ is equivalent to requiring that the system of constraints \eqref{Eq:GaugeFix} is second-class and thus that one can construct a Dirac bracket with respect to these constraints. We shall not enter into the details of this construction here and refer to the standard references~\cite{dirac1964lectures,Henneaux:1992ig} for a thorough treatment.

In the case that we are considering here, let us note that the $\gf$-weak Poisson brackets between the constraints \eqref{Eq:GaugeFix} read (see Equation \eqref{Eq:PBCg})
\begin{subequations}\label{Eq:MatrixPB}
\begin{equation}
\Bigl\lbrace \Cct\ti{1}(x), \Cct\ti{2}(y) \Bigr\rbrace_{\Act} \gf \Bigl\lbrace g\ti{1}(x), g\ti{2}(y) \Bigr\rbrace_{\Act} \gf 0
\end{equation}
and
\begin{equation*}
g\ti{2}(y)^{-1}\Bigl\lbrace \Cct\ti{1}(x), g\ti{2}(y) \Bigr\rbrace_{\Act} \gf - g\ti{1}(x)^{-1}\Bigl\lbrace g\ti{1}(x), \Cct\ti{2}(y) \Bigr\rbrace_{\Act} \gf C\ti{12}\, \delta_{xy}.
\end{equation*}
\end{subequations}
The fact that the system of constraints \eqref{Eq:GaugeFix} is second-class is contained here in the fact that the tensor $C\ti{12}$ appearing in the Poisson bracket above between $\Cct$ with $g$ is invertible (as written in a basis $\lbrace I^a \rbrace_{a=1,\cdots,\dim\g}$ of $\g$ it is given by the inverse $\kappa_{ab}$ of the Killing metric, see Subsection \ref{SubSubSec:Conventions}).

Following the standard procedure (see~\cite{dirac1964lectures,Henneaux:1992ig}), one can inverse the ``matrix'' of Poisson brackets \eqref{Eq:MatrixPB} and compute the Dirac bracket $\lbrace \mathcal{F}, \mathcal{G} \rbrace^*$ between two elements $\mathcal{F}$ and $\mathcal{G}$ of $\Act$. After a few manipulations, one finds
\begin{align}\label{Eq:Dirac}
\lbrace \mathcal{F}, \mathcal{G} \rbrace^* =& \lbrace \mathcal{F}, \mathcal{G} \rbrace_{\Act}\, - \int_\D \dd x \; \kappa\Bigl( \bigl\lbrace \Cct(x), \mathcal{F} \bigr\rbrace_{\Act}, g(x)^{-1}\bigl\lbrace g(x), \mathcal{G} \bigr\rbrace_{\Act} \Bigr) \\
 & \hspace{80pt} + \int_\D \dd x \; \kappa\Bigl( \bigl\lbrace \Cct(x), \mathcal{G} \bigr\rbrace_{\Act}, g(x)^{-1}\bigl\lbrace g(x), \mathcal{F} \bigr\rbrace_{\Act} \Bigr). \notag
\end{align}
Using the Poisson brackets \eqref{Eq:MatrixPB}, one checks that the definition of this Dirac bracket is made in such a way that
\begin{equation*}
\bigl\lbrace \mathcal{F}, \Cct(x) \bigr\rbrace^* \gf 0 \;\;\;\;\;\; \text{ and } \;\;\;\;\;\; \bigl\lbrace \mathcal{F}, g(x) \bigr\rbrace^* \gf 0
\end{equation*}
for any $\mathcal{F}$ in $\Act$. This shows that the Poisson structure induced by the Dirac bracket is consistent with imposing the constraints \eqref{Eq:GaugeFix} originating from the gauge fixing. In other words, the Dirac Poisson structure is compatible with the $\gf$-weak equality, in the sense that
\begin{equation*}
\lbrace \mathcal{F}', \mathcal{G}' \rbrace^* \gf \lbrace \mathcal{F}, \mathcal{G} \rbrace^* \;\;\;\;\;\; \text{ if } \;\;\; \mathcal{F}\gf\mathcal{F}' \;\; \text{and} \;\; \mathcal{G}\gf\mathcal{G}'.
\end{equation*}
The above equation shows that the Dirac bracket induces a well-defined Poisson structure on the space of physical observables of the model, constructed as $\Act$ with the $\gf$-weak equality.\\

Recall from the previous paragraph that this space of physical observables can be identified with the subspace $\Ac\otimes 1$. Let us consider two elements $\mathcal{F}\otimes 1$ and $\mathcal{G}\otimes 1$ in this subspace. Such elements Poisson commute with the field $g$, as it belongs to the component $\Ac_{G_0}$ of $\Act$. Thus, the Dirac bracket between these two elements is given by
\begin{equation*}
\lbrace \mathcal{F}\otimes 1, \mathcal{G}\otimes 1 \rbrace^* \gf \lbrace \mathcal{F}\otimes 1, \mathcal{G}\otimes 1 \rbrace_{\Act} = \lbrace \mathcal{F}, \mathcal{G} \rbrace \otimes 1,
\end{equation*}
where the last bracket is understood as a bracket in $\Ac$. Thus, we see that the space of physical observables of the new model with spectral parameter $\zt$ can be identified with the Poisson algebra $\Ac$ of the initial model with spectral parameter $z$, in a way which preserves the corresponding Poisson structures. In the rest of this section, we shall work with this identification and in particular write $\mathcal{F}\otimes 1$ simply as $\mathcal{F}$.

\subsubsection{Equivalence of the dynamics}
\label{SubSubSec:EquivDynamics}

\paragraph{Gaudin Lax matrix.} Recall the decomposition \eqref{Eq:DecompoGt} of the Gaudin Lax matrix $\Gt(\zt,x)$. The component $\Gt_\infty(\zt,x)$ in this decomposition is given by Equation \eqref{Eq:GtInf}. The simple pole at $\zt=\zi$ in this expression is proportional to the constraint $\Cct(x) = \sum_{\alpha\in\Sit} \Jtt\alpha 0(x)$ and thus is $\gf$-weakly vanishing. Moreover, by Equation \eqref{Eq:NewCurrentInf}, the double pole is proportional to $j(x)$, which is also $\gf$-weakly equal to 0, according to Equation \eqref{Eq:jWWeak}. Thus, the component $\Gt_\infty(\zt,x)$ $\gf$-weakly vanishes and we get
\begin{equation*}
\Gt(\zt,x) \gf \Gt_0(\zt,x).
\end{equation*}
From Equation \eqref{Eq:Gt1form}, one then obtain
\begin{equation*}
\Gamma(z,x) \dd z \gf \Gt(\zt,x)\dd\zt.
\end{equation*}

Thus, we see that the Gaudin Lax matrix transforms as a 1-form. This is the equivalent for the present case of Equation \eqref{Eq:Gamma1form} in Subsection \ref{SubSec:ChangeSpec}, where we proved the equivalence between two constrained model, in which case we saw that the Gaudin Lax matrix transforms as a 1-form when considering $\approx$-weak equalities. Here, we want to prove the equivalence of a non-constrained model with a constrained model, so we need to consider $\gf$-weak equalities, as the gauge fixing is necessary to identify the physical observables of the constrained model with the ones of the non-constrained model.

\paragraph{Zeroes of the twist function and quadratic charges.} As we constructed the new twist function $\vpt(\zt)$ in such a way that $\vp(z)\dd z = \vpt(\zt)\dd\zt$, it is clear that its zeroes $\zet_i$ are related to the ones $\ze_i$ of $\vp(z)$ by
\begin{equation*}
\zet_i = f(\ze_i), \;\;\;\;\; \text{ for } \; i\in\lbrace 1,\cdots,M \rbrace.
\end{equation*}
Depending on $f$, these zeroes $\zet_i$'s can either be all finite, so that we are in the case (i) of Subsection \ref{SubSec:Zeroes}, or either be all finite expect for one, then corresponding to the case (ii) of Subsection \ref{SubSec:Zeroes}.

In both these cases, one defines the quadratic charges $\widetilde{\Q}_i$, $i\in\lbrace 1,\cdots,M\rbrace$, associated with the zeroes as the residues
\begin{equation}\label{Eq:QTilde}
\widetilde{\Q}_i = \res_{\zt=\zet_i} \widetilde{\Q}(\zt)\dd \zt, \;\;\;\;\;\; \text{ where } \;\;\;\;\;\; \widetilde{\Q}(\zt) = -\frac{1}{2\vp(\zt)} \int_{\D} \dd x \; \kappa\Bigl( \Gt(\zt,x),\Gt(\zt,x) \Bigr).
\end{equation}
As both the Gaudin Lax matrix and the twist function transform as 1-forms (at least $\gf$-weakly), so does the quadratic charge $\Q(z)\dd z$. By invariance of the residues of a 1-form under a change of coordinate, we thus get that the quadratic charges $\widetilde{\Q}_i$ are $\gf$-weakly equal to the ones introduced in Subsection \ref{SubSec:NonConst} for the non-constrained model:
\begin{equation}\label{Eq:QiInv}
\Q_i \gf \widetilde{\Q}_i.
\end{equation}

\paragraph{Hamiltonian.}\label{Par:NewHam} To achieve the construction of the constrained model $\mathbb{M}^{\vpt,\pit}_{\bm{\widetilde{\epsilon}}}$ with twist function $\vpt(\zt)$, we have to specify its first-class Hamiltonian, through a choice of parameters $\bm{\widetilde{\epsilon}}=(\widetilde{\epsilon}_1,\cdots,\widetilde{\epsilon}_M)$:
\begin{equation*}
\widetilde{\Hc}_0 = \sum_{i=1}^M \widetilde{\epsilon}_i \, \widetilde{\Q}_i.
\end{equation*}
We will make the choice here that these parameters take the same values as the one defining the non-constrained model $\mathbb{M}^{\vp,\pi}_{\eb}$, so that $\widetilde{\epsilon}_i=\epsilon_i$ for every $i\in\lbrace 1,\cdots,M \rbrace$. It is then clear from Equation \eqref{Eq:QiInv} that the above Hamiltonian $\gf$-weakly coincides with the one of the non-constrained model and thus that they have the same dynamics. This ends the demonstration of the equivalence of the two models $\mathbb{M}^{\vp,\pi}_{\eb}$ and $\mathbb{M}^{\vpt,\pit}_{\eb}$, which was the main goal of this subsection.

\paragraph{Lax pair.} As for the equivalence of constrained models under a change of spectral parameter considerer in Subsection \ref{SubSec:ChangeSpec}, it is natural to wonder how the Lax pairs of these models are related. The Lax matrices of the two models are respectively defined as
\begin{equation*}
\Lc(z,x) = \frac{\Gamma(z,x)}{\vp(z)} \;\;\;\;\; \text{ and } \;\;\;\;\; \widetilde{\Lc}(\zt,x) = \frac{\Gt(\zt,x)}{\vpt(\zt)}.
\end{equation*}
As both the Gaudin Lax matrix and the twist function $\gf$-weakly transform as 1-forms, we get that the Lax matrix transforms as a function (as in Subsection \ref{SubSec:ChangeSpec} or in~\nm). After a few computations using the expressions of the temporal component $\Mc(z,x)$ of the initial Lax pair (from~\nm, Equation (2.40)) and of the component $\widetilde{\Mc}(\zt,x)$ of the new Lax pair (from Theorem \ref{Thm:Lax}), one can prove that this is also the case for this component. Thus, one has
\begin{equation}\label{Eq:TransfoLaxPair}
\Lc(z,t) \gf \widetilde{\Lc}(\zt,x) \;\;\;\;\;\; \text{ and } \;\;\;\;\;\, \Mc(z,t) \gf \widetilde{\Mc}(\zt,x).
\end{equation}
It is clear that the zero curvature equation for $\bigl(\Lc(z),\Mc(z)\bigr)$ is equivalent to the one for $\bigl(\widetilde{\Lc}(\zt),\widetilde{\Mc}(\zt)\bigr)$.

\subsection{Gauging integrable field theories}
\label{SubSec:Gauging}

\subsubsection{Gauging procedure}

\paragraph{Principle.} Let us summarise what we have done so far. We started with a non-constrained realisation of a local affine Gaudin model $\mathbb{M}^{\vp,\pi}_{\eb}$, as the ones considered in~\nm. Performing a change of spectral parameter $z \mapsto \zt$ which brought the site at infinity of this model in the finite complex plane, we constructed a constrained model $\mathbb{M}^{\vpt,\pit}_{\eb}$, at the cost of introducing new degrees of freedom $g(x)$ and $X(x)$. This apparent addition of degrees of freedom is counter-balanced by the introduction of a gauge symmetry in the model $\mathbb{M}^{\vpt,\pit}_{\eb}$, due to the presence of the first-class constraint. We then went on to prove that the non-constrained model $\mathbb{M}^{\vp,\pi}_{\eb}$ is identical to some gauge-fixed formulation of the constrained model $\mathbb{M}^{\vpt,\pit}_{\eb}$. Thus, it is also equivalent to the non-gauge-fixed formulation of the model. This provides a systematic procedure to ``gauge'' the models considered in~\nm.

\paragraph{Lagrangian formulation.} All the examples of models that we will consider in this article possess a Lagrangian formulation. Indeed, their algebra of observables $\Ac_M$ is generated by canonical fields on the cotangent bundle $T^*M$ of some manifold $M$. One can then perform an inverse Legendre transform to obtain an action $S[\Phi]$ for the model, where $\Phi(x,t)$ is a $M$-valued field.

The procedure of gauging described above then allows to construct an equivalent model on the algebra $\Ac_M\otimes\Ac_{G_0}$, with an additional $G_0$-gauge symmetry. As this algebra is the canonical algebra of fields on the cotangent bundle $T^*(M\times G_0)$, one can also perform an inverse Legendre transform for this model, yielding an action $S_g[\Phi,g]$ with $g(x,t)$ the $G_0$-valued field introduced in $\Ac_{G_0}$. This action is then invariant under a gauge transformation, which in particular acts on the field $g(x,t)$ as
\begin{equation*}
g(x,t) \longmapsto g(x,t)h(x,t),
\end{equation*}
with $h(x,t)$ a local parameter in $G_0$.

It is then clear that $g(x,t)=\Id$ is a good gauge-fixing condition. One then recovers $\Phi(x,t)$ as the only physical degree of freedom of the model and obtains back the action of the original model we started with:
\begin{equation*}
S_g[ \Phi , \Id ] = S[ \Phi ].
\end{equation*}

\paragraph{Gauge symmetry.} It is natural to wonder what form the gauge transformation takes on the initial $M$-valued field $\Phi(x,t)$. Recall that the gauge transformation \eqref{Eq:NewGauge} is generated by the constraint \eqref{Eq:Ctilde}. The action of this gauge transformation on an observable in the component $\Ac_M$ of the algebra $\Act=\Ac_M \otimes \Ac_{G_0}$ is then generated only by the part $\sum_{\alpha\in\Si} \J\alpha p(x)$ of the constraint, as the field $\Jtt\is 0(x)$ in the component $\Ac_{G_0}$ Poisson commutes with this observable. Yet, the quantity $\sum_{\alpha\in\Si} \J\alpha p(x)$ is the density of the moment map of the global diagonal symmetry of the initial non-constrained model (see~\nm, Subsection 2.2.4). Thus, the action of the gauge transformation on the observable in $\Ac_M$ is a local version of this diagonal symmetry. In particular, coming back to the Lagrangian formulation, it means that the $M$-valued field $\Phi(x,t)$ transforms under the gauge transformation as
\begin{equation}\label{Eq:GaugePhi}
\Phi(x,t) \longmapsto \Delta^\text{diag}_{h(x,t)} \bigl[ \Phi(x,t) \bigr],
\end{equation}
where $\Delta^\text{diag}_{h}$ for $h\in G_0$ constant represents the action of the diagonal global symmetry of the initial model on the field $\Phi(x,t)$.

\subsubsection{Global symmetries}
\label{SubSubSec:GlobalSymGauging}

\paragraph{Diagonal symmetry.} Let us consider the initial non-gauged model $\mathbb{M}^{\vp,\pi}_{\eb}$. As explained in~\nm, Subsection 2.2.4, and mentioned above, it possesses a global diagonal $G_0$-symmetry, with moment map
\begin{equation*}
K^\is = \int_{\D} \dd x\; \sum_{\alpha\in\Si} \J\alpha 0(x).
\end{equation*}
This symmetry acts on the Takiff currents $\J\alpha p(x)$ of the model as
\begin{equation*}
\J\alpha p(x) \longmapsto h^{-1} \, \J\alpha p(x) \, h,
\end{equation*}
with $h$ a constant parameter in $G_0$.

We saw in the previous paragraph that under the gauging procedure, this diagonal symmetry is ``promoted'' to the gauge symmetry of the model $\mathbb{M}^{\vpt,\pit}_{\eb}$. Indeed, this model is invariant under the local transformation
\begin{equation*}
\J\alpha p(x) \longmapsto h^{-1}(x,t) \, \J\alpha p(x) \, h(x,t) + \ls\alpha p \, h^{-1}(x,t)\p_x h(x,t) \;\;\;\;\;\; \text{ and } \;\;\;\;\;\; g(x,t) \longmapsto g(x,t)h(x,t),
\end{equation*}
with $h(x,t)$ an arbitrary $G_0$-valued field. It is then natural to wonder how the global diagonal symmetry of the initial model is realised in the gauged formulation $\mathbb{M}^{\vpt,\pit}_{\eb}$. This is the subject of this subsection.

\paragraph{Global $\bm{G_0}$-symmetry associated with the PCM+WZ realisation.} Let us recall that in the gauging procedure, we have introduced new degrees of freedom $g(x)$ and $X(x)$, in the component $\Ac_{G_0}$ of the algebra $\Act=\Ac\otimes\Ac_{G_0}$. These new degrees of freedom are attached to the site $\is$ of the model $\mathbb{M}^{\vpt,\pit}_{\eb}$, which is a site of multiplicity two realised by the PCM+WZ realisation \eqref{Eq:NewCurrentInf}. In~\nm, Subsection 3.3.4, it is explained how every such realisation comes with a global symmetry of the model. Let us briefly explain the main properties of this symmetry in the present case, based on the general discussion in~\nm. It is generated by the moment map
\begin{equation}\label{Eq:MomentMapTilde}
\Kt^\is = -\int_\D \dd x \; g(x) \left( X(x) + \frac{1}{2} \lst\is 0 W(x) - \frac{1}{2} \lst\is 0 j(x) \right) g(x)^{-1}.
\end{equation}
It is clear that this moment map Poisson commutes with the Takiff currents $\Jtt\alpha p(x)$ associated with sites $\alpha\in\Si$ different from $\is$, which are then invariant under the transformation generated by $\Kt^\is$. Moreover, according to~\nm, Subsection 3.3.4, the Takiff currents $\Jtt\is p(x)$ associated with the site $\is$ are also invariant under this transformation. As the Hamiltonian of the model $\mathbb{M}^{\vpt,\pit}_{\eb}$ is constructed from the Takiff currents $\Jtt\alpha p(x)$, $\alpha\in\Si\sqcup\lbrace\is\rbrace$, it is then also invariant, hence the symmetry.

Although this symmetry leaves invariant all the Takiff currents $\Jtt\alpha p(x)$, its action on $\Act$ is not trivial. In particular, it acts on the field $g(x)$ of the PCM+WZ realisation by left multiplication:
\begin{equation*}
g(x,t) \longmapsto h^{-1}\,g(x,t),
\end{equation*}
with $h\in G_0$ constant.

\paragraph{Equivalence with the global diagonal symmetry.} Let us now argue that the global symmetry of the gauged model $\mathbb{M}^{\vpt,\pit}_{\eb}$ introduced above is in fact the diagonal symmetry of the initial non-gauged model $\mathbb{M}^{\vp,\pi}_{\eb}$. Recall that the equivalence between the models $\mathbb{M}^{\vp,\pi}_{\eb}$ and $\mathbb{M}^{\vpt,\pit}_{\eb}$ was established through a gauge-fixing of $\mathbb{M}^{\vpt,\pit}_{\eb}$, mathematically encoded by considering $\gf$-weak equalities. Under this gauge fixing, the fields $g$, $j$, $W$ and $X$ associated with the site $\is$ satisfy the $\gf$-weak equalities \eqref{Eq:GaugeFix}, \eqref{Eq:jWWeak} and \eqref{Eq:XGaugeFix}. Inserting these in Equation \eqref{Eq:MomentMapTilde}, one finds the gauge-fixed expression of the moment map $\Kt^\is$:
\begin{equation*}
\Kt^\is \gf \int_{\D} \dd x\; \sum_{\alpha\in\Si} \J\alpha 0(x).
\end{equation*}
This expression coincides with the moment map $K^\is$ generating the diagonal symmetry of the initial model $\mathbb{M}^{\vp,\pi}_{\eb}$. Thus, under the gauge-fixing procedure, the symmetry generated by $K^\is$ and $\Kt^\is$ are identical, as announced.

\subsubsection{Integrable structure}
\label{SubSubSec:IntegrabilityGauging}

In this subsection, we are considering a general procedure to find a gauged formulation of the models considered in~\nm. One of the main property of these models is their integrability. Let us discuss how this integrable structure is modified under the gauging procedure.

\paragraph{Integrable structure of the gauged model.} The integrability of the initial model $\mathbb{M}^{\vp,\pi}_{\eb}$ rests on the existence of a Lax pair formulation of its dynamic and on the fact that its Lax matrix $\Lc(z,x)$ satisfies a Maillet bracket (with twist function $\vp(z)$), as explained in~\nm.

The integrable structure of the new model $\mathbb{M}^{\vpt,\pit}_{\eb}$ in its non-gauge-fixed formulation is described in Subsection \ref{SubSec:Integrability}. It admits a Lax formulation by Theorem \ref{Thm:Lax} and its Lax matrix $\Lct(\zt,x)$ satisfies a Maillet bracket with twist function $\vpt(\zt)$ (see Paragraph \ref{SubSubSec:Maillet}). Note that this Maillet bracket is with respect to the Poisson structure $\lbrace\cdot,\cdot\rbrace_{\Act}$ of the algebra $\Act$ (hence before gauge-fixing) ; in fact, it is even a strong bracket, in the sense that it is true even without having to impose the constraint $\Cct(x) \approx 0$.

Finally, let us note that the integrability of both the models $\mathbb{M}^{\vp,\pi}_{\eb}$ and $\mathbb{M}^{\vpt,\pit}_{\eb}$ also manifests itself in the existence of infinite hierarchies of local conserved charges in involution, associated with each zero of their respective twist function (see Equation (2.41) of~\nms and Equations \eqref{Eq:Hierarchy} and \eqref{Eq:HierarchyInfinity} of the present article). As the zeroes of $\vp(z)$ and $\vpt(\zt)$ are in one-to-one correspondence, there are the same number of these hierarchies in both models. In fact, these hierarchies coincide when gauge-fixing the model $\mathbb{M}^{\vpt,\pit}_{\eb}$, \textit{i.e.} under $\gf$-weak equalities.

\paragraph{Integrable structure of the gauge-fixed model.} Let us now turn ourselves to the integrable structure of the gauge-fixed formulation of the model $\mathbb{M}^{\vpt,\pit}_{\eb}$. According to Equation \eqref{Eq:TransfoLaxPair}, its Lax pair is directly related to the one of the model $\mathbb{M}^{\vp,\pi}_{\eb}$ by the change of variable $z \mapsto \zt=f(z)$. As the Lax matrix $\Lc(z,x)$ satisfies a Maillet bracket with twist function $\vp(z)$, the Lax matrix $\Lct(\zt,x)$ of the gauge-fixed model also satisfies a Maillet bracket, with $\Rc$-matrix:
\begin{equation*}
\widetilde{\Rc}^{\text{GF}}\ti{12}(\zt,\wt) = \Rc\bigl( f^{-1}(\zt), f^{-1}(\wt) \bigr) = \frac{C\ti{12}}{f^{-1}(\wt)-f^{-1}(\zt)} \vp\bigl( f^{-1}(\wt) \bigr)^{-1}.
\end{equation*}
Considering the expression \eqref{Eq:NewTwistConst} of the twist function $\vpt(\zt)$, we then have
\begin{equation}\label{Eq:RMatGF1}
\widetilde{\Rc}^{\text{GF}}\ti{12}(\zt,\wt) = \widetilde{\Rc}\hspace{1pt}^0(\zt,\wt) \vpt(\wt)^{-1}, \;\;\;\;\;\;\;\; \text{ where } \;\;\;\;\;\;\;\; \widetilde{\Rc}\hspace{1pt}^0(\zt,\wt)= \frac{C\ti{12}}{f^{-1}(\wt)-f^{-1}(\zt)} \frac{ad-bc}{(c \wt-a)^2}.
\end{equation}
A simple computation then yields
\begin{equation}\label{Eq:RMatGF}
\widetilde{\Rc}\hspace{1pt}^0(\zt,\wt) = \frac{C\ti{12}}{\wt-\zt} - \frac{C\ti{12}}{\wt-\zi},
\end{equation}
where we recall that $\zi=a/c$ is the position of the site $\is$. The gauge-fixed formulation of the model is thus not exactly a model with twist function, as the matrix $\widetilde{\Rc}\hspace{1pt}^0(\zt,\wt)$ is not the standard $\Rc$-matrix $C\ti{12}/(\wt-\zt)$. In fact, it is equal to the standard $\Rc$-matrix plus a term depending only on the second spectral parameter. This type of $\Rc$-matrices has been considered before in~\cite{Lacroix:2017isl} (see also \cite{Lacroix:2018njs}). In particular, it was proven in~\cite{Lacroix:2017isl} that they also lead to the existence of integrable hierarchies of local conserved charges in involution. This shows that the hierarchies mentioned above for the non-constrained model and the non-gauge-fixed model can also be constructed directly at the level of the gauge-fixed formulation.

\paragraph{Gauge-fixed Maillet bracket from the Dirac bracket.} For completeness and as we shall use a similar approach in Subsection \ref{SubSec:DiagYB}, let us mention another method to obtain the Maillet bracket of the gauge-fixed model computed in the previous paragraph. In the first method above, we computed the Poisson bracket of the gauge-fixed Lax matrix $\Lct(\zt,x)$ as an element of the algebra $\Ac$ with Poisson bracket $\lbrace\cdot,\cdot\rbrace$, by relating it to the Lax matrix $\Lc(z,x)$ of the initial model. However, one can also consider the Lax matrix $\Lct(\zt,x)$ as an element of the algebra $\Act$: the gauge-fixing is then taken into account by considering $\gf$-weak equalities and by replacing the Poisson bracket $\lbrace\cdot,\cdot\rbrace_{\Act}$ by the Dirac bracket $\lbrace\cdot,\cdot\rbrace^*$.

Thus one can compute the Poisson bracket of the gauge-fixed Lax matrix $\Lct(\zt,x)$ as a Dirac bracket. For brevity, we give the details of this computation in Appendix \ref{App:DiracUndeformed}. In the end, one finds the same results as with the first method above, \textit{i.e.} that the gauge-fixed Lax matrix $\Lct(\zt,x)$ satisfies a Maillet bracket with $\Rc$-matrix \eqref{Eq:RMatGF1}.

\subsection{Diagonal Yang-Baxter deformation}
\label{SubSec:DiagYB}

\subsubsection{Setting up the deformation}

\paragraph{Principle.} We will use the notations and results of the previous subsections. In particular, $\mathbb{M}^{\vp,\pi}_{\eb}$ is a non-constrained model as considered in~\nm. Under the gauging procedure described above, it is related to the gauged model $\mathbb{M}^{\vpt,\pit}_{\eb}$ by gauge-fixing. Let us recall also that the model $\mathbb{M}^{\vpt,\pit}_{\eb}$ possesses a site $\is$ with multiplicity two, realised by the PCM+WZ realisation \eqref{Eq:NewCurrentInf}.

It is explained in Subsection 4.2 of~\nms that models which possess such a realisation can be deformed while preserving their integrability by the Yang-Baxter deformation procedure. In this subsection, we apply such a deformation to the gauged model $\mathbb{M}^{\vpt,\pit}_{\eb}$. The deformed model obtained this way still possesses a gauge symmetry. This symmetry can then be gauge-fixed in a similar way than the one considered in the previous subsections, yielding an integrable deformation of the initial model $\mathbb{M}^{\vp,\pi}_{\eb}$.\\

As explained in~\nm, the Yang-Baxter deformation procedure has for effect to break the global symmetry naturally associated with the PCM+WZ realisation. For the gauged model $\mathbb{M}^{\vpt,\pit}_{\eb}$ considered here, we have studied the global symmetry associated with the PCM+WZ realisation attached to the site $\is$ in Subsection \ref{SubSubSec:GlobalSymGauging}. In particular, we have shown that through the gauge fixing, this symmetry is identified with the global diagonal symmetry of the initial model $\mathbb{M}^{\vp,\pi}_{\eb}$. Thus, the Yang-Baxter deformation considered here corresponds to a breaking of the diagonal symmetry of the model $\mathbb{M}^{\vp,\pi}_{\eb}$. For this reason, we shall call it the \textit{diagonal Yang-Baxter deformation}.\\

The article~\nms considers two types of Yang-Baxter deformations: the homogeneous and inhomogeneous ones. Both these deformations can be applied to the case considered here. We will develop mostly the case of an inhomogeneous deformation and will introduce the homogeneous one at the end of the subsection, as a limit of the inhomogeneous deformation. Moreover, the reference~\nms mostly treats the inhomogeneous Yang-Baxter deformation of a PCM realisation without a Wess-Zumino term (see Appendix \ref{App:PCM+WZ}), although the possibility of a more general deformation of any PCM+WZ realisation is mentioned at the end of Subsection 4.2.3, based on the results of~\cite{Delduc:2014uaa}. For simplicity, we shall also restrict here to the case of a PCM realisation without Wess-Zumino term.

\paragraph{The undeformed model.} Let us consider the twist function $\vpt(\zt)$ of the model $\mathbb{M}^{\vpt,\pit}_{\eb}$, whose partial fraction decomposition is given by Lemma \ref{Lem:NewLevels2}. In particular, $\vpt(\zt)$ possesses a double pole at $\zt=\zi$, corresponding to the site $\is$. We would like to apply an inhomogeneous Yang-Baxter deformation to the realisation attached to this site. As explained above, we will restrict here to the case where this realisation is a PCM realisation without Wess-Zumino term. Concretely, this means that the residue $\lst\is 0$ of $\vpt(\zt)$ at $\zt=\zi$ should vanish (see Appendix \ref{App:PCM+WZ}). Considering the expression of this residue given in Lemma \ref{Lem:NewLevels2}, we thus suppose that the levels $\ls\alpha 0$, $\alpha\in\Si$, of the initial model are such that
\begin{equation}\label{Eq:NoWZ}
\sum_{\alpha\in\Si} \ls\alpha 0 = -\res_{z=\infty} \vp(z)\dd z = 0.
\end{equation}
The twist function of the model $\mathbb{M}^{\vpt,\pit}_{\eb}$ can then be written
\begin{equation*}
\vpt(\zt) = \vpt_\Si(\zt) + \frac{\lst\is 1}{(\zt-\zi)^2}, \;\;\;\;\;\;\;\; \text{ with } \;\;\;\;\;\; \vpt_\Si(\zt) = \sum_{\alpha\in\Si} \sum_{p=0}^{m_\alpha-1} \frac{\lst\alpha p}{(\zt-\zt_\alpha)^{p+1}}.
\end{equation*}

Let us now turn to the Gaudin Lax matrix of the model $\mathbb{M}^{\vpt,\pit}_{\eb}$. As the site $\is$ is realised by a PCM realisation, without Wess-Zumino term, it reads
\begin{equation}\label{Eq:UndefGt}
\Gt(\zt,x) = \Gt_\Si(\zt,x) + \frac{X(x)}{\zt-\zi} + \frac{\lst\is 1\, j(x)}{(\zt-\zi)^2}, \;\;\;\;\;\;\;\; \text{ with } \;\;\;\;\;\; \Gt_\Si(\zt,x) = \sum_{\alpha\in\Si} \sum_{p=0}^{m_\alpha-1} \frac{\Jtt\alpha p(x)}{(\zt-\zt_\alpha)^{p+1}}.
\end{equation}

\paragraph{Deforming the twist function.} 
We then introduce the deformed twist function as
\begin{equation*}
\vpt_\eta(\zt) = \vpt_\Si(\zt) + \frac{\lst\is 1}{(\zt-\zi)^2-\cc^2 \eta^2},
\end{equation*}
with $\eta$ a real deformation parameter and $\cc$ being either $1$ or $i$, in the so-called split or non-split cases. It is clear that the deformed twist function $\vpt_\eta(\zt)$ reduces to the initial one $\vpt(\zt)$ in the limit $\eta\to0$. Instead of a double pole at $\zi$, the deformed twist function now possesses two simple poles at
\begin{equation}\label{Eq:ZtPM}
\zt_\pm = \zi \pm \cc\,\eta.
\end{equation}
The partial fraction decomposition of the deformed twist function reads
\begin{equation}\label{Eq:DefTwistDSE}
\vpt_\eta(\zt) = \sum_{\alpha\in\Si} \sum_{p=0}^{m_\alpha-1} \frac{\lst\alpha p}{(\zt-\zt_\alpha)^{p+1}} + \frac{1}{2\cc\gamma} \frac{1}{\zt-\zt_+} - \frac{1}{2\cc\gamma} \frac{1}{\zt-\zt_-}, \;\;\;\;\;\;\;\;\; \text{ with } \;\;\;\;\;\;\;\;\; \gamma = \frac{\eta}{\lst\is 1}.
\end{equation}
An affine Gaudin model with twist function $\vpt_\eta(\zt)$ would then have sites $\Sit^{(\eta)}=\Si \sqcup \lbrace (+),(-)\rbrace$. The sites $\alpha\in\Si$ keep the same positions $\zt_\alpha$ and levels $\lst\alpha p$ as in $\mathbb{M}^{\vpt,\pit}_{\eb}$. The sites $\pms$ have multiplicity 1, positions $\zt_\pm$ and levels $\lst\pms0=\pm (2\cc\gamma)^{-1}$.

\paragraph{Deforming the Gaudin Lax matrix.} To construct a realisation of affine Gaudin model with deformed twist function $\vpt_\eta(\zt)$, one needs to find a realisation of its Takiff algebra with sites $\Sit^{(\eta)}=\Si \sqcup \lbrace (+),(-)\rbrace$. The sites $\alpha\in\Si$ of the model have the same levels as in the undeformed model $\mathbb{M}^{\vpt,\pit}_{\eb}$: thus, they can be realised by the same Takiff currents $\Jtt\alpha p(x)$, in the algebra of observables $\Ac$, given by Equation \eqref{Eq:NewCurrents2}.

We now have to find a realisation of the sites $\pms$. As these sites have multiplicities $1$ and opposite levels $\lst\pms 0$, they can be realised by the inhomogeneous Yang-Baxter realisation, whose construction is recalled in Appendix \ref{App:iYB}. This realisation is valued in the same Poisson algebra $\Ac_{G_0}$ as the PCM realisation ; the Takiff currents attached to the sites $\pms$ are then expressed in terms of the fields $g(x)$, $X(x)$ and $j(x)$ in $\Ac_{G_0}$ as:
\begin{equation}\label{Eq:CurrentYBDef}
\Jtt\pms 0(x) = \frac{1}{2\cc} \left( \cc X(x) \mp R_gX(x) \pm \frac{1}{\gamma} j(x) \right).
\end{equation}
In this equation, $R_g=\Ad_g^{-1} \circ R \circ \Ad_g$, where $R$ is a skew-symmetric solution of the mCYBE \eqref{Eq:mCYBE}, as explained in Appendix \ref{App:iYB}. As the inhomogeneous Yang-Baxter realisation is valued in the same algebra $\Ac_{G_0}$ as the PCM realisation, the deformed model shares the same observables $\Act = \Ac \otimes \Ac_{G_0}$ as the model $\mathbb{M}^{\vpt,\pit}_{\eb}$.\\

The expression of the Gaudin Lax matrix $\Gt_\eta(\zt,x)$ of the deformed model derives directly from the choice of realisation made above. It is given by
\begin{equation}\label{Eq:DeformGt}
\Gt_\eta(\zt,x) = \Gt_\Si(\zt,x) + \frac{\Jtt{(+)}0(x)}{\zt-\zt_+} +  \frac{\Jtt{(-)}0(x)}{\zt-\zt_-},
\end{equation}
where $\Gt_\Si(\zt,x)$ is defined as in the undeformed Gaudin Lax matrix \eqref{Eq:UndefGt}. Taking the limit $\eta\to 0$, one checks that we recover this undeformed Gaudin Lax matrix\footnote{Note that taking this limit requires a bit of care. Indeed, the terms $1/\gamma$ in the currents $\Jtt\pms 0(x)$ diverge when $\eta\to0$. However, this divergence cancels with the fact that $(\zt-\zt_+)^{-1} - (\zt-\zt_-)^{-1}$ converges to $0$ in the limit, creating a derivative eventually responsible for the appearance of the double pole at $\zt=\zi$ in $\Gt(\zt,x)$.}:
\begin{equation*}
\Gt_\eta(\zt,x) \xrightarrow{\eta\to0} \Gt(\zt,x).
\end{equation*}

\paragraph{Constraint and gauge symmetry.} Recall that the model $\mathbb{M}^{\vpt,\pit}_{\eb}$ possesses a first-class constraint $\Cct(x)\approx 0$, which generates a gauge symmetry. This constraint is given by Equation \eqref{Eq:Ctilde}, where one has to remember that in the present case, the Takiff current $\Jtt\is0(x)$ is the one of a PCM realisation without Wess-Zumino term, \textit{i.e.} simply $X(x)$ (see Appendix \ref{App:PCM+WZ}). Thus, here, the constraint is
\begin{equation}\label{Eq:CTildeNoWZ}
\Cct(x) = \sum_{\alpha\in\Si} \J\alpha 0(x) + X(x).
\end{equation}

Let us now show that the deformed model with twist function $\vpt_\eta(\zt)$ is subject to the same constraint and gauge symmetry. It is clear from Equation \eqref{Eq:DefTwistDSE} that
\begin{equation*}
\res_{\zt = \infty} \vpt_\eta(\zt) \dd\zt = -\sum_{\alpha\in\Si} \lst\alpha 0.
\end{equation*}
However, recall from Lemma \ref{Lem:NewLevels2} that the levels $\lst\alpha 0$ of Takiff mode 0 are equal to the ones $\ls\alpha 0$ of the initial model. Moreover, recall that we supposed that these levels satisfy the condition \eqref{Eq:NoWZ}, ensuring that the site $\is$ is realised by a PCM realisation without Wess-Zumino term. Thus, we have
\begin{equation*}
\res_{\zt = \infty} \vpt_\eta(\zt) \dd\zt = 0.
\end{equation*}
The deformed twist function $\vp_\eta(\zt)$ then satisfies the first-class condition \eqref{Eq:ResidueInf}, as $\vpt(\zt)$. This condition ensures that we can define a first-class constraint in the deformed model (see Subsection \ref{SubSec:Gauge}). This constraint is defined as
\begin{equation*}
- \res_{\zt=\infty} \Gt_\eta(\zt,x) \dd \zt =  \sum_{\alpha\in\Si} \Jtt\alpha 0(x) + \Jtt{(+)} 0(x) + \Jtt{(-)}0(x). 
\end{equation*}
Recall from Equation \eqref{Eq:NewCurrents2} that for $\alpha\in\Si$, one has $\Jtt\alpha 0(x)=\J\alpha 0(x)$. Moreover, it is clear from the expression \eqref{Eq:CurrentYBDef} of the currents $\Jtt\pms 0(x)$ that their sum is simply equal to the field $X(x)$. Thus, the constraint above exactly coincides with the one \eqref{Eq:CTildeNoWZ} of the undeformed model $\mathbb{M}^{\vpt,\pit}_{\eb}$.

This means that the deformed model possesses the same gauge symmetry as the undeformed one. In particular, this gauge symmetry acts on the field $g(x)$ by right translation, as in Equation \eqref{Eq:GaugeOnG}. As for the undeformed model $\mathbb{M}^{\vpt,\pit}_{\eb}$, one can gauge-fix the deformed model by eliminating the field $g(x)$, requiring $g(x)\gf \Id$ as in Equation \eqref{Eq:GaugeFix}. Similarly to what we showed in Subsection \ref{SubSubSec:GaugeFixing}, through this gauge fixing, the space of physical observables of the deformed model can then be identified with the observables $\Ac$ of the initial non-gauged model $\mathbb{M}^{\vp,\pi}_{\eb}$.

\paragraph{Hamiltonian.} Let us achieve the construction of the deformed model by specifying its Hamiltonian. Recall that the first-class Hamiltonian of the model $\mathbb{M}^{\vpt,\pit}_{\eb}$ is defined as $\widetilde{\Hc}_0=\sum_{i=1}^M \epsilon_i \widetilde{\Q}_i$. This definition is based on a choice of coefficients $\eb=(\epsilon_1,\cdots,\epsilon_M)$ (recall from Paragraph \ref{Par:NewHam} that these coefficients also coincide with the ones defining the initial non-constrained model $\mathbb{M}^{\vp,\pi}_{\eb}$) and the construction of the quadratic charges $\widetilde{\Q}_i$, associated with the zeroes $\widetilde{\zeta}_i$ of $\vpt(\zt)$, which are extracted from $\Gt(\zt,x)$ and $\vpt(\zt)$ as in Equation \eqref{Eq:QTilde}. 

As the twist function $\vpt_\eta(\zt)$ of the deformed model is a deformation of $\vpt(\zt)$, its zeroes $\widetilde{\ze}^{(\eta)}_i$ are deformations of the zeroes $\widetilde{\ze}_i$ of $\vpt(\zt)$. Let us then introduce the quadratic charges
\begin{equation*}
\widetilde{\Q}^{(\eta)}_i = \res_{\;\;\zt=\widetilde{\ze}^{(\eta)}_i} \widetilde{\Q}^{(\eta)}(\zt)\dd\zt, \;\;\;\;\;\; \text{ where } \;\;\;\;\;\; \widetilde{\Q}^{(\eta)}(\zt) = -\frac{1}{2\vp_\eta(\zt)} \int_{\D} \dd x \; \kappa\Bigl( \Gt_\eta(\zt,x),\Gt_\eta(\zt,x) \Bigr).
\end{equation*}
It is clear that in the limit $\eta\to 0$, these charges tend to the charges $\widetilde{\Q}_i$. Let us then define the first-class Hamiltonian of the deformed model as
\begin{equation*}
\widetilde{\Hc}^{(\eta)}_0 = \sum_{i=1}^M \epsilon_i\,\widetilde{\Q}^{(\eta)}_i,
\end{equation*}
using the same coefficients $\epsilon_i$ as in the undeformed model $\mathbb{M}^{\vpt,\pit}_{\eb}$. By construction, this Hamiltonian then satisfies
\begin{equation*}
\widetilde{\Hc}^{(\eta)}_0 \xrightarrow{\eta\to 0} \widetilde{\Hc}_0.
\end{equation*}
The deformed model is the constrained realisation of affine Gaudin model $\mathbb{M}^{\vpt_\eta,\pit_\eta}_{\eb}$, where $\pit_\eta$ is the realisation constructed from the choice of Takiff currents made above. As such, it is an integrable field theory. Thus, we have constructed an integrable deformation of the model $\mathbb{M}^{\vpt,\pit}_{\eb}$ (and, as a consequence, of the model $\mathbb{M}^{\vp,\pi}_{\eb}$, which is equivalent to $\mathbb{M}^{\vpt,\pit}_{\eb}$).

Let us end this paragraph by the following remark. If the initial model $\mathbb{M}^{\vp,\pi}_{\eb}$ is Lorentz invariant, the coefficients $\epsilon_i$ are equal to either $+1$ or $-1$ (see Subsection 2.4.2 of~\nms and Subsection \ref{SubSubSec:Lorentz} of this article). As we suppose that the deformed model $\mathbb{M}^{\vpt_\eta,\pit_\eta}_{\eb}$ is defined using the same parameters $\eb$, it follows that it is then also automatically Lorentz invariant. This ensures that the diagonal Yang-Baxter deformation of a relativistic model is still relativistic.

\paragraph{Poisson-Lie $\bm q$-deformed diagonal symmetry.} As explained in~\nm, based on the results of~\cite{Delduc:2013fga} (see also~\cite{Kawaguchi:2012ve,Kawaguchi:2011pf,Kawaguchi:2012gp}), the inhomogeneous Yang-Baxter deformation of an affine Gaudin model with a PCM realisation breaks the global symmetry associated with this realisation and in fact transforms it into a $q$-deformed Poisson-Lie symmetry (for models defined on the real line $\D=\R$). Thus, the deformed model $\mathbb{M}^{\vpt_\eta,\pit_\eta}_{\eb}$ that we constructed in the previous paragraphs possesses a Poisson-Lie symmetry. This symmetry is a $q$-deformation of the global symmetry of the undeformed model $\mathbb{M}^{\vpt,\pit}_{\eb}$ associated with the PCM realisation attached to the site $\is$.

The infinitesimal action of this undeformed symmetry on the field $g(x)$ of the PCM realisation takes the form of a left multiplication $\delta g(x) = -\epsilon \, g(x)$, with $\epsilon$ in $\g_0$ (see~\nm, Subsection 3.3.4). In the deformed model, this action is replaced by a Poisson-Lie transformation, which has been described explicitly in~\cite{Delduc:2016ihq}: it also takes the form of a left multiplication on $g$, but with the infinitesimal parameter $\epsilon$ replaced by a non-local field-dependent quantity.

This describes the action of the $q$-deformed symmetry in the gauged model $\mathbb{M}^{\vpt_\eta,\pit_\eta}_{\eb}$. Recall however that to see this model as a deformation of the initial model $\mathbb{M}^{\vp,\pi}_{\eb}$, one has to fix the gauge-symmetry of the model. In particular, under this gauge-fixing procedure, the field $g(x)$ is eliminated (as the gauge-fixing condition is $g(x)\gf\Id$). It is thus natural to wonder how the $q$-deformed symmetry is realised in the gauge-fixed formulation of the model $\mathbb{M}^{\vpt_\eta,\pit_\eta}_{\eb}$. In the undeformed model, we showed in Subsection \ref{SubSubSec:GlobalSymGauging} that under the gauge-fixing, the undeformed symmetry is identified with the diagonal symmetry of the initial model, which in particular acts on the Takiff currents $\J\alpha p(x)$ by conjugation: $\J\alpha p(x) \mapsto h^{-1}\,\J\alpha p(x) h$. Thus, one should expect the $q$-deformed symmetry to be a modification of this diagonal symmetry. It would be interesting to study this question in more details.

\subsubsection{Deformation of the initial non-constrained model}

Recall that the initial model $\mathbb{M}^{\vp,\pi}_{\eb}$ that we started from is identified with the model $\mathbb{M}^{\vpt,\pit}_{\eb}$ through a change of spectral parameter (see Subsection \ref{SubSec:ChangeSpec2}) and a gauge-fixing procedure (see Subsection \ref{SubSec:Equivalence}). Thus, to see the deformed model $\mathbb{M}^{\vpt_\eta,\pit_\eta}_{\eb}$ as a deformation of the initial model $\mathbb{M}^{\vp,\pi}_{\eb}$, one has to perform similar manipulations. Let us start by considering the change of spectral parameter.

\paragraph{Back to the initial spectral parameter.} The passage from $\mathbb{M}^{\vp,\pi}_{\eb}$ to $\mathbb{M}^{\vpt,\pit}_{\eb}$ was made by considering the new spectral parameter $\zt=f(z)$, where $f$ is a Mobius transformation \eqref{Eq:Mobius}. Let us now perform the inverse transformation $\zt\mapsto z=f^{-1}(\zt)$ on the deformed model $\mathbb{M}^{\vpt_\eta,\pit_\eta}_{\eb}$, to obtain an equivalent model $\mathbb{M}^{\vp_\eta,\pi_\eta}_{\eb}$, whose spectral parameter is the initial parameter $z$. The model $\mathbb{M}^{\vpt_\eta,\pit_\eta}_{\eb}$ has sites $\alpha\in\Sit^{(\eta)}=\Si\sqcup\lbrace (+),(-)\rbrace$, with positions $\zt_\alpha$. For $\alpha\in\Si$, this position $\zt_\alpha$ is the same as in the undeformed model $\mathbb{M}^{\vpt,\pit}_{\eb}$ and is thus equal to $\zt_\alpha=f(z_\alpha)$, where $z_\alpha$ is the position of the site $\alpha\in\Si$ in the initial model $\mathbb{M}^{\vp,\pi}_{\eb}$ (see Equation \eqref{Eq:NewPositions2}). Thus, under the inverse transformation $\zt \mapsto z=f^{-1}(\zt)$, the sites $\alpha \in \Si \subset \Sit^{(\eta)}$ are sent back to their initial positions $z_\alpha$.

Let us now consider the sites $(\pm)$. Their positions are a deformation $\zt_\pm=\zt_\is\pm\cc\,\eta$ of the position $\zt_\is$ of the site $\is$ of the model $\mathbb{M}^{\vpt,\pit}_{\eb}$. Recall from Equation \eqref{Eq:PosInf} that this position is equal to $\zt_\is=f(\infty)=a/c$. Under the inverse transformation $\zt \mapsto z=f^{-1}(\zt)$, the sites $(\pm)$ are then sent to the positions\footnote{Note here the different notations for the numbers $c$ and $\cc$, which have different origins and thus should not be confused one with another.}
\begin{equation}\label{Eq:Zpm}
z_\pm = f^{-1}(\zt_\pm) = - \frac{d}{c} \mp \frac{ad-bc}{c^2} \frac{1}{\cc \,\eta}.
\end{equation}

Note in particular that for a non-zero deformation parameter $\eta$, these positions $z_\pm$ are in the finite complex plane $\C$. The model $\mathbb{M}^{\vp_\eta,\pi_\eta}_{\eb}$ that we are constructing by performing the inverse transformation $\zt \mapsto z=f^{-1}(\zt)$ on the deformed model $\mathbb{M}^{\vpt_\eta,\pit_\eta}_{\eb}$ will then not have a site at infinity. Thus, it will also be a constrained model and the change of spectral parameter should then be performed following the method of Subsection \ref{SubSec:ChangeSpec}. Let us pause here to analyse the situation. Although the change of spectral parameter $z\mapsto\zt$ performed originally in Subsection \ref{SubSec:ChangeSpec2} relates the non-constrained model $\mathbb{M}^{\vp,\pi}_{\eb}$ to the constrained model $\mathbb{M}^{\vpt,\pit}_{\eb}$, the inverse transformation $\zt\mapsto z$, when applied to the deformed model $\mathbb{M}^{\vpt_\eta,\pit_\eta}_{\eb}$, yields another constrained model $\mathbb{M}^{\vp_\eta,\pi_\eta}_{\eb}$. This is due to the fact that to perform the deformation, we split the site $\is$ into two sites $\pms$. This site $\is$ keeps track in the constrained model $\mathbb{M}^{\vpt,\pit}_{\eb}$ of the special status of the site at $z=\infty$ in the initial model $\mathbb{M}^{\vp,\pi}_{\eb}$, which in particular is responsible for the absence of constraints in this model. By deforming the site $\is$, we broke this structure, resulting in the appearance of a constraint in the model $\mathbb{M}^{\vp_\eta,\pi_\eta}_{\eb}$.

\paragraph{The deformed twist function $\bm{\vp_\eta(z)}$.} Let us now proceed to describe the model $\mathbb{M}^{\vp_\eta,\pi_\eta}_{\eb}$. Its twist function $\vp_\eta(z)$ is related to $\vpt_\eta(\zt)$ by
\begin{equation*}
\vp_\eta(z) \dd z = \vpt_\eta(\zt) \dd\zt.
\end{equation*}
Let us notice that Equation \eqref{Eq:DefTwistDSE} can be rewritten as
\begin{equation*}
\vpt_\eta(\zt) = \vpt(\zt) + \xi(\zt), \;\;\;\;\;\;  \text{ with } \;\;\;\;\; \xi(\zt) = \frac{1}{2\cc \gamma}\frac{1}{\zt-\zt_+} - \frac{1}{2\cc \gamma}\frac{1}{\zt-\zt_-} - \frac{\lst\is 1}{(\zt-\zt_\is)^2}.
\end{equation*}
Yet, one has by construction $\vpt(\zt)\dd\zt=\vp(z)\dd z$, hence
\begin{equation*}
\vp_\eta(z)\dd z = \vp(z) \dd z + \xi(\zt)\dd\zt = \left( \vp(z) - \frac{ad-bc}{(cz+d)^2} \; g\left( \frac{a z+b}{c z+d} \right) \right) \dd z.
\end{equation*}
Recall from Lemma \ref{Lem:NewLevels2} the relation between $\lst\is 1$ and the constant term $\ell^\infty$ of $\vp(z)$. A direct computation then yields:
\begin{equation}\label{Eq:DeformTwist}
\vp_\eta(z) = \sum_{\alpha\in\Si} \sum_{p=0}^{m_\alpha-1} \frac{\ls\alpha p}{(z-z_\alpha)^{p+1}} + \frac{1}{2\cc \gamma}\frac{1}{z-z_+} - \frac{1}{2\cc \gamma}\frac{1}{z-z_-}.
\end{equation}
The sites of the model $\mathbb{M}^{\vp_\eta,\pi_\eta}_{\eb}$ with twist function $\vp_\eta(z)$ can be described by the same labels $\Sit^{(\eta)}=\Si \sqcup \lbrace (+),(-) \rbrace$ as the model $\mathbb{M}^{\vpt_\eta,\pit_\eta}_{\eb}$, as one can expect from the general method of Subsection \ref{SubSec:ChangeSpec}.

\paragraph{The Gaudin Lax matrix $\bm{\Gamma_\eta(z,x)}$.} Let us now describe the Gaudin Lax matrix $\Gamma_\eta(z,x)$ of the model $\mathbb{M}^{\vp_\eta,\pi_\eta}_{\eb}$. It has poles of order $m_\alpha$ at the positions $z_\alpha$ of the sites $\alpha\in\Sit^{(\eta)}=\Si \sqcup \lbrace (+),(-) \rbrace$. We want to determine the coefficients $\J{\alpha,(\eta)}p(x)$, $p\in\lbrace 0,\cdots,m_\alpha-1\rbrace$, of these poles. For that, recall that we are performing the change of spectral parameter $\zt \mapsto z=f^{-1}(\zt)$, from the model $\mathbb{M}^{\vpt_\eta,\pit_\eta}_{\eb}$ to the model $\mathbb{M}^{\vp_\eta,\pi_\eta}_{\eb}$. Following Subsection \ref{SubSec:ChangeSpec}, the currents $\J{\alpha,(\eta)}p (x)$ are then obtained from the Takiff currents $\Jtt\alpha p(x)$ of the model $\mathbb{M}^{\vpt_\eta,\pit_\eta}_{\eb}$ by applying Proposition \ref{Prop:NewTakiff}, but replacing the currents $\J\alpha p(x)$ by $\Jtt\alpha p(x)$ and the coefficients $(a,b,c,d)$ by the corresponding coefficients of the inverse M\"obius transformation $f^{-1}$.

Let us first consider the sites $\alpha\in\Si \subset \Sit^{(\eta)}$. In this case, the Takiff currents $\Jtt\alpha p(x)$ of the model $\mathbb{M}^{\vpt_\eta,\pit_\eta}_{\eb}$ are the same as the ones of the undeformed model $\mathbb{M}^{\vpt,\pit}_{\eb}$. Thus, they are obtained from the currents $\J\alpha p(x)$ of the initial model $\mathbb{M}^{\vp,\pi}_{\eb}$ by the transformation \eqref{Eq:NewCurrents2}, which was a direct application of Proposition \ref{Prop:NewTakiff}. To obtain the currents $\J{\alpha,(\eta)}p (x)$, one then apply the inverse transformation, corresponding to the M\"obius map $f^{-1}$. Thus, the currents $\J{\alpha,(\eta)}p (x)$ actually coincide with the initial currents $\J\alpha p(x)$ for the sites $\alpha\in\Si$.

Let us now turn our attention to the sites $\pms$. These are sites of multiplicity one, which are then associated with only one current: $\Jtt\pms 0(x)$ in the model $\mathbb{M}^{\vpt_\eta,\pit_\eta}_{\eb}$ and $\J{\pms,(\eta)}0(x)$ in the model $\mathbb{M}^{\vp_\eta,\pi_\eta}_{\eb}$. According to Proposition \ref{Prop:NewTakiff}, such currents of Takiff mode 0 are not modified under a change of spectral parameter. Thus, the currents $\J{\pms,(\eta)}0(x)$ are simply equal to the currents $\Jtt\pms 0(x)$, which are given by Equation \eqref{Eq:CurrentYBDef}. In the end, the Gaudin Lax matrix $\Gamma_\eta(z,x)$ of the model $\mathbb{M}^{\vp_\eta,\pi_\eta}_{\eb}$ is then
\begin{equation}\label{Eq:GammaDef}
\Gamma_\eta(z,x) = \sum_{\alpha\in\Si} \sum_{p=0}^{m_\alpha-1} \frac{\J\alpha p(x)}{(z-z_\alpha)^{p+1}} + \frac{\Jtt {(+)} 0(x)}{z-z_+} + \frac{\Jtt {(-)} 0(x)}{z-z_-}.
\end{equation}
Let us note that it can be expressed simply in terms of the Gaudin Lax matrix \eqref{Eq:GammaNonConst} of the initial model $\mathbb{M}^{\vp,\pi}_{\eb}$ as
\begin{equation*}
\Gamma_\eta(z,x) = \Gamma(z,x) + \frac{\Jtt {(+)} 0(x)}{z-z_+} + \frac{\Jtt {(-)} 0(x)}{z-z_-}.
\end{equation*}

\paragraph{Constraint and gauge-fixing.} As expected from the general method developed in Subsection \ref{SubSec:ChangeSpec}, the model $\mathbb{M}^{\vp_\eta,\pi_\eta}_{\eb}$ is a constrained model. Its first-class constraint $\Cc(x)$ can be extracted as (minus) the residue at $z=\infty$ of the 1-form $\Gamma_\eta(z,x)\dd z$. As explained in Subsection \ref{SubSec:ChangeSpec}, this constraint in fact coincides with the one $\Cct(x)$ of the model $\mathbb{M}^{\vpt_\eta,\pit_\eta}_{\eb}$ before the change of spectral parameter $\zt \mapsto z$ and is thus given by Equation \eqref{Eq:CTildeNoWZ}. In particular, it generates a gauge symmetry which acts on the field $g(x)$ by right multiplication.

The final step to describe the model $\mathbb{M}^{\vp_\eta,\pi_\eta}_{\eb}$ as a deformation of the initial model $\mathbb{M}^{\vp,\pi}_{\eb}$ is to fix its gauge-symmetry. As previously, we do it by imposing the gauge-fixing condition $g(x)\gf\Id$, in addition of the initial constraint $\Cc(x)\gf 0$. As $\Cc(x)$ is given by Equation \eqref{Eq:CTildeNoWZ}, this gauge-fixing allows us to eliminate the degrees of freedom $g(x)$ and $X(x)$ attached to the sites $\pms$ of the model:
\begin{equation*}
g(x)\gf\Id \;\;\;\;\;\;\;\; \text{ and } \;\;\;\;\;\;\;\; X(x) \gf - \sum_{\alpha\in\Si} \J\alpha 0(x).
\end{equation*}
The algebra of observables of the model is then identified with the observables $\Ac$ of the initial model $\mathbb{M}^{\vp,\pi}_{\eb}$ (see Subsection \ref{SubSubSec:GaugeFixing}).

Let us now describe the Gaudin Lax matrix $\Gamma_\eta(z,x)$ of the gauge-fixed deformed model $\mathbb{M}^{\vp_\eta,\pi_\eta}_{\eb}$. For that, we start with its non-gauged fixed expression \eqref{Eq:GammaDef} and insert the above gauge-fixing conditions in the definition \eqref{Eq:CurrentYBDef} of the currents $\Jtt\pms 0(x)$. We then get
\begin{equation}\label{Eq:GammaEtaGF}
\Gamma_\eta(z,x) \gf \sum_{\alpha\in\Si} \sum_{p=0}^{m_\alpha-1} \frac{\J\alpha p(x)}{(z-z_\alpha)^{p+1}} + \frac{1}{2\cc} \left( \frac{R^-}{z-z_+} - \frac{R^+}{z-z_-} \right) \left( \,\sum_{\alpha\in\Si} \J\alpha 0(x) \right),
\end{equation}
where the operators $R^\pm$ are defined as $R^\pm=R\pm \cc\,\Id$. Recall that in the limit $\eta\to 0$, the posititions $z_\pm$ diverge to $\infty$ (see Equation \eqref{Eq:Zpm}). Thus, in this limit, we have
\vspace{-4pt}\begin{equation*}
\lim_{\eta\to 0} \,\Gamma_\eta(z,x) \gf \Gamma(z,x), \vspace{-4pt}
\end{equation*}
as expected.

\subsubsection{Integrable structure of the deformed model}
\label{SubSubSec:IntegrabilityGaugingDeformed}

In the previous subsections, we constructed an integrable deformation of the model $\mathbb{M}^{\vp,\pi}_{\eb}$. This deformed model have several formulations, using different spectral parameters and either being considered a constrained model or a gauge-fixed model. Let us now describe the integrable structure of this model in its different formulations.

\paragraph{The constrained models.} Let us first consider the models $\mathbb{M}^{\vp_\eta,\pi_\eta}_{\eb}$ and $\mathbb{M}^{\vpt_\eta,\pit_\eta}_{\eb}$ in their constrained formulations. These formulations are described by the formalism developed in Section \ref{Sec:AGM}. In particular, their integrability first manifests itself in the fact that their equations of motion take the form of a zero curvature equation, according to Theorem \ref{Thm:Lax}. Note that their Lax pairs $\bigl(\Lc_\eta(z,x),\Mc_\eta(z,x)\bigr)$ and $\bigl(\Lct_\eta(\zt,x),\widetilde{\Mc}_\eta(\zt,x)\bigr)$ are equal, at least $\approx$-weakly, according to Paragraph \ref{Par:TransfoLax}.

Moreover, the Hamiltonian integrability of the constrained models $\mathbb{M}^{\vp_\eta,\pi_\eta}_{\eb}$ and $\mathbb{M}^{\vpt_\eta,\pit_\eta}_{\eb}$ is ensured by the fact that their Lax matrices $\Lc_\eta(z,x)$ and $\Lct_\eta(\zt,x)$ satisfy a Maillet bracket, following the results of Subsection \ref{SubSubSec:Maillet}. The $\Rc$-matrices of these Maillet brackets are characterised by their respective twist functions $\vp_\eta(z)$ and $\vpt_\eta(\zt)$, given by Equations \eqref{Eq:DeformTwist} and \eqref{Eq:DefTwistDSE}:
\begin{equation}\label{Eq:RMatDeformed}
\Rc^\eta\ti{12}(z,w) = \frac{C\ti{12}}{w-z} \vp_\eta(z)^{-1} \;\;\;\;\;\;\; \text{ and } \;\;\;\;\;\;\; \Rct^\eta\ti{12}(\zt,\wt) = \frac{C\ti{12}}{\wt-\zt} \vpt_\eta(\zt)^{-1}
\end{equation}
Let us note here that although we are considering constrained models, the Maillet bracket satisfied by $\Lc_\eta(z,x)$ and $\Lct_\eta(\zt,x)$ holds strongly, \textit{i.e.} without imposing the constraint.

As explained in Subsection \ref{SubSubSec:Maillet}, the fact that the models $\mathbb{M}^{\vp_\eta,\pi_\eta}_{\eb}$ and $\mathbb{M}^{\vpt_\eta,\pit_\eta}_{\eb}$ are described by Maillet brackets with twist function implies the existence of infinite hierarchies of conserved local charges in involution in these models, associated with the zeroes of their twist functions. The  hierarchies of these two models, which are also a strong mark of their integrable properties, coincide, at least $\approx$-weakly.

\paragraph{The gauge-fixed formulation of the model $\pmb{\mathbb{M}}\bm{^{\vpt_\eta,\pit_\eta}_{\eb}}$.} Let us now consider the model $\mathbb{M}^{\vpt_\eta,\pit_\eta}_{\eb}$ in its gauge-fixed formulation. It is obvious that the equations of motion of this model are described by a Lax pair, which is given by gauge-fixing the Lax pair $\bigl(\Lct_\eta(\zt,x),\widetilde{\Mc}_\eta(\zt,x)\bigr)$. However, the gauge-fixing procedure does not preserve the Poisson structure of the model and thus in particular the Maillet bracket of the Lax matrix. Let us then determine what form does the Poisson bracket of $\Lct_\eta(\zt,x)$ take in the gauge-fixed model.

To do so, let us take some inspiration from a similar computation that we did in Subsection \ref{SubSubSec:IntegrabilityGauging} and Appendix \ref{App:DiracUndeformed}, where we determined the Poisson bracket of the gauge-fixed undeformed Lax matrix $\Lct(\zt,x)$ using the Dirac bracket $\lbrace \cdot,\cdot\rbrace^*$. We shall use the same method here. For simplicity, we present the details of the computation in Appendix \ref{App:DiracDeformed}. In the end, we find that the bracket of the gauge-fixed Lax matrix $\Lct_\eta(\zt,x)$ takes the form of a Maillet bracket with $\Rc$-matrix (see \eqref{Eq:RMatGFDeformed})
\begin{equation}\label{Eq:RMatYB}
\Rct^{\eta,GF}\ti{12}(\zt,\wt) = \Rct^{\eta,0}\ti{12}(\zt,\wt) \vpt_\eta(\wt)^{-1},
\end{equation}
where
\begin{equation*}
\Rct^{\eta,0}\ti{12}(\zt,\wt) = \frac{C\ti{12}}{\wt-\zt} - \frac{1}{2} \left( \frac{1}{\wt-\zt_+} + \frac{1}{\wt-\zt_-} \right) C\ti{12} - \frac{1}{2\cc} \left( \frac{1}{\wt-\zt_+} - \frac{1}{\wt-\zt_-} \right) R\ti{12},
\end{equation*}
with $R\ti{12}=R\ti{1}\,C\ti{12}$. These equations have to be compared with the equivalent ones \eqref{Eq:RMatGF1} and \eqref{Eq:RMatGF} for the undeformed model. Let us recall that in the undeformed limit $\eta\to 0$, we have $\zt_\pm \to \zi$. It is then clear that
\begin{equation*}
\Rct^{\eta,0}\ti{12}(\zt,\wt) \xrightarrow{\eta\to 0} \Rct^{0}\ti{12}(\zt,\wt),
\end{equation*}
as one should expect.

The matrix $\Rct^{\eta,0}\ti{12}(\zt,\wt)$ is similar to the $\Rc$-matrix obtained in Section 6 of~\cite{Delduc:2015xdm} (see also~\cite{Lacroix:2018njs,Lacroix:2017isl}), in the context of the study of the bi-Yang-Baxter model in its gauge-fixed formulation. Let us describe its properties. It is a solution of the classical Yang-Baxter equation
\begin{equation}\label{Eq:CYBERmat}
\left[ \Rct^{\eta,0}\ti{12}(\zt_1,\zt_2), \Rct^{\eta,0}\ti{13}(\zt_1,\zt_3) \right] + \left[ \Rct^{\eta,0}\ti{12}(\zt_1,\zt_2), \Rct^{\eta,0}\ti{23}(\zt_2,\zt_3) \right] + \left[ \Rct^{\eta,0}\ti{32}(\zt_3,\zt_2), \Rct^{\eta,0}\ti{13}(\zt_1,\zt_3) \right] = 0,
\end{equation}
ensuring that the Maillet bracket of $\Lct_\eta(\zt,x)$ satisfies the Jacobi identity. In order to prove this, one needs to use the property \eqref{Eq:ComCasimir} of the split Casimir $C\ti{12}$ and the fact that the operator mCYBE \eqref{Eq:mCYBE} on $R$ translates to the ``matrix'' mCYBE
\begin{equation*}
\left[ R\ti{12}, R\ti{13} \right] + \left[ R\ti{12}, R\ti{23} \right] + \left[ R\ti{13},R\ti{23} \right] = -\cc^2 \left[ C\ti{12}, C\ti{13} \right].
\end{equation*}
Moreover, similarly to $\Rct^{0}\ti{12}(\zt,\wt)$, $\Rct^{\eta,0}\ti{12}(\zt,\wt)$ is equal to the standard $\Rc$-matrix $C\ti{12}/(\wt-\zt)$ plus a matrix depending only on the second spectral parameter. As already mentioned for $\Rct^{0}\ti{12}(\zt,\wt)$, this type of $\Rc$-matrix have been studied in~\cite{Lacroix:2017isl}, where it was proved that models with such $\Rc$-matrices also possess integrable hierarchies of conserved local charges in involution associated with the zeroes of the twist function. In the present case, these hierarchies will coincide with the ones mentioned in the previous paragraph for the non-gauge-fixed model, when considered under $\gf$-weak equalities.

\paragraph{The gauge-fixed formulation of the model $\pmb{\mathbb{M}}\bm{^{\vp_\eta,\pi_\eta}_{\eb}}$.} Finally, let us consider the model $\bm{\mathbb{M}^{\vp_\eta,\pi_\eta}_{\eb}}$ after gauge-fixing. According to the general results of Subsection \ref{SubSubSec:EquivChangeSpec}, its Lax matrix is directly related to the one of the model $\bm{\mathbb{M}^{\vpt_\eta,\pit_\eta}_{\eb}}$ by the change of spectral parameter $\zt\mapsto z$:
\begin{equation*}
\Lc_\eta(z,x) \approx \Lct_\eta(\zt,x) = \Lct_\eta\bigl( f(z),x \bigr).
\end{equation*}
As this equation is true for $\approx$-weak equalities, it is also true for $\gf$-weak equalities, \textit{i.e.} after gauge fixing. Thus, one can determine the Poisson bracket of the gauge-fixed Lax matrix $\Lc_\eta(z,x)$ from the one of $\Lct_\eta(\zt,x)$ (similarly to the first method we used in Subsection \ref{SubSubSec:IntegrabilityGauging}). More precisely, we find that $\Lc_\eta(z,x)$ satisfies a Maillet bracket with $\Rc$-matrix
\begin{equation*}
\Rc^{\eta,GF}\ti{12}(z,w) = \Rct^{\eta,GF}\ti{12}\bigl( f(z), f(w) \bigr).
\end{equation*}
After a few manipulations, we get $\Rc^{\eta,GF}\ti{12}(z,w)=\Rc^{\eta,0}\ti{12}(z,w)\vp_\eta(w)^{-1}$ with
\begin{equation}\label{Eq:R0YBgf}
\Rc^{\eta,0}\ti{12}(z,w) = \frac{C\ti{12}}{w-z} - \frac{1}{2} \left( \frac{1}{w-z_+} + \frac{1}{w-z_-} \right) C\ti{12} - \frac{1}{2\cc} \left( \frac{1}{w-z_+} - \frac{1}{w-z_-} \right) R\ti{12}.
\end{equation}
This $\Rc$-matrix satisfies similar properties than the one $\Rct^{\eta,0}\ti{12}(\zt,\wt)$.

\subsubsection{Homogeneous diagonal Yang-Baxter deformation}
\label{SubSubSec:Homogeneous}

\paragraph{Principle.} So far, we have constructed and studied the diagonal inhomogeneous Yang-Baxter deformation of the model $\mathbb{M}^{\vp,\pi}_{\eb}$. Let us end this subsection by saying a few words about the corresponding homogeneous deformation (such a deformation was first considered in~\cite{Kawaguchi:2014qwa} in the context of the $AdS_5\times S^5$ superstring). The inhomogeneous deformation is obtained by deforming the PCM realisation at site $\is$ into an inhomogeneous Yang-Baxter realisation, attached to two sites $\pms$. Similarly, the homogeneous Yang-Baxter deformation can be obtained by deforming the PCM realisation into a homogeneous Yang-Baxter realisation (see Subsection 4.1 of~\nm). In this case, contrarily to the inhomogeneous one, one does not split the site $\is$ of multiplicity 2 into two sites of multiplicity 1: instead, the homogenenous Yang-Baxter realisation is still attached to the site $\is$ with the same multiplicity 2 and with the same levels $\lst\is p$. The deformation only enters the way the site $\is$ is realised, \textit{i.e.} the expression of the Takiff currents $\Jtt\is p(x)$ attached to the site. More precisely, these currents now become~\cite{Vicedo:2015pna} (see also~\nm, Subsection 4.1.1)
\begin{equation*}
\Jtt{\is,\text{hYB}} 0 (x) = X(x) \;\;\;\;\;\;\;\; \text{ and } \;\;\;\;\;\;\;\; \Jtt{\is,\text{hYB}} 1 (x) = \lst\is 1\, j(x) - R_gX (x).
\end{equation*}
In this equation, the deformation is introduced by the appearance of the term $R_gX(x)$ in the current $\Jtt{\is,\text{hYB}} 1 (x)$, where $R_g = \Ad_g^{-1} \circ R \circ \Ad_g$ and $R:\g_0\rightarrow\g_0$ is a skew-symmetric solution of the CYBE
\begin{equation}\label{Eq:CYBE}
[RX,RY] - R\bigl( [RX,Y] + [X,RY] \bigr) = 0, \;\;\;\;\;\; \forall \, X,Y\in\g_0.
\end{equation}
The twist function of the homogeneously deformed model is thus the same as the one of the undeformed model, \textit{i.e.} $\vpt(\zt)$, and its Gaudin Lax matrix is given by
\begin{equation}\label{Eq:GthYB}
\Gt_{\text{hYB}}(\zt,x) = \Gt_\Si(\zt,x) + \frac{X(x)}{\zt-\zi} + \frac{\lst\is 1\, j(x)-R_gX(x)}{(\zt-\zi)^2},
\end{equation}
with $\Gt_\Si(\zt,x)$ defined as in the expression \eqref{Eq:UndefGt} of the undeformed Gaudin Lax matrix $\Gt(\zt,x)$. It is clear from Equation \eqref{Eq:UndefGt} that $\Gt_{\text{hYB}}(\zt,x)$ is a deformation of $\Gt(\zt,x)$, as one recovers the latter simply by taking $R=0$.

The choice of Takiff currents above corresponds to a choice of realisation $\pit_{\hYB}$, which is different from the realisation $\pit$, although they are both valued in the same algebra $\Act=\Ac\otimes\Ac_{G_0}$ and realise the same Takiff algebra, with sites $\alpha\in\Sit$ and levels $\lst\alpha p$. One then constructs the homogeneously deformed model as the realisation of affine Gaudin model $\mathbb{M}^{\vpt,\pit_{\hYB}}_{\eb}$. By construction, this is an integrable deformation of the model $\mathbb{M}^{\vpt,\pit}_{\eb}$, with $R$ playing the role of deformation parameter.

As the twist function of this deformed model is the same $\vpt(\zt)$ as the undeformed one, it satisfies the same first-class condition \eqref{Eq:NewResidueInf}. This ensures that the deformed model possesses a first-class constraint, as the undeformed one. This constraint is obtained as minus the residue at $\zt=\infty$ of  the 1-form $\Gt_{\hYB}(\zt,x)$. One verifies easily that this constraint coincides with the one $\Cct(x)$ of the undeformed model. This ensures that the deformed model shares the same gauge symmetry than the undeformed one and thus the same physical degrees of freedom.\\

Let us note that, as in the rest of this subsection, we supposed here that the site $\is$ in the undeformed model $\mathbb{M}^{\vpt,\pit}_{\eb}$ is associated with a PCM realisation without Wess-Zumino term (and thus that the level $\lst\is 0$ is equal to $0$). It is not known whether one could apply a similar homogeneous Yang-Baxter deformation to a model where the site $\is$ would be associated with a general PCM+WZ realisation, as no homogeneous Yang-Baxter realisation with a Wess-Zumino term has been exhibited in the literature so far.

\paragraph{Homogeneous deformation as a limit of the inhomogeneous one.} Instead of performing the analysis of the homogeneous deformation following the same steps as we did before for the inhomogeneous one, we shall follow here a quicker approach. Indeed, one can see the homogeneous deformation as an appropriate limit of the inhomogeneous one and thus describe most of its properties from the results that we found in the rest of this subsection. Let us explain what is this limit.

Recall that the inhomogeneous Yang-Baxter deformation came with the introduction of a skew-symmetric operator $R$, solution of the mCYBE \eqref{Eq:mCYBE}. The homogeneous CYBE \eqref{Eq:CYBE} can be seen as the limit $\cc\to 0$ of the mCYBE \eqref{Eq:mCYBE}. Although we introduced the parameter $\cc$ to be either equal to $1$ or $i$ to distinguish the split and non-split cases, one easily checks that all the results presented in the previous paragraphs also hold for a general $\cc\in\C^*$. As we shall see, one can see the homogeneous Yang-Baxter deformation as the limit $\cc\to 0$ of the inhomogeneous one.

For that, let us consider the Gaudin Lax matrix $\Gt_\eta(\zt,x)$ of the inhomogeneous deformation, given by Equation \eqref{Eq:DeformGt}. Recalling the expressions \eqref{Eq:ZtPM} and \eqref{Eq:CurrentYBDef} of the positions $\zt_\pm$ and the currents $\Jtt\pms 0(x)$, and in particular their dependence on $\cc$, one can compute the limit $\cc\to 0$ of $\Gt_\eta(\zt,x)$. One then finds
\begin{equation*}
\Gt_{\eta}(\zt,x) \xrightarrow{\cc\to 0} \Gt_\Si(\zt,x) + \frac{X(x)}{\zt-\zi} + \frac{\lst\is 1\, j(x)-\eta\,R_gX(x)}{(\zt-\zi)^2},
\end{equation*}
This limit coincides with the expression \eqref{Eq:GthYB} of the Gaudin Lax matrix $\Gt_{\text{hYB}}(\zt,x)$ of the homogeneous deformation, with the operator $\eta R$ instead of $R$. However, this difference is irrelevant, as if $R$ is a solution of the CYBE \eqref{Eq:CYBE}, so is $\eta R$ (hence the characterisation ``homogeneous''), so that it can be reabsorbed in a redefinition of $R$, or equivalently by setting $\eta=1$. Thus, the homogeneous deformation coincides with the limit $\cc\to 0$ of the inhomogeneous deformation with $\eta=1$.

\paragraph{Integrable structure of the deformed model $\pmb{\mathbb{M}}\,\bm{^{\vpt,\pit_{\hYB}}_{\eb}}$.} By construction, the model $\mathbb{M}^{\vpt,\pit_{\hYB}}_{\eb}$ is integrable. In its non-gauge-fixed formulation, its integrable structure is described by the general results of Section \ref{Sec:AGM}. It possesses a Lax pair formulation and its Lax matrix $\Lct_{\hYB}(\zt,x)$ satisfies a strong Maillet bracket. The $\Rc$-matrix of the latter is the same as the one of the undeformed model $\mathbb{M}^{\vpt,\pit}_{\eb}$, as they share the same twist function $\vpt(\zt)$.\\

Let us now consider the deformed model $\mathbb{M}^{\vpt,\pit_{\hYB}}_{\eb}$ under the gauge fixing $g(x)\gf\Id$ (this still defines a good gauge-fixing condition as the deformed model possesses the same gauge symmetry as the undeformed one). To study the properties of this model, let us use the limit $\cc\to 0$ described above. The Lax matrix of the gauge-fixed inhomogeneous Yang-Baxter deformation satisfies a Maillet bracket with $\Rc$-matrix \eqref{Eq:RMatYB}. Considering the limit $\cc\to 0$ of this bracket, we obtain that the gauge-fixed Lax matrix $\Lct_{\hYB}(\zt,x)$ of the homogeneous deformation also satisfies a Maillet bracket, with $\Rc$-matrix
\begin{equation*}
\Rct^{\hYB,GF}\ti{12}(\zt,\wt) = \Rct^{\hYB,0}\ti{12}(\zt,\wt) \vpt(\wt)^{-1},
\end{equation*}
where
\begin{equation*}
\Rct^{\hYB,0}\ti{12}(\zt,\wt) = \frac{C\ti{12}}{\wt-\zt} -  \frac{ C\ti{12}}{\wt-\zi}  - \frac{R\ti{12}}{(\wt-\zi)^2},
\end{equation*}
with $R\ti{12}=R\ti{1}\,C\ti{12}$. The matrix $\Rct^{\hYB,0}\ti{12}(\zt,\wt)$ is a deformation of the corresponding matrix \eqref{Eq:RMatGF} for the undeformed model. As a limit of $\Rct^{\eta,0}\ti{12}(\zt,\wt)$, it satisfies similar properties, as for example the CYBE \eqref{Eq:CYBERmat}.

\paragraph{Back to the initial spectral parameter.} Finally, let us discuss the model obtained after performing the inverse transformation of spectral parameter $\zt \mapsto z=f^{-1}(\zt)$, which is then a direct deformation of the initial model $\mathbb{M}^{\vp,\pi}_{\eb}$. This model can be seen as the limit $\cc\to0$ of the inhomogeneous deformation $\mathbb{M}^{\vp_\eta,\pi_\eta}_{\eb}$. Its twist function is then the limit when $\cc$ goes to $0$ of the twist function $\vp_\eta(z)$. One checks that this coincides with the twist function $\vp(z)$ of the initial model $\mathbb{M}^{\vp,\pi}_{\eb}$, given by Equation \eqref{Eq:TwistNonConst}. One can expect this, as homogeneous Yang-Baxter deformations do not change the levels of the sites and thus the twist function.

Similarly, the Gaudin Lax matrix $\Gamma_{\hYB}(z,x)$ of the homogeneous deformation can be computed as the limit $\cc\to 0$ of $\Gamma_\eta(z,x)$. Recall that the latter is expressed as \eqref{Eq:GammaDef}, in terms of the currents $\Jtt\pms 0(x)$ and the positions $z_\pm$. Considering their expressions \eqref{Eq:CurrentYBDef} and \eqref{Eq:Zpm}, and in particular their dependence on $\cc$, one can take the limit $\cc\to 0$ in order to determine $\Gamma_{\hYB}(z,x)$, yielding:
\begin{equation}\label{Eq:GammahYB}
\Gamma_{\hYB}(z,x) = \sum_{\alpha\in\Si} \sum_{p=0}^{m_\alpha-1} \frac{\J\alpha p(x)}{(z-z_\alpha)^{p+1}} - \ell^\infty \, j(x) - \frac{c^2}{ad-bc} R_gX(x),
\end{equation}
where we used the relation between $\lst\is 1$ and $\ell^\infty$ found in Lemma \ref{Lem:NewLevels2}.

The Lax matrix of the model is equal to $\Lc_{\hYB}(z,x)=\Gamma_{\hYB}(z,x)/\vp(z)$ and can also be seen as the limit $\cc\to0$ of the Lax matrix $\Lc_\eta(z,x)$ of the inhomogeneous deformation. Its time evolution takes the form of a zero curvature equation and its Poisson bracket is a Maillet bracket with $\Rc$-matrix
\begin{equation*}
\Rc(z,w) = \frac{C\ti{12}}{w-z} \vp(w)^{-1},
\end{equation*}
ensuring the integrability of the model.

\paragraph{Structure of the model.} Let us make a few comments about the structure of the model described above. The matrix $\Gamma_{\hYB}(z,x)$ cannot be interpreted as the Gaudin Lax matrix of a realisation of affine Gaudin model as considered in~\nm. Indeed, contrarily to the Gaudin Lax matrices in~\nm, the matrix $\Gamma_{\hYB}(z,x)$ possesses a constant term. In fact, the model obtained here after the inverse change of spectral parameter $\zt\mapsto z$ belongs to a larger class of realisations of affine Gaudin models than the one considered in~\nm, which fits into the original and most general construction of affine Gaudin models in~\bg. As the ones of~\nm, these more general models possess a site of multiplicity 2 at $z=\infty$. In~\nm, this site only takes the form of the constant term $\ell^\infty$ in the Lax matrix $\vp(z)$. In addition to this level $\ell^\infty$, the models in this larger class also possess a constant term $-\mathcal{J}^\infty (x)$ in their Gaudin Lax matrix. In the case considered above of the homogeneous deformation, this additional current is given by
\begin{equation*}
\mathcal{J}^\infty_{\hYB}(x) = \ell^\infty \, j(x) + \frac{c^2}{ad-bc} R_gX (x).
\end{equation*}
In general, this current should satisfy the following Poisson brackets (see \bg, Corollary 4.7)
\begin{equation*}
\left\lbrace \mathcal{J}^\infty\ti{1}(x), \mathcal{J}^\infty\ti{2}(y) \right\rbrace = 0 \;\;\;\;\;\;\; \text{ and } \;\;\;\;\;\;\; \left\lbrace \J\alpha p\,\ti{1}(x), \mathcal{J}^\infty\ti{2}(y) \right\rbrace = 0, \;\;\;\; \forall \, \alpha\in\Si .
\end{equation*}
The models considered in~\nms simply correspond to the simplest choice $\mathcal{J}^\infty (x)=0$, but most of the results of~\nms still hold for this larger class of models. One can check that the current $\mathcal{J}^\infty_{\hYB}(x)$ indeed satisfies the bracket above (see for example~\nm, Equation (4.5c)) and thus that the homogeneous diagonal Yang-Baxter deformation (when seen as a model with spectral parameter $z$) fits into this more general formalism. Recall however that this deformed model possesses a first-class constraint and a corresponding gauge symmetry, which allow to eliminate the additional degrees of freedom $g(x)$ and $X(x)$ appearing in the current $\mathcal{J}^\infty_{\hYB}(x)$. At the moment, it is not clear how constraints can be introduced in the general treatment of this larger class of models. It would be interesting to study this question in more details.

\paragraph{Structure of the gauge-fixed model.} To end this subsection, let us describe the integrable structure of the model obtained above after gauge-fixing. As before, this gauge-fixing takes the form of the constraints $g(x)\gf\Id$ and $X(x)\gf-\sum_{\alpha\in\Si} \J\alpha 0(x)$. The gauge-fixed expression of the Gaudin Lax matrix $\Gamma_{\hYB}(z,x)$ can then be obtained directly from Equation \eqref{Eq:GammahYB}:
\begin{equation*}
\Gamma_{\hYB}(z,x) = \sum_{\alpha\in\Si} \sum_{p=0}^{m_\alpha-1} \frac{\J\alpha p(x)}{(z-z_\alpha)^{p+1}} + \frac{c^2}{ad-bc} R \left( \sum_{\alpha\in\Si} \J\alpha 0(x) \right).
\end{equation*}
Alternatively, this expression can be computed as the $\cc\to 0$ limit of Equation \eqref{Eq:GammaEtaGF}. The gauge-fixed Lax matrix $\Lc_{\hYB}(z,x)$ of the model then satisfies a Maillet bracket whose $\Rc$-matrix $\Rc^{\hYB,GF}\ti{12}(z,w)$ is the limit when $\cc$ goes to 0 of the $\Rc$-matrix $\Rc^{\eta,GF}\ti{12}(z,w)$ of the gauge-fixed inhomogeneous deformation, given in Equation \eqref{Eq:RMatDeformed}. One finds
\begin{equation*}
\Rc^{\hYB,GF}\ti{12}(z,w) = \Rc^{\hYB,0}\ti{12}(z,w) \vp(w)^{-1}, \;\;\;\;\;\; \text{ with } \;\;\;\;\;\; \Rc^{\hYB,0}\ti{12}(z,w) = \frac{C\ti{12}}{w-z} - \frac{c^2}{ad-bc} R\ti{12}.
\end{equation*}
The matrix $\Rc^{\hYB,0}\ti{12}(z,w)$ satisfies the classical Yang-Baxter equation \eqref{Eq:CYBERmat}.

\section[Application to integrable coupled $\s$-models]{Application to integrable coupled $\bm{\s}$-models}
\label{Sec:SigmaModels}

In this section, we apply the results of Sections \ref{Sec:AGM} and \ref{Sec:Gauging} to the study of integrable $\s$-models. In particular we will discuss the gauged formulation and the homogeneous diagonal Yang-Baxter deformation of the integrable coupled $\s$-models recently introduced in~\cite{Delduc:2018hty}, using the results of Section \ref{Sec:Gauging}. Before that, let us illustrate these ideas on a simpler example mentioned in the introduction of this article as a motivation for the general construction: the Principal Chiral Model (PCM).

\subsection{The principal chiral model}
\label{SubSec:PCM}

\subsubsection{Non-gauged formulation}

The PCM is the simplest example of an integrable $\s$-model. Its interpretation in terms of realisation of affine Gaudin models was first pointed out in~\bgs (see also Subsection 3.2 of~\nm). In these references, the PCM is seen as a non-constrained model, without any gauge symmetry. Let us briefly review its construction, following~\nm. Its twist function is given by
\begin{equation*}
\vp_\pcm(z) = K \frac{1-z^2}{z^2} = \frac{K}{z^2} - K.
\end{equation*}
This twist function corresponds to an affine Gaudin model with one site $(1)$ in the complex plane, at position $z=0$, and one site at $z=\infty$. The site $(1)$ has multiplicity two and levels $\ls{(1)}1=K$ and $\ls{(1)}0=0$, while the site at infinity corresponds to the constant term $\ell^\infty=K$ in the twist function.

The site $(1)$ is associated with a PCM realisation, whose construction is recalled in Appendix \ref{App:PCM+WZ} (note that here we are considering a PCM realisaiton without Wess-Zumino term, as the level $\ls{(1)}0$ vanishes). This realisation is expressed in terms of the fields $\gb 1(x)$ and $\Xb 1(x)$ of the Poisson algebra $\Ac_{G_0}$, describing canonical fields in the cotangent bundle $T^*G_0$. The Gaudin Lax matrix of the model takes the form:
\begin{equation*}
\Gamma_\pcm(z,x) = \frac{\Xb 1(x)}{z} + \frac{K\,\jb 1(x)}{z^2}.
\end{equation*}
The twist function $\vp_\pcm(z)$ possesses two zeroes $\ze_1=+1$ and $\ze_2=-1$. Following the general construction of~\nm, we associate to these zeroes the quadratic charges
\begin{equation*}
\Q^\pcm_i = \res_{z=\ze_i} \Q^\pcm(z)\dd z, \;\;\;\;\; \text{ with } \;\;\;\;\; \Q^\pcm(z) = -\frac{1}{2\vp_\pcm(z)} \int_\D \dd x \; \kappa\bigl( \Gamma_\pcm(z,x), \Gamma_\pcm(z,x) \bigr).
\end{equation*}
The Hamiltonian of the model is then defined as $\Hc_\pcm=\Q_1-\Q_2$ and can be expressed explicitely in terms of the fields $\gb 1(x)$ and $\Xb 1(x)$ (see~\nm, Equation (3.18), with $\kay=0$).\\

One can then compute the time evolution $\p_t \gb 1 = \lbrace \Hc_\pcm, \gb 1 \rbrace$ of the field $\gb 1$, and express it in terms of the ``momentum'' $\Xb 1$. Introducing the temporal current $\jb 1_0=g^{(1)\,-1}\p_t \gb 1$, one finds
\begin{equation}\label{Eq:PCMLagX}
\Xb 1 = K\,\jb 1_0.
\end{equation} 
Expressing the momentum $\Xb 1$ in terms of the time derivative of $\gb 1$ is the first step towards the inverse Legendre transform of the model. The next step is the computation of the action, defined as
\begin{equation*}
S_\pcm\bigl[\gb 1 \bigr] = \iint_{\R\times\D} \dd t\,\dd x \; \kappa\bigl(\Xb 1,\jb 1_0\bigr) - \int_{\R} \dd t \; \Hc_\pcm,
\end{equation*}
where $\Xb 1$ should be replaced by its Lagrangian expression \eqref{Eq:PCMLagX}. In the end, one finds the standard action of the PCM
\begin{equation}\label{Eq:ActionPCM}
S_\pcm\bigl[\gb 1 \bigr] = \frac{K}{2} \iint_{\R\times\D} \dd t\,\dd x \; \kappa\bigl(\jb 1_+,\jb 1_-\bigr),
\end{equation}
where $\jb 1_\pm=g^{(1)\,-1}\p_\pm \gb 1$ are the light-cone currents, obtained form the derivatives $\p_\pm = \p_t \pm \p_x$ with respect to the light-cone coordinates $x^\pm = (t\pm x)/2$.

\subsubsection{Gauged formulation}

In the previous paragraph, we presented the non-gauged formulation of the PCM. According to the general results of Section \ref{Sec:Gauging}, the PCM also admits a gauged formulation. Such a formulation was already known before and has been recalled in the introduction of this article, to motivate the search for similar formulations for arbitrary realisations of affine Gaudin models. Let us then prove that indeed, the results mentioned in the introduction fit into the more general framework developed in Section \ref{Sec:Gauging}.

\paragraph{Change of spectral parameter.} According to Section \ref{Sec:Gauging}, the gauged formulation of the PCM is obtained by considering a change of spectral parameter $z\mapsto \zt=f(z)$, which brings the site at $z=\infty$ of the model to a finite position. In the present case, we will consider the following transformation of the spectral parameter, as in Equation \eqref{Eq:IntroChangeSpecPCM} of the introduction:
\begin{equation*}
f(z) = \frac{1+z}{1-z}.
\end{equation*}
This change of spectral parameter is a Mobius transformation, of the form \eqref{Eq:Mobius} (with corresponding coefficients $a=b=d=1$ and $c=-1$). It sends the site $(1)$ at $z=0$ to the position $\zt_{(1)}=f(0)=+1$ and the site at $z=\infty$ to the finite position $\zi=f(\infty)=-1$. The twist function $\vpt_\pcm(\zt)$ of the model with spectral parameter $\zt$ satisfies
\begin{equation*}
\vp_\pcm(z)\dd z = \vpt_\pcm(\zt)\dd\zt.
\end{equation*}
One then finds, in agreement with Equation \eqref{Eq:IntroTwistPCMGauged} of the introduction,
\begin{equation}\label{Eq:TwistGaugedPCM}
\vpt_\pcm(\zt) = \frac{8K\,\zt}{(\zt^2-1)^2}.
\end{equation}
As expected, this twist function has double poles at $\zt_{(1)}=+1$ and $\zi=-1$, the positions of the sites $(1)$ and $\is$ for the new spectral parameter. The corresponding levels can be checked to be
\begin{equation*}
\lst{(1)}0 = \lst\is 0 = 0 \;\;\;\;\;\; \text{ and } \;\;\;\;\;\; \lst{(1)}1= - \lst\is 1 = 2K.
\end{equation*}
Alternatively, these levels can be obtained using Lemma \ref{Lem:NewLevels2}.

\paragraph{Gaudin Lax matrix.} The Takiff currents $\Jtt{(1)}p(x)$ associated with the site $(1)$ in the gauged model are related to the ones $\J{(1)}p(x)$ of the non-gauged formulation by Proposition \ref{Prop:NewTakiff}. Applying this proposition in the present case, one finds
\begin{equation*}
\Jtt{(1)}0(x) = \J{(1)}0(x) = \Xb 1(x) \;\;\;\;\;\;\;\; \text{ and } \;\;\;\;\;\;\;\; \Jtt{(1)}1(x) = 2\J{(1)}1(x) = 2 K\,\jb 1(x).
\end{equation*}
It is clear that these currents are Takiff currents with levels $\lst{(1)}0=0$ and $\lst{(1)}1=2K$.

As explained for the general construction in Subsection \ref{SubSubSec:NewCurrents}, the site $\is$ is realised by a PCM+WZ realisation. In the present case, as the level $\lst\is 0$ is zero, it is in fact a PCM realisation without Wess-Zumino term. In Subsection \ref{SubSubSec:NewCurrents}, we denoted the canonical fields attached to this realisation by $g(x)$ and $X(x)$: in the present case, in order to keep the notations uniform, we shall instead denote them by $\gb 2(x)$ and $\Xb 2(x)$. The corresponding Takiff currents are then
\begin{equation*}
\Jtt\is 0 (x) = \Xb 2(x) \;\;\;\;\;\;\;\; \text{ and } \;\;\;\;\;\;\;\; \Jtt\is 1(x) = -2K\,\jb 2(x),
\end{equation*}
where $\jb 2=g^{(2)\,-1}\p_x\gb 2$. The Gaudin Lax matrix of the gauged model then reads
\begin{equation*}
\Gt_\pcm(\zt,x) = \frac{\Xb 1(x)}{\zt-1} + \frac{2K\,\jb 1(x)}{(\zt-1)^2} + \frac{\Xb 2(x)}{\zt+1} - \frac{2K\,\jb 2(x)}{(\zt+1)^2}.
\end{equation*}

\paragraph{Constraint and gauge symmetry.} The constraint of the gauged model is defined as
\begin{equation*}
\Cct(x) = - \res_{\zt=\infty} \Gt_\pcm(\zt,x) = \Xb 1(x) + \Xb 2(x).
\end{equation*}
One easily verifies that this constraint is first-class (\textit{i.e.} satisfies Equation \eqref{Eq:FirstClass}), as expected from Subsection \ref{SubSubSec:Gauging}. Thus it generates a gauge symmetry. One checks that this gauge symmetry acts on the fields $\gb 1$ and $\gb 2$ as
\begin{equation}\label{Eq:GaugePCM}
\gb 1(x) \longmapsto \gb 1(x)h(x,t) \;\;\;\;\;\;\;\; \text{ and } \;\;\;\;\;\;\;\; \gb 2(x) \longmapsto \gb 2(x) h(x,t).
\end{equation}

\paragraph{Zeroes of the twist function and Hamiltonian.} From Equation \eqref{Eq:TwistGaugedPCM}, one gets that the twist function 1-form $\vpt_\pcm(\zt)\dd\zt$ possesses two zeroes $\zet_1=\infty$ and $\zet_2=0$. One checks that these zeroes are the images $\zet_i=f(\ze_i)$ under the transformation $z\mapsto\zt=f(z)$ of the zeroes $\ze_1=+1$ and $\ze_2=-1$ of the original twist function $\vp_\pcm(z)$. Note that in the nomenclature of Subsection \ref{SubSec:Zeroes}, this gauged formulation of the PCM belongs to the case (ii), where the twist function has a zero at infinity. Let us introduce the quadratic charges associated with these zeroes:
\begin{equation*}
\widetilde{\Q}^\pcm_i = \res_{\zt=\zet_i} \widetilde{\Q}^\pcm(\zt)\dd \zt, \;\;\;\;\; \text{ with } \;\;\;\;\; \widetilde{\Q}^\pcm(\zt) = -\frac{1}{2\vpt_\pcm(\zt)} \int_\D \dd x \; \kappa\Bigl( \Gt_\pcm(\zt,x), \Gt_\pcm(\zt,x) \Bigr).
\end{equation*}
The Hamiltonian of the model is then defined by
\begin{equation*}
\widetilde{\Hc}_\pcm = \widetilde{\Q}^\pcm_1 - \widetilde{\Q}^\pcm_2 + \int_\D \dd x\; \kappa\bigl( \Cct(x), \mu(x) \bigr),
\end{equation*}
where $\mu(x)$ is a $\g_0$-valued Lagrange multiplier.

\paragraph{Dynamics of $\bm{\gb 1}$ and $\bm{\gb 2}$.} The definition of the Hamiltonian $\widetilde{\Hc}_\pcm$ above specifies the dynamic $\p_t \approx \lbrace \widetilde{\Hc}_\pcm,\cdot \rbrace$ of the model. In particular, one can compute the expression of the temporal currents $\jb 1_0=g^{(1)\,-1}\p_t \gb 1$ and $\jb 2_0=g^{(2)\,-1}\p_t \gb 2$. After a few manipulations, one finds
\begin{align}
\jb 1_0 &\approx \frac{\Xb 1-\Xb 2}{4K} + \frac{\jb 1+\jb 2}{2} + \mu, \\
\jb 2_0 &\approx -\frac{\Xb 1-\Xb 2}{4K} +  \frac{\jb 1+\jb 2}{2} + \mu.
\end{align}
Using the constraint $\Xb 1+\Xb 2 \approx 0$ and taking the difference of the above equations, we get
\begin{equation}\label{Eq:PCMLagXgauged}
\Xb 1 \approx -\Xb 2 \approx K\bigl(\jb 1_0-\jb 2_0\bigr).
\end{equation}

\paragraph{Inverse Legendre transformation.} The action of the model is obtained by the Legendre inverse transformation
\begin{equation*}
S^{\text{gauge}}_\pcm \bigl[\gb 1,\gb 2] = \iint_{\R\times\D} \dd t\,\dd x \; \Bigl( \kappa\bigl(\Xb 1,\jb 1_0\bigr) + \kappa\bigl(\Xb 2,\jb 2_0) \Bigr) - \int_{\R} \dd t \; \Hc_\pcm,
\end{equation*}
where one should replace $\Xb 1$ and $\Xb 2$ by their Lagrangian expression \eqref{Eq:PCMLagXgauged}. Introducing the light-cone currents
\begin{equation*}
\jb 1_\pm = g^{(1)\,-1} \p_\pm \gb 1 = \jb 1_0 \pm \jb 1 \;\;\;\;\;\;\;\; \text{ and } \;\;\;\;\;\;\;\; \jb 2_\pm = g^{(2)\,-1} \p_\pm \gb 2 = \jb 2_0 \pm \jb 2,
\end{equation*}
we find after a few manipulations
\begin{equation}\label{Eq:ActionPCMGauged}
S^{\text{gauge}}_\pcm \bigl[\gb 1,\gb 2 \bigr] = \frac{K}{2} \iint_{\R\times\D} \dd t\,\dd x \; \kappa\bigl( \jb 1_+-\jb 2_+, \jb 1_--\jb 2_- \bigr).
\end{equation}
As announced, we recover by this method the action \eqref{Eq:IntroActionGauged} defined in the introduction.

\paragraph{Gauge symmetry and gauge-fixing.} The action above should be invariant under the local transformation \eqref{Eq:GaugePCM}, as it is a gauge symmetry of the Hamiltonian model. Under this transformation, the light-cone currents $\jb 1_\pm$ and $\jb 2_\pm$ transform as
\begin{equation*}
\jb 1_\pm \longmapsto h^{-1} \jb 1_\pm\hspace{1pt} h + h^{-1} \p_\pm h \;\;\;\;\;\;\;\; \text{ and } \;\;\;\;\;\;\;\; \jb 2_\pm \longmapsto h^{-1} \jb 2_\pm\hspace{1pt} h + h^{-1} \p_\pm h.
\end{equation*}
Thus, the combination $\jb 1_\pm-\jb 2_\pm$ transforms covariantly and the invariance of the action \eqref{Eq:ActionPCMGauged} follows from the conjugacy-invariance of the bilinear form $\kappa$.

One can gauge-fix the model by imposing $\gb 2(x,t) = \Id$. Under this gauge-fixing condition, the currents $\jb 2_\pm$ vanish and the action $S^{\text{gauge}}_\pcm\bigl[\gb 1,\Id]$ then reduces to the non-gauged action $S_\pcm \bigl[\gb 1 \bigr]$. Thus, we have constructed a gauged formulation of the PCM, as announced, \textit{via} the introduction of the new field $\gb 2$. This result illustrates the general gauging procedure outlined in Subsection \ref{SubSec:Gauging}. Indeed, we started with an integrable action $S_\pcm\bigl[\gb 1 \bigr]$, depending on a $G_0$-valued field $\gb 1$ ($\gb 1$ here plays the role of the field $\Phi$ in Subsection \ref{SubSec:Gauging} and $G_0$ the role of the manifold $M$). By applying the gauging procedure, we then constructed an action $S^{\text{gauge}}_\pcm\bigl[\gb 1,\gb 2 \bigr]$, depending on an additional field $\gb 2$ in $G_0$ (which was denoted by $g$ in Subsection \ref{SubSec:Gauging}). This action is invariant under a gauge transformation which in particular acts on the field $\gb 2$ by right multiplication $\gb 2(x,t)\mapsto \gb 2(x,t)h(x,t)$. The gauge-fixing condition $\gb 2(x,t)=\Id$ then produces back the original action $S_\pcm\bigl[\gb 1 \bigr]$, as described in Subsection \ref{SubSec:Gauging}.

\subsubsection{Symmetries and deformations}
\label{SubSubSec:PCMSymAndDef}

\paragraph{Symmetries.} The PCM in its non-gauged formulation \eqref{Eq:ActionPCM} possesses two global symmetries: the left multiplication $\gb 1\mapsto h^{-1}\gb 1$ and the right multiplication $\gb 1\mapsto \gb 1h$. In the language of realisations of affine Gaudin models, the left multiplication symmetry is the global symmetry associated with the PCM realisation at the site $(1)$ of the model (indeed, recall from Subsection 3.3.4 of~\nms that any model with a PCM or PCM+WZ realisation possesses a global symmetry, which acts on the corresponding field $\gb 1$ by left multiplication). Similarly, the right multiplication symmetry is identified with the diagonal symmetry of the model.\\

Let us consider now the gauged formulation \eqref{Eq:ActionPCMGauged} of the PCM. As mentioned above, the gauge symmetry acts on the additional field $\gb 2(x,t)$ by right multiplication. Moreover, it also acts on the initial field $\gb 1(x,t)$ by right multiplication $\gb 1(x,t)\mapsto \gb 1(x,t)h(x,t)$. This is in agreement with the general discussion of Subsection \ref{SubSec:Gauging}. Indeed, as explained in Equation \eqref{Eq:GaugePhi}, the gauge symmetry of the gauged model acts on the initial field of the model by a local version of the diagonal symmetry of the initial model, in this case the right multiplication on $\gb 1$.

The gauged-model \eqref{Eq:ActionPCMGauged} also possesses two global symmetries: the left multiplications $\gb 1\mapsto h^{-1}\gb 1$ and $\gb 2 \mapsto h^{-1} \gb 2$ on $\gb 1$ and $\gb 2$. To understand to which symmetries of the initial model they correspond to, let us recall that one recovers the initial action \eqref{Eq:ActionPCM} from the gauged one \eqref{Eq:ActionPCMGauged} by imposing the gauge-fixing condition $\gb 2(x,t)=\Id$. This can be seen as performing a gauge transformation by $\gb 2(x,t)^{-1}$, under which the first field $\gb 1(x,t)$ becomes the combination $g(x,t)=\gb 1(x,t)\gb 2(x,t)^{-1}$. By gauge invariance, we have $S^{\text{gauge}}_\pcm[\gb 1,\gb 2] = S^{\text{gauge}}_\pcm[g,\Id]$, which is then identified with the non-gauged action $S_\pcm[g]$. The gauge-fixing can thus be understood as considering the gauge-invariant field $g(x,t)$ instead of $\gb 1(x,t)$. It is clear that the left multiplication $\gb 1 \mapsto h^{-1}\gb 1$ on $\gb 1$ induces the left multiplication $g \mapsto h^{-1}g$ on $g=\gb 1\,g^{(2)\,-1}$. Similarly, the left multiplication $\gb 2 \mapsto h^{-1}\gb 2$ on $\gb 2$ induces the right multiplication $g \mapsto gh$ on $g$. Thus, under the gauge-fixing, the left multiplication symmetries on $\gb 1$ and $\gb 2$ are identified with the left and right multiplication of the initial PCM. This is in agreement with the general discussion in Subsection \ref{SubSubSec:GlobalSymGauging}.

\paragraph{Yang-Baxter deformations.} For simplicity, we shall drop the exponent $(1)$ of the field $\gb 1$ in the rest of this subsection (or equivalently, we will consider the model on the gauge-invariant field $g=\gb 1 g^{(2)\,-1}$ instead of $\gb 1$, as in the previous paragraph). There are two ways of applying Yang-Baxter deformations to the PCM, which correspond to the breaking of its two different (left and right) symmetries. The first one follows the Yang-Baxter deformation procedure of~\nms (see Subsection 4.2), and simply changes the PCM realisation at site $(1)$ into a Yang-Baxter realisation. The model obtained by this procedure has been described in Subsection 4.2.2 of~\nms and coincides with the ``left'' Yang-Baxter model~\cite{Klimcik:2002zj,Klimcik:2008eq}
\begin{equation}\label{Eq:LeftYB}
S_{\text{iYB-left}}[g] = \frac{K}{2} \iint_{\R\times\D} \dd t\,\dd x \; \kappa\left(j_+,\frac{1}{1-\eta R_g} j_-\right),
\end{equation}
where $R_g=\Ad_g^{-1}\circ R \circ \Ad_g$. This deformation breaks the left symmetry $g\mapsto h^{-1}g$ of the PCM but preserves its right symmetry $g\mapsto gh$ (which can still be interpreted as the diagonal symmetry of the underlying affine Gaudin model).

The second possible Yang-Baxter deformation of the PCM comes from applying the diagonal Yang-Baxter deformation procedure described in Subsection \ref{SubSec:DiagYB} of this article. We will not enter into the details of this construction here and just give the final result (modulo some possible redefinition of parameters). The action of the deformed model in its gauge fixed formulation reads
\begin{equation*}
S_{\text{iYB-right}}[g] = \frac{K}{2} \iint_{\R\times\D} \dd t\,\dd x \; \kappa\left(j_+,\frac{1}{1-\eta R} j_-\right).
\end{equation*}
This deformation breaks the right symmetry $g\mapsto gh$ of the PCM, as it is identified with its diagonal symmetry (which is broken by the construction of Subsection \ref{SubSec:DiagYB}), but preserves the left symmetry $g\mapsto h^{-1}g$.

Note that one can pass from one deformation to another by considering the redefinition of the field $g \mapsto g^{-1}$, so that the two models are in fact equivalent. At the level of the gauged model, this redefinition of the field in fact corresponds to changing the gauge-fixing condition $\gb 2=\Id$ to $\gb 1=\Id$, or equivalently to consider the gauge-invariant field $g=\gb 2\,g^{(1)\,-1}$ instead of $g=\gb 1\,g^{(2)\,-1}$ (see discussion above). Exchanging the roles of $\gb 1$ and $\gb 2$ in the gauged-model indeed coincides with exchanging the left and right symmetries of the gauge-fixed model.

\paragraph{Multi-parameters deformations.} One can push further this game of deformations. For example, the left Yang-Baxter deformation \eqref{Eq:LeftYB} is still a non-constrained realisation of affine Gaudin model belonging to the class considered in~\nm. Thus, one can further apply to it the diagonal Yang-Baxter deformation procedure described in Subsection \ref{SubSec:DiagYB}. This produces a two-parameter deformation of the PCM, which combines both left and right Yang-Baxter deformations into one model, breaking all its symmetries. This model has been identified in~\bgs with the so-called bi-Yang-Baxter model~\cite{Klimcik:2008eq,Klimcik:2014bta} (although in~\bgs this model was not introduced as the result of a diagonal Yang-Baxter deformation but instead starting from the Hamiltonian analysis of the bi-Yang-Baxter model performed in~\cite{Delduc:2015xdm}).

Let us note that there are other possibilities for constructing multi-parameter deformations of the PCM, using the techniques developed in this article. For example, one can also consider Yang-Baxter deformations of the PCM plus Wess-Zumino term, both on its left symmetry and its right (diagonal) symmetry. This would then result in a three-parameter deformed model, which we conjecture to coincide with the model introduced in~\cite{Delduc:2017fib}. In fact the model of~\cite{Delduc:2017fib} possesses $(3+\text{rk}(\g))$ deformation parameters, where the last $\text{rk}(\g)$ ones are obtained by a TsT transformation: we thus conjecture that the model mentioned above coincides with the model of~\cite{Delduc:2017fib} without TsT transformations. Although this model was shown to admit a Lax pair, its Hamiltonian integrability has not been proven yet. The conjectured identification with an affine Gaudin model made above would be a proof of this Hamiltonian integrability, as by construction, affine Gaudin models satisfy a Maillet bracket.

Another possibility is to start with the $\lambda$-model~\cite{Sfetsos:2013wia}, which is an integrable deformation of the non-abelian T-dual of the PCM. It was explained in~\bnms that this model is also a non-constrained realisation of affine Gaudin model. Thus, one can also apply a diagonal Yang-Baxter deformation to it, to obtain a two-parameter integrable deformation of the non-abelian T-dual of the PCM. We conjecture that this model is identical to the generalised $\lambda$-model introduced in~\cite{Sfetsos:2015nya} as a two-parameter deformation of the non-abelian T-dual which possesses a Lax pair. Such an identification would also prove the Hamiltonian integrability of this model.

The study of multi-parameter deformations of the PCM (or its non-abelian T-dual) in the context of constrained affine Gaudin models is a promising path to explore to understand better the panorama of integrable deformed $\s$-models and their underlying algebraic structure. Some results in this direction will be presented in a separate article~\cite{Lacroix:toAppear} by the author.

\subsection[Integrable coupled $\s$-models: non-gauged formulation]{Integrable coupled $\bm\s$-models: non-gauged formulation}
\label{SubSec:NonGaugedSigma}

Let us now get to the main subject of this section, the integrable coupled $\s$-models introduced in~\cite{Delduc:2018hty}. Following the ideas of Section \ref{Sec:Gauging}, we will exhibit a gauged formulation of these models and use it to construct their homogeneous diagonal Yang-Baxter deformation. First, let us recall briefly their non-gauged formulation, as described in Subsection 3.3 of~\nms (we refer to this subsection for details).

\subsubsection{Hamiltonian formulation and affine Gaudin model interpretation}

\paragraph{Sites and twist function.} The integrable coupled $\s$-model with $N$ copies is constructed as a realisation of affine Gaudin model $\mathbb{M}^{\vp,\pi}_{\eb}$ with $N$ sites $(r)$, $r\in\lbrace 1,\cdots,N\rbrace$, of multiplicity two. We will denote by $z_r$ the corresponding positions of these sites. The twist function of the model then takes the form
\begin{equation}\label{Eq:TwistSigmaProd}
\vp(z) = - \ell^\infty \frac{\prod_{i=1}^{2N}(z-\ze_i)}{\prod_{r=1}^N (z-z_r)^2},
\end{equation}
where $\ze_1,\cdots,\ze_{2N}$ are the zeroes of the twist function. These zeroes $\ze_i$, the positions $z_r$ and the constant term $\ell^\infty$ are then taken as the defining parameters of the model.

The zeroes $\ze_i$ are divided into two sets, labelled by indices $i$ either in $I_+$ or in $I_-$, whether the corresponding coefficient $\epsilon_i$ appearing in the definition of the Hamiltonian is equal to $+1$ or $-1$. These two sets $I_\pm$ are supposed to be of equal size $N$. We then define the following factorisation of the twist function:
\begin{equation*}
\vp(z) = -\ell^\infty \vp_+(z)\vp_-(z), \;\;\;\;\;\;\; \text{ where } \;\;\;\;\;\;\; \vp_\pm(z) = \frac{\prod_{i\in I_\pm}(z-\ze_i)}{\prod_{r=1}^N (z-z_r)}.
\end{equation*}
We will need the following functions of the spectral parameter:
\begin{equation*}
\vppm r(z) = (z-z_r)\vp_\pm(z),
\end{equation*}
which are defined in such a way that they are regular at $z=z_r$. The partial fraction decomposition
\begin{equation}\label{Eq:DSETwistSigma}
\vp(z) = \sum_{r=1}^N \left( \frac{\lc r}{(z-z_r)^2} - \frac{2\kc r}{z-z_r} \right) - \ell^\infty
\end{equation}
of the twist function is then given in terms of the functions $\vppm r(z)$ by the formulas
\begin{equation}\label{Eq:LevelsSigma}
\lc r = - \ell^\infty \vpp r(z_r) \vpm r(z_r) \;\;\;\;\;\; \text{and} \;\;\;\;\;\; \kc r = \frac{\ell^\infty}{2}  \Bigl( \vpp r(z_r) \vpm r'(z_r) + \vpp r'(z_r) \vpm r(z_r) \Bigr).
\end{equation}

\paragraph{Realisation and Gaudin Lax matrix.} To each site $(r)$ of the model is attached a PCM+WZ realisation (see Appendix \ref{App:PCM+WZ}), with canonical fields $\gb r(x)$ and $\Xb r(x)$. The corresponding Takiff currents are then
\begin{equation}\label{Eq:TakiffSigma}
\J{(r)}0(x) = \Xb r (x) - \kc r \, \Wb r(x) - \kc r \, \jb r (x) \;\;\;\;\;\;\;\; \text{ and } \;\;\;\;\;\;\;\; \J{(r)}1(x) = \lc r\, \jb r(x).
\end{equation}
This choice of Takiff currents defines a realisation $\pi:\Tc_{\lt} \rightarrow \Ac$ of the Takiff Poisson algebra $\Tc_{\lt}$ corresponding to the twist function $\vp(z)$ into the algebra of observables $\Ac=\Ac_{G_0}^{\otimes N}$ generated by $N$ copies of $\Ac_{G_0}$. This algebra of observables can be seen as the algebra of canonical fields on the cotangent bundle $T^*G_0^N$ of the $N^{\rm{th}}$ Cartesian product of $G_0$.

The Gaudin Lax matrix $\Gamma(z,x)$ of the model is then entirely determined by this choice of realisation and simply reads
\begin{equation}\label{Eq:GammaSigma}
\Gamma(z,x) = \sum_{r=1}^N \left( \frac{\J{(r)}1(x)}{(z-z_r)^2} + \frac{\J{(r)}0(x)}{z-z_r} \right).
\end{equation}

\paragraph{Hamiltonian and Lax pair.} Let us consider the quadratic charges associated with the zeroes of the twist function:
\begin{equation}\label{Eq:QiSigma}
\Q_i = - \frac{1}{2\vp'(\ze_i)} \int_\D \dd x \; \kappa\bigl( \Gamma(\ze_i,x), \Gamma(\ze_i,x) \bigr).
\end{equation}
The Hamiltonian of the model is then defined as
\begin{equation}\label{Eq:HamSigma}
\Hc = \sum_{i=1}^M \epsilon_i \Q_i = \sum_{i\in I_+} \Q_i - \sum_{i\in I_-} \Q_i.
\end{equation}
This determines the dynamic $\p_t=\lbrace \Hc,\cdot\rbrace$ of the model. In particular, the equations of motion of the model are equivalent to the zero curvature equation of the Lax pair
\begin{equation}\label{Eq:LaxSigma}
\Lc(z,x) = \sum_{i=1}^M \frac{1}{\vp'(\ze_i)} \frac{\Gamma(\ze_i,x)}{z-\ze_i} \;\;\;\;\;\;\;\; \text{ and } \;\;\;\;\;\;\;\;  \Mc(z,x) = \sum_{i=1}^M \frac{\epsilon_i}{\vp'(\ze_i)} \frac{\Gamma(\ze_i,x)}{z-\ze_i}.
\end{equation}

\subsubsection{Lagrangian formalism}
\label{SubSubSec:LagNonGauged}

\paragraph{Lagrangian Lax pair through interpolation.} One shows that the Poisson bracket of the field $\gb r$ with the Gaudin Lax matrix evaluated at the value $\ze_i$ of the spectral parameter is
\begin{equation}\label{Eq:PBgGamma}
 \left\lbrace \gb r\ti1(x), \Gamma\ti{2}(\ze_i,y) \right\rbrace = \gb r\ti1(x) \frac{C\ti{12}}{\po_r-\ze_i} \delta_{xy},
\end{equation}
hence
\begin{equation*}
\gb r(x)^{-1}\left\lbrace \Q_i, \gb r(x) \right\rbrace = \frac{1}{\vp'(\ze_i)} \frac{\Sg(\ze_i,x)}{\po_r - \ze_i}.
\end{equation*}
Using the expression \eqref{Eq:HamSigma} of the Hamiltonian, one can then compute the temporal current $\jb r_0 = \gb r\null^{\,-1} \p_t \gb r$ from the equation above, yielding
\begin{equation}\label{Eq:j0}
\jb r_0(x) = \gb r(x)^{-1} \left\lbrace \Hc, \gb r(x) \right\rbrace = \sum_{i=1}^M \frac{\epsilon_i}{\vp'(\ze_i)} \frac{\Sg(\ze_i,x)}{\po_r - \ze_i}.
\end{equation}
One recognizes in the right hand side of this equation the temporal part $\Mc(z,x)$ of the Lax pair \eqref{Eq:LaxSigma}, evaluated at $z=z_r$. Thus, we have $\jb r_0(x)=\Mc(z_r,x)$. A similar computation yields $\jb r(x)=\Lc(z_r,x)$ for the spatial component of the current. Passing to light-cone currents $\jb r_\pm$ and Lax pair $\Lc_\pm(z)$, we then get
\begin{equation}\label{Eq:jLaxEval}
\jb r_\pm = \Lc_\pm(z_r), \;\;\;\;\; \forall \, r\in\lbrace 1,\cdots,N \rbrace.
\end{equation}
The light-cone component $\Lc_\pm(z)$ of the Lax pair only possesses poles at the zeroes $\ze_i$, for $i\in I_\pm$ (see Equation \eqref{Eq:LaxSigma} or Equation (2.60) of~\nm). All these poles are simple and $\Lc_\pm(z)$ does not possess a polynomial part. Moreover, there are $N=|I_\pm|$ of such poles. Thus, $\Lc_\pm(z)$ is fully determined by its evaluation at $N$ points, pairwise distinct and different from the $\ze_i$'s, $i\in I_\pm$. In the present case, we know $N$ such evaluations by Equation \eqref{Eq:jLaxEval}. This forces the Lax pair $\Lc_\pm(z)$ to be given by
\begin{equation}\label{Eq:LaxLagSigma}
\Lc_\pm(z) = \sum_{r=1}^N \frac{\vppm r(z_r)}{\vppm r(z)}\jb r_\pm.
\end{equation}
Indeed, one checks that the above expression of $\Lc_\pm(z)$ satisfies the interpolation condition \eqref{Eq:jLaxEval}, using the fact that
\begin{equation*}
\left. \frac{\vppm r(z_r)}{\vppm r(z)} \right|_{z=z_s} = \delta_{rs}.
\end{equation*}

\paragraph{Lagrangian expression of $\bm{\Xb r}$.} Recall that the Lax pair of the model is given by $\Lc(z,x)=\Gamma(z,x)/\vp(z)$. Combining Equations \eqref{Eq:DSETwistSigma}, \eqref{Eq:TakiffSigma} and \eqref{Eq:GammaSigma}, one is able to extract the field $\Xb r$ from the Lax matrix of the model:
\begin{equation*}
\Xb r - \kc r \, \Wb r = \lc r \,\Lc'(\po_r) - \kc r \,\jb r = \frac{\lc r}{2} \Bigl( \Lc'_+(\po_r) - \Lc'_-(\po_r) \Bigr) - \frac{\kc r}{2} \Bigl( \jb r_+ - \jb r_- \Bigr).
\end{equation*}
Using the Lagrangian expression \eqref{Eq:LaxLagSigma} of the light-cone Lax pair $\Lc_\pm(z)$, one can then compute the Lagrangian expression of $\Xb r$. After a few manipulations, one finds (see~\nm, Subsection 3.3.2):
\begin{equation}\label{Eq:XLag}
\Xb r - \kc r \, \Wb r = \sum_{s=1}^N \Bigl( \rho_{sr} \, \jb s_+ + \rho_{rs} \, \jb s_- \Bigr),
\end{equation}
with
\begin{subequations}\label{Eq:Rho}
\begin{align}
\rho_{rr} &= \frac{\ell^\infty}{4}  \Bigl( \vpp r'(z_r) \vpm r(z_r) - \vpp r(z_r) \vpm r'(z_r) \Bigr),\label{Eq:Rhorr} \\
\rho_{rs} &= \frac{\ell^\infty}{2} \frac{\vpp r(z_r)\vpm s(z_s)}{\po_r-\po_s}, \;\;\;\;\; \text{ for } \; r\neq s. \label{Eq:Rhors}
\end{align}
\end{subequations}

\paragraph{Action.} Now that we have the Lagrangian expression \eqref{Eq:XLag} of the currents $\Xb r$, one can perform the inverse Legendre transform of the model and obtain its action. According to Equation (3.48) of~\nm, this inverse Legendre transform takes the form
\begin{equation}\label{Eq:LegendreInverseNSites}
S\bigl[ \gb r\bigr] = \sum_{r=1}^N \left( \frac{1}{2} \iint_{\R \times \D} \dd t \, \dd x \; \kappa\left( \Xb r - \kc r \, \Wb r, \jb r_+ + \jb r_- \right) \right) - \int_\R \dd t \; \Hc + \sum_{r=1}^N \kc r \; \Ww {\gb r},
\end{equation}
where $\Ww {\gb r}$ is the Wess-Zumino term of the field $\gb r$. One can then introduce the Lagrangian expression \eqref{Eq:XLag} of $\Xb r - \kc r\,\Wb r$ in this equation.

Moreover, one also needs to obtain the Lagrangian expression of the Hamiltonian $\Hc$. According to Equations \eqref{Eq:QiSigma} and \eqref{Eq:HamSigma}, this Hamiltonian only contains terms of the form $\kappa\bigl(\Gamma(\ze_i),\Gamma(\ze_i)\bigr)$, for $i\in\lbrace 1,\cdots,2N\rbrace$. In particular, one sees from Equation (2.60) of~\nms that $\Gamma(\ze_i)$ can be extracted either from $\Lc_+(z)$ ot $\Lc_-(z)$, depending whether $i$ belongs to $I_+$ or $I_-$. Thus $\Gamma(\ze_i)$ only contains terms proportional to the $\jb r_+$'s or the $\jb r_-$'s, but does not mix different ``light-cone chiralities''. Therefore, the Hamiltonian only contains terms of the form $\kappa\bigl(\jb r_+,\jb s_+\bigr)$ or $\kappa\bigl(\jb r_-,\jb s_-\bigr)$. As these terms are non-Lorentz-invariant and the model should be relativistic (because we chose all coefficients $\epsilon_i$ to be either $+1$ or $-1$), they should totally disappear in the computation of the action. Indeed, it is explained in~\nms that these terms simplify with the non-relativistic terms coming from the parts $\kappa\left( \Xb r - \kc r \, \Wb r, \jb r_+ + \jb r_- \right)$ of Equation \eqref{Eq:LegendreInverseNSites}. In the end one finds that the action of the model is
\begin{equation}\label{Eq:ActionNonGauged}
S\bigl[ \gb 1, \cdots, \gb N \bigr] = \iint_{\R\times\D} \dd t \, \dd x \; \sum_{r,s=1}^N \; \rho_{rs} \, \kappa\left(\jb r_+, \jb s_-\right)  + \sum_{r=1}^N \kc r \; \Ww {\gb r}.
\end{equation}
Let us note here that the coefficients $\rho_{rs}$ and $\kc r$ appearing in this action can be re-expressed as residues of well-chosen forms on the Riemann sphere $\CP$, providing a more geometric and unified interpretation of these parameters. Although we shall not need this property in what follows, it is an interesting result in itself and is explained in more details in Appendix \ref{App:Identities}.

\paragraph{Symmetries.} The model \eqref{Eq:ActionNonGauged} possesses two types of global symmetries. The first one is the diagonal symmetry, which acts simultaneously on all the copies of $G_0$ by a right multiplication:
\begin{equation}\label{Eq:DiagoSigma}
\gb 1(x,t) \longmapsto \gb 1(x,t)\,h, \;\;\;\;\;\; \cdots, \;\;\;\;\;\; \gb N(x,t) \longmapsto \gb N(x,t)\,h,
\end{equation}
with $h$ a constant parameter in $G_0$. Moreover, the coupled model also possesses $N$ left $G_0$-symmetries, which act on the $N$ copies of the model independently:
\begin{equation*}
\gb 1(x,t) \longmapsto h_1^{-1}\,\gb 1(x,t), \;\;\;\;\;\; \cdots, \;\;\;\;\;\; \gb N(x,t) \longmapsto h_N^{-1}\,\gb N(x,t),
\end{equation*}
where $h_1,\cdots,h_N$ are $N$ independent constant parameters in $G_0$.

\subsection[Integrable coupled $\s$-models: gauged formulation]{Integrable coupled $\bm\s$-models: gauged formulation}
\label{SubSec:GaugedSigma}

\subsubsection{Hamiltonian formulation}
\label{SubSubSec:HamGauged}

\paragraph{Change of spectral parameter and twist function.} Let us now apply the results of Section \ref{Sec:Gauging} to construct a gauged formulation $\mathbb{M}^{\vpt,\pit}_{\eb}$ of the integrable coupled $\s$-models introduced in~\prl~and whose construction as a realisation of affine Gaudin model $\mathbb{M}^{\vp,\pi}_{\eb}$ was recalled in the previous subsection. Following Subsection \ref{SubSec:ChangeSpec2}, we introduce a change of spectral parameter $z \mapsto \zt=f(z)$, with $f$ the M\"obius transformation \eqref{Eq:Mobius}. The site $(r)$, with position $z_r$ in the model $\mathbb{M}^{\vp,\pi}_{\eb}$, now has position $\zt_r=f(z_r)$ in the model $\mathbb{M}^{\vpt,\pit}_{\eb}$, that we will suppose finite. The model $\mathbb{M}^{\vpt,\pit}_{\eb}$ also possesses a site at $\zi=f(\infty)=a/c$, that we denoted $\is$ in Section \ref{Sec:Gauging}, which is the image under $f$ of the site at $z=\infty$ in the model $\mathbb{M}^{\vp,\pi}_{\eb}$. To make the notations uniform, we shall now call this site $(N+1)$ and denote its position by $\zt_{N+1}=\zi$.\\

The twist function $\vpt(\zt)$ of the model $\mathbb{M}^{\vpt,\pit}_{\eb}$ can be expressed from the partial fraction decomposition \eqref{Eq:DSETwistSigma} of $\vp(z)$ using Lemma \ref{Lem:NewLevels2}. More precisely, one gets
\begin{equation}\label{Eq:DSETwistTildeSigma}
\vpt(\zt) = \sum_{r=1}^{N+1} \frac{\elt r}{(\zt-\zt_r)^2} - \frac{2\kc r}{\zt-\zt_r}.
\end{equation}
In this equation, $\kc r$ is the same as in the previous subsection for $r \in \lbrace 1,\cdots,N \rbrace$ and we introduce
\begin{equation*}
\elt r = \frac{ad-bc}{(c z_r+d)^2} \,\lc r = \frac{(c \zt_r-a)^2}{ad-bc} \,\lc r,
\end{equation*}
for $r\in\lbrace 1,\cdots,N \rbrace$, as well as
\begin{equation}\label{Eq:klN+1}
\kc{N+1} = - \sum_{r=1}^N \kc r, \;\;\;\;\;\;\;\;\; \elt{N+1} = - \frac{\ell^\infty (ad-bc)}{c^2}.
\end{equation}

However, recall that in the non-gauged formulation, we used mostly the factorised expression \eqref{Eq:TwistSigmaProd} of the twist function $\vp(z)$ in terms of its zeroes. Let us search for a similar expression of $\vpt(\zt)$. We denote by $\zet_i=f(\ze_i)$ the image under $f$ of the zeroes of the twist function $\vp(z)$ of the model $\mathbb{M}^{\vp,\pi}_{\eb}$. These are the zeroes of the new twist function $\vpt(\zt)$. We will suppose here that these are all finite\footnote{We choose here to not treat the case where one of the zeroes is send to infinity by $f$ for brevity. Such a model could be studied in a similar way and would fall into the case (ii) of Subsection \ref{SubSec:Zeroes}}. The twist function $\vpt(\zt)$ satisfies $\vpt(\zt)=\vp\bigl(f^{-1}(\zt)\bigr)/f'\bigl(f^{-1}(\zt)\bigr)$. Starting from Equation \eqref{Eq:TwistSigmaProd}, we get
\begin{equation*}
\vpt(\zt) = - \widetilde{\ell}^\infty \frac{\prod_{i=1}^{2N}(\zt-\zet_i)}{\prod_{r=1}^{N+1} (\zt-\zt_r)^2},
\end{equation*}
where
\begin{equation}\label{Eq:LtildeInf}
\widetilde{\ell}^\infty = \ell^\infty \frac{ad-bc}{c^2\,\omega_+\omega_-}, \;\;\;\;\;\; \text{ with } \;\;\;\;\;\; \omega_\pm = \frac{\prod_{r=1}^N (c\, z_r+d)}{\prod_{i\in I_\pm}^{2N}(c\,\ze_i+d)} = \frac{\prod_{i\in I_\pm}(c\,\zet_i-a)}{\prod_{r=1}^N (c\, \zt_r-a)}.
\end{equation}

\paragraph{Factorisation of the twist function.} Following the treatment of the non-gauged model in the previous subsection, we introduce the following factorisation of the twist function:
\begin{equation*}
\vpt(\zt) = -\widetilde{\ell}^\infty \vpt_+(\zt) \vpt_-(\zt), \;\;\;\;\;\;\; \text{ with } \;\;\;\;\;\;\; \vpt_\pm(\zt) = \frac{\prod_{i\in I_\pm}(\zt-\zet_i)}{\prod_{r=1}^{N+1} (\zt-\zt_r)} 
\end{equation*}
Similarly to the non-gauged case, we also define the functions
\begin{equation*}
\vptpm r(\zt) = (\zt-\zt_r)\vpt_\pm(\zt),
\end{equation*}
which are regular at $\zt=\zt_r$. One can then express the levels $\elt r$ and $\kc r$ appearing in the partial fraction decomposition  \eqref{Eq:DSETwistTildeSigma} as
\begin{equation}\label{Eq:LevelsTildeSigma}
\elt r = - \widetilde{\ell}^\infty \, \vptp r(\zt_r) \, \vptm r(\zt_r) \;\;\;\;\;\; \text{and} \;\;\;\;\;\; \kc r = \frac{\widetilde{\ell}^\infty}{2}  \Bigl( \vptp r(\zt_r) \, \vptm r '(\zt_r) + \vptp r'(\zt_r) \, \vptm r(\zt_r) \Bigr)
\end{equation}
similarly to Equation \eqref{Eq:LevelsSigma} in the non-gauged case. The functions $\vpt_\pm(\zt)$ and $\vptpm r(\zt)$ are related to the ones $\vp_\pm(z)$ and $\vppm r(z)$ used in the non-gauged case by
\begin{equation}\label{Eq:TwistPMRel}
\vpt_\pm(\zt) = \frac{\omega_\pm}{\zt-\zt_{N+1}} \vp_\pm\bigl( f^{-1}(\zt) \bigr) \;\;\;\;\;\; \text{ and } \;\;\;\;\;\; \vptpm r(\zt) = \frac{\omega_\pm\,c\,(c\,\zt_r-a)}{ad-bc} \vppm r\bigl( f^{-1}(\zt) \bigr),
\end{equation}
for $r\in\lbrace 1,\cdots,N \rbrace$ and with $\omega_\pm$ defined as in Equation \eqref{Eq:LtildeInf}. One can then check that inserting Equation \eqref{Eq:LevelsSigma} in the formulas for $\elt r$ and $\kc r$ in the previous paragraph gives an expression in agreement with the one above, using the identities \eqref{Eq:LtildeInf} and \eqref{Eq:TwistPMRel}.

\paragraph{Gaudin Lax matrix and realisation.} The gauged model $\mathbb{M}^{\vpt,\pit}_{\eb}$ possesses $N+1$ sites with multiplicity 2. Each site $(r)$, $r\in\lbrace 1,\cdots,N+1 \rbrace$, is associated with a copy of the PCM+WZ realisation, with canonical fields $\gb r(x)$ and $\Xb r(x)$. The corresponding Takiff currents are
\begin{equation}\label{Eq:NewCurrentsSigma}
\Jtt {(r)} 0 (x) =\Xb r (x) - \kc r \, \Wb r(x) - \kc r \, \jb r (x) \;\;\;\;\;\;\;\; \text{ and } \;\;\;\;\;\;\;\; \Jtt {(r)}1(x) = \elt r\, \jb r(x).
\end{equation}
For $r\in\lbrace 1,\cdots,N \rbrace$, these are related to the ones $\J{(r)}p(x)$ used in the non-gauged formulation by
\begin{equation*}
\Jtt {(r)} 0 (x) = \J {(r)} 0 (x) \;\;\;\;\;\; \text{ and } \;\;\;\;\;\; \Jtt {(r)} 1 (x) =  \frac{ad-bc}{(c z_r+d)^2} \J {(r)} 1 (x).
\end{equation*}
This is in agreement with the discussion of Subsection \ref{SubSubSec:NewCurrents}, where it is explained that the current of the gauged formulation should be related to the ones of the non-gauged formulation by applying Proposition \ref{Prop:NewTakiff}.

In Subsection \ref{SubSubSec:NewCurrents}, it is also explained that the additional site $\is$ appearing in the gauged model should be realised by a PCM+WZ realisation, with fields that we denoted $g(x)$ and $X(x)$. Recall that here, we renamed this site $(N+1)$ to keep the notations as uniform as possible: for the same reason, we changed the notations of the fields $g(x)$ and $X(x)$ to $\gb{N+1}(x)$ and $\Xb{N+1}(x)$, so that Equation \eqref{Eq:NewCurrentsSigma} holds for every $r\in\lbrace 1,\cdots,N+1 \rbrace$.

The Takiff currents \eqref{Eq:NewCurrentsSigma} specify the realisation $\pit$ used to construct the model. This realisation is valued in the Poisson algebra $\Act=\Ac_{G_0}^{\otimes N+1}$ generated by $N+1$ copies of the canonical fields on $T^*G_0$, which is then the canonical algebra on $T^*G_0^{N+1}$. Note that $\Act=\Ac \otimes \Ac_{G_0}$, where $\Ac=\Ac_{G_0}^{\otimes N}$ is the algebra of observables of the non-gauged model $\mathbb{M}^{\vp,\pi}_{\eb}$. This is in agreement with the general discussion of Subsection \ref{SubSubSec:NewCurrents}: the additional factor $\Ac_{G_0}$ in $\Act$ corresponds to the additional fields $\gb{N+1}(x)$ and $\Xb{N+1}(x)$ introduced in the model.\\

As usual, the Gaudin Lax matrix can be read easily from the choice of Takiff currents \eqref{Eq:NewCurrentsSigma} as
\begin{equation}\label{Eq:GtSigma}
\Gt(\zt,x) = \sum_{r=1}^{N+1} \left( \frac{\Jtt{(r)}1(x)}{(\zt-\zt_r)^2} + \frac{\Jtt{(r)}0(x)}{\zt-\zt_r} \right).
\end{equation}

\paragraph{Constraint and gauge symmetry.} The model $\mathbb{M}^{\vpt,\pit}_{\eb}$ admits the first-class constraint
\begin{equation*}
\Cct(x) = \sum_{r=1}^{N+1} \Jtt{(r)}0(x) = \sum_{r=1}^{N+1} \left( \Xb r (x) - \kc r \, \Wb r(x) - \kc r \, \jb r (x) \right).
\end{equation*}
From the Poisson bracket of the fields $\Xb s$, $\Wb s$ and $\jb s$ with $\gb r$ (see Appendix \ref{App:TStarG}), we see that the bracket of this constraint with the field $\gb r$ is given by
\begin{equation}\label{Eq:PBCgSigma}
\bigl\lbrace \Cct\ti{1}(x), \gb r\ti{2}(y) \bigr)\rbrace = \gb r\ti{2}(y)\,C\ti{12}\, \delta_{xy}.
\end{equation}
The gauge symmetry generated by the constraint then acts on the fields $\gb r(x)$ as a diagonal right multiplication:
\begin{equation*}
\gb 1(x) \longmapsto \gb 1(x)\,h(x,t), \;\;\;\;\;\; \cdots, \;\;\;\;\;\; \gb {N+1}(x) \longmapsto \gb {N+1}(x)\,h(x,t),
\end{equation*}
with $h(x,t)$ a local parameter in $G_0$. As expected from Subsection \ref{SubSec:Gauging}, this gauge symmetry acts on the additional field $\gb {N+1}(x)$ by right multiplication and on the initial fields $\gb 1(x),\cdots,\gb N(x)$ by a local version of the diagonal symmetry \eqref{Eq:DiagoSigma} of the non-gauged model.

\paragraph{Hamiltonian and Lax pair.} Following the general procedure of Subsection \ref{SubSec:Zeroes}, we define quadratic charges $\widetilde{\Q}_i$ associated with the zeroes $\zet_i$'s of the twist function. According to Equation \eqref{Eq:QHZeros}, these charges can be expressed as
\begin{equation}\label{Eq:QiSigmaGauged}
\widetilde{\Q}_i = - \frac{1}{2\vpt\hspace{1pt}'\bigl(\zet_i\bigr)} \int_\D \dd x \; \kappa\Bigl( \Gt\bigl(\zet_i,x\bigr), \Gt\bigl(\zet_i,x\bigr) \Bigr),
\end{equation}
similarly to Equation \eqref{Eq:QiSigma} in the non-gauged formulation.

The Hamiltonian of the model is then defined from these quadratic charges $\widetilde{\Q}_i$ and the choice of parameters $\eb=(\epsilon_1,\cdots,\epsilon_{2N})$ (which we take to be the same as in the non-gauged formulation, following Subsection \ref{SubSubSec:EquivDynamics}). It reads
\begin{equation}\label{Eq:HamSigmaGauged}
\widetilde{\Hc} = \sum_{i=1}^{2N} \epsilon_i \, \widetilde{\Q}_i + \int_\D \dd x \; \kappa\bigl( \Cct(x), \mu(x) \bigr) = \sum_{i\in I_+}  \widetilde{\Q}_i  - \sum_{i\in I_-}  \widetilde{\Q}_i + \int_\D \dd x \; \kappa\bigl( \Cct(x), \mu(x) \bigr),
\end{equation}
where $\mu(x)$ is a $\g_0$-valued Lagrange multiplier. This specifies the dynamic $\p_t \approx \lbrace \widetilde{\Hc}, \cdot \rbrace$ of the model. According to Theorem \ref{Thm:Lax}, this dynamic takes the form of a zero curvature equation, with Lax pair\footnote{Note that Theorem \ref{Thm:Lax} also provides a (weak) expression of $\Lct(\zt,x)$ similar to the one of $\widetilde{\Mc}(\zt,x)$. However, we shall not need this expression here.}
\begin{equation}\label{Eq:LaxSigmaGauged}
\Lct(\zt,x) = \frac{\Gt(\zt,x)}{\vpt(\zt)} \;\;\;\;\;\; \text{ and } \;\;\;\;\;\; \widetilde{\Mc}(\zt,x) = \sum_{i=1}^{2N} \frac{\epsilon_i}{\vpt\hspace{1pt}'\bigl(\zet_i\bigr)} \frac{\Gt\bigl(\zet_i,x\bigr)}{\zt-\zet_i} + \mu(x).
\end{equation}

\subsubsection{Lagrangian formulation}
\label{SubSubSec:LagSigmaGauged}

\paragraph{Dynamics of the fields $\bm{\gb r}$.} Let us now turn to the Lagrangian formulation of the gauged model $\mathbb{M}^{\vpt,\pit}_{\eb}$. As we shall see, determining its action will follow very closely the same steps as in the non-gauged case (see Subsection \ref{SubSubSec:LagNonGauged}). One starts by computing the Poisson bracket of the field $\gb r(x)$ with the current $\Gt\bigl(\zet_i,y\bigr)$. Similarly to Equation \eqref{Eq:PBgGamma} in the non-gauged case, we find
\begin{equation*}
 \left\lbrace \gb r\ti1(x), \Gt\ti{2}\bigl(\zet_i,y\bigr) \right\rbrace = \gb r\ti1(x) \frac{C\ti{12}}{\zt_r-\zet_i} \delta_{xy}.
\end{equation*}
Using the expression \eqref{Eq:QiSigmaGauged} of the charge $\widetilde{\Q}_i$, we then get
\begin{equation}\label{Eq:PBgQi}
\gb r(x)^{-1}\left\lbrace \widetilde{\Q}_i, \gb r(x) \right\rbrace = \frac{1}{\vpt\hspace{1pt}'\bigl(\zet_i\bigr)} \frac{\Gt\bigl(\zet_i,x\bigr)}{\zt_r - \zet_i}.
\end{equation}
These brackets for $i\in\lbrace 1,\cdots,2N\rbrace$ control part of the dynamics of the field $\gb r (x)$ under the flow of the Hamiltonian \eqref{Eq:HamSigmaGauged}. The rest of the dynamics involves the Lagrange multiplier $\mu(x)$. Starting from the Poisson bracket \eqref{Eq:PBCgSigma} between the constraint and the field $\gb r$, we see that this Lagrange multiplier enters the dynamics of $\gb r$ simply through
\begin{equation*}
\gb r(x)^{-1}\left\lbrace \int_\D \dd x\; \kappa\bigl( \Cct(x), \mu(x) \bigr), \gb r(x) \right\rbrace \approx \mu(x).
\end{equation*}

Combining the last two equations, one obtains the expression of the temporal current $\jb r_0 = \gb r\null^{\,-1} \p_t \gb r$ , similarly to Equation \eqref{Eq:j0} in the non-gauged case:
\begin{equation*}
\jb r_0(x) \approx \gb r(x)^{-1} \left\lbrace \widetilde{\Hc}, \gb r(x) \right\rbrace \approx \sum_{i=1}^M \frac{\epsilon_i}{\vpt\hspace{1pt}'\bigl(\zet_i\bigr)} \frac{\Gt\bigl(\zet_i,x\bigr)}{\zt_r - \zet_i} + \mu(x).
\end{equation*}
As in the non-gauged case, comparing the right-hand side of this equation with the expression \eqref{Eq:LaxSigmaGauged} of the temporal part $\widetilde{\Mc}(\zt,x)$ of the Lax pair, we find that $\jb r_0(x)=\Mc(z_r,x)$.

From the definition $\Lct(\zt,x)=\Gt(\zt,x)/\vpt(\zt)$ of the Lax matrix in terms of the twist function and the Gaudin Lax matrix, as well as their expressions \eqref{Eq:DSETwistTildeSigma} and \eqref{Eq:GtSigma}, one finds (see also Equation (2.22) of~\nm)
\begin{equation*}
\Lct(\zt_r,x) = \frac{\Jtt{(r)}1(x)}{\elt r} = \jb r(x),
\end{equation*}
where we also used the expression \eqref{Eq:NewCurrentsSigma} of $\Jtt{(r)}1(x)$. Combining these two results for $\Lct(\zt_r)$ and $\widetilde{\Mc}(\zt_r)$, we then get for the light-cone currents $\jb r_\pm$:
\begin{equation}\label{Eq:jLaxEvalGauge}
\jb r_\pm = \Lct_\pm(\zt_r), \;\;\;\;\; \forall \, r\in\lbrace 1,\cdots,N+1 \rbrace.
\end{equation}

\paragraph{Lagrangian Lax pair through interpolation.} Equation \eqref{Eq:jLaxEvalGauge} is similar to the one \eqref{Eq:jLaxEval} in the non-gauged case. Recall that in this case, we used it to determine the Lagrangian expression of the light-cone Lax pair $\Lc_\pm(z)$ by interpolation. We shall use a similar method here. Corollary \ref{Cor:LightConeLax} fixes the analytic properties of $\Lct_\pm(\zt)$ as a function of the spectral parameter $\zt$: it possesses a constant term and $N$ simple poles, at the zeroes $\zet_i$ with $i\in I_\pm$ (recall here that we are in the case (ii), where $\zt=\infty$ is not a zero of the twist function). Thus, $\Lct_\pm(\zt)$ is determined by its evaluation at $N+1$ points and in particular by Equation \eqref{Eq:jLaxEvalGauge}. Similarly to the non-gauged case, $\Lct_\pm(\zt)$ is then given by
\begin{equation}\label{Eq:LaxLagSigmaGauge}
\Lct_\pm(\zt) = \sum_{r=1}^{N+1} \frac{\vptpm r(\zt_r)}{\vptpm r(\zt)}\jb r_\pm.
\end{equation}
One checks easily that this expression satisfies Equation \eqref{Eq:jLaxEvalGauge}.

The main difference with the non-gauged case is the following. In the non-gauged model, the function $\vppm r(z)$ is the ratio of a polynomial of degree $N$ (with roots $\ze_i$, $i\in I_\pm$) and a polynomial of degree $N-1$ (with roots $z_s$, $s\in\lbrace 1,\cdots,N\rbrace$ different from $r$). In particular, its inverse vanishes when $z$ goes to $\infty$. In the gauged case, $\vptpm r(\zt)$ is the ratio of two polynomials of degree $N$, as there are $N+1$ points $\zt_r$ and not only $N$. Thus, the function $1/\vptpm r(\zt)$ has a finite limit (in fact equal to 1), when $\zt$ goes to infinity. Thus, the above expression for $\Lct_\pm(\zt)$ possesses a constant term. The rest of its partial fraction decomposition is composed by simple poles at the zeroes of $\vptpm r(\zt)$, \textit{i.e.} at the $\zet_i$, $i\in I_\pm$. Thus, the above expression has the correct rational dependence on $\zt$ to be $\Lct_\pm(\zt)$. As it also satisfies the interpolation formula \eqref{Eq:jLaxEvalGauge}, it coincides with $\Lct_\pm(\zt)$.

\paragraph{Lagrangian expression of $\bm{\Xb r}$.} Following the computation of Subsection \ref{SubSubSec:LagNonGauged} in the non-gauged model, we extract the field $\Xb r$ from the Lax matrix by the formula
\begin{equation*}
\Xb r - \kc r \, \Wb r = \elt r \,\Lct'(\zt_r) - \kc r \,\jb r = \frac{\elt r}{2} \Bigl( \Lct'_+(\zt_r) - \Lct'_-(\zt_r) \Bigr) - \frac{\kc r}{2} \Bigl( \jb r_+ - \jb r_- \Bigr).
\end{equation*}
On then compute the Lagrangian expression of $\Xb r$ by inserting the one \eqref{Eq:LaxLagSigmaGauge} of the light-cone Lax pair $\Lct_\pm(\zt)$ in the above equation. As in the non-gauged case (see Equation \eqref{Eq:XLag}), we find
\begin{equation}\label{Eq:XLagGauged}
\Xb r - \kc r \, \Wb r = \sum_{s=1}^{N+1} \Bigl( \rhot_{sr} \, \jb s_+ + \rhot_{rs} \, \jb s_- \Bigr),
\end{equation}
with
\begin{subequations}\label{Eq:RhoTilde}
\begin{align}
\rhot_{rr} &= \frac{\widetilde{\ell}^\infty}{4}  \Bigl( \vptp r'(\zt_r) \,\vptm r(\zt_r) - \vptp r(\zt_r) \,\vpm r'(\zt_r) \Bigr),\label{Eq:Rhotrr} \\
\rhot_{rs} &= \frac{\widetilde{\ell}^\infty}{2} \frac{\vptp r(\zt_r)\,\vptm s(\zt_s)}{\zt_r-\zt_s}, \;\;\;\;\; \text{ for } \; r\neq s. \label{Eq:Rhotrs}
\end{align}
\end{subequations}
As in the non-gauged case (see discussion after Equation \eqref{Eq:ActionNonGauged}), the coefficients $\rhot_{rs}$ can be re-interpreted as residues of certain forms on the Riemann sphere $\CP$. This is explained in Appendix \ref{App:Identities}. In particular, this property can be used to prove the following identity:
\begin{equation}\label{Eq:IdentityRhot}
\sum_{s=1}^{N+1} \rhot_{rs} - \frac{\kc r}{2} = \sum_{s=1}^{N+1} \rhot_{sr} + \frac{\kc r}{2} = 0.
\end{equation}
Let us note that, from Equation \eqref{Eq:XLagGauged}, one gets
\begin{equation*}
\sum_{r=1}^{N+1} \left( \Xb r (x) - \kc r \, \Wb r(x) - \kc r \, \jb r (x) \right) = \sum_{r=1}^{N+1} \left( \sum_{s=1}^{N+1} \rhot_{rs} - \frac{\kc r}{2} \right) \jb r_+ + \left( \sum_{s=1}^{N+1} \rhot_{sr} + \frac{\kc r}{2} \right) \jb r_- .
\end{equation*}
Thus, the identity \eqref{Eq:IdentityRhot} ensures that the constraint $\Cct = \displaystyle \sum_{r=1}^{N+1} \left( \Xb r - \kc r \, \Wb r - \kc r \, \jb r \right)$ vanishes on-shell, as one should expect.

\paragraph{Action.} Starting from the Lagrangian expression \eqref{Eq:XLagGauged} of $\Xb r$, the computation of the action of the gauged model follows the exact same steps than the one of the non-gauged model. We shall skip the details of this computation here and refer to Subsection \ref{SubSubSec:LagNonGauged} of this article and to Subsection 3.3.2 of~\nm. In the end, one finds (the subscript $g$ indicates here that this is the ``gauged'' action)
\begin{equation}\label{Eq:ActionGauged}
S_g\bigl[ \gb 1, \cdots, \gb {N+1} \bigr] = \iint_{\R\times\D} \dd t \, \dd x \; \sum_{r,s=1}^{N+1} \; \rhot_{rs} \, \kappa\left(\jb r_+, \jb s_-\right)  + \sum_{r=1}^{N+1} \kc r \; \Ww {\gb r}.
\end{equation}
Let us note that this action possesses $N+1$ global $G_0$-symmetries, acting on the $N+1$ fields $\gb r$ by left multiplication
\begin{equation*}
\gb 1(x,t) \longmapsto h_1^{-1}\,\gb 1(x,t), \;\;\;\;\;\; \cdots, \;\;\;\;\;\; \gb {N+1}(x,t) \longmapsto h_{N+1}^{-1}\,\gb {N+1}(x,t),
\end{equation*}
with $h_1,\cdots,h_{N+1}$ independent $G_0$-valued constant parameters. We will discuss the gauge symmetry of the model in the following paragraph.

\subsubsection{Gauge symmetry and gauge fixing}
\label{SubSubSec:GaugeSymSigma}

\paragraph{Gauge invariance.} As explained in Subsection \ref{SubSubSec:HamGauged}, the model possesses a gauge symmetry which acts on the fields $\gb r$ by a diagonal right multiplication:
\begin{equation}\label{Eq:GaugeTransfoSigma}
\gb 1(x,t) \longmapsto \gb 1(x,t)\,h(x,t), \;\;\;\;\;\; \cdots, \;\;\;\;\;\; \gb {N+1}(x,t) \longmapsto \gb {N+1}(x,t)\,h(x,t).
\end{equation}
Thus, this model can be seen as an integrable $\s$-model on the coset $G_0^{N+1}/G_{\text{diag}}$, where $G_{\text{diag}} = \bigl\lbrace (h,\cdots,h), \; h\in G_0 \bigr\rbrace$, is the diagonal subgroup of $G_0^{N+1}$, which acts by right multiplication on $G_0^{N+1}$ as in Equation \eqref{Eq:GaugeTransfGamma}.

Let us check the invariance of the action \eqref{Eq:ActionGauged} under this local transformation. For $r\in\lbrace 1,\cdots,N+1\rbrace$, the light-cone currents $\jb r_\pm=\gb r\null^{-1} \p_\pm \gb r$ transform as
\begin{equation*}
\jb r_\pm \longmapsto h^{-1} \Bigl( \jb r_\pm + \nu_\pm \Bigr) h, \;\;\;\;\;\;\;\; \text{ with } \;\;\;\;\;\;\;\; \nu_\pm = \bigl( \p_\pm h \bigr)  h^{-1}.
\end{equation*}
The Wess-Zumino term of the product $\gb r\,h$ can be expressed using the Polyakov-Wiegmann formula~\cite{Polyakov:1983tt}, yielding the transformation of $\Ww{\gb r}$ under the gauge transformation:
\begin{equation*}
\Ww{\gb r} \longmapsto \Ww{\gb r\,h} = \Ww{\gb r} + \Ww{h} + \frac{1}{2} \iint_{\R\times\D} \dd t\, \dd x\;\Bigl( \kappa\bigl(\nu_+,\jb r_-\bigr) - \kappa\bigl(\jb r_+,\nu_-\bigr) \Bigr).
\end{equation*}
Inserting the two equations above in the action \eqref{Eq:ActionGauged} and using the conjugacy-invariance of the form $\kappa$, we find
\begin{align*}
S_g\bigl[ \gb 1, \cdots, \gb {N+1} \bigr] \longmapsto \hspace{5pt} S_g\bigl[ \gb 1, \cdots, \gb {N+1} \bigr] + \left( \sum_{r=1}^{N+1} \kc r \right) \Ww h + \left( \sum_{r,s=1}^{N+1} \rhot_{rs} \right) \iint \dd t\, \dd x \; \kappa(\nu_+,\nu_-) \\
 + \hspace{6pt} \iint \dd t\,\dd x \; \left[ \sum_{s=1}^{N+1} \left( \sum_{r=1}^{N+1} \rhot_{sr} - \frac{\kc s}{2} \right) \kappa\bigl(\jb s_+,\nu_-\bigr) + \sum_{s=1}^{N+1} \left( \sum_{r=1}^{N+1} \rhot_{rs} + \frac{\kc s}{2} \right) \kappa\bigl(\jb s_-,\nu_+\bigr) \right] \hspace{25pt}
\end{align*}
The second line vanishes by the identity \eqref{Eq:IdentityRhot}. Moreover, let us note that $\sum_{r=1}^{N+1} \kc r$ is equal to 0 by the definition \eqref{Eq:klN+1} of $\kc{N+1}$ (this is in fact the first-class condition $\res_{\zt=\infty} \vp(\zt)\dd\zt=0$). Finally, summing the identity \eqref{Eq:IdentityRhot} over $r\in\lbrace 1,\cdots,N+1 \rbrace$ and using the fact that the sum of the $\kc r$'s is zero, one finds that the sum of the $\rhot_{rs}$ over $r,s\in\lbrace 1,\cdots,N+1 \rbrace$ also vanishes. This proves the invariance of the action \eqref{Eq:ActionGauged} under the gauge transformation \eqref{Eq:GaugeTransfoSigma}.

\paragraph{Gauge fixing.} Let us now consider the model \eqref{Eq:ActionGauged} under the gauge-fixing condition $\gb{N+1}=\Id$. The action becomes
\begin{equation*}
S_g\bigl[ \gb 1, \cdots, \gb N, \Id \bigr] = \iint_{\R\times\D} \dd t \, \dd x \; \sum_{r,s=1}^N \; \rhot_{rs} \, \kappa\left(\jb r_+, \jb s_-\right)  + \sum_{r=1}^N \kc r \; \Ww {\gb r}.
\end{equation*}
According to the general results of Subsection \ref{SubSec:Gauging}, this gauge-fixed action should coincide with the action \eqref{Eq:ActionNonGauged} of the initial non-gauged $\s$-model. It is clear that the Wess-Zumino terms appearing in the equation above are the same as in the action \eqref{Eq:ActionNonGauged}. The equivalence of the two action then reduces to the simple identity
\begin{equation}\label{Eq:RhoEqualRhot}
\rho_{rs} = \rhot_{rs}, \;\;\;\;\;\;\;\; \forall \, r,s\in\lbrace 1,\cdots,N \rbrace.
\end{equation}
One can verify this equality directly from the definitions \eqref{Eq:Rho} and \eqref{Eq:RhoTilde} of $\rho_{rs}$ and $\rhot_{rs}$ in terms of the functions $\vppm r(z)$ and $\vptpm r(\zt)$, together with the relation \eqref{Eq:TwistPMRel} between these functions. We give an alternative proof of this identity in Appendix \ref{App:Identities}, using the interpretation of the coefficients $\rho_{rs}$ and $\rhot_{rs}$ as residues of bi-differentials.\\

Alternatively, this gauge-fixing procedure can be seen as performing a gauge transformation \eqref{Eq:GaugeTransfoSigma} with parameter $h(x,t)=\gb{N+1}(x,t)^{-1}$. One then re-expresses the action in terms of the $N$ gauge-invariant fields $\gb 1\hspace{1pt}'(x,t)=\gb 1(x,t)\gb{N+1}(x,t)^{-1}, \cdots, \gb N\hspace{1pt}'(x,t)=\gb N(x,t)\gb{N+1}(x,t)^{-1}$ as
\begin{equation*}
S_g\bigl[ \gb 1\hspace{1pt}', \cdots, \gb N\hspace{1pt}', \Id \bigr] = S\bigl[ \gb 1\hspace{1pt}', \cdots, \gb N\hspace{1pt}'\bigr]
\end{equation*}
This formulation makes it easier to identify the symmetries of the gauged model with the symmetries of the non-gauged model. Indeed, it is clear that for $r\in\lbrace 1,\cdots,N \rbrace$, the left multiplication $\gb r \mapsto h^{-1}_r\,\gb r$ translates to the left multiplication $\gb r\hspace{1pt}' \mapsto h^{-1}_r\,\gb r\hspace{1pt}'$, while the left multiplication $\gb {N+1} \mapsto h^{-1}\,\gb {N+1}$ translates to the diagonal right multiplication $\bigl(\gb 1\hspace{1pt}',\cdots,\gb 1\hspace{1pt}'\bigr) \mapsto \bigl(\gb 1\hspace{1pt}'\,h,\cdots,\gb 1\hspace{1pt}'\,h\bigr)$.

\paragraph{Discrete symmetries in the space of parameters.} As an application of the formalism developed above, let us discuss some discrete symmetries of the parameters defining the integrable coupled $\s$-model considered here. Let us first consider the non-gauged action \eqref{Eq:ActionNonGauged}. If $\s \in \mathfrak{S}_N$ is a permutation of the labels $\lbrace 1,\cdots,N \rbrace$, it is clear that the model with the coefficients $\rho_{rs}$ and $\kc r$ replaced by $\rho_{\s(r)\s(s)}$ and $\kc{\s(r)}$ is equivalent to the initial one, as this transformation just corresponds to renaming the fields $\gb r$ as $\gb{\s(r)}$. The action \eqref{Eq:ActionNonGauged} then possesses a $\mathfrak{S}_N$ discrete symmetry (here the notion of symmetry is understood in the sense of a transformation of the parameters defining the model which yields an equivalent model, not in the sense of a transformation of the fields which leaves the action invariant).

Let us now turn to the gauged action \eqref{Eq:ActionGauged}. It is clear that this model possesses a discrete $\mathfrak{S}_{N+1}$ symmetry, which permutes the labels $r,s\in\lbrace 1,\cdots,N+1 \rbrace$. However, the gauged action \eqref{Eq:ActionGauged} describes the same model as the non-gauged action \eqref{Eq:ActionNonGauged}. Thus, one should be able to realise this $\mathfrak{S}_{N+1}$ symmetry at the non-gauged level. Recall that one passes from the gauged formalism to the non-gauged one by imposing the gauge-fixing condition $\gb{N+1}=\Id$. If we consider a permutation in $\mathfrak{S}_{N+1}$ which belongs to the $\mathfrak{S}_{N}$-subgroup of permutations which fix $N+1$, then it mixes the fields $\gb{1},\cdots,\gb{N}$ but does not modify the gauge-fixing condition $\gb{N+1}=\Id$. It then just results in the corresponding permutation of the fields $\gb{1},\cdots,\gb N$ in the gauge-fixed model. The obvious $\mathfrak{S}_{N}$-symmetry of the non gauged model is thus identified with the subgroup of $\mathfrak{S}_{N+1}$ which fixes the label $N+1$.

However, if one considers a permutation in $\mathfrak{S}_{N+1}$ which does not fix $N+1$, this leads to a discrete symmetry which acts on the parameters of the non-gauged model in a more involved way. One can realise this discrete symmetry by performing the permutation at the level of the coefficients $\rhot_{rs}$ and $\kc r$ for $r,s\in\lbrace 1,\cdots,N+1 \rbrace$ and then eliminating the field $\gb{N+1}$ by gauge-fixing it. One can then express the action of the symmetry in terms of the original coefficients $\rho_{rs}$ and $\kc r$, using the fact that $\rhot_{rs}=\rho_{rs}$ for $r,s\in\lbrace 1,\cdots,N \rbrace$ by Equation \eqref{Eq:RhoEqualRhot} and that the coefficients of the form $\rhot_{N+1,r}$ and $\rhot_{r,N+1}$ can be expressed in terms of the others using the identity \eqref{Eq:IdentityRhot}. Equivalently, the involution $(r,N+1)$ in $\mathfrak{S}_{N+1}$ which exchanges $N+1$ with a label $r\in\lbrace 1,\cdots,N\rbrace$ (which, together with the subgroup $\mathfrak{S}_{N}$, generates $\mathfrak{S}_{N+1}$) can be understood as changing the gauge-fixing condition $\gb{N+1}=\Id$ into $\gb r=\Id$: it is clear that the resulting action takes the same form as \eqref{Eq:ActionNonGauged}, but with different coefficients.\\

Let us illustrate these discrete symmetries on examples. We start with the simplest example of a unique PCM with Wess-Zumino term, \textit{i.e.} the case $n=1$. In this case, one verifies that the only symmetry in $\mathfrak{S}_{2}$ leaves the coefficient $\rho_{11}$ invariant but changes the level $\kc 1$ into $-\kc 1$. This is equivalent to considering the gauge-fixing $\gb 1=\Id$ instead of $\gb 2=\Id$, or equivalently to using the gauge-invariant field $\gb 2\,\gb 1\null^{-1}$ instead of $\gb 1\,\gb 2\null^{-1}$. This involutive symmetry can thus be seen as the map $g \mapsto g^{-1}$ on the field of the model. It is a classic result that such a transformation leaves invariant the quadratic term $\kappa(j_+,j_-)$ but sends the Wess-Zumino term $\Ww g$ to its opposite.

Let us end this paragraph by an example involving more than one copy, by considering $n=2$. Let us describe the action of the involution $(2,3) \in \mathfrak{S}_3$. After a few manipulations, one finds that it acts on the coefficients $\rho_{rs}$ and $\kc r$ as\vspace{-5pt}
\begin{equation*}
\rho_{11} \mapsto \rho_{11}, \;\;\;\;\;\;\; \rho_{12} \mapsto -\rho_{11}-\rho_{12}+\frac{\kc 1}{2}, \;\;\;\;\;\;\; \rho_{21} \mapsto -\rho_{11}-\rho_{21}-\frac{\kc 1}{2}, \vspace{-4pt}
\end{equation*}
\begin{equation*}
\rho_{22} \mapsto \rho_{11}+\rho_{12}+\rho_{21}+\rho_{22}, \;\;\;\;\;\;\; \kc 1 \mapsto \kc 1, \;\;\;\;\;\;\; \kc 2 \mapsto -\kc 1-\kc 2.
\end{equation*}
One checks that this indeed defines an involution on the parameters $\rho_{rs}$ and $\kc r$. This transformation can also be thought of as replacing the gauge-fixing condition $\gb 3=\Id$ by $\gb 2=\Id$. Seeing these gauge-fixings as forming good gauge-invariant combinations of the fields $\gb 1$, $\gb 2$ and $\gb 3$, one finds that in the non-gauged model, it amounts to a redefinition of the fields $\gb 1 \mapsto \gb 1\,\gb 2\null^{-1}$ and $\gb 2 \mapsto \gb 2\null^{-1}$. One can perform this transformation at the level of the action, using the Polyakov-Wigmman formula~\cite{Polyakov:1983tt}, to recover the above change of parameters. Note that not all redefinitions of the fields lead to an action of the same form: for example, $\gb 1 \mapsto \gb 1\,\gb 2\null^{-1}$ and $\gb 2 \mapsto \gb 2\null$ does not. Considering redefinitions which originate from changing the gauge-fixing condition ensures that the action stays of the same form, as above.

\subsection[Diagonal homogeneous Yang-Baxter deformation of the coupled $\s$-model]{Diagonal homogeneous Yang-Baxter deformation of the coupled $\bm{\s}$-model}
\label{SubSec:YBSigma}

\paragraph{Principle.} Let us end this section by an application of the formalism developed in Section \ref{Sec:Gauging}: the construction of a new integrable $\s$-model as a diagonal Yang-Baxter deformation of the integrable coupled $\s$-model considered above. Indeed, as explained in Subsection \ref{SubSec:DiagYB}, having found a gauged formulation of this model in Subsection \ref{SubSec:GaugedSigma}, one can then construct a Yang-Baxter deformation breaking its diagonal symmetry. For brevity, we will focus here on the simplest possible deformation, the homogeneous one (see Subsection \ref{SubSubSec:Homogeneous}).

Recall that such a deformation can be applied only if the site $\is$ of the gauged model, in the notation of Section \ref{Sec:Gauging}, is associated with a PCM realisation with Wess-Zumino term. In this section, we renamed the site $\is$ as $(N+1)$. We shall thus suppose that the Takiff currents at this site are the ones of a PCM realisation without Wess-Zumino term:
\begin{equation*}
\Jtt{(N+1)}0(x) = \Xb{N+1}(x) \;\;\;\;\;\;\; \text{ and } \;\;\;\;\;\;\; \Jtt{(N+1)}1(x) = \elt{N+1}\,\jb{N+1}(x).
\end{equation*}
This is the case if and only if the level $\lst{(N+1)} 0$ of the site $(N+1)$ is equal to zero, \textit{i.e.} if and only if
\begin{equation*}
\kc{N+1} = -\sum_{r=1}^N \kc r = 0.
\end{equation*}
We shall then suppose that this condition is satisfied through this subsection.

\paragraph{Hamiltonian formulation.} The homogeneous diagonal Yang-Baxter deformation of the gauged model $\mathbb{M}^{\vpt,\pit}_{\eb}$ shares the same twist function $\vpt(\zt)$. However, it differs from the undeformed model by its choice of Takiff currents. Following the general procedure explained in Subsection \ref{SubSubSec:Homogeneous}, we keep the Takiff currents \eqref{Eq:NewCurrentsSigma} attached to the sites $(1)$ to $(N)$ of the model $\mathbb{M}^{\vpt,\pit}_{\eb}$ but change the currents attached to the site $(N+1)$ into the ones of a homogeneous Yang-Baxter realisation:
\begin{equation*}
\Jtt{(N+1),\hYB}0\null(x) = \Xb{N+1}(x) \;\;\;\;\;\;\; \text{ and } \;\;\;\;\;\;\; \Jtt{(N+1),\hYB}1(x) = \elt{N+1}\,\jb{N+1}(x) - \Rb{N+1}\Xb{N+1}(x).
\end{equation*}
In this equation, $\Rb{N+1}=\Ad_{\gb{N+1}}^{-1} \circ R \circ \Ad_{\gb{N+1}}$, with $R:\g_0\rightarrow\g_0$ a skew-symmetric solution of the CYBE \eqref{Eq:CYBE}. The Gaudin Lax matrix is then
\begin{equation*}
\Gt_\hYB(\zt,x) = \sum_{r=1}^{N} \left( \frac{\Jtt{(r)}1(x)}{(\zt-\zt_r)^2} + \frac{\Jtt{(r)}0(x)}{\zt-\zt_r} \right) + \frac{\Jtt{(N+1),\hYB}1(x)}{(\zt-\zt_{N+1})^2} + \frac{\Jtt{(N+1),\hYB}0(x)}{\zt-\zt_{N+1}}.
\end{equation*}
As in the undeformed case, the Hamiltonian is defined as
\begin{equation*}
\widetilde{\Hc}_\hYB = \sum_{i=1}^{2N} \epsilon_i \, \widetilde{\Q}_i + \int_\D \dd x \; \kappa\bigl( \Cct(x), \mu(x) \bigr), \;\; \text{ where } \;\; \widetilde{\Q}_i = - \frac{1}{2\vpt\hspace{1pt}'\bigl(\zet_i\bigr)} \int_\D \dd x \; \kappa\Bigl( \Gt_\hYB\bigl(\zet_i,x\bigr), \Gt_\hYB\bigl(\zet_i,x\bigr) \Bigr).
\end{equation*}

\paragraph{Lagrangian Lax pair.} Applying a homogeneous Yang-Baxter deformation to the non-gauged model has been done in~\nms (see Appendix D): the resulting action is then computed in a very similar way to the undeformed case, using a variant of the interpolation technique. As we have seen in Subsection \ref{SubSubSec:LagSigmaGauged}, this interpolation technique also applies straightforwardly to the gauged model. Similarly, the method developed in~\nms to treat the homogeneous Yang-Baxter deformation is also easily transposed in the gauged case. For brevity, we shall only sketch the main steps of this computation.

One can compute the Poisson bracket of the field $\gb r$ with the Gaudin Lax matrix $\Gt_\hYB(\zt)$ and then with the quadratic charge $\widetilde{\Q}_i$, as done in Equation \eqref{Eq:PBgQi} in the undeformed case. Recall that in the undeformed case, this computation led to the simple formula \eqref{Eq:jLaxEvalGauge} relating the current $\jb r_\pm$ with the evaluation $\Lct_\pm(\zt_r)$ of the Lax pair at the position $\zt_r$. In the deformed case, this relation still holds for $r\in\lbrace 1,\cdots,N \rbrace$ as the site $(1),\cdots,(N)$ are still associated with PCM+WZ realisations. However, it is modified for $r=N+1$ and now reads
\begin{equation}\label{Eq:jLaxEvalYB}
\jb {N+1}_\pm = \Lct_\pm(\zt_{N+1}) + \Rb {N+1} \Lct_\pm'(\zt_{N+1}),
\end{equation}
where $\Lct_\pm'(\zt)$ denotes the derivative of $\Lct_\pm(\zt)$ with respect to the spectral parameter $\zt$. The equation above is the equivalent of Equation (D.7) in~\nms and is obtained in a very similar way (the main difference being that one now has to take into account the contribution of the Lagrange multiplier $\mu$ to the dynamic, as we did in Subsection \ref{SubSubSec:LagSigmaGauged} for the undeformed gauged model).\\

We now determine the Lagrangian expression of the Lax pair $\Lct_\pm(\zt)$ using a similar interpolation formula than the one \eqref{Eq:LaxLagSigmaGauge} in the undeformed case. To take into account the deformation, this formula is modified into the form
\begin{equation*}
\Lct_\pm(\zt) = \sum_{r=1}^N \frac{\vptpm r(\zt_r)}{\vptpm r(\zt)}\jb r_\pm + \frac{\vptpm {N+1}(\zt_{N+1})}{\vptpm {N+1}(\zt)}\Jb {N+1}_\pm.
\end{equation*}
As the evaluation equation \eqref{Eq:jLaxEvalGauge} is still valid for $r\in\lbrace 1,\cdots,N \rbrace$, we do not modify the corresponding terms in the equation above. Thus, we only change the current $\jb{N+1}_\pm$ into a new current $\Jb{N+1}_\pm$, which we now have to determine, using Equation \eqref{Eq:jLaxEvalYB}. One then finds (see Equation (D.10) of~\nm)
\begin{equation*}
\Jb {N+1}_\pm = \frac{1}{1 \pm \vt \, \Rb{N+1}} \left( \jb {N+1}_\pm \mp \frac{2}{\elt {N+1}} \sum_{r=1}^N a^\pm_{r} \,\Rb {N+1}\jb r_\pm \right),
\end{equation*}
where we introduced the parameters
\begin{equation*}
\vt = \frac{2\rhot_{N+1,N+1}}{\elt{N+1}}, \;\;\;\;\; a^+_r = \rhot_{r,N+1} \;\;\;\;\; \text{ and } \;\;\;\;\; a^-_r = \rhot_{N+1,r}.
\end{equation*}

\paragraph{Gauged action.} Following the method used in Subsection D.4 of~\nm, one can perform the inverse Legendre transform of the deformed model introduced above. We shall skip here the details of this computation and just give the final result. The action of the (gauged) homogeneous Yang-Baxter deformed model is given by
\begin{equation}\label{Eq:ActionYBGauged}
S^\hYB_g\bigl[ \gb 1, \cdots, \gb {N+1} \bigr] = \iint_{\R\times\D} \dd t \, \dd x \; \sum_{r,s=1}^{N+1} \; \, \kappa\left(\jb r_+, \Oc rs^g\, \jb s_-\right)  + \sum_{r=1}^N \kc r \; \Ww {\gb r},
\end{equation}
where the operator $\Oc rs^g$ is defined as
\begin{equation*}
\Oc rs^g = \rhot_{rs}\,\Id + \frac{2\rhot_{r,N+1}\rhot_{N+1,s}}{\elt {N+1}} \frac{\Rb{N+1}}{1-\vt\,\Rb{N+1}}.
\end{equation*}
It is clear that this operator reduces to $\rhot_{rs}\,\Id$ in the undeformed limit $R = 0$, so that the action $S^\hYB_g$ is indeed a deformation of the action $S_g$, given by Equation \eqref{Eq:ActionGauged}.

One can check that this action is invariant under the gauge transformation \eqref{Eq:GaugeTransfoSigma} by a similar method than the one used in the non-gauged case (see Subsection \ref{SubSubSec:GaugeSymSigma}). In particular, one has to use the identity
\begin{equation*}
\sum_{s=1}^{N+1} \Oc rs^g - \frac{\kc r}{2} \,\Id = \sum_{s=1}^{N+1} \Oc sr^g + \frac{\kc r}{2}\, \Id = 0,
\end{equation*}
which is similar to the relation \eqref{Eq:IdentityRhot} and is directly derived from it (recalling that here $\kc{N+1}=0$, as supposed from the beginning of this subsection).

\paragraph{Gauge fixing.} To see the deformed model constructed above as a deformation of the initial non-gauged model with action \eqref{Eq:ActionNonGauged}, one has to perform the gauge-fixing $\gb{N+1}=\Id$. Under this gauge-fixing condition, the operator $\Rb{N+1}=\Ad_{\gb{N+1}}^{-1} \circ R \circ \Ad_{\gb{N+1}}$ simply reduces to $R$. Moreover, recall from Equation \eqref{Eq:RhoEqualRhot} that the coefficient $\rhot_{rs}$ is equal to $\rho_{rs}$ for $r,s\in\lbrace 1,\cdots,N\rbrace$. Thus, the gauge-fixed action of the deformed model takes the form
\begin{equation}\label{Eq:ActionYB}
S^\hYB\bigl[ \gb 1, \cdots, \gb N \bigr] = \iint_{\R\times\D} \dd t \, \dd x \; \sum_{r,s=1}^N \; \, \kappa\left(\jb r_+, \Oc rs\, \jb s_-\right)  + \sum_{r=1}^N \kc r \; \Ww {\gb r},
\end{equation}
with the operators
\begin{equation*}
\Oc rs = \rho_{rs}\,\Id + \frac{2\rhot_{r,N+1}\rhot_{N+1,s}}{\elt {N+1}} \frac{R}{1-\vt\,R}.
\end{equation*}
This action is clearly a deformation of the initial action \eqref{Eq:ActionNonGauged}. Moreover, note that the various parameters appearing in the expression of the operators $\Oc rs$ can be expressed in terms of the coefficients $\rho_{rs}$, $\kc r$ and $\ell^\infty$ of the original model using the identities \eqref{Eq:klN+1}, \eqref{Eq:IdentityRhot} and \eqref{Eq:RhoEqualRhot}.

The action \eqref{Eq:ActionYB} is invariant under the global left-translation symmetries $\gb 1 \mapsto h_1^{-1} \gb 1, \, \cdots, \gb N \mapsto h_N^{-1} \gb N$, with constant parameters $h_1,\cdots,h_N$ in $G_0$, similarly to the undeformed action \eqref{Eq:ActionNonGauged}. However, the diagonal right-translation symmetry  $\gb 1 \mapsto \gb 1 h, \cdots, \gb N \mapsto \gb N h$ is broken by the presence of the operators $\Oc rs$ in the action, as one should expect from the general diagonal Yang-Baxter deformation procedure.

\section{Conclusion}

In this article, we reviewed and pushed further the study of realisations of affine Gaudin models with gauge symmetries. In particular, we have shown that any non-gauged realisation of affine Gaudin model as considered in the article~\nms possesses an equivalent reformulation as such a gauged model, \textit{via} an appropriate change of spectral parameter. In addition to a better understanding of the structure of realisations of affine Gaudin models, this result allowed us to introduce the notion of diagonal Yang-Baxter deformation procedure, which provides a systematic way of deforming the models considered in~\nms while preserving their integrability, but breaking their diagonal symmetry. Moreover, we applied these general results to the study of integrable $\s$-models and in particular constructed the gauged formulation and the homogeneous diagonal Yang-Baxter deformation of the integrable coupled $\s$-model recently introduced in~\prl.\\

The results of Subsection \ref{SubSec:ChangeSpec} and Section \ref{Sec:Gauging} provide some interesting perspective on the role played by the spectral parameter in the construction of affine Gaudin models and their realisations, as also did Subsection 2.3.2 of~\nm. Let us recall the results of the latter: as the rest of the article~\nm, it concerned realisations of affine Gaudin models which possess a site of multiplicity two at infinity. This site is treated in a slightly different way than the finite ones: in particular, it only contributes to the model through the appearance of a constant term in the twist function and it is responsible for the absence of a gauge symmetry in the model. Although the construction of such a model depends on the choice of spectral parameter (for example through the positions of the sites of the model), it was then shown in Subsection 2.3.2 of~\nms that two models related by a change of spectral parameter $z \mapsto \zt=f(z)=az+b$ are in fact equivalent. The form of the transformation $f$ considered in this case, \textit{i.e.} the combination of a translation and a dilation, is justified by the fact that it preserves the point $z=\infty$ and thus does not break the particular status of the site at infinity in this model. In particular, it ensures that this change of spectral parameter relates two non-gauged affine Gaudin models.

More general transformations of the spectral parameter were considered in the present article, in Subsection \ref{SubSec:ChangeSpec} and \ref{SubSec:ChangeSpec2}. More precisely, we considered general M\"obius transformations, \textit{i.e.} biholomorphisms of the Riemann sphere $\CP=\C\sqcup\lbrace\infty\rbrace$. If one requires such a transformation $f$ to stabilise the point $\infty$ of the Riemann sphere, $f$ reduces to the combination of a translation and a dilation, as considered above. However, there exist more general M\"obius transformations, which sends the infinity to the finite complex plane $\C$. The action of such a transformation on a non-constrained model as in~\nms was considered in Subsection \ref{SubSec:ChangeSpec2}: this is in fact the starting point of the gauging procedure developed in Section \ref{Sec:Gauging}. The result of this procedure is then that the model obtained after performing the change of spectral parameter $z\mapsto \zt=f(z)$, although being a model with gauge symmetry, is equivalent to the initial non-gauged model with spectral parameter $z$, thus providing another result of invariance of affine Gaudin models under the choice of spectral parameter.

Finally, such a statement was also obtained in Subsection \ref{SubSec:ChangeSpec}, this time relating two constrained realisations of affine Gaudin models. This was ensured by considering only M\"obius transformations $f$ which send the finite sites of the first model to positions which are still finite: this is the case for example if $f$ is just the combination of a dilation and a translation, but also for a larger class of M\"obius transformations, as one relaxed the condition that $f$ should fix the point at infinity.

These different results suggest the existence of a deeper formulation of realisations of affine Gaudin models, which would be automatically independent of the choice of spectral parameter (indeed, although the various results above proved such an independence at least in some particular cases, their proofs are rather long and involved, and do not particularly shed light on the underlying structure of the model which is ultimately responsible for this invariance, especially for the transformations relating a non-gauged model to a gauged one). In such a formulation, the defining ``geometric'' data of the model, instead of being the positions of the sites of the model, would be given by the choice of a $N$-punctured Riemann sphere. The choice of spectral parameter $z$ would then be equivalent to the choice of a local coordinate on the Riemann sphere: such a choice of coordinate would in particular specify a point with a different status than the others (as not belonging to the coordinate chart), the infinity point $\infty$, and would also attach to each puncture a position $z_\alpha$, which could then be interpreted as the position of a site $\alpha$ in the usual formulation of affine Gaudin models. Following the results obtained in Subsections \ref{SubSec:ChangeSpec} and \ref{SubSec:ChangeSpec2}, the key ingredients of the affine Gaudin model, such as its Gaudin Lax matrix and its twist function, should then be thought of as 1-forms on the underlying punctured Riemann sphere. The search for such a geometric formulation is a very interesting direction for a deeper understanding of affine Gaudin models.

It is worth mentioning that such a formulation already exists for finite Gaudin models~\cite{Gaudin_book83}, associated with finite Kac-Moody algebras instead of affine ones, in the form of the so-called Hitchin systems~\cite{Hitchin:1987mz}. In fact, Hitchin systems can be defined starting from punctured Riemann surfaces of arbitrary genus. Finite Gaudin models then correspond to the case of genus 0, \textit{i.e.} the Riemann sphere. It would be interesting to see if the construction of Hitchin systems can be generalised to the affine case and if it could provide a way to construct new integrable field theories starting with higher genus surfaces.

Let us finally note, concerning the geometric aspects of affine Gaudin models and their spectral parameter, that the results of Appendix \ref{App:Identities} also provides an interesting geometric reformulation of the integrable coupled $\s$-model with $N$ copies, by interpreting the coefficients appearing in the action of the model (either gauged or non-gauged) as double residues of a bi-differential $\Theta$. Indeed, this result encodes all the information on the model in a unique geometric object $\Theta$, which is independent of the choice of spectral parameter. This makes obvious the invariance of the model under a change of spectral parameter, as residues of bi-differentials are invariant under a change of coordinate of the underlying Riemann surface.\\

It would also be interesting to develop further the applications of the diagonal Yang-Baxter deformation procedure to the study of integrable $\s$-models. In this article, we presented a simple example of such a deformation: the homogeneous diagonal Yang-Baxter deformation of the integrable coupled $\s$-model with $N$ sites introduced in~\prl, resulting in the action \eqref{Eq:ActionYB}. This deformation corresponds to the breaking of the diagonal symmetry $(\gb 1,\cdots,\gb N) \mapsto (\gb 1 h,\cdots,\gb N h)$ of the model. It can be combined with other deformations, corresponding to breaking the remaining global symmetries of the model. For example, one can also deform the PCM+WZ realisation attached to one of the sites $(r)$ of the model into a Yang-Baxter realisation (either homogeneous or inhomogeneous), thus breaking the left symmetry $\gb r \mapsto h^{-1}\gb r$ associated with this site. In total, one can then construct an integrable coupled model with $N+1$ Yang-Baxter deformations, breaking all symmetries of the undeformed model. Similarly, one can also consider $\lambda$-realisations in some of the sites and obtain other types of deformed $\s$-models. It would be interesting to explore the panorama of possibilities of such deformations and to study their properties.

As explained in Subsection \ref{SubSubSec:PCMSymAndDef}, this is already an interesting question for models with only one copy, which correspond to multi-parameter deformations of the Principal Chiral Model (or its non-abelian T-dual). It would be for example desirable to compare the models obtained this way with some of the multi-parameter deformations already existing in the literature, such as the generalised $\lambda$-deformed model introduced in~\cite{Sfetsos:2015nya} and the 3-parameter model constructed in~\cite{Delduc:2017fib}. These models have been shown to possess a Lax pair formulation at the Lagrangian level ; however, their Hamiltonian analysis and the study of the Poisson bracket of their Lax matrix has not been performed yet. Identifying these models with constrained realisations of affine Gaudin models would then show their Hamiltonian integrability and demonstrate that these various deformations all fit into one unique formalism. Some results in this direction will be presented in another article~\cite{Lacroix:toAppear} by the author.\\

Another interesting direction to explore is the study of the $\mathbb{Z}_T$-cyclotomic affine Gaudin models, using the terminology of~\bg, which are affine Gaudin models possessing additional equivariance properties with respect to an automorphism $\s$ of the Lie algebra $\g$ with finite order $T$. It was shown in~\bgs that the integrable $\s$-model on the $\mathbb{Z}_T$-coset $G_0/H$~\cite{Young:2005jv}, where $H$ is the subgroup of fixed points under the automorphism $\s$, can be seen as a realisation of such a cyclotomic affine Gaudin model. More precisely, this model is realised as a field theory on the group $G_0$ with a gauge symmetry acting by multiplication by the subgroup $H$. This gauge symmetry is implemented at the level of the affine Gaudin model by the presence of a first-class constraint extracted from the Gaudin Lax matrix at infinity, similarly to what is considered in this article. The integrable $\s$-model on $G_0/H$ corresponds to a cyclotomic model with only one site (associated with a PCM realisation in $\Ac_{G_0}$): applying the same idea that led to the construction of the integrable coupled $\s$-model on $G_0^N$ in~\nm, it would be interesting to consider similar models with $N$ sites, each realised in a copy of $\Ac_{G_0}$. As already announced in~\prl, we expect these models to describe integrable $\s$-models on $G_0^N/H^{\text{diag}}$, in a similar way than the non-cyclomotic constrained model studied in Subsection \ref{SubSec:GaugedSigma} of this article describes an integrable $\s$-model on $G_0^{N+1} / G_0^{\text{diag}}$.\\

It would also be interesting to study the quantisation of affine Gaudin models and their hierarchies of local charges and thus in particular the quantisation of the constrained models considered in this article (see Subsection \ref{SubSec:Integrability} for the construction of the hierarchies in this case). This was in fact the main motivation for the reinterpretation of integrable $\s$-models as affine Gaudin models in~\bgs (or similarly of the Drinfeld-Sokolov hierarchies in~\cite{Feigin:2007mr}). It is worth noticing that the quantisation of finite Gaudin models~\cite{Gaudin_book83}, associated with finite Kac-Moody algebras, has been studied extensively in the literature. In particular, their integrable hierarchies have been quantised in~\cite{Feigin:1994in}, leading to the description of their spectrum in terms of opers (see \textit{e.g.}~\cite{Frenkel:2004qy}), which are differential operators naturally associated with the finite Kac-Moody algebra underlying the Gaudin model. Recent progresses have been made in the search for a quantisation of affine Gaudin models. By analogy with the finite case, several conjectures about the quantisation of their hierarchies have been proposed in~\cite{Lacroix:2018fhf}, by studying the structure of affine opers, which are the generalisation for affine algebras of the differential operators involved in the quantisation of finite Gaudin models. Moreover, some of these conjectures were verified on particular cases in~\cite{Lacroix:2018itd}, through the explicit quantisation of some of the charges belonging to these hierarchies\footnote{More precisely the cubic charges for untwisted affine Kac-Moody algebras of type A.}.

Recall that in the affine Gaudin models considered in this article, the hierarchies of local charges whose quantisation we are discussing are invariant under the gauge symmetry of the model. Recall also that this gauge symmetry is obtained from the natural diagonal symmetry present in affine Gaudin models regular at infinity (more precisely, as explained in Subsection \ref{SubSec:Gauge}, this diagonal symmetry is transformed into a gauge symmetry by considering its moment map, the current extracted from the Gaudin Lax matrix at infinity, as a first-class constraint). It is interesting to note that this diagonal symmetry was also considered in the article~\cite{Lacroix:2018itd}, where it was proved (Theorem 3.8.iii) that it leaves invariant the quantum charges constructed in this article. This is a strong sign for the existence of a quantisation of these hierarchies which preserves their gauge-invariance.

As explained at the beginning of this conclusion, the results of~\nms and of this article suggest the existence of a more geometric formulation of affine Gaudin models which would be independent of a choice of spectral parameter, reminiscent of the Hitchin systems construction for finite Kac-Moody algebras. The quantisation of finite Gaudin models in terms of opers mentioned above fits into a more general program of quantisation of Hitchin systems~\cite{Beilinson:1995}, which is related to the Geometric Langlands Correspondence. It would be interesting to see if this program can give some insight about the quantisation in the affine case. Moreover, it is interesting to notice that the behaviour of affine opers under a change of coordinate in the underlying Riemann surface have been studied in~\cite{Lacroix:2018fhf}, Section 6. In particular, it was shown that functions on affine opers over the Riemann sphere with marked points, which are conjectured to describe the spectrum of quantised affine Gaudin models, are invariant under such a change. This suggests that the results of invariance under a change of spectral parameter mentioned above should also extend to the quantised models. Note however that Equation (6.8a) of~\cite{Lacroix:2018fhf} shows that in the quantum case, the twist function should not be thought of as a meromorphic 1-form anymore but instead as a meromorphic connection on the canonical line bundle over the Riemann sphere. It would be interesting to explore these aspects further.

\paragraph{Acknowledgements.} The author is particularly grateful to M. Magro for useful discussions and comments on the manuscript and to G. Arutyunov, F. Delduc and B. Vicedo for fruitful discussions. The author would also like to thank C. Bassi, R. Belliard, I. Buri\'c and T. Figiel for interesting discussions. This work is funded by the Deutsche Forschungsgemeinschaft (DFG, German Research Foundation) under Germany's Excellence Strategy -- EXC 2121 ``Quantum Universe'' -- 390833306.

\appendix

\section{Constrained Hamiltonian field theories and gauges symmetries}
\label{App:Constraints}

\subsection{Constraints}

Let $\Ac$ be the Poisson algebra of functionals on fundamental fields $\phi^i(x)$, $i \in\lbrace1,\cdots,m\rbrace$, on the one-dimensional space $\D=\R$ or $\D=\mathbb{S}$. We are interested here in Hamiltonian field theories defined on $\Ac$ and subject to a set of constraints
\begin{equation}\label{Eq:AppConstraints}
\Cc^a(x) \approx 0, \;\;\;\;\; a=1,\cdots,n,
\end{equation}
where the $\Cc^a(x)$'s are fields in $\Ac$. The notation $\approx$ in the above equation, that we will also use in the rest of this article, denotes a weak equality, \textit{i.e.} an equality which is true only when the constraints are imposed but not in $\Ac$ in general. We will restrict here to the simple case where the constraints $\Cc^a(x)$ satisfy a Poisson bracket of the form:
\begin{equation}\label{Eq:AppAlgebraConstraints}
\bigl\lbrace \Cc^a(x), \Cc^b(y) \bigr\rbrace = \Sc abc \, (x) \; \Cc^c(x) \, \delta_{xy},
\end{equation}
for $\Sc abc$ some fields in $\Ac$ (and where the sum over the repeated index $c$ is implicit). In particular, this implies that the $\Cc^a(x)$'s are first-class constraints, as their Poisson bracket weakly vanishes:
\begin{equation}\label{Eq:AppFirstClass}
\bigl\lbrace \Cc^a(x), \Cc^b(y) \bigr\rbrace \approx 0.
\end{equation}

\subsection{Hamiltonian and dynamics}

In order to define a constrained dynamical system on $\Ac$ we consider an Hamiltonian, given by a local functional
\begin{equation*}
\Hc_0 = \int_{\D} \dd x \; h(x),
\end{equation*}
whose density $h(x)$ is a function of the $\phi^i(x)$'s and their derivatives. We will suppose that the Hamiltonian $\Hc_0$ is first-class, in the sense that its Poisson bracket with all constraints $\Cc^a(x)$ weakly vanishes:
\begin{equation}\label{Eq:AppH0FirstClass}
\bigl\lbrace \Hc_0, \Cc^a(x) \bigr\rbrace \approx 0.
\end{equation}
This ensures that the the Hamiltonian flow generated by $\Hc_0$ preserves the constraints \eqref{Eq:AppConstraints}. As the above Poisson bracket weakly vanishes, it is proportional to the constraints. We will suppose here that it can be written as\footnote{As $\Hc_0$ is defined as a local functional, its Poisson bracket with $\Cc^a(x)$ is a local field evaluated at the point $x$. In the most general situation, this bracket can contain terms proportional to the spatial derivatives of the constraints evaluated at $x$. Here we will restrict to the case where only the constraints and their first spatial derivatives appear.}
\begin{equation}\label{Eq:AppPBHConstraints}
\bigl\lbrace \Hc_0, \Cc^a(x) \bigr\rbrace = \Tcc ab\, (x) \; \Cc^b(x) + \Uc ab\, (x) \; \p_x\Cc^b(x),
\end{equation}
for some fields $\Tcc ab\,(x)$ and $\Uc ab\,(x)$ in $\Ac$.\\

As we are considering an Hamiltonian field theory with constraints \eqref{Eq:AppConstraints}, there exists an ambiguity in the definition of the Hamiltonian. Indeed, one could add to $\Hc_0$ any quantity proportional to the constraints, as these constraints vanish for physical states. Thus, one should consider the so-called total hamiltonian
\begin{equation}\label{Eq:AppTotalHam}
\Hc = \Hc_0  + \int_{\D} \dd x \; \mu_a(x) \, \Cc^a(x).
\end{equation}
The fields $\mu_a$ appearing in this expression are called the Lagrange multipliers associated with the constraints $\Cc^a$ and are considered as independent auxiliary degrees of freedom. The time evolution of an observable $\mathcal{O}$ in $\Ac$ (which does not depend explicitly on the time coordinate $t$) is then defined by the Hamiltonian flow
\begin{equation}\label{Eq:AppTime}
\p_t \mathcal{O} \approx \lbrace \Hc, \mathcal{O} \rbrace \approx \lbrace \Hc_0, \mathcal{O} \rbrace + \int_\D \dd x \; \mu_a(x) \, \bigl\lbrace \Cc^a(x), \mathcal{O} \bigr\rbrace.
\end{equation}
In particular, as both $\Hc_0$ and the constraints are first-class (see equation \eqref{Eq:AppFirstClass} and \eqref{Eq:AppH0FirstClass}), the constraints \eqref{Eq:AppConstraints} are preserved under the time evolution:
\begin{equation*}
\p_t \Cc^a(x) \approx 0.
\end{equation*}

Let us end this paragraph by discussing how the introduction of Lagrange multipliers resolves the ambiguity in the choice of the Hamiltonian. Let us imagine that we started with another local Hamiltonian $\widetilde{\Hc}_0$, which coincides weakly with $\Hc_0$. This means that their exist fields $\nu_a(x)$ in $\Ac$ such that
\begin{equation*}
\widetilde{\Hc}_0 = \Hc_0 + \int_\D \dd x \; \nu_a(x) \, \Cc^a(x).
\end{equation*}
Thus, one can write the total Hamiltonian of the model as
\begin{equation*}
\Hc = \widetilde{\Hc}_0  + \int_{\D} \dd x \; \widetilde{\mu}_a(x) \, \Cc^a(x),
\end{equation*}
with a new set of Lagrange multiplier
\begin{equation*}
\widetilde{\mu}_a(x) = \mu_a(x) - \nu_a(x).
\end{equation*}
The dynamic of the model is then governed by the time evolution
\begin{equation*}
\p_t \mathcal{O} \approx \lbrace \Hc, \mathcal{O} \rbrace  \approx \bigl\lbrace \widetilde{\Hc}_0, \mathcal{O} \bigr\rbrace + \int_\D \dd x \; \widetilde{\mu}_a(x) \, \bigl\lbrace \Cc^a(x), \mathcal{O} \bigr\rbrace  ,
\end{equation*}
which takes the same form as \eqref{Eq:AppTime} but with $\Hc_0$ replaced by $\widetilde{\Hc}_0$ and $\mu_a$ replaced by $\widetilde{\mu}_a$. Thus, changing the initial Hamiltonian of the model leads to an equivalent system through an appropriate redefinition of the Lagrange multipliers.

\subsection{Gauge symmetries}

In this subsection, we will prove that the canonical transformations generated by the constraints $\Cc^a(x)$ are gauge symmetries of the constrained hamiltonian model defined above. Let us then consider the infinitesimal canonical transformation which sends any observable $\Oo$ in $\Ac$ to
\begin{equation}\label{Eq:AppGauge}
\Oot = \Oo + \delta_\epsilon \Oo + o(\epsilon), \;\;\;\;\; \text{ with } \;\;\;\;\; \delta_\epsilon \Oo = \left\lbrace \int_{\D} \dd x \; \epsilon_a(x,t) \, \Cc^a(x), \Oo \right\rbrace \approx \int_{\D} \dd x \; \epsilon_a(x,t) \bigl\lbrace \Cc^a(x), \Oo \bigr\rbrace,
\end{equation}
where the infinitesimal parameters $\epsilon_a(x,t)$ ($a\in\lbrace 1,\cdots,n\rbrace$) of the transformation are taken as arbitrary functions of the space-time coordinates (which can depend on the dynamical fields of the model). Recall that the treatment of first-class constraints in the previous subsection required the introduction of additional auxiliary degrees of freedom in the model, the Lagrange multipliers $\mu_a(x)$. One can then suppose that these degrees of freedom are also modified by the the transformation and consider new Lagrange multipliers
\begin{equation}\label{Eq:AppTransfoLag}
\widetilde\mu_a(x) = \mu_a(x) + \delta_\epsilon \mu_a(x) + o(\epsilon).
\end{equation}
Before going further, let us note that Equation \eqref{Eq:AppFirstClass} implies that
\begin{equation*}
\delta_\epsilon \Cc^a(x) \approx 0,
\end{equation*}
for all $a\in\lbrace 1,\cdots,n \rbrace$. Thus the transformation $\Oo \mapsto \Oot$ preserves the constraints \eqref{Eq:AppConstraints} and therefore maps admissible physical states of the model to other admissible physical states.\\

We want to prove that the transformation defined above is a symmetry of the model, \textit{i.e.} that it preserves its equations of motion. Using Equation \eqref{Eq:AppTime}, the equation of motion of the observable $\Oo$ can be written as
\begin{equation}\label{Eq:AppEoMF}
\p_t \Oo \approx \mathcal{F}[\phi^i,\mu_a],
\end{equation}
where $\mathcal{F}$ is the following functional of the dynamical fields $\phi^i$ and the Lagrange multipliers $\mu_a$:
\begin{equation}\label{Eq:AppFunctionalEOM}
\mathcal{F}[\phi^i,\mu_a] = \lbrace \Hc_0, \Oo \rbrace + \int_{\D} \dd x \; \mu_a(x) \bigl\lbrace \Cc^a(x), \Oo \bigr\rbrace.
\end{equation}
The time evolution of the transformed observable $\Oot$ is given by
\begin{equation*}
\p_t \Oot = \p_t \Oo + \int_\D \dd x \; \epsilon_b(x,t) \p_t \bigl\lbrace \Cc^b(x), \Oo \bigr\rbrace + \int_{\D} \dd x \; \p_t\epsilon_b(x,t)  \bigl\lbrace \Cc^b(x), \Oo \bigr\rbrace + o(\epsilon).
\end{equation*}
The time derivative in the second term can be computed using Equation \eqref{Eq:AppTime}, yielding
\begin{eqnarray*}
\p_t \Oot &\approx & \p_t \Oo + \int_{\D} \dd x \; \epsilon_b(x,t) \Bigl\lbrace \Hc_0, \bigl\lbrace \Cc^b(x), \Oo \bigr\rbrace \Bigr\rbrace  + \int_\D \dd x \int_\D \dd y \; \epsilon_b(x,t) \mu_c(y) \Bigl\lbrace \Cc^c(y), \bigl\lbrace \Cc^b(x), \Oo \bigr\rbrace \Bigr\rbrace \\
 & & \hspace{50pt} + \int_{\D} \dd x \; \p_t\epsilon_b(x,t)  \bigl\lbrace \Cc^b(x), \Oo \bigr\rbrace + o(\epsilon).
\end{eqnarray*}
Using Equation \eqref{Eq:AppEoMF}, the Jacobi identity and the skew-symmetry of the Poisson bracket, we get
\begin{eqnarray*}
\p_t \Oot &\approx & \mathcal{F}[\phi^i,\mu_a] + \int_{\D} \dd x \; \epsilon_b(x,t) \Bigl(  \bigl\lbrace \Cc^b(x),  \bigl\lbrace \Hc_0,  \Oo \bigr\rbrace \bigr\rbrace + \bigl\lbrace \bigl\lbrace  \Hc_0, \Cc^b(x) \bigr\rbrace , \Oo \bigr\rbrace  \Bigr) \\
 & & \hspace{20pt} + \int_\D \dd x \int_\D \dd y \; \epsilon_b(x,t) \, \mu_c(y) \Bigl(  \bigl\lbrace  \Cc^b(x), \bigl\lbrace \Cc^c(y), \Oo \bigr\rbrace \bigr\rbrace - \bigl\lbrace \bigl\lbrace \Cc^b(x), \Cc^c(y)\bigr\rbrace, \Oo  \bigr\rbrace \Bigr) \\
 & & \hspace{20pt} + \int_{\D} \dd x \; \p_t\epsilon_a(x,t)  \bigl\lbrace \Cc^a(x), \Oo \bigr\rbrace + o(\epsilon).
\end{eqnarray*}
Rearranging the terms and using Equations \eqref{Eq:AppAlgebraConstraints} and \eqref{Eq:AppPBHConstraints}, we get
\begin{eqnarray*}
\p_t \Oot &\approx & \mathcal{F}[\phi^i,\mu_a] + \int_{\D} \dd x \; \epsilon_b(x,t) \left\lbrace \Cc^b(x), \bigl\lbrace \Hc_0,  \Oo \bigr\rbrace + \int_\D \dd y \; \mu_c(y) \bigl\lbrace \Cc^c(y), \Oo \bigr\rbrace \right\rbrace \\
 & & \hspace{20pt} + \int_{\D} \dd x \; \epsilon_b(x,t) \bigl\lbrace \Tcc b a (x)\, \Cc^a(x) + \Uc ba\, (x) \; \p_x\Cc^a(x), \Oo \bigr\rbrace \\
 & & \hspace{20pt} - \int_\D \dd x \; \epsilon_b(x,t) \, \mu_c(x)  \bigl\lbrace \Sc bca (x) \, \Cc^a(x), \Oo \bigr\rbrace \bigr\rbrace \\
 & & \hspace{30pt} + \int_{\D} \dd x \; \p_t\epsilon_a(x,t)  \bigl\lbrace \Cc^a(x), \Oo \bigr\rbrace + o(\epsilon).
\end{eqnarray*}
We recognize in the right term of the Poisson bracket of the first line the functional \eqref{Eq:AppFunctionalEOM}. Finally, performing an integration by parts and noting that, for any observables $f$ and $g$, one has $\lbrace f \, \Cc^a(x), g \rbrace \approx f \lbrace \Cc^a(x), g \rbrace$, we obtain
\begin{eqnarray}\label{Eq:AppEomTransfo}
\p_t \Oot &\approx & \mathcal{F}[\phi^i,\mu_a] + \int_{\D} \dd x \; \epsilon_b(x,t) \bigl\lbrace \Cc^b(x), \mathcal{F}[\phi^i,\mu_a] \bigr\rbrace  + o(\epsilon) \\
 & &  + \int_{\D} \dd x \; \Bigl( \p_t\epsilon_a(x,t) - \p_x \bigl( \epsilon_b(x,t) \, \Uc ba (x) \bigr) + \epsilon_b(x,t) \, \Tcc b a (x) - \epsilon_b(x,t) \, \mu_c(x) \, \Sc bca (x) \, \Bigr) \bigl\lbrace \Cc^a(x) , \Oo \bigr\rbrace. \notag
\end{eqnarray}

The transformation $\Oo \mapsto \Oot$ is a symmetry of the model if and only if the equation of motion of the transformed observable $\Oot$ is the same as the one of $\Oo$, \textit{i.e.} if and only if
\begin{equation}\label{Eq:AppSym}
\p_t \Oot \approx \mathcal{F}\bigl[\, \widetilde{\phi}\,^i, \widetilde{\mu}_a \bigr],
\end{equation}
with $\mathcal{F}$ the same functional \eqref{Eq:AppFunctionalEOM} as in \eqref{Eq:AppEoMF}. The left-hand side of this equality is given by Equation \eqref{Eq:AppEomTransfo}. Let us now compute its right-hand side
\begin{equation*}
\mathcal{F}\bigl[ \, \widetilde{\phi}\,^i, \widetilde{\mu}_a \bigr] = \mathcal{F}[\phi^i, \mu_a] + \delta_\epsilon \mathcal{F}[\phi^i,\mu_a] + o(\epsilon).
\end{equation*}
The variation $\delta_\epsilon \mathcal{F}$ of the functional $\mathcal{F}$ comes from two sources: the variation of the fundamental fields $\phi^i$, which is given by the canonical transformation \eqref{Eq:AppGauge}, and the variation of the Lagrange multipliers, which is given by \eqref{Eq:AppTransfoLag}. As the Lagrange multipliers only appear in the second term of \eqref{Eq:AppFunctionalEOM}, we have
\begin{equation*}
\mathcal{F}\bigl[ \, \widetilde{\phi}\,^i, \widetilde{\mu}_a \bigr] \approx \mathcal{F}[\phi^i,\mu_a] + \int_{\D} \dd x \; \epsilon_b(x,t) \bigl\lbrace \Cc^b(x), \mathcal{F}[\phi^i,\mu_a] \bigr\rbrace + \int_\D \dd x \; \delta_\epsilon \mu_a(x) \bigl\lbrace \Cc^a(x), \Oo \bigr\rbrace + o(\epsilon).
\end{equation*}
Comparing this expression to Equation \eqref{Eq:AppEomTransfo}, it is clear that the symmetry condition \eqref{Eq:AppSym} is satisfied if we choose the variation of the Lagrange multipliers to be
\begin{equation}\label{Eq:AppTransfMult}
\delta_\epsilon \mu_a(x) = \p_t\epsilon_a(x,t) - \p_x \bigl( \epsilon_b(x,t) \, \Uc ba (x) \bigr) + \epsilon_b(x,t)\, \Tcc b a (x) - \epsilon_b(x,t) \, \mu_c(x)\, \Sc bca (x).
\end{equation}
This proves that the canonical transformation \eqref{Eq:AppGauge} generated by the constraints, together with the transformation \eqref{Eq:AppTransfMult} of the Lagrange multipliers, is a symmetry of the model. As this transformation depends on arbitrary functions $\epsilon_a(x,t)$ of space-time coordinates, it is a gauge symmetry.

\section{Change of spectral parameters and Takiff realisations}
\label{App:ChangeSpec}

The goal of this appendix is to prove Proposition \ref{Prop:NewTakiff}, which explains how to construct Takiff currents $\Jtt\alpha p(x)$ with levels $\lst\alpha p$ from the Takiff current $\J\alpha p(x)$ with levels $\ls\alpha p$, supposing that the levels $\ls\alpha p$ and $\ls\alpha p$ are related by Lemma \ref{Lem:NewLevels}. As the statements of Lemma \ref{Lem:NewLevels} and Proposition \ref{Prop:NewTakiff} do not mix the levels and currents at different sites, one can focus on a particular site and drop the superscripts $\alpha\in\Si$. We thus want to prove the following.

\begin{proposition}\label{Prop:AppTakiff}
Let $\J\null p(x)$ be Takiff currents with multiplicity $m$ and levels $\ls\null p$. For $a,b,c,d$ and $z$ complex numbers such that $ad-bc$ and $cz+d$ are non-zero, we define
\begin{equation*}
\lst \null 0 = \ls \null 0 \;\;\;\;\;\; \text{ and } \;\;\;\;\;\; \lst \null p = \sum_{q=p}^{m-1} {q-1 \choose p-1 } \frac{(-c)^{q-p}(ad-bc)^p}{(cz+d)^{p+q}} \ls\null q, \;\; \text{ for } \; 0<p<m.
\end{equation*}
Then the currents
\begin{equation*}
\Jtt \null 0(x) = \J \null 0(x) \;\;\;\;\;\; \text{ and } \;\;\;\;\;\; \Jtt \null p(x) = \sum_{q=p}^{m-1} {q-1 \choose p-1 } \frac{(-c)^{q-p}(ad-bc)^p}{(cz+d)^{p+q}} \J\null q(x), \;\; \text{ for } \; 0<p<m.
\end{equation*}
are Takiff currents with levels $\lst\null p$.
\end{proposition}

Let us first consider the case $c=0$ (in terms of the change of spectral parameter $z\mapsto \zt=f(z)$, this case corresponds to a M\"obius transformation $f$ which is simply the combination of a translation and a dilation). We then simply have
\begin{equation*}
\lst\null p = \left(\frac{a}{d}\right)^p \ls\null p \;\;\;\;\;\;\;\; \text{ and } \;\;\;\;\;\;\;\; \Jtt\null p = \left(\frac{a}{d}\right)^p \J\null p(x).
\end{equation*}
The proposition is then direct to prove (this was in fact already stated in Equation (2.46) of~\nm).

For the rest of this appendix, we will thus consider the case $c\neq 0$. Let us start by introducing some simplifying notations. We define the parameters
\begin{equation*}
\Delta=-\frac{ad-bc}{c(cz+d)} \;\;\;\;\;\;\; \text{ and } \;\;\;\;\;\;\; \Omega = -\frac{c}{cz+d}.
\end{equation*}
One then has, for $p$ in $\lbrace 1,\cdots,m-1 \rbrace$,
\begin{equation*}
\lst\null p = \sum_{q=1}^{m-1} \alpha_{p,q} \,\ls\null q \;\;\;\;\;\;\;\; \text{ and } \;\;\;\;\;\;\;\; \Jtt\null p(x) = \sum_{q=1}^{m-1} \alpha_{p,q} \,\J\null q(x),
\end{equation*}
where the coefficients $\alpha_{p,q}$ ($p,q\in\lbrace 1,\cdots,m-1\rbrace$) are defined as
\begin{equation}\label{Eq:Alpha}
\alpha_{p,q} = \left\lbrace \begin{array}{ll}
\;\;0 & \text{ if }\, 0 < q < p, \\[2pt]
\displaystyle {q-1 \choose p-1 } \Delta^p \, \Omega^q & \text{ if }\, p \leq q < m.
\end{array}\right.
\end{equation}

\subsection{Automorphisms of Takiff Lie algebras}

To prove the proposition, let us first state a more general result, which is motivated from the original construction of Takiff algebra (although in the end we will give an independent direct proof of it). Recall (see for example~\bg) that the Takiff Lie algebra $\mathrm{T}^m\mathfrak{f}$ associated with a Lie algebra $\mathfrak{f}$ and with multiplicity $m$ is defined as the tensor product $\mathfrak{f} \otimes \bigl( \C[\xi]/\xi^m\C[\xi] \bigr)$, where $\C[\xi]$ denotes the algebra of complex polynomials in an abstract variable $\xi$ and $\xi^m\C[\xi]$ its ideal generated by $\xi^m$. In the abstract construction of affine Gaudin model in~\bg, the Takiff currents originate from the Takiff Lie algebra $\mathrm{T}^m\mathfrak{f}$ where $\mathfrak{f}$ is the untwisted affine Kac-Moody algebra constructed from the finite algebra $\g$. The relation that we are aiming to prove between two sets of Takiff currents $\J\null p(x)$ and $\Jtt\null p(x)$ can be interpreted as coming from an automorphism of the Takiff Lie algebra $\mathrm{T}^m\mathfrak{f}$.

More precisely, it corresponds to an automorphism of the component $\C[\xi]/\xi^m\C[\xi]$ of the Takiff Lie algebra $\mathrm{T}^m\mathfrak{f}$ (which then acts trivially on the factor $\mathfrak{f}$). Let us discuss how to construct natural automorphisms of $\C[\xi]/\epsilon^m\C[\xi]$. This algebra has generator $\xi$ and is subject to the nilpotence relation $\xi^m=0$. Let us consider the change of variable $\xi \mapsto \widetilde{\xi}=\sum_{p=0}^{m-1} \vpi_p\,\xi^p$. One can extend it to a morphism on the whole algebra $\C[\xi]/\xi^m\C[\xi]$ if the new variable $\widetilde{\xi}$ is also nilpotent of order $m$. One easily sees that this is the case if the coefficient $\vpi_0$ is zero (if not, the power $\widetilde{\xi}\,^m$ would contain a non-zero constant term $\vpi_0^m$) and the coefficient $\vpi_1$ is non-zero (if not, there would exist a smaller power $\widetilde{\xi}\,^k$ which vanishes). One also verifies that in this case, this morphism is bijective and thus an automorphism.

The automorphism we constructed sends $\xi$ to $\sum_{q=1}^{m-1} \vpi_q\,\xi^q$. Thus, it sends $\xi^p$ to the element
\begin{equation*}
\left(\sum_{q=1}^{m-1} \vpi_q\,\xi^q \right)^p = \sum_{q=1}^{m-1} \beta_{p,q}\,\xi^q, \;\;\;\;\;\;\; \text{ with } \;\;\;\;\;\;\; \beta_{p,q} = \sum_{\substack{ r_1,\cdots,r_p=1 \\ r_1+\cdots+r_p=q }}^{m-1} \vpi_{r_1} \cdots \vpi_{r_p}.
\end{equation*}
A little bit of algebra shows that the coefficients $\beta_{p,q}$ are uniquely characterised by the relation
\begin{equation}\label{Eq:RelBeta}
\sum_{s=0}^r \beta_{p,s}\beta_{q,r-s} = \beta_{p+q,r},
\end{equation}
for all $p,q,r\in \lbrace 1,\cdots,m-1\rbrace$. Let us come back to Takiff currents. In the abstract construction of affine Gaudin models in~\bg, the Takiff current $\J\null p(x)$ and levels $\ls\null p$ of mode $p\in\lbrace 0,\cdots,m-1\rbrace$ correspond to elements of the Takiff Lie algebra with $\xi^p$ in the tensor  factor $\C[\xi]/\xi^m\C[\xi]$. We thus translate the general discussion above to the following Lemma on Takiff currents.

\begin{lemma}\label{Lem:AutoTakiff}
Let $\J\null p(x)$ be Takiff currents with multiplicity $m$ and levels $\ls\null p$. We define
\begin{equation*}
\lst\null 0 = \ls\null 0 \;\;\;\;\;\;\;\; \text{ and } \;\;\;\;\;\;\;\; \lst\null p = \sum_{q=1}^{m-1} \beta_{p,q} \,\ls\null q.
\end{equation*}
If the coefficients $\beta_{p,q}$ satisfy relation \eqref{Eq:RelBeta}, then the currents
\begin{equation*}
\Jtt\null 0(x) = \J\null 0(x)\;\;\;\;\;\;\;\; \text{ and } \;\;\;\;\;\;\;\; \Jtt\null p(x) = \sum_{q=1}^{m-1} \beta_{p,q} \,\J\null q(x).
\end{equation*}
are Takiff currents of levels $\lst\null p$.
\end{lemma}

\begin{proof}
The proof is rather direct. For simplicity, we extend the notations $\J\null p(x)$ and $\ls\null p$ to $p\geq m$ by letting them be 0. This allows to write everything with infinite sum instead of finite ones without having to worry about there domains. In particular, for $p$ and $q$ greater that $1$, we have
\begin{align*}
\left\lbrace \Jtt\null p\,\ti{1}(x), \Jtt\null q\,\ti{2}(y) \right\rbrace &= \sum_{s,t=1}^{\infty} \beta_{p,s} \beta_{q,t} \left\lbrace \J\null s\,\ti{1}(x), \J\null t\,\ti{2}(y) \right\rbrace \\
 &= \sum_{s,t=1}^{\infty} \beta_{p,s} \beta_{q,t} \left( \left[ C\ti{12}, \J\null {s+t}\,\ti{1}(x) \right] \delta_{xy} - C\ti{12}\, \ls\null{s+t}\,\delta'_{xy} \right).
\end{align*}
The terms proportional to $\J\null {r}(x)$ and $\ls\null {r}$ in this equation are obtained when $s+t=r$. Thus, we have
\begin{equation*}
\left\lbrace \Jtt\null p\,\ti{1}(x), \Jtt\null q\,\ti{2}(y) \right\rbrace = \sum_{r=1}^{\infty} \left( \sum_{s=0}^r \beta_{p,s}\beta_{q,r-s} \right) \left( \left[ C\ti{12}, \J\null {r}\,\ti{1}(x) \right] \delta_{xy} - C\ti{12}\, \ls\null{r}\,\delta'_{xy} \right). 
\end{equation*}
By assumption, the coefficients $\beta_{p,q}$ satisfy the condition \eqref{Eq:RelBeta}, hence
\begin{equation*}
\left\lbrace \Jtt\null p\,\ti{1}(x), \Jtt\null q\,\ti{2}(y) \right\rbrace = \left[ C\ti{12}, \Jtt\null {p+q}\,\ti{1}(x) \right] \delta_{xy} - C\ti{12}\, \lst\null{p+q}\,\delta'_{xy}.
\end{equation*}
This almost ends the proof of the lemma. There is left to study the Poisson bracket when $p$ and/or $q$ is equal to $1$. These follow easily (without having to use the condition \eqref{Eq:RelBeta} on the coefficients $\beta_{p,q}$).
\end{proof}

\subsection{Proof of Proposition B.1}

Let us now prove Proposition \ref{Prop:AppTakiff} (and thus Proposition \ref{Prop:NewTakiff}). By Lemma \ref{Lem:AutoTakiff}, it is enough to prove that the coefficients $\alpha_{p,q}$ introduced in Equation \eqref{Eq:Alpha} satisfy the condition \eqref{Eq:RelBeta}, hence the following lemma.

\begin{lemma}
The coefficients $\alpha_{p,q}$ defined in Equation \eqref{Eq:Alpha} satisfy
\begin{equation*}
\sum_{s=0}^r \alpha_{p,s}\alpha_{q,r-s} = \alpha_{p+q,r}, \;\;\;\; \forall\, p,q,r\in \lbrace 1,\cdots,m-1\rbrace.
\end{equation*}
\end{lemma}

\begin{proof}
Recall that $\alpha_{p,q}=0$ for $q<p$. Thus the sum $\sum_{s=0}^r \alpha_{p,s}\alpha_{q,r-s}$ is empty in the case where $r<p+q$, and therefore is equal to $\alpha_{p+q,r}=0$. Let us now consider the case $r\geq p+q$:
\begin{align*}
\sum_{s=0}^r \alpha_{p,s}\alpha_{q,r-s} &= \sum_{s=p}^{r-q}  \alpha_{p,s}\alpha_{q,r-s} = \sum_{s=p}^{r-q}  {s-1 \choose p-1 } \Delta^p \, \Omega^s {r-s-1 \choose q-1 } \Delta^q \, \Omega^{r-s} \\
&= \Delta^{p+q} \Omega^{r} \sum_{s=p}^{r-q}  {s-1 \choose p-1 } {r-s-1 \choose q-1 }.
\end{align*}
The result then follows from the following standard identity on binomial coefficients (this is a variant of the so-called Vandermonde identity, which can be proved by the generating binomial series):
\begin{equation*}
\sum_{s=p}^{r-q}  {s-1 \choose p-1 } {r-s-1 \choose q-1 } = {r-1 \choose p+q-1}. \qedhere
\end{equation*}
\end{proof}

\section[Realisations in $\Ac_{G_0}$]{Takiff realisations in $\bm{\Ac_{G_0}}$}
\label{App:Realisations}

\subsection[Fields on $T^*G_0$]{Fields on $\bm{T^*G_0}$}
\label{App:TStarG}

\paragraph{Canonical fields.} In this appendix, we recall the construction of two Takiff realisations: the PCM+WZ one and the inhomogeneous Yang-Baxter one. They both take value in the Poisson algebra $\Ac_{G_0}$ of canonical fields on the cotangent bundle $T^*G_0$ of a real semi-simple Lie group $G_0$. Let us first recall briefly the description of these fields (we refer for example to~\nm, Subsection 3.1.1, for a more detailed review).

The algebra $\Ac_{G_0}$ contains coordinate fields on the Lie group $G_0$, which are naturally gathered in a $G_0$-valued field $g(x)$. It also contains momenta fields, which are canonically conjugate to the coordinate fields and valued in the cotangent space $T_{g(x)}^*G_0$ of $G_0$ at the point $g(x)$. Using the left translation by $g(x)^{-1}$, one can bring this cotangent space to $T^*_{\Id}G_0 = \g_0^*$, \textit{i.e.} to the dual of the Lie algebra $\g_0$ of $G_0$. As $G_0$ is semi-simple, the Killing form on $\g_0$ is non-degenerate, which allows to further identify $\g_0^*$ with $\g_0$. This allows to gather the momenta fields in a $\g_0$-valued field $X(x)$.

The canonical brackets between the coordinate and momenta fields translate to the following Poisson bracket between the fields $g(x)$ and $X(x)$, written in the standard tensorial notation (see Subsection \ref{SubSubSec:Conventions}):
\begin{subequations}\label{Eq:PBTstarG}
\begin{align}
\left\lbrace g\ti{1}(x), g\ti{2}(y) \right\rbrace & = 0, \\
\left\lbrace X\ti{1}(x), g\ti{2}(y) \right\rbrace & = g\ti{2}(x) C\ti{12} \delta_{xy},\label{Eq:PBXg} \\
\left\lbrace X\ti{1}(x), X\ti{2}(y) \right\rbrace & = \left[ C\ti{12}, X\ti{1}(x) \right] \delta_{xy}.\label{Eq:PBXX}
\end{align}
\end{subequations}

\paragraph{Currents $\bm{j}$ and $\bm{W}$.} We will need later the following $\g_0$-valued current:
\begin{equation*}
j(x) = g(x)^{-1} \p_x g(x).
\end{equation*}
Its Poisson bracket with $g(x)$ and $X(x)$ can be computed easily from Equation \eqref{Eq:PBTstarG} and read
\begin{subequations}\label{Eq:PBj}
\begin{align}
\left\lbrace g\ti{1}(x), j\ti{2}(y) \right\rbrace &= 0,\label{Eq:PBjg}\\
\left\lbrace j\ti{1}(x), j\ti{2}(y) \right\rbrace &= 0, \\
\left\lbrace X\ti{1}(x), j\ti{2}(y) \right\rbrace &= \bigl[ C\ti{12}, j\ti{1}(x) \bigr] \delta_{xy} - C\ti{12} \delta'_{xy}.
\end{align}
\end{subequations}

We will also need another $\g_0$-valued current, denoted as $W(x)$ in~\nm, and constructed from the canonical 3-form on $G_0$ and the derivatives of the coordinate field $g(x)$. We will not enter into details about the construction of this current and refer to the Subsection 3.1.2 of~\nm. In particular, this field is related to the Wess-Zumino term of $g$ by the following identity:
\begin{equation*}
\W g = \iint_{\mathbb{W}} \dd t \, \dd x \; \kappa\bigl( W, g^{-1} \p_t g \bigr).
\end{equation*}
One can compute the Poisson bracket of the current $W(x)$ with other fields in $\Ac_{G_0}$. In particular, one has
\begin{equation}\label{Eq:PbW1}
\bigl\lbrace g\ti{1}(x), W\ti{2}(y) \bigr\rbrace = 0, \;\;\;\;\;\;\; \bigl\lbrace j\ti{1}(x), W\ti{2}(y) \bigr\rbrace = 0
\end{equation}
and
\begin{equation}\label{Eq:PbW2}
\bigl\lbrace X\ti{1}(x), W\ti{2}(y) \bigr\rbrace + \bigl\lbrace W\ti{1}(x), X\ti{2}(y) \bigr\rbrace = \bigl[ C\ti{12}, W\ti{1}(x)-j\ti{1}(x) \bigr] \delta_{xy}.
\end{equation}

\subsection{The PCM+WZ realisation}
\label{App:PCM+WZ}

Let us now describe the PCM+WZ realisation. It is a Takiff realisation with one site of multiplicity 2, valued in the algebra $\Ac_{G_0}$. The corresponding Takiff currents take the form
\begin{equation*}
\J\null 0(x) = X(x) + \frac{\ell_0}{2} W(x) + \frac{\ell_0}{2} j(x) \;\;\;\;\;\; \text{ and } \;\;\;\;\;\;\; \J\null 1(x) = \ell_1 \, j(x), 
\end{equation*}
where $\ell_0$ and $\ell_1$ denote the levels of the Takiff realisation. One checks from the Poisson bracket of the previous subsection that
\begin{subequations}
\begin{eqnarray}
\bigl\lbrace \J\null0\null\ti1(x),\J\null0\null\ti2(y) \bigr\rbrace &=& \bigl[ C\ti{12}, \J\null0\null\ti1(x) \bigr] \delta_{xy} - \ell_0 \, C\ti{12} \delta'_{xy}, \\
\bigl\lbrace \J\null0\null\ti1(x),\J\null1\null\ti2(y) \bigr\rbrace &=& \bigl[ C\ti{12}, \J\null1\null\ti1(x) \bigr] \delta_{xy} - \ell_1 \, C\ti{12} \delta'_{xy}, \\
\bigl\lbrace \J\null1\null\ti1(x),\J\null1\null\ti2(y) \bigr\rbrace &=& 0.
\end{eqnarray}
\end{subequations}
Thus, the currents $\J\null p(x)$ are indeed Takiff currents of multiplicity 2 with levels $\ls\null p$.

Let us note that the unique site of the PCM+WZ realisation is real. As expected, the corresponding levels $\ls\null p$ are real numbers and the Takiff currents $\J\null p(x)$ are valued in the real form $\g_0$. Thus, the construction above satisfy the appropriate reality conditions of a Takiff realisation.\\

Because of the presence of the field $W(x)$ in the Takiff current $\J\null 0(x)$ of the PCM+WZ realisation, the action of an affine Gaudin model which possesses such a realisation will involve a WZ term $\W g$, proportional to the level $\ell_0$ (see for example Section 3 in~\nm). When this level is zero, the WZ term vanishes from the action: in this case, we speak of a PCM realisation, in opposition to the general PCM+WZ realisation. The Takiff currents of the PCM realisation are simply
\begin{equation*}
\J\null 0(x) = X(x) \;\;\;\;\;\; \text{ and } \;\;\;\;\;\;\; \J\null 1(x) = \ell_1 \, j(x).
\end{equation*}

\subsection{The inhomogeneous Yang-Baxter realisation}
\label{App:iYB}

Let us now describe the inhomogeneous Yang-Baxter realisation. It is also a realisation in the algebra $\Ac_{G_0}$ of canonical fields on $T^*G_0$. Contrarily to the PCM+WZ realisation, it has two sites $\pms$ of multiplicity 1, whose levels we shall parametrise as
\begin{equation*}
\ls\pms 0 = \pm \frac{1}{2\cc\gamma},
\end{equation*}
with $\gamma$ a real number and $\cc$ equal either to $1$ or $i$. In the case $\cc=1$, that we shall refer as the split case, the two sites $\pms$ are real, whereas in the non-split case $\cc=i$, the two sites are complex conjugate one to another.

The inhomogeneous Yang-Baxter deformation is constructed from a choice of a $R$-operator on the Lie algebra $\g_0$. This is a linear map $R:\g_0\rightarrow\g_0$, skew-symmetric with respect to the Killing form of $\g_0$, satisfying the modified Classical Yang-Baxter Equation (mCYBE)
\begin{equation}\label{Eq:mCYBE}
[RX,RY] - R\bigl( [RX,Y] + [X,RY] \bigr) = -\cc^2 [X,Y], \;\;\;\;\;\; \forall \, X,Y\in\g_0.
\end{equation}
Let us also define the map
\begin{equation*}
R_g = \Ad_g^{-1} \circ R \circ \Ad_g,
\end{equation*}
with $g$ the $G_0$-valued field introduced in Subsection \ref{App:TStarG}. The Kac-Moody currents defining the inhomogeneous Yang-Baxter realisation are then defined as~\cite{Delduc:2013fga}
\begin{equation*}
\J \pms 0(x) = \frac{1}{2\cc} \left( \cc X(x) \mp R_gX(x) \pm \frac{1}{\gamma} j(x) \right).
\end{equation*}
One checks that these are indeed commuting Kac-Moody currents, \textit{i.e.} that they satisfy the expected Poisson brackets
\begin{equation*}
\bigl\lbrace \J{(\pm)}0\null\ti{1}(x), \J{(\pm)}0\null\ti{2}(y) \bigr\rbrace = \bigl[ C\ti{12}, \J{(\pm)}0\null\ti{1}(x) \bigr] \delta_{xy} - \ls\pms 0\, C\ti{12} \delta'_{xy} \;\;\;\;\;\; \text{and} \;\;\;\;\;\; \bigl\lbrace \J{(\pm)}0\null\ti{1}(x), \J{(\mp)}0\null\ti{2}(y) \bigr\rbrace = 0.
\end{equation*}

Note that in the split case $\cc=1$, the currents $\J\pms 0(x)$ are valued in the real form $\g_0$, consistently with the facts that the sites $\pms$ are both real. In the non-split case $\cc=i$, these currents are valued in the complexified Lie algebra $\g$ and are conjugate one to another by the involutive anti-automorphism $\tau$ defining the real form $\g_0$, as expected.

\section{Maillet brackets from Dirac brackets}
\label{App:Dirac}

\subsection{Undeformed case}
\label{App:DiracUndeformed}

This appendix is a technical addition to Subsection \ref{SubSubSec:IntegrabilityGauging}. We shall then use the definitions and notations of this subsection without recalling them. The goal of the appendix is to detail the computation of the Poisson bracket of the gauge-fixed Lax matrix $\Lct(\zt,x)$ using the Dirac bracket.\\

As $\Lct(\zt,x)$ is the Lax matrix of the model $\mathbb{M}^{\vpt,\pit}_{\eb}$, its bracket in terms of the Poisson structure $\lbrace\cdot,\cdot\rbrace_{\Act}$ of the algebra $\Act$ takes the form of a Maillet bracket with twist function $\vpt(\zt)$, \textit{i.e.} with $\Rc$-matrix
\begin{equation*}
\Rct\ti{12}(\zt,\wt) = \frac{C\ti{12}}{\wt-\zt} \vpt(\wt)^{-1}.
\end{equation*}
Recall that the Dirac bracket is defined in terms of the bracket $\lbrace\cdot,\cdot\rbrace_{\Act}$ by Equation \eqref{Eq:Dirac}. Thus, the Dirac bracket of the Lax matrix $\Lct(\zt,x)$ with itself takes the following form:
{\small
\begin{align}\label{Eq:DiracL}
\bigl\lbrace \Lct\ti{1}(\zt,x), \Lct\ti{2}(\wt,y) \bigr\rbrace^*
&= \bigl[ \Rct\ti{12}(\zt,\wt), \Lct\ti{1}(\zt,x) \bigr] \delta_{xy} - \int_\D \dd \xi \; \kappa\ti{3}\Bigl( \bigl\lbrace \Cct\ti{3}(\xi), \Lct\ti{1}(\zt,x) \bigr\rbrace_{\Act}, g\ti{3}(\xi)^{-1}\bigl\lbrace g\ti{3}(\xi), \Lct\ti{2}(\wt,y) \bigr\rbrace_{\Act} \Bigr) \notag \\
& \hspace{10pt} - \bigl[ \Rct\ti{21}(\wt,\zt), \Lc\ti{2}(\wt,x) \bigr] \delta_{xy} + \int_\D \dd \xi \; \kappa\ti{3}\Bigl( \bigl\lbrace \Cct\ti{3}(\xi), \Lct\ti{2}(\wt,y) \bigr\rbrace_{\Act}, g\ti{3}(\xi)^{-1}\bigl\lbrace g\ti{3}(\xi), \Lct\ti{1}(\zt,x) \bigr\rbrace_{\Act} \Bigr) \notag \\
& \hspace{20pt} - \left( \Rct\ti{12}(\zt,\wt) + \Rct\ti{21}(\wt,\zt) \right) \delta'_{xy}.
\end{align}}

The Poisson bracket of the field $g(\xi)$ with the Lax matrix $\Lct(\zt,x)$ is directly related to its bracket with the Gaudin Lax matrix $\Gt(\zt,x)$, as $\Lct(\zt,x)=\Gt(\zt,x)/\vpt(\zt)$. Recall from Equation \eqref{Eq:GammaTilde} that the latter is constructed from the Takiff currents $\Jtt\alpha p(x)$, with $\alpha\in\Sit=\Si \sqcup \lbrace \is \rbrace$. As the field $g(\xi)$ belongs to the component $\Ac_{G_0}$ of $\Act=\Ac\otimes\Ac_{G_0}$, it Poisson commutes with the Takiff currents $\Jtt\alpha p(x)$, for $\alpha\in\Si$. Using the expressions \eqref{Eq:NewCurrentInf} of the currents $\Jtt\is p(x)$, one can then compute explicitly the Poisson bracket of $g(\xi)$ with $\Lct(\zt,x)$. More precisely, using the brackets \eqref{Eq:PBTstarG}, \eqref{Eq:PBjg} and \eqref{Eq:PbW1}, one finds
\begin{equation*}
g\ti{3}(\xi)^{-1}\bigl\lbrace g\ti{3}(\xi), \Lct\ti{1}(\zt,x) \bigr\rbrace_{\Act} = - \frac{\vpt(\zt)^{-1}}{\zt-\zi}  C\ti{13} \, \delta_{x\xi}.
\end{equation*}
Reinserting this equation (and its equivalent with $\Lct\ti{1}(\zt,x)$ replaced by $\Lct\ti{2}(\wt,y)$) in the Dirac bracket above, we get
\begin{align*}
\bigl\lbrace \Lct\ti{1}(\zt,x), \Lct\ti{2}(\wt,y) \bigr\rbrace^*
&= \bigl[ \Rct\ti{12}(\zt,\wt), \Lct\ti{1}(\zt,x) \bigr] \delta_{xy} - \frac{\vpt(\wt)^{-1}}{\wt-\zi}  \bigl\lbrace \Lct\ti{1}(\zt,x), \Cct\ti{2}(y) \bigr\rbrace_{\Act} \\
& \hspace{15pt} - \bigl[ \Rct\ti{21}(\wt,\zt), \Lc\ti{2}(\wt,x) \bigr] \delta_{xy} - \frac{\vpt(\zt)^{-1}}{\zt-\zi} \bigl\lbrace \Cct\ti{1}(x), \Lct\ti{2}(\wt,y) \bigr\rbrace_{\Act} \\
& \hspace{15pt} - \left( \Rct\ti{12}(\zt,\wt) + \Rct\ti{21}(\wt,\zt) \right) \delta'_{xy}. \notag
\end{align*}
The Poisson bracket of the constraint $\Cct(y)$ with the Lax matrix $\Lct(\zt,x)$ is computed by considering Equation \eqref{Eq:PbCL} for the model $\mathbb{M}^{\vpt,\pit}_{\eb}$. One then finds
\begin{equation}\label{Eq:PBLC}
\bigl\lbrace \Lct\ti{1}(\zt,x), \Cct\ti{2}(y) \bigr\rbrace_{\Act} = \bigl[ C\ti{12}, \Lct\ti{1}(\zt,x) \bigr] \delta_{xy} - C\ti{12} \, \delta'_{xy}.
\end{equation}
Inserting this equation in the computation above, one finds that the Dirac bracket of the Lax matrix $\Lct(\zt,x)$ with itself takes the form of a Maillet bracket with $\Rc$-matrix
\begin{equation*}
\widetilde{\Rc}^{\text{GF}}\ti{12}(\zt,\wt) = \widetilde{\Rc}\ti{12}(\zt,\wt) - \frac{C\ti{12}}{\wt-\zi} \vpt(\wt)^{-1} = \left( \frac{C\ti{12}}{\wt-\zt} - \frac{C\ti{12}}{\wt-\zi} \right) \vpt(\wt)^{-1}.
\end{equation*}
We recognize here the $\Rc$-matrix \eqref{Eq:RMatGF1} found in Subsection \ref{SubSubSec:IntegrabilityGauging} by another method.

\subsection{Deformed case}
\label{App:DiracDeformed}

In this appendix, we detail a computation needed for Subsection \ref{SubSubSec:IntegrabilityGaugingDeformed}. We shall then use the definitions and notations of this subsection without recalling them. Our goal here is to compute the Poisson bracket of the gauge-fixed Lax matrix $\Lct_\eta(\zt,x)$ of the Yang-Baxter deformed model $\mathbb{M}^{\vpt_\eta,\pit_\eta}_{\eb}$, using the Dirac bracket. This is the equivalent for the deformed model of the computation detailed in the previous subsection of this appendix, which concerned the undeformed model $\mathbb{M}^{\vpt,\pit}_{\eb}$. We shall follow here the exact same steps as the ones of this previous subsection.\\

As the gauge-fixing conditions are the same in the undeformed and deformed models, the Dirac bracket in the present case is defined by the same equation \eqref{Eq:Dirac}. Similarly to the undeformed case (see Equation \eqref{Eq:DiracL}), and recalling that in terms of the bracket $\lbrace\cdot,\cdot\rbrace_{\Act}$ the Lax matrix $\Lct_\eta(\zt,x)$ satisfies a Maillet bracket with $\Rc$-matrix $\Rct^\eta\ti{12}(\zt,\wt)$, the Dirac bracket of the Lax matrix $\Lct_\eta(\zt,x)$ can then be expressed as
{\small
\begin{align}\label{Eq:DiracLDeformed}
\bigl\lbrace \Lct_\eta\null\ti{1}(\zt,x), \Lct_\eta\null\ti{2}(\wt,y) \bigr\rbrace^*
&= \bigl[ \Rct^\eta\ti{12}(\zt,\wt), \Lct_\eta\null\ti{1}(\zt,x) \bigr] \delta_{xy} - \bigl[ \Rct^\eta\ti{21}(\wt,\zt), \Lc_\eta\null\ti{2}(\wt,x) \bigr] \delta_{xy} - \left( \Rct^\eta\ti{12}(\zt,\wt) + \Rct^\eta\ti{21}(\wt,\zt) \right) \delta'_{xy} \notag \\
& \hspace{20pt} - \int_\D \dd \xi \; \kappa\ti{3}\Bigl( \bigl\lbrace \Cct\ti{3}(\xi), \Lct_\eta\null\ti{1}(\zt,x) \bigr\rbrace_{\Act}, g\ti{3}(\xi)^{-1}\bigl\lbrace g\ti{3}(\xi), \Lct_\eta\null\ti{2}(\wt,y) \bigr\rbrace_{\Act} \Bigr) \notag \\
& \hspace{20pt} + \int_\D \dd \xi \; \kappa\ti{3}\Bigl( \bigl\lbrace \Cct\ti{3}(\xi), \Lct_\eta\null\ti{2}(\wt,y) \bigr\rbrace_{\Act}, g\ti{3}(\xi)^{-1}\bigl\lbrace g\ti{3}(\xi), \Lct_\eta\null\ti{1}(\zt,x) \bigr\rbrace_{\Act} \Bigr) 
\end{align}}

Recall that $\Lct_\eta(\zt,x)=\Gt_\eta(\zt,x)/\vpt_\eta(\zt)$ and that the Gaudin Lax matrix $\Gt_\eta(\zt,x)$ can be expressed as \eqref{Eq:DeformGt}. As the field $g(\xi)$ as a trivial Poisson bracket with $\Gt_\Si(\zt,x)$, we then get
\begin{equation*}
g\ti{3}(\xi)^{-1}\bigl\lbrace g\ti{3}(\xi), \Lct_\eta\null\ti{1}(\zt,x) \bigr\rbrace_{\Act} = \vpt_\eta(\zt)^{-1} \, g\ti{3}(\xi)^{-1}\left\lbrace g\ti{3}(\xi), \frac{\Jtt{(+)}0\null\ti{1}(x)}{\zt-\zt_+} +  \frac{\Jtt{(-)}0\null\ti{1}(x)}{\zt-\zt_-} \right\rbrace_{\Act}.
\end{equation*}
From the expression \eqref{Eq:CurrentYBDef} of the currents $\Jtt\pms 0(x)$ and the Poisson brackets \eqref{Eq:PBTstarG} and \eqref{Eq:PBjg}, we get
\begin{equation*}
g\ti{3}(\xi)^{-1}\bigl\lbrace g\ti{3}(\xi),\Jtt\pms 0\null\ti{1}(x) \bigr\rbrace_{\Act} = -\frac{1}{2\cc} \bigl( \cc \,\Id \mp R_g \bigr)\ti{1} C\ti{13}\,\delta_{x\xi}.
\end{equation*}
To simplify this expression, let us use the gauge-fixing condition $g\gf\Id$, which implies $R_g\gf R$. Reinserting this bracket in the one above, we get
\begin{equation*}
g\ti{3}(\xi)^{-1}\bigl\lbrace g\ti{3}(\xi), \Lct_\eta\null\ti{1}(\zt,x) \bigr\rbrace_{\Act} \gf \frac{\vpt_\eta(\zt)^{-1}}{2\cc} \left( \frac{R-\cc \,\Id }{\zt-\zt_+} -  \frac{R+\cc \,\Id}{\zt-\zt_-} \right)\ti{1}  C\ti{13}\, \delta_{x\xi}.
\end{equation*}
The Poisson bracket \eqref{Eq:DiracLDeformed} then becomes
{\small
\begin{align}
\bigl\lbrace \Lct_\eta\null\ti{1}(\zt,x), \Lct_\eta\null\ti{2}(\wt,y) \bigr\rbrace^*
&\gf \bigl[ \Rct^\eta\ti{12}(\zt,\wt), \Lct_\eta\null\ti{1}(\zt,x) \bigr] \delta_{xy} - \bigl[ \Rct^\eta\ti{21}(\wt,\zt), \Lc_\eta\null\ti{2}(\wt,x) \bigr] \delta_{xy} - \left( \Rct^\eta\ti{12}(\zt,\wt) + \Rct^\eta\ti{21}(\wt,\zt) \right) \delta'_{xy} \notag \\
& \hspace{20pt} + \frac{\vpt_\eta(\wt)^{-1}}{2\cc} \left( \frac{R-\cc \,\Id }{\wt-\zt_+} -  \frac{R+\cc \,\Id}{\wt-\zt_-} \right)\ti{2}\bigl\lbrace \Lct_\eta\null\ti{1}(\zt,x), \Cct\ti{2}(y) \bigr\rbrace_{\Act}  \notag \\
& \hspace{20pt} + \frac{\vpt_\eta(\zt)^{-1}}{2\cc} \left( \frac{R-\cc \,\Id }{\zt-\zt_+} -  \frac{R+\cc \,\Id}{\zt-\zt_-} \right)\ti{1} \bigl\lbrace \Cct\ti{1}(x), \Lct_\eta\null\ti{2}(\wt,y) \bigr\rbrace_{\Act}
\end{align}}

In the undeformed case, we computed the Poisson bracket of the constraint $\Cct$ with the Lax matrix $\Lct$ applying the general Equation \eqref{Eq:PbCL} to the model $\mathbb{M}^{\vpt,\pit}_{\eb}$ (yielding the bracket \eqref{Eq:PBLC}). This equation also applies to the model $\mathbb{M}^{\vpt_\eta,\pit_\eta}_{\eb}$ and we then get
\begin{equation*}
\bigl\lbrace \Lct_\eta\null\ti{1}(\zt,x), \Cct\ti{2}(y) \bigr\rbrace_{\Act} = \bigl[ C\ti{12}, \Lct_\eta\null\ti{1}(\zt,x) \bigr] \delta_{xy} - C\ti{12} \, \delta'_{xy}.
\end{equation*}
Reinserting this in the Poisson bracket above, we obtain a Maillet bracket
\begin{align*}
\bigl\lbrace \Lct_\eta\null\ti{1}(\zt,x), \Lct_\eta\null\ti{2}(\wt,y) \bigr\rbrace^*
&\gf \bigl[ \Rct^{\eta,GF}\ti{12}(\zt,\wt), \Lct_\eta\null\ti{1}(\zt,x) \bigr] \delta_{xy} - \bigl[ \Rct^{\eta,GF}\ti{21}(\wt,\zt), \Lc_\eta\null\ti{2}(\wt,x) \bigr] \delta_{xy} \\
& \hspace{40pt} - \left( \Rct^{\eta,GF}\ti{12}(\zt,\wt) + \Rct^{\eta,GF}\ti{21}(\wt,\zt) \right) \delta'_{xy},
\end{align*}
with $\Rc$-matrix
\begin{equation*}
\Rct^{\eta,GF}\ti{12}(\zt,\wt) = \Rct^{\eta}\ti{12}(\zt,\wt) + \frac{\vpt_\eta(\wt)^{-1}}{2\cc} \left( \frac{R-\cc \,\Id }{\wt-\zt_+} -  \frac{R+\cc \,\Id}{\wt-\zt_-} \right)\ti{2} C\ti{12}.
\end{equation*}
Let us introduce the matrix
\begin{equation*}
R\ti{12} = R\ti{1}\,C\ti{12} = -R\ti{2}\,C\ti{12} = -R\ti{21},
\end{equation*}
where the equality in the middle is a reformulation of the skew-symmetry of $R$ with respect to the Killing form. Using the expression \eqref{Eq:RMatDeformed} of $\Rct^{\eta}\ti{12}(\zt,\wt)$, we can finally write the $\Rc$-matrix as
\begin{equation}\label{Eq:RMatGFDeformed}
\Rct^{\eta,GF}\ti{12}(\zt,\wt) = \Rct^{\eta,0}\ti{12}(\zt,\wt) \vpt_\eta(\wt)^{-1},
\end{equation}
where
\begin{equation*}
\Rct^{\eta,0}\ti{12}(\zt,\wt) = \frac{C\ti{12}}{\wt-\zt} - \frac{1}{2} \left( \frac{1}{\wt-\zt_+} + \frac{1}{\wt-\zt_-} \right) C\ti{12} - \frac{1}{2\cc} \left( \frac{1}{\wt-\zt_+} - \frac{1}{\wt-\zt_-} \right) R\ti{12}.
\end{equation*}

\section[Coefficients in the action of the coupled $\s$-model as residues]{Coefficients in the action of the coupled $\bm\s$-model as residues}
\label{App:Identities}

\paragraph{Non-gauged model.} Let us consider the non-gauged integrable coupled $\s$-model, with action \eqref{Eq:ActionNonGauged}. By definition, the coefficients $\kc r$ of the Wess-Zumino terms in this action can be expressed as the following residues (see Equation \eqref{Eq:DSETwistSigma}):
\begin{equation*}
\kc r = - \frac{1}{2} \res_{z=z_r} \vp(z)\dd z.
\end{equation*}
Let us now introduce
\begin{equation*}
\Theta(z,w) = \frac{\ell^\infty}{2} \frac{\vp_+(z)\vp_-(w)}{z-w}.
\end{equation*}
For $r,s\in\lbrace 1,\cdots,N \rbrace$, the coefficient $\rho_{rs}$ appearing in the action \eqref{Eq:ActionNonGauged} and defined in Equation \eqref{Eq:Rho} is related to the following double residues:
\begin{equation*}
\res_{z=z_r} \left( \res_{w=z_s} \Theta(z,w)\, \dd z \, \dd w \right) = \rho_{rs} + \frac{\kc r}{2}\delta_{rs} \;\;\;\;\;\;\; \text{ and } \;\;\;\;\;\;\; \res_{w=z_s} \left(  \res_{z=z_r} \Theta(z,w)\, \dd z \, \dd w \right) = \rho_{rs} - \frac{\kc r}{2}\delta_{rs}.
\end{equation*}
Let us precise what we mean by these double residues, focusing for example on the left-hand side of the first equation: this is obtained by taking first the residue at $w=z_s$, seeing $z$ as an auxiliary parameter, and then by taking the residue at $z=z_r$ of the resulting 1-form (see also below). Note that due to the presence of the term $1/(z-w)$ in the definition of $\Theta$, the order in which we take the residues matters (if we take the residues at the same point for $z$ and $w$). Let us also introduce the following 1-forms, defined by taking a residue only in one of the parameters of $\Theta$:
{\small\begin{equation*}
\Theta_{+,r}(z) = \res_{w=z_r} \Theta(z,w)\,\dd w = \frac{\ell^\infty}{2} \frac{\vpp r(z) \vpm r(z_r)}{(z-z_r)^2}, \;\;\;\;\;\;\; \Theta_{-,r}(w) = \res_{z=z_r} \Theta(z,w)\,\dd z = -\frac{\ell^\infty}{2} \frac{\vpm r(w) \vpp r(z_r)}{(w-z_r)^2}.
\end{equation*}}
Then one has
\begin{equation*}
\res_{z=z_s} \Theta_{+,r}(z)\,\dd z = \rho_{sr} + \frac{\kc r}{2}\delta_{rs}  \;\;\;\;\;\;\; \text{ and } \;\;\;\;\;\;\; \res_{w=z_s} \Theta_{-,r}(w)\,\dd w = \rho_{rs} - \frac{\kc r}{2}\delta_{rs}.
\end{equation*}

\paragraph{Gauged model.} The results of the previous paragraph generalise straightforwardly to the case of the gauged model with action \eqref{Eq:ActionGauged}. One has
\begin{equation*}
\kc r = -\frac{1}{2} \res_{\zt=\zt_r} \vpt(\zt)\dd\zt
\end{equation*}
and
\begin{equation}\label{Eq:RhotThetat}
\res_{\zt=\zt_r} \, \left( \res_{\wt=\zt_s} \widetilde\Theta(\zt,\wt)\,\dd\zt\,\dd\wt \right) = \rhot_{rs} + \frac{\kc r}{2}\delta_{rs} \;\;\;\;\;\;\; \text{ and } \;\;\;\;\;\;\;  \res_{\wt=\zt_s} \left( \res_{\zt=\zt_t} \widetilde\Theta(\zt,\wt)\,\dd \zt\,\dd\wt \right) = \rhot_{rs} - \frac{\kc r}{2}\delta_{rs},
\end{equation}
for $r$ and $s$ in $\lbrace 1,\cdots,N+1 \rbrace$, where $\widetilde{\Theta}$ is defined as
\begin{equation*}
\widetilde{\Theta}(\zt,\wt) = \frac{\widetilde{\ell}^\infty}{2} \frac{\vpt_+(\zt)\vpt_-(\wt)}{\zt-\wt}.
\end{equation*}
Moreover, we have
\begin{equation}\label{Eq:RhotThetat2}
\res_{\zt=\zt_s} \widetilde\Theta_{+,r}(\zt)\dd \zt = \rhot_{sr} + \frac{\kc r}{2}\delta_{rs}  \;\;\;\;\;\;\; \text{ and } \;\;\;\;\;\;\; \res_{\wt=\zt_s} \widetilde\Theta_{-,r}(\wt) \dd\wt = \rhot_{rs} - \frac{\kc r}{2}\delta_{rs},
\end{equation}
with
{\small \begin{equation}\label{Eq:ThetaTilde}
\widetilde\Theta_{+,r}(\zt) = \res_{\wt=\zt_r} \widetilde\Theta(\zt,\wt)\, \dd\wt = \frac{\widetilde\ell^\infty}{2} \frac{\vptp r(\zt) \vptm r(\zt_r)}{(\zt-\zt_r)^2} \;\;\;\;\;\;\; \widetilde\Theta_{-,r}(\wt) = \res_{\zt=\zt_r} \widetilde\Theta(\zt,\wt)\,\dd \zt = -\frac{\widetilde\ell^\infty}{2} \frac{\vptm r(\wt) \vptp r(\zt_r)}{(\wt-\zt_r)^2}.
\end{equation}}

\paragraph{Change of spectral parameter.} Recall that under the change of spectral parameter $z \mapsto \zt$, the twist function transforms as 1-form, \textit{i.e.} $\vp(z)\,\dd z=\vpt(\zt)\,\dd\zt$. In particular, this justifies that the levels $\kc r$ appearing in front of the Wess-Zumino terms are the same in the non-gauged action \eqref{Eq:ActionNonGauged} and the gauged one \eqref{Eq:ActionGauged}, as the residues of a 1-form are invariant under a change of coordinates.

Similarly, one can show that under the change of spectral parameter $z\mapsto\zt$, the quantity $\Theta$ transforms as a bi-differential, \textit{i.e.} as the product of two 1-forms:
\begin{equation*}
\Theta(z,w)\,\dd z\,\dd w = \widetilde{\Theta}(\zt,\wt) \,\dd\zt\,\dd\wt.
\end{equation*}
In particular, using the expression in the previous paragraph of the coefficients $\rho_{rs}$ and $\rhot_{rs}$ as double residues of this quantity, we obtain
\begin{equation*}
\rho_{rs}=\rhot_{rs}, \;\;\;\;\;\; \forall \, r,s\in\lbrace 1,\cdots,N\rbrace.
\end{equation*}
This is Equation \eqref{Eq:RhoEqualRhot}, which allowed us to identify the non-gauged action \eqref{Eq:ActionNonGauged} as a gauge-fixing of the gauged one \eqref{Eq:ActionGauged}. The reasoning above provides a more geometric proof of this statement, using forms on the ``spectral'' Riemann sphere $\CP$.

Note that $\Theta_{r,\pm}$, defined as a residue of $\Theta$ in only one of its coordinates, also transforms as a 1-form, in the sense that $\Theta_{r,\pm}(z)\,\dd z$ and $\widetilde\Theta_{r,\pm}(\zt)\,\dd \zt$ are the same 1-form, expressed in the different coordinates $z$ and $\zt$.\\

In the computations above, we considered the change of spectral parameter $z\mapsto\zt$ from the non-gauged model $\mathbb{M}^{\vp,\pi}_{\eb}$ to the gauged one $\mathbb{M}^{\vpt,\pit}_{\eb}$. Recall from Subsection \ref{SubSec:ChangeSpec} that one can also consider a change of spectral parameter which relates a gauged model to another gauged model, corresponding to a M\"obius transformation which does not send any site of the initial model to infinity (see Equation \eqref{Eq:NewSitesFinite}). Let us then consider a change of spectral parameter $\zt\mapsto\widehat{z}=g(\zt)$ such that the positions $\zt_r$ of the model $\mathbb{M}^{\vpt,\pit}_{\eb}$ are sent to finite points $\widehat{z}_r=g(\zt_r)$. This transformation maps the gauged $\s$-model $\mathbb{M}^{\vpt,\pit}_{\eb}$ to another gauged $\s$-model $\mathbb{M}^{\widehat{\vp},\widehat{\pi}}_{\eb}$ whose twist function satisfies
\begin{equation*}
\vp(z)\dd z = \vpt(\zt) \dd\zt = \widehat\vp(\widehat z)\dd\widehat z.
\end{equation*}
This gauged $\s$-model $\mathbb{M}^{\widehat{\vp},\widehat{\pi}}_{\eb}$ should be equivalent to $\mathbb{M}^{\vpt,\pit}_{\eb}$ and in particular should share the same action \eqref{Eq:ActionGauged}. Thus, the coefficients $\widehat\rho_{rs}$ naturally appearing in the action of $\mathbb{M}^{\widehat{\vp},\widehat{\pi}}_{\eb}$ should coincide with the coefficients $\rhot_{rs}$. This can be verified by seeing the coefficients $\rhot_{rs}$ as residues of $\widetilde{\Theta}$ as in Equation \eqref{Eq:RhotThetat} and by introducing the corresponding quantity $\widehat\Theta$ in the model $\mathbb{M}^{\widehat{\vp},\widehat{\pi}}_{\eb}$. Indeed, one then shows that these quantities represent the same bi-differential under the change of spectral parameter $(\zt,\wt)\mapsto(\widehat z,\widehat w)$:
\begin{equation*}
\widetilde{\Theta}(\zt,\wt)\,\dd\zt\,\dd\wt = \widehat{\Theta}(\widehat z,\widehat w)\,\dd\widehat z\,\dd\widehat w.
\end{equation*}
The equality of the coefficients $\rhot_{rs}$ and $\widehat\rho_{rs}$ then follows directly from the invariance of the double residues of a bi-differential under a change of coordinate.

\paragraph{An identity.} Finally, let us use the interpretation of the coefficients $\rhot_{rs}$ and $\kc r$ as residues to derive the identity \eqref{Eq:IdentityRhot} that we used in Subsection \ref{SubSec:GaugedSigma}. From its definition \eqref{Eq:ThetaTilde}, one sees that in the complex plane, the 1-form $\widetilde{\Theta}_{\pm,r}(\zt)\,\dd\zt$ has poles only at the positions $\zt=\zt_s$, $s\in\lbrace 1,\cdots,N+1 \rbrace$. Moreover, one checks that it is regular at $\zt=\infty$ by performing the change of variables $\zt=1/u$ and using the limit
\begin{equation*}
\vptpm r\left(\frac{1}{u}\right) \xrightarrow{u\to0} 1.
\end{equation*}
As the residues of a 1-form on the Riemann sphere sum to zero, we then get
\begin{equation*}
\sum_{s=1}^{N+1} \res_{\zt=\zt_s} \widetilde{\Theta}_{+,r}(\zt)\,\dd\zt = \sum_{s=1}^{N+1} \res_{\wt=\zt_s} \widetilde{\Theta}_{-,r}(\wt)\,\dd\wt = 0.
\end{equation*}
Using Equation \eqref{Eq:RhotThetat2}, one then get the identity \eqref{Eq:IdentityRhot}.

\bibliographystyle{JHEP}

\providecommand{\href}[2]{#2}\begingroup\raggedright\endgroup

\end{document}